\documentclass[11pt,letterpaper]{article}
\usepackage[margin=1in]{geometry}
\usepackage{amsmath,mathtools,mathrsfs}
\usepackage{amssymb}
\usepackage{amsthm}
\usepackage{microtype}
\usepackage{array,booktabs,tabularx,caption}
\usepackage{enumitem}
\usepackage{xcolor}

\makeatletter
\providecommand{\papersection}{}
\renewcommand{\papersection}[2]{%
  \@ifundefined{r@#1}{\S\,#2}{\hyperref[#1]{\S\,\ref*{#1}}}}
\makeatother
\usepackage{braket}
\usepackage{hyperref}
\hypersetup{colorlinks=true,linkcolor=blue!55!black,citecolor=blue!55!black,urlcolor=blue!55!black}

\newtheorem{theorem}{Theorem}[section]
\newtheorem{proposition}[theorem]{Proposition}
\newtheorem{lemma}[theorem]{Lemma}
\newtheorem{corollary}[theorem]{Corollary}

\theoremstyle{definition}
\newtheorem{definition}[theorem]{Definition}
\newtheorem{remark}[theorem]{Remark}

\newcommand{\mcP}{\mathcal P}

\renewcommand{\P}{\mathbb P}
\newcommand{\cF}{\mathcal F}

\newcommand{\E}{\mathbb E}

\newcommand{\Tr}{\operatorname{Tr}}
\newcommand{\supp}{\operatorname{supp}}

\newcommand{\Stab}{\mathsf{stab}}
\newcommand{\PStab}{\mathsf{pstab}}
\newcommand{\Prod}{\mathsf{prod}}
\newcommand{\BlockProd}{\mathsf{blockprod}}
\newcommand{\Cliff}{\mathsf{cliff}}
\newcommand{\CliffT}{\mathsf{cliffT}}
\newcommand{\RevMag}{\mathsf{revmag}}

\newcommand{\op}{\mathrm{op}}

\newcommand{\eps}{\varepsilon}

\newcommand{\hatM}{\widehat{M}}
\newcommand{\barQ}{\overline{Q}}

\newcommand{\oag}[1]{\textcolor{purple}{\textbf{OMAR NOTE:} #1}}
\newcommand{\bk}[1]{\textcolor{red}{\textbf{Bobak NOTE:} #1}}
\newcommand{\dg}[1]{\textcolor{blue}{\textbf{DAVID NOTE:} #1}}

\renewcommand{\oag}[1]{}
\renewcommand{\bk}[1]{}
\renewcommand{\dg}[1]{}
\hypersetup{pdftitle={Beyond Light Cones: State Preparation Complexity in Quantum Spin Glasses},pdfauthor={Omar Al-Ghattas, David Gamarnik, Bobak T. Kiani}}
\title{Beyond light cones: State preparation complexity\\in quantum spin glasses}
\author{Omar Al-Ghattas\thanks{Department of Mathematics, Massachusetts Institute of Technology, Cambridge, MA, USA. \texttt{oag@mit.edu}.}
\and David Gamarnik\thanks{Sloan School of Management, Operations Research Center, and Institute for Data, Systems, and Society, Massachusetts Institute of Technology, Cambridge, MA, USA. \texttt{gamarnik@mit.edu}.}
\and Bobak T. Kiani\thanks{Department of Computer Science, Bowdoin College, Brunswick, ME, USA. \texttt{b.kiani@bowdoin.edu}.}}
\date{}
\begin{document}

\hypersetup{pageanchor=false}
\begin{titlepage}
\maketitle
\begin{abstract}
\normalsize

We introduce a general technique for studying state preparation
complexity in the quantum $p$-spin model on $n$ qubits, a disordered
quantum spin system with random interactions on every set of $p$
qubits. We reduce the task of proving state preparation lower
bounds to bounding the \textit{local profile complexity} of the
class, a quantity derived from the metric entropy of the set of
vectors recording each state's expectation values against all
weight-$p$ Pauli strings. We show that, for sufficiently large fixed $p$, any class
with local profile complexity $O(n^2)$ remains separated
from the ground-state energy by $\Omega(\sqrt n)$, a gap
on the leading energy scale. When the profile complexity
is $o(n^2)$, we prove the stronger conclusion that the
class cannot outperform a benchmark class of simpler
states at leading order. Product states always provide such a benchmark, while additional structure can permit a smaller
benchmark and a sharper comparison.

Our entropy-based framework applies to a broad range of state
classes and, beyond this generality, strengthens known bounds
and yields bounds in settings not covered by previous methods.
These include examples where circuit lightcones extend across
the entire system, preventing a direct application of standard
lightcone arguments. We obtain an $\Omega(n^2/\log n)$ gate lower bound for preparing
near-ground states and depth--width tradeoffs, allowing arbitrary
one- and two-qubit gates and arbitrarily many discardable ancillas.
Further applications include lower bounds on entanglement depth
and matrix product state bond dimension, and energy obstructions
for forward and reverse magic at every fixed level of Parham's
magic hierarchy \cite{parham2026quantum}.
At the first level, our bounds allow arbitrarily many ancillas:
discardable for forward magic and clean for reverse magic.
For Clifford circuits doped with $T$-gates, we obtain sharper
energy comparisons using product stabilizer states as a more
restrictive benchmark.
The proof combines an adaptation of a nonsymmetric quantum
de Finetti theorem \cite{berta2022semidefinite} with metric-entropy bounds for the Gaussian process indexed by
the local profiles of states in the class.
\end{abstract}
\thispagestyle{empty}
\end{titlepage}
\hypersetup{pageanchor=true}

\pagenumbering{roman}
\tableofcontents
\clearpage
\pagenumbering{arabic}
\pagestyle{plain}

\section{Introduction}
State preparation is perhaps the most fundamental primitive
of quantum computation. Quantum algorithms rely on suitable
input states that can be prepared reliably across repeated
runs \cite{fomichev2024initial}. The resources needed to do so
are quantified by their \emph{state preparation complexity}.
Lower bounds on this complexity establish unavoidable costs
for any algorithm that requires these states, regardless
of how efficiently its subsequent steps can be implemented.

In this work, we study how closely the energies of states prepared by
a class $\mathcal C$ of circuits with bounded resources
can approach the ground-state energy of a quantum Hamiltonian $H$. Our focus is the dense quantum $p$-spin Hamiltonian
\cite{bray1980replica,erdHos2014phase,baldwin2020quenched,anschuetz2025bounds,zlokapa2025average,al2026circuit},
an important model of disordered quantum many-body systems
that extends the classical $p$-spin model
\cite{TalagrandBook,MezardMontanariBook}
to a noncommutative setting.

More broadly, understanding the complexity of quantum state
preparation is an important problem across computational
complexity, quantum information, many-body physics, quantum
field theory, and quantum gravity
\cite{freedman2014quantum,anshu2022nlts,Brown2018second,baiguera2026quantum,chapman2022quantum,chen2025incompressibility}.
Although these fields employ different notions of complexity,
they share a common question: \emph{which physically relevant
states are difficult to prepare?}

Answering this question is challenging. Although Haar-random quantum states typically
require exponentially many local gates to prepare \cite{knill1995approximation}, proving that near-ground states are
difficult to prepare requires understanding how the Hamiltonian
interacts with the structure of states produced by restricted circuits.

Many existing proofs of quantum circuit lower bounds are based on bounded light cones or short-depth perturbations of stabilizer states.  For code Hamiltonians,
for example, circuit lower bounds are proven by showing that the lightcones of low-depth preparation circuits are too small to reproduce the local indistinguishability and error-correcting structure of low-energy states
\cite{bravyi2006lieb,eldar2017local,anshu2020circuit,anshu2022nlts}.  
Recent work of the first two authors on the quantum
$p$-spin model also uses circuit light cones, adapting a
residual argument to the causal geometry of shallow circuits
to obtain lower bounds with an arbitrary number of ancillary
qubits \cite{al2026circuit}. Circuit lightcones are essential to these proofs and typically lead to logarithmic lower bounds on circuit depth. They do not directly apply
to state classes in which local observables spread globally. Lower bounds for the magic hierarchy use a different kind of
class-specific structure, namely the algebraic rigidity of stabilizer codes
\cite{parham2026quantum,wei2026longrangenonstabilizernessquantumcodes}.
At the first level of forward magic, \cite{parham2026quantum} exploits the fact that a shallow
circuit acts on an underlying stabilizer state in a controlled way, allowing
stabilizer-code structure to be used to derive circuit lower bounds. 

A notable exception arises for strongly interacting fermionic models such as the Sachdev-Ye-Kitaev (SYK) model.  There, strong anticommutation of the Hamiltonian terms makes the
disorder variance of the energy uniformly small for every fixed state, as
captured, for example, by the commutation index
\cite{hastings2022perturbation,anschuetz2025strongly}.  This concentration, which also holds for the sparse nonlocal
Pauli model studied in \cite{chen2024sparse}, allows a simple
union bound to control large classes of candidate states,
although the resulting circuit bounds do not cover arbitrary
ancillary workspace. For more general Hamiltonians, the corresponding variance
need not vanish uniformly over states, so the same union-bound
argument is too weak to control comparably large state classes.

We develop a general framework for proving energy lower bounds
across qualitatively different restricted state classes in
the bounded-locality setting of the dense quantum $p$-spin model.
The energy depends on a state only through its \emph{weight-$p$ Pauli profile},
the vector of its expectation values against weight-$p$
Pauli operators. We use the metric entropy of these profiles
to control the energy attainable by an entire state class.

The proof combines a conditioning argument adapted from
the nonsymmetric quantum de Finetti theorem of
\cite{berta2022semidefinite} with bounds on Gaussian-process
suprema. This yields a quadratic threshold for profile
complexity. Classes of at most quadratic complexity exhibit
an $\Omega(\sqrt n)$ ground-state gap for sufficiently large
fixed $p$: the optimal energy attainable within the class
remains separated from the ground-state energy on its
natural $\sqrt n$ scale. For classes of subquadratic
complexity, the energy advantage over a suitable benchmark
is at most $o(\sqrt n)$. Product states always provide such
a benchmark, and additional structure can permit a sharper
comparison.

To apply this principle, we bound the metric entropy of
the local Pauli profiles associated with each state class
and identify the parameter regimes in which the resulting
complexity is quadratic or subquadratic. We carry this out
for circuits with local gates, states with restricted
entanglement, the magic hierarchy of
\cite{parham2026quantum}, and Clifford+$T$ circuits.
These applications are summarized in
Table~\ref{tab:energy-summary} and described in Section~\ref{ssec:applications}.
The two result columns give sufficient conditions for
a ground-state gap and a vanishing leading-order advantage
over the benchmark, respectively.

\begin{table}[!ht]
\centering
\begingroup
\small
\setlength{\tabcolsep}{6pt}
\renewcommand{\arraystretch}{1.05}
\newcommand{\regime}[1]{{\footnotesize$#1$}}
\newcommand{\classentry}[3]{%
  \textbf{#1}~{\footnotesize (#2)}%
  \if\relax\detokenize{#3}\relax
  \else\newline{\footnotesize #3}\fi}
\begin{tabularx}{\linewidth}{@{}
  >{\raggedright\arraybackslash}p{.35\linewidth}
  >{\centering\arraybackslash}p{.10\linewidth}
  >{\hsize=.92\hsize\linewidth=\hsize\raggedright\arraybackslash}X
  >{\hsize=1.08\hsize\linewidth=\hsize\raggedright\arraybackslash}X@{}}
\toprule
\textbf{Class}
  & \textbf{Width $M$}
  & \textbf{Ground-state gap}
  & \textbf{Benchmark comparison} \\
\midrule
\classentry{Circuit size $G$}
  {\papersection{sec:gate-count}{3.1}}
  {}
  & {\footnotesize Arbitrary}
  & \regime{G=O(n^2/\log n)}
  & \regime{G=o(n^2/\log n)} \\
\addlinespace[4pt]
\classentry{Circuit depth $D$}
  {\papersection{subsec:plain-shallow}{3.2}}
  {$\widehat M=\min\{M,n2^D\}$}
  & {\footnotesize Arbitrary}
  & \regime{D\widehat M=O(n^2/\log n)}
  & \regime{D\widehat M=o(n^2/\log n)} \\
\addlinespace[4pt]
\classentry{Block-entangled}
  {\papersection{subsec:block-product}{3.3}}
  {$b$: maximum block size}
  & {\footnotesize Arbitrary}
  & \regime{b\le\log_2 n}
  & \regime{b\le c\log_2 n} \\
\addlinespace[4pt]
\classentry{Matrix Product States}
  {\papersection{subsec:tensor-networks}{3.4}}
  {$\chi$: bond dimension}
  & ---
  & \regime{\chi=O(\sqrt{n/\log n})}
  & \regime{\chi=o(\sqrt{n/\log n})} \\
\addlinespace[4pt]
\classentry{Magic, fixed $k$ rounds}
  {\papersection{sec:higher-magic-levels}{4.2}}
  {Both orientations}
  & $O(n)$
  & \regime{D=O(n/\log n)}
  & --- \\
\addlinespace[4pt]
\classentry{Forward magic $QC$}
  {\papersection{subsec:forward-profile-complexity}{4.3}}
  {}
  & $O(n)$
  & \regime{D=O(n/\log n)}
  & \regime{D\le c\log_2 n} \\
\addlinespace[2pt]

  & {\footnotesize Arbitrary}
  & \regime{D\le c\log_2 n}
  & \regime{D\le c\log_2 n} \\
\addlinespace[4pt]
\classentry{Reverse magic $CQ$}
  {\papersection{sec:revmag-clean}{4.4}}
  {Clean ancillas}
  & {\footnotesize Arbitrary}
  & \regime{D=O(1)}
  & --- \\
\addlinespace[4pt]
\classentry{Clifford+$T$}
  {\papersection{subsec:sublinear-t-count}{5.4}}
  {$t$: total number of $T$-gates}
  & {\footnotesize Arbitrary}
  & \regime{t=O(n)}
  & \regime{t=o(n)\,{}^{\dagger}} \\
\bottomrule
\end{tabularx}

\endgroup
\caption{Summary of applications of
Theorem~\ref{thm:local-profile-complexity-principle}. Section references in the first column indicate where
the corresponding results are proved. Entries give sufficient conditions for $\Omega(\sqrt n)$
ground-state separation or $o(\sqrt n)$ benchmark advantage,
as in parts~\textup{(i)} and~\textup{(ii)}, respectively.
The benchmark is product states, except for $\dagger$, which
uses product stabilizer states.
Here $n$ counts retained outputs, $M$ is total width, and $c>0$
is a sufficiently small fixed constant. In magic circuits,
$D$ bounds each shallow block.
Ancillas are discardable unless marked clean and are confined
to individual blocks for block-entangled states.
A dash means not applicable in the width column and no assertion
in the comparison column.}
\label{tab:energy-summary}
\end{table}

\subsection{The quantum $p$-spin model}
We introduce the dense quantum $p$-spin Hamiltonian and fix
the notation for quantum states and operators used throughout
the paper. Let
\[
    X=
    \begin{pmatrix}
        0&1\\
        1&0
    \end{pmatrix},
    \qquad
    Y=
    \begin{pmatrix}
        0&-i\\
        i&0
    \end{pmatrix},
    \qquad
    Z=
    \begin{pmatrix}
        1&0\\
        0&-1
    \end{pmatrix}
\]
denote the non-identity single-qubit Pauli matrices.  Let $[n] = \{1,\dots, n\}$ and assume $n \ge p \ge 2$ throughout.  For a $p$-element subset $I=\{i_1,\ldots,i_p\}\subseteq[n]$, with $i_1<\cdots<i_p$, and
$a=(a_1,\ldots,a_p)\in\{X,Y,Z\}^p$, let $P_I^a$ denote the
$n$-qubit operator that acts as $a_j$ on qubit $i_j$ and as the
identity on every qubit outside $I$.  In particular, $X_i,Y_i,Z_i$
denote the corresponding single-site Pauli operators.  We write
\[
    \mathcal P_{n,p}
    :=
    \left\{
        P_I^a:
        I\in\binom{[n]}p,\;
        a\in\{X,Y,Z\}^p
    \right\}
\]
for the set of weight-$p$ Pauli strings.  The support of an operator $A$, denoted $\operatorname{supp}(A)$ is the set of qubits it acts nontrivially on. In particular, $\operatorname{supp}(P_I^a)=I$.

The dense quantum $p$-spin Hamiltonian is
\[
    H_{n,p}
    :=
    \frac{1}{\sqrt{\binom np}}
    \sum_{P\in\mathcal P_{n,p}}
    g_P P,
\]
where the coefficients $g_P$ are independent standard Gaussian random
variables.  

We write $\Tr$ for the trace and $A^\dagger$ for the adjoint of a
matrix $A$.  The Euclidean vector norm is denoted by $\|\cdot\|_2$,
and the trace and operator norms are
$\|A\|_1:=\Tr\sqrt{A^\dagger A}$ and
$\|A\|_{\mathrm{op}}:=\sup_{\|v\|_2=1}\|Av\|_2$, respectively.

An $n$-qubit state is represented by a positive semidefinite operator $\rho$ on
$(\mathbb C^2)^{\otimes n}$ satisfying $\Tr(\rho)=1$.
For a state $\rho$ on a set $A$ of qubits, $\Tr_R$ denotes the
partial trace over $R\subseteq A$, and
$\rho_R:=\Tr_{A\setminus R}\rho$ denotes its reduced state on $R$.
We use the same marginal convention for other states, with a single
site $i$ standing for $\{i\}$. Throughout this work, all state classes are nonempty,
deterministic, and fixed independently of the Gaussian disorder.
The state selected from such a class may depend on the
realized Hamiltonian.

\subsection{Local-profile complexity}
For a class $\mathcal C_n$ of $n$-qubit states, the energy of
$\rho\in\mathcal C_n$ is
\[
    \Tr(H_{n,p}\rho)
    =
    \binom np^{-1/2}
    \sum_{P\in\mathcal P_{n,p}} g_P\Tr(P\rho).
\]
This depends on $\rho$ only through its \emph{weight-$p$ Pauli profile}
\[
    f_\rho(P):=\Tr(P\rho),
    \qquad P\in\mathcal P_{n,p}.
\]
The random energies can therefore be studied as a centered Gaussian process
indexed by the profile class $\mathcal F_{n,p}(\mathcal C_n):=\{f_\rho:\rho\in\mathcal C_n\}$. The \emph{canonical metric} of this process is a metric on the profile class, defined by
\[
    d_{n,p}(f_\rho,f_\sigma)^2
    :=
    \operatorname{Var}\left(
        \Tr(H_{n,p}\rho)-\Tr(H_{n,p}\sigma)
    \right)
    =
    \frac{1}{\binom np}
    \sum_{P\in\mathcal P_{n,p}}
    \bigl(f_\rho(P)-f_\sigma(P)\bigr)^2,
\]
and we write $\|f\|_{n,p}^2:=\binom np^{-1}\sum_{P\in\mathcal P_{n,p}}f(P)^2$ for the corresponding norm. Notably, two states with identical profiles have the same energy under $H_{n,p}$. 

Although our motivating question concerns the ground-state energy, we work
with the maximal energy for convenience.  Since the Gaussian disorder is
distributionally symmetric, the problems of maximizing and minimizing the
energy are equivalent in distribution.

The maximal energy attainable by states in $\mathcal C_n$ is
\[
    \mathcal{E}_{n,p}(\mathcal C_n)
    :=
    \sup_{\rho\in\mathcal C_n}
    \Tr(H_{n,p}\rho).
\]
 We write $\lambda_{\max}(H_{n,p})
    :=
    \sup_{\rho}\Tr(H_{n,p}\rho)$ for the unrestricted maximal energy, equivalently the largest eigenvalue of
$H_{n,p}$, and let
\[
    \Prod_n
    :=
    \left\{
        \bigotimes_{i=1}^n |\psi_i\rangle\langle\psi_i|
        :
        |\psi_i\rangle\in\mathbb C^2,\;
        \|\psi_i\|_2=1
    \right\}
\]
denote the class of pure product states.  For a reference class $\mathcal B_n$, define

\begin{align}
    \Delta^{\mathrm{opt}}_{n,p}(\mathcal C_n)
    &:=
    \lambda_{\max}(H_{n,p})
    -
    \mathcal{E}_{n,p}(\mathcal C_n),
    \label{eq:intro-optimality-gap}
    \\
    \Delta^{\mathrm{ref}}_{n,p}(\mathcal C_n;\mathcal B_n)
    &:=
    \left[
        \mathcal{E}_{n,p}(\mathcal C_n)
        -
        \mathcal{E}_{n,p}(\mathcal B_n)
    \right]_+,
    \label{eq:intro-benchmark-advantage}
\end{align}
where $[x]_+:=\max\{x,0\}$. Thus $\Delta^{\mathrm{opt}}_{n,p}(\mathcal C_n)$
measures how far the best energy attainable in $\mathcal C_n$ falls
below the unrestricted maximum, and $\Delta^{\mathrm{ref}}_{n,p}(\mathcal C_n;\mathcal B_n)$ measures
how much this energy exceeds that attainable in the benchmark.

Let $\operatorname{conv}(\mathcal B_n)$ denote the set of convex mixtures
of states in $\mathcal B_n$. Our theory requires that these mixtures
approximate the local Pauli profiles of $\mathcal C_n$ on a suitably
chosen subset of qubits, in the sense of the following definition.

\begin{definition}[Conditioning benchmark]
\label{def:conditioning-benchmark}
$(\mathcal B_n)_{n\ge1}$ is a conditioning benchmark for
$(\mathcal C_n)_{n\ge1}$ if the following approximation property holds
for every fixed $p\ge2$ and $\alpha\in(0,1)$.
Each state $\rho\in\mathcal C_n$ admits a set $J\subseteq[n]$,
with $|J|\ge(1-\alpha) n$, and a state
$\sigma\in\operatorname{conv}(\mathcal B_n)$ satisfying
\begin{equation}
\label{eq:conditioning-benchmark-asymptotic}
    \frac1{\binom np}
    \sum_{\substack{P\in\mathcal P_{n,p}\\
                    \operatorname{supp}(P)\subseteq J}}
    |f_\rho(P)-f_\sigma(P)|^2
    =
    O_{p,\alpha}(n^{-1})
\end{equation}
as $n\to\infty$, uniformly over $\rho\in\mathcal C_n$.
\end{definition}
This definition is motivated by the de Finetti construction in
Section~\ref{subsec:revmag-definetti}, which shows that $\Prod_n$
is a conditioning benchmark for every state class. Allowing convex mixtures does not change the benchmark energy,
since $\mathcal{E}_{n,p}(\operatorname{conv}(\mathcal B_n))
=\mathcal{E}_{n,p}(\mathcal B_n)$ by linearity.

Finally, we will relate the attainable energy of $\mathcal C_n$ to the metric entropy of its profile class. Recall that for a subset $\mathcal A$ of a metric space $(\mathcal X,d)$, $N(\varepsilon,\mathcal A,d)$ denotes the covering number, namely, the smallest number of closed $d$-balls of radius $\varepsilon$, whose union contains $\mathcal A$. The centers of the covering balls may lie anywhere in the ambient space $\mathcal X$. The metric entropy of
$\mathcal A$ at scale $\varepsilon$ is $\log N(\varepsilon,\mathcal A,d)$.

We now state our main result.

\begin{theorem}[Local-profile complexity principle]
\label{thm:local-profile-complexity-principle}
Let $a_{n,p},b_{n,p}\ge0$ and $B_{n,p}\ge1$ satisfy,
for integers $n\ge p\ge2$ and $0<\varepsilon<1$,
\[
    \log N\left(
        \varepsilon,
        \mathcal F_{n,p}(\mathcal C_n),
        \|\cdot\|_{n,p}
    \right)
    \le
    a_{n,p}
    +
    b_{n,p}\log\left(\frac{B_{n,p}}{\varepsilon}\right).
\]
Let $\Gamma_{n,p}
    :=
    a_{n,p}
    +
    b_{n,p}\log(eB_{n,p}n)$ denote the effective profile complexity.

\begin{enumerate}[label=\textup{(\roman*)}]
    \item
    \textbf{Quadratic ground-state separation.}
    For every $L \ge 0$ and every sufficiently large fixed
    $p$ depending only on $L$, the condition $\limsup_{n\to\infty}\Gamma_{n,p}/n^2\le L$ implies, for some $c_{p,L}>0$,
    \[
        \mathbb P\left(
            \Delta^{\mathrm{opt}}_{n,p}(\mathcal C_n)
            \ge c_{p,L}\sqrt n
        \right)
        \longrightarrow1.
    \]

    \item
    \textbf{Subquadratic benchmark comparison.}
    Let $\mathcal B_n$ be a conditioning benchmark for $\mathcal C_n$
    in the sense of Definition~\ref{def:conditioning-benchmark}.
    For every fixed $p\ge2$, if $\Gamma_{n,p}=o(n^2)$, then, in $L^1$ and in probability as $n\to\infty$,
    \[
        \frac{
            \Delta^{\mathrm{ref}}_{n,p}(\mathcal C_n;\mathcal B_n)
        }{
            \sqrt n
        }
        \longrightarrow0.
    \]
\end{enumerate}
\end{theorem}

Part~\textup{(i)} is our main ground-state separation result:
for sufficiently large fixed $p$, quadratic profile complexity
precludes approaching the unrestricted optimum. Applying it requires
only a metric-entropy bound, with no further class-specific energy
analysis.

Part~\textup{(ii)} shows that subquadratic profile complexity prevents
$\mathcal C_n$ from outperforming any of its conditioning benchmarks
at the leading $\sqrt n$ energy scale. Since $\Prod_n$ is always
such a benchmark, additional quantum resources give no leading-order
advantage over product states. Further structure can yield a sharper benchmark, as we demonstrate
for stabilizer states and Clifford+$T$ circuits with sublinear
$T$-count in Section~\ref{subsec:intro-stabilizer-t-count}.

In what follows, we demonstrate Theorem~\ref{thm:local-profile-complexity-principle} by applying it to a variety of state classes.

\subsection{Applications}\label{ssec:applications}

We write $\ket{0^M}:=\ket0^{\otimes M}$. In our circuit applications,
$M_n\ge n$ denotes the total number of qubits, initialized in
$\ket{0^{M_n}}$, with $n$ designated output qubits. The remaining
$M_n-n$ qubits may be traced out at the end and are called
\emph{discardable ancillas}. Unless otherwise stated, circuit gates act on one or two qubits, and we
place no a priori restriction on the total circuit width $M_n$.

\subsubsection{Circuit size and depth}
Our most general bound is in terms of the total gate count $G_n$, with no restriction on circuit depth, geometry, or ancillary
workspace.
\begin{corollary}[Circuit size]
\label{cor:intro-circuit-size}
For every $L<\infty$, the conclusion of
Theorem~\ref{thm:local-profile-complexity-principle}\textup{(i)} holds
whenever $G_n\le Ln^2/\log n$ for all sufficiently large $n$.  The
conclusion of
Theorem~\ref{thm:local-profile-complexity-principle}\textup{(ii)}, with benchmark $\mathcal B_n=\Prod_n$, holds
whenever $G_n=o(n^2/\log n)$.
\end{corollary}

Thus any leading-order advantage over product states requires nearly
quadratically many local gates, and even a fixed-constant multiple of
$n^2/\log n$ gates remains insufficient to approach the ground-state
energy for sufficiently large fixed $p$, even with arbitrary discardable
ancillas. The $n^2/\log n$ threshold for outperforming product states
is sharp up to the logarithmic factor. In Proposition~\ref{prop:beatingProds}, we construct a class of states,
each preparable using at most $\binom n2$ two-qubit gates and no
ancillas, whose maximal energy exceeds the product-state optimum
by a positive multiple of $\sqrt n$ with high probability,
for sufficiently large fixed $p$.

Specializing to depth-bounded circuits gives substantially stronger depth
bounds than previous results in the literature
\cite{anschuetz2025bounds,al2026circuit}.  Suppose the circuit has depth
$D_n$, meaning that it consists of $D_n$ layers of disjoint gates.  We denote the effective workspace by ${\widehat M}_n
    :=
    \min\{M_n,n2^{D_n}\}$, where the second term bounds the number of qubits contained in the union of
the backward lightcones of the $n$ retained output qubits. The following result quantifies a continuous tradeoff between circuit depth and ancillary workspace.

\begin{corollary}[Depth--width tradeoff]
\label{cor:intro-shallow-depth}
For every $L<\infty$, the conclusion of
Theorem~\ref{thm:local-profile-complexity-principle}\textup{(i)} holds
whenever $D_n{\widehat M}_n\le Ln^2/\log n$ for all sufficiently
large $n$.  The conclusion of
Theorem~\ref{thm:local-profile-complexity-principle}\textup{(ii)}, with benchmark $\mathcal B_n=\Prod_n$, holds
whenever $D_n{\widehat M}_n=o(n^2/\log n)$.
\end{corollary}

Our result can be compared directly with existing bounds for the
quantum $p$-spin model. The work~\cite{anschuetz2025bounds}
determines the asymptotic maximal product-state energy
$\sqrt{2n\log p}$ as $n\to\infty$ followed by $p\to\infty$.
The subsequent work~\cite{al2026circuit} establishes that, for
sufficiently large fixed $p$, near-ground-state preparation
requires depth $\Omega_p(\log n)$ with high probability,
even with arbitrarily many discardable ancillas.

Corollary~\ref{cor:intro-shallow-depth} strengthens this depth lower
bound to $\Omega(n/\log n)$ for circuits acting on $O(n)$ qubits,
for sufficiently large fixed $p$, while retaining logarithmic-depth
lower bounds when the ancillary workspace is unrestricted.
More generally, if the total circuit width is $O(n^\gamma)$ for
some fixed $1\le\gamma<2$, the threshold becomes
$D_n=O(n^{2-\gamma}/\log n)$. For every fixed $p\ge2$, the
corresponding little-$o$ condition
$D_n=o(n^{2-\gamma}/\log n)$ ensures that such circuits can
improve on the maximal product-state energy by only $o(\sqrt n)$.
At logarithmic depths, the backward-lightcone bound
$\widehat M_n\le n2^{D_n}$ controls the effective workspace
independently of the total circuit width.

We next obtain a lower bound on entanglement depth. We say that a pure
$n$-qubit state is $b_n$-producible if it factors over a
partition of the physical qubits into blocks of size at most
$b_n$, with arbitrary entanglement allowed within each block.
A mixed state is $b_n$-producible if it is a convex mixture
of such pure states, whose block partitions may differ.
The entanglement depth of a state is the smallest integer $b$
for which it is $b$-producible. We allow each block to be
prepared using arbitrarily many discardable ancillas, with
no restrictions on the circuit depth or gate count within
the block.

\begin{corollary}[Entanglement depth]
\label{cor:intro-entanglement-depth}
For every $L<\infty$, the conclusion of
Theorem~\ref{thm:local-profile-complexity-principle}\textup{(i)} holds for
$b_n$-producible states whenever $2^{b_n}/b_n\le Ln/\log n$ for all
sufficiently large $n$.  The conclusion of
Theorem~\ref{thm:local-profile-complexity-principle}\textup{(ii)}, with benchmark $\mathcal B_n=\Prod_n$, holds
whenever $2^{b_n}\log n/b_n=o(n)$.
\end{corollary}

Thus, near-ground states must have entanglement depth
at least of order $\log n$. Under the stronger little-$o$ condition, allowing
entanglement within smaller blocks does not improve upon product states at the
leading $\sqrt n$ scale.

\subsubsection{Matrix product states}
The matrix product states (MPS) are a standard tensor-network ansatz for representing many-body states
\cite{cirac2021matrix,fannes1992finitely}. Fixing an ordering of the
qubits, an open-boundary MPS assigns a matrix $A^{[j]}_{x_j}$ to each site
$j$ and basis value $x_j\in\{0,1\}$, and writes the state as
\begin{align*}
    \ket{\psi}
    =
    \sum_{x_1,\ldots,x_n\in\{0,1\}}
    A^{[1]}_{x_1}A^{[2]}_{x_2}\cdots A^{[n]}_{x_n}
    \ket{x_1\cdots x_n}.
\end{align*}
The matrices $A^{[1]}_{x_1},A^{[2]}_{x_2},\dots A^{[n]}_{x_n}$ have dimensions compatible for matrix multiplication, with the first a row vector and
the last a column vector, so their product is a scalar. The largest intermediate matrix dimension, denoted by $\chi$, is
called the \emph{bond dimension}. The smallest bond dimension needed
to represent a state equals its largest Schmidt rank across the
contiguous cuts of the chosen ordering. Bond dimension one gives
precisely the product states, while bond dimension two already
includes the GHZ state, whose correlations extend across the entire
system. The entanglement entropy across each cut is bounded by
$\log\chi$~\cite{cirac2021matrix}.

MPS thus permit global correlations while their local-profile
complexity remains controlled by the bond dimension, providing
a qualitatively different application of our framework.

\begin{corollary}[MPS bond dimension]
\label{cor:intro-mps}
Consider the $n$-qubit open-boundary MPS of bond dimension at most
$\chi_n$. For every $L<\infty$, the conclusion of
Theorem~\ref{thm:local-profile-complexity-principle}\textup{(i)}
holds whenever $\chi_n^2\log(n\chi_n)\le Ln$ for all sufficiently
large $n$. The conclusion of
Theorem~\ref{thm:local-profile-complexity-principle}\textup{(ii)},
with benchmark $\mathcal B_n=\Prod_n$, holds whenever
$\chi_n^2\log(n\chi_n)=o(n)$.
Both statements remain valid for convex mixtures and when the
ordering of the physical qubits may vary with the state.
\end{corollary}

In particular, near-ground states require bond dimension $\Omega(\sqrt{n/\log n})$, while under the stronger little-$o$ condition, MPS have no leading-order advantage over product states. To our knowledge, these are the first results showing that low-energy states of the quantum $p$-spin model require asymptotically growing MPS bond dimension. For SYK, \cite{anschuetz2025strongly} show that near-ground states require polynomially growing bond dimension using a commutation-index argument. The work \cite{navascues2018bond} gives examples of constant-bond-dimension MPS classes that fail to reach the ground energy of local Hamiltonians, while \cite{anshu2023concentration} gives concentration bounds for a restricted class of injective MPS of constant bond dimension.

\subsubsection{Magic hierarchy}

We next consider the magic hierarchy introduced in \cite{parham2026quantum}. Recall that a unitary $C$ is \emph{Clifford} if, for every Pauli string $P$,
there exists a Pauli string $P'$ such that $CPC^\dagger=\pm P'$. At its first level, \emph{forward magic} applies an
unrestricted Clifford circuit $C$ followed by a shallow circuit $Q$, producing states of the form $QC\ket{0^{M_n}}$, while \emph{reverse magic} uses
the opposite ordering, producing states of the form $CQ\ket{0^{M_n}}$.  At higher levels there are correspondingly two orientations: one obtained by
continuing the alternation starting from forward magic, and the other by
continuing the alternation starting from reverse magic. We call
the classes containing $k$ shallow blocks the $k$-round forward- and
reverse-magic classes.  The shallow blocks may have different depths, and we denote the maximum depth as $D_n$. Their precise definitions and their relation to the levels of Parham's hierarchy are given in Definition~\ref{def:higher-magic-levels} and
Remark~\ref{rem:higher-parham-notation}. 

We recall from \cite{parham2026quantum} that \emph{clean ancillas}, as opposed to discardable ancillas, enforce the requirement that the $M_n-n$ ancillary qubits begin and end in the state $\ket{0}$.

\begin{corollary}[Magic hierarchy]
\label{cor:intro-magic}
The following statements hold for every fixed $L<\infty$.
\begin{enumerate}[label=\textup{(\roman*)}]
    \item
    For every fixed number $k\ge1$ of shallow blocks with total width $M_n=O(n)$, the
    conclusion of
    Theorem~\ref{thm:local-profile-complexity-principle}\textup{(i)} holds
    for both orientations whenever $D_n\le Ln/\log n$ for all sufficiently
    large $n$.  
    
    \item
    For forward magic with total width $M_n=O(n)$, the conclusion of
    Theorem~\ref{thm:local-profile-complexity-principle}\textup{(i)}
    holds whenever $D_n\le Ln/\log n$.
    With arbitrarily many discardable ancillas, the conclusion of
    part~\textup{(ii)}, with benchmark $\mathcal B_n=\Prod_n$,
    holds whenever $D_n\le c\log_2 n$ for any fixed
    $c\in(0,1/2)$.

    \item
    For reverse magic with total width $M_n=O(n)$, the conclusion of
    Theorem~\ref{thm:local-profile-complexity-principle}\textup{(i)} holds
    whenever $D_n\le Ln/\log n$ for all sufficiently large $n$.
    With an arbitrary number of clean ancillas, the same conclusion holds
    for every fixed depth $D$.
\end{enumerate}
\end{corollary}

In forward magic, a local
observable can be pulled backward through the final shallow circuit while
remaining supported on a controlled lightcone.  In contrast, for reverse magic the final
unrestricted Clifford can map a weight-$p$ Pauli observable to a Pauli string with large support, eliminating this locality.  Our
argument does not require the pulled-back observable to remain local. The corresponding weight-$p$ profile class can still have controlled metric entropy, and this is enough to obtain the ground-state obstruction.  The same
phenomenon occurs at higher levels of the hierarchy, where repeated
alternations of shallow and Clifford blocks can produce globally supported
observables while the profile complexity remains controlled.

Parham's results concern first-level forward magic and give robust
state-preparation lower bounds for certain explicit target states, including
ground states of topologically ordered Hamiltonians and nonstabilizer quantum
codes \cite{parham2026quantum}.  Their general state-preparation framework
allows polynomially many clean ancillas.  In particular, for an
$\ell$-local code of distance $\theta$ whose codespace dimension is not a power of
two, their code-based results give a lower bound of
$\frac12\log_2(\theta/\ell)$ on the depth of the shallow circuit at an appropriate
approximation accuracy, and these bounds extend to an arbitrary number of
clean ancillas.

Our results address a different question. Rather than ruling out
approximation of a particular target state or codespace, we show that entire
circuit classes remain separated from the ground-state energy of the quantum
$p$-spin model.  The two types of results are therefore not directly
comparable in strength.  At the same time, our framework applies to portions
of the magic hierarchy that are not accessible to the forward-magic
lightcone argument. For near-ground-state preparation with linear total circuit width,
the preceding corollary addresses \cite[Open Problem~4]{parham2026quantum}, which asks for lower bounds against
first-level reverse magic or higher levels of the hierarchy.

\subsubsection{Stabilizer states and Clifford$+T$ circuits}
\label{subsec:intro-stabilizer-t-count}

Our final application demonstrates how additional structure in the class of interest can be used to choose a sharper benchmark. Unrestricted Clifford circuits prepare the pure stabilizer states $\Stab_n$ from $\ket{0^n}$. Setting the shallow
circuit to the identity in Corollary~\ref{cor:intro-magic} shows
that stabilizer states cannot outperform product states at leading
order. This benchmark, however, allows arbitrary Bloch-sphere
directions and overlooks the discrete local structure of stabilizer
states.

We sharpen the comparison using the \emph{product stabilizer states},
defined by $\PStab_n:=\Stab_n\cap\Prod_n$. We show that the long-range
entanglement generated by unrestricted Clifford circuits gives no
leading-order energy advantage over this smaller class. Moreover, the same
comparison holds for Clifford+$T$ circuits with a sublinear number of $T$-gates, even with arbitrary discardable ancillas.

Specifically, for $M_n\ge n$ and $s\le t$, the $n$-qubit outputs of these circuits take the
following form:
\begin{align*}
    \rho
    &=
    \operatorname{Tr}_{A}\left(
        U\ket{0^{M_n}}\bra{0^{M_n}}U^\dagger
    \right),
    \qquad
    U
    =C_sT_{i_s}C_{s-1}\cdots C_1T_{i_1}C_0,
\end{align*}
where each $C_j$ is an arbitrary Clifford circuit on the $M_n$ qubits,
each $T_{i_j}$ applies
$T=\operatorname{diag}(1,e^{i\pi/4})$ to a single qubit, and $A$ is a set
of $M_n-n$ discarded output qubits.  We denote the resulting class of
$n$-qubit states by $\CliffT^{\mathrm{anc}}_{n,t}$.  Thus the only
restricted resource is the total number of $T$-gates, since the Clifford
circuits, their depth and geometry, and the total workspace $M_n$ are all
unrestricted, and $T$-gates acting on subsequently discarded qubits are
still included in the count.

\begin{corollary}[Clifford+$T$ circuits]
\label{cor:intro-t-count}
Let $(t_n)_{n\ge1}$ be a sequence of nonnegative integers.
\begin{enumerate}[label=\textup{(\roman*)}]
    \item
    For every $0\le\kappa<\infty$, the conclusion of
    Theorem~\ref{thm:local-profile-complexity-principle}\textup{(i)}
    holds for $\mathcal C_n=\CliffT^{\mathrm{anc}}_{n,t_n}$ whenever
    $t_n\le\kappa n$ for all sufficiently large $n$.

    \item
    If $t_n=o(n)$, then, for every fixed $p\ge2$,
    \[
        \frac{
            \mathcal{E}_{n,p}(\CliffT^{\mathrm{anc}}_{n,t_n})
            -
            \mathcal{E}_{n,p}(\PStab_n)
        }{\sqrt n}
        \longrightarrow0
    \]
    in $L^1$ and in probability.  In this regime, for every
    $\varepsilon>0$, there exists an integer $p_\varepsilon$ such that for all fixed $p\ge p_\varepsilon$,
    \[
        \P\left(
            \left|
                \frac{\mathcal{E}_{n,p}(\CliffT^{\mathrm{anc}}_{n,t_n})}{\sqrt n}
                -
                \sqrt{2\log6}
            \right|
            \le\varepsilon
        \right)
        \longrightarrow1.
    \]
\end{enumerate}
\end{corollary}
Thus, ground-state separation persists with a linear number of
$T$-gates for sufficiently large fixed $p$, while a sublinear
number gives no leading-order energy advantage over product
stabilizer states. Both statements allow arbitrary discardable
ancillas, with the number of $T$-gates measured relative to the
$n$ retained outputs.

The energy comparison in part~\textup{(ii)} follows from
Theorem~\ref{thm:local-profile-complexity-principle}\textup{(ii)}.
When $t_n=0$, we apply it directly with conditioning benchmark
$\PStab_n$. For general $t_n=o(n)$, we use a larger benchmark
consisting of product-stabilizer states with at most $t_n$ non-stabilizer
factors. Lemma~\ref{lem:pstab-exceptional-site-cost} shows that
this benchmark has only an $o(\sqrt n)$ expected energy advantage
over $\PStab_n$, yielding the stated comparison.

The constant $\sqrt{2\log6}$, which we derive in
Proposition~\ref{prop:pstab-value}, is the sharp large-$p$
normalized energy of product stabilizer states. Its form reflects
an equivalent description of $\PStab_n$: each state is a tensor
product whose factors are chosen from the six one-qubit
eigenstates of $X$, $Y$, and $Z$.

Clifford circuits with few non-Clifford gates have been studied
in connection with classical simulation, quantum-state learning,
and lower bounds on magic resources
\cite{bravyi2016improved,leone2024learning,beverland2020lower}.

For explicitly constructed local Hamiltonians,
\cite{coble2023local} established energy obstructions for
stabilizer states, while \cite{coble2023hamiltonians} showed
that sufficiently low-energy states can require $\Omega(n)$
$T$-gates. Their constructions include Hamiltonians obtained
by conjugating the NLTS Hamiltonians of \cite{anshu2022nlts}
by a tensor product of identical single-qubit rotations.
Our results complement these constructions by identifying a
sharp leading-order comparison with product stabilizer states
for the dense random quantum $p$-spin model, allowing a sublinear
number of $T$-gates and arbitrary discardable ancillas.

Finally, the restriction on the total number of $T$-gates controls a
different resource from the magic hierarchy of
\cite{parham2026quantum}. At each fixed level, the hierarchy
allows only a fixed number of shallow arbitrary-gate blocks,
each of which may contain an extensive number of non-Clifford
gates. Here, the total number of $T$-gates is restricted, but
the number of alternations with unrestricted Clifford circuits
may grow with $n$. Parham's approximation lower bounds concern
explicit target states and codes with clean ancillas, whereas
our results concern the energy optimum of a random Hamiltonian
and allow arbitrary discardable ancillas.

\subsection{Overview and key ideas}

We first explain how a local approximation leads to an energy
comparison, then describe the conditioning construction and how
it yields sharper benchmarks. We conclude by outlining how the profile-entropy bounds
are obtained for the different state classes.

\subsubsection{Energy decomposition}
Fix $\rho\in\mathcal C_n$, a set $S_\rho\subseteq[n]$ of
qubits to be excluded, and write $J_\rho=S_\rho^c$.
We decompose $H_{n,p}=H_{J_\rho}+H_{\mathrm{hit}(S_\rho)}$, where $H_{J_\rho}$ contains the Hamiltonian terms whose
$p$-qubit supports lie entirely in $J_\rho$, and
$H_{\mathrm{hit}(S_\rho)}$ contains those whose supports
intersect $S_\rho$.
For any separable state
$\sigma_\rho\in\operatorname{conv}(\Prod_n)$, meaning a convex
mixture of product states, we decompose
\[
\begin{aligned}
    \Tr(H_{n,p}\rho)
    &=
    \underbrace{\Tr(H_{n,p}\sigma_\rho)}_{\text{product benchmark}}
    +
    \underbrace{\Tr\!\left(H_{J_\rho}(\rho-\sigma_\rho)\right)}
               _{\text{approximation error}}
    +
    \underbrace{\Tr\!\left(
        H_{\mathrm{hit}(S_\rho)}(\rho-\sigma_\rho)
    \right)}
               _{\text{excluded interactions}}.
\end{aligned}
\]
By linearity, the energy of $\sigma_\rho$ is an average
of the energies of its constituent product states, and so
the first term is at most $\mathcal{E}_{n,p}(\Prod_n)$.
The last term is at most
$2\|H_{\mathrm{hit}(S_\rho)}\|_{\mathrm{op}}$.
Consequently,
\begin{equation}
\label{eq:proof-idea-energy-comparison}
    \Tr(H_{n,p}\rho)
    \le
    \mathcal{E}_{n,p}(\Prod_n)
    +\Tr\!\left(H_{J_\rho}(\rho-\sigma_\rho)\right)
    +2\|H_{\mathrm{hit}(S_\rho)}\|_{\mathrm{op}}.
\end{equation}
The choice of $S_\rho$ involves a tradeoff. Controlling the
second term requires choosing $\sigma_\rho$ whose weight-$p$
Pauli expectations on $J_\rho$ closely approximate those of
$\rho$. Obtaining a better approximation may require excluding
more qubits, whereas controlling the third term requires keeping
$|S_\rho|$ small, since it accounts for interactions meeting
the excluded set.

The nonsymmetric quantum de Finetti argument of
Berta et al.~\cite[Theorem~2.5]{berta2022semidefinite}
provides the starting point for this approximation.
We adapt their proof to impose an upper bound on the size of $S_\rho$.
Fix $0<\alpha<1/2$ and let $m=\lfloor\alpha n\rfloor$.
For all sufficiently large $n$ and every $\rho\in\mathcal C_n$,
we show that one can choose a set $S_\rho$ with
$|S_\rho|\le m$ and a separable state $\sigma_\rho$ such that
\begin{equation}
\label{eq:proof-idea-profile-approximation}
    \|\mathsf M_{J_\rho}(f_\rho-f_{\sigma_\rho})\|_{n,p}
    \le \frac{A_{p,\alpha}}{\sqrt n},
\end{equation}
where $A_{p,\alpha}$ depends only on $p$ and $\alpha$.
Here $\mathsf M_J$ is the coordinate projection defined by
$(\mathsf M_J f)(P):=\mathbf 1_{\{\operatorname{supp}(P)\subseteq J\}}f(P)$.

\subsubsection{Uniform control by profile entropy}
We now explain how the approximation-error and
excluded-interaction terms in
\eqref{eq:proof-idea-energy-comparison} are controlled
uniformly over $\rho\in\mathcal C_n$.

For each $\rho$, choose $J_\rho$ and $\sigma_\rho$ satisfying
\eqref{eq:proof-idea-profile-approximation}, independently
of the Hamiltonian disorder, and define the residual by
$h_\rho:=\mathsf M_{J_\rho}(f_\rho-f_{\sigma_\rho})$.
The class of residual profiles, with zero adjoined,
\[
    \mathcal R_n:=\{h_\rho:\rho\in\mathcal C_n\}\cup\{0\},
\]
lies in the $\|\cdot\|_{n,p}$-ball of radius $r_n:=A_{p,\alpha}/\sqrt n$. Moreover, the residual energies
$X_{h_\rho}:=\Tr(H_{J_\rho}(\rho-\sigma_\rho))$,
together with $X_0:=0$, form a centered Gaussian process
indexed by $\mathcal R_n$ with the same canonical metric
as the original energy process. Dudley's entropy inequality
\cite[Theorem~8.1.3]{vershynin2018high} therefore gives
\[
    \E\sup_{h\in\mathcal R_n}X_h
    \le
    C\int_0^{2r_n}
        \sqrt{
            \log N\!\left(
                \varepsilon,\mathcal R_n,\|\cdot\|_{n,p}
            \right)
        }
    \,d\varepsilon.
\]
To bound the metric entropy term, we group the
residuals by their unmeasured set:
$\mathcal R_{n,J}:=\{h_\rho:\rho\in\mathcal C_n,\ J_\rho=J\}\cup\{0\}$.
Since $\mathcal R_n=\bigcup_{|S|\le m}\mathcal R_{n,S^c}$,
covering each group separately gives
\[
    \log N\!\left(\varepsilon,\mathcal R_n,\|\cdot\|_{n,p}\right)
    \le n\log2+
    \max_{|S|\le m}
    \log N\!\left(
        \varepsilon,\mathcal R_{n,S^c},\|\cdot\|_{n,p}
    \right).
\]
Within each fixed-$J$ group, we choose the approximating
profile by metric projection onto the closed convex hull
of the benchmark profiles restricted to $J$.
This retains the approximation bound
\eqref{eq:proof-idea-profile-approximation}, and the
corresponding residual map does not increase distances:
$\|h_\rho-h_{\rho'}\|_{n,p}
\le \|f_\rho-f_{\rho'}\|_{n,p}$. A cover of the original profile class therefore yields a cover
of each residual group at the same radius, with one additional
ball centered at zero. The union over possible
excluded sets consequently adds only $O(n)$ to the metric
entropy; see Lemma~\ref{lem:profile-entropy-comparison}. It is important to note that the crucial gain relative to applying Dudley directly to the original class is the integration range $2r_n=O_{p,\alpha}(n^{-1/2})$. For fixed $p$, the original profile class typically has diameter of order one, so applying
Dudley's inequality directly can give bounds well above the $\sqrt n$ energy scale. Working with the residual class reduces this diameter while preserving control of the metric entropy.

Substituting the entropy assumption of
Theorem~\ref{thm:local-profile-complexity-principle} into
Dudley's integral then gives,
for fixed $p,\alpha$ and sufficiently large $n$,
\[
    \E\sup_{h\in\mathcal R_n}X_h
    \le
    C A_{p,\alpha}
    \sqrt{\frac{n+\Gamma_{n,p}}{n}}.
\]

The excluded-interaction term is controlled separately in Lemma~\ref{lem:revmag-hit-bound}, which utilizes a Gaussian
matrix-series estimate \cite[Theorem~4.1]{tropp2012user}, giving
\[
    \E\max_{|S|\le m}
        \|H_{\mathrm{hit}(S)}\|_{\mathrm{op}}
    \le C3^{p/2}\sqrt{p\alpha n}.
\]
Taking the maximum over $S$ covers every admissible
excluded set, including the different sets $S_\rho$
selected for different input states.

Combining these two estimates with
\eqref{eq:proof-idea-energy-comparison} gives
\begin{equation}
\label{eq:proof-idea-uniform-comparison}
    \frac{
        \E\Delta^{\mathrm{ref}}_{n,p}
        (\mathcal C_n;\Prod_n)
    }{\sqrt n}
    \le
    C A_{p,\alpha}
    \frac{\sqrt{n+\Gamma_{n,p}}}{n}
    +
    C3^{p/2}\sqrt{p\alpha}.
\end{equation}
If $\Gamma_{n,p}=o(n^2)$, the first term vanishes as
$n\to\infty$ for every fixed $\alpha$.
Letting $\alpha\downarrow0$ afterward eliminates the
second term, proving the comparison with product states
at the leading $\sqrt n$ energy scale, in $L^1$ and
in probability, as asserted in
Theorem~\ref{thm:local-profile-complexity-principle}\textup{(ii)}.

At quadratic effective profile complexity, the first
term need not vanish. For part~\textup{(i)} of
Theorem~\ref{thm:local-profile-complexity-principle},
we use the explicit dependence of $A_{p,\alpha}$ on
$p$ and $\alpha$ to balance the two terms in
\eqref{eq:proof-idea-uniform-comparison}; see
Corollary~\ref{cor:quadratic-profile-complexity}.  Comparing this upper bound with the ground-state energy lower
bound of \cite{anschuetz2025strongly}, recalled in Lemma~\ref{lem:revmag-ground-benchmark}, shows that
it remains below the unrestricted optimum for sufficiently
large fixed $p$.

\subsubsection{Constructing the approximating state}
The argument so far hinges on the construction of a separable approximation $\sigma_\rho$ in
\eqref{eq:proof-idea-profile-approximation}. We now explain how this can be done. For a given $\rho$, we first apply a fixed one-qubit
positive operator-valued measure (POVM) separately to the
qubits in $S_\rho$. The POVM is chosen to retain enough
information about the state after measurement to control
the approximation error, as quantified by
\eqref{eq:measurement-trace-comparison}.
Write $z$ for the joint measurement outcome,
$q_\rho(z)$ for its probability, and $\rho_{J_\rho\mid z}$
for the normalized conditional density matrix on $J_\rho$.
The measurement and conditional states are defined in
Section~\ref{subsubsec:conditioning-measurements}.

Averaging the conditional density matrices recovers the
original marginal,
\[
    \rho_{J_\rho}
    =\sum_z q_\rho(z)\rho_{J_\rho\mid z}.
\]
For each outcome of positive probability, let $\rho_{i\mid z}
    :=\Tr_{J_\rho\setminus\{i\}}(\rho_{J_\rho\mid z})$ for $i\in J_\rho$, be the one-qubit marginals of $\rho_{J_\rho\mid z}$.
We construct an approximation to $\rho_{J_\rho}$ by
replacing each $\rho_{J_\rho\mid z}$ in the preceding
average with the product of these marginals:
\[
    \eta_{J_\rho}
    :=
    \sum_z q_\rho(z)\bigotimes_{i\in J_\rho}\rho_{i\mid z}.
\]
In particular,
$\eta_{J_\rho}\in\operatorname{conv}(\Prod_{J_\rho})$.
Extending $\eta_{J_\rho}$ with $|0\rangle\langle0|$
on each qubit in $S_\rho$ gives an $n$-qubit separable
state $\sigma_\rho$. The de Finetti argument in
Lemma~\ref{lem:measurement-parametrized-definetti}
guarantees a choice of $S_\rho$, with $|S_\rho|\le m$,
for which this construction satisfies the canonical
approximation bound
\eqref{eq:proof-idea-profile-approximation};
see Corollary~\ref{cor:conditioning-canonical}.
This applies to arbitrary input states $\rho$.
For the entropy argument, we may then replace this
approximation by the choice described above, whose residual
map is distance-nonincreasing for fixed $J$, without
increasing the approximation error.

\subsubsection{Sharper benchmarks}
We can sharpen the benchmark by inspecting the conditional
one-qubit marginals $\rho_{i\mid z}$ and choosing the POVM
to preserve additional structure of the input class.
If each product $\bigotimes_{i\in J_\rho}\rho_{i\mid z}$
belongs to the convex hull of a smaller class
$\mathcal B_{J_\rho}$, then so does $\eta_{J_\rho}$,
with the same approximation error.

For example, for stabilizer inputs, we can choose a POVM
satisfying \eqref{eq:measurement-trace-comparison} for which
conditioning preserves stabilizer structure.
The conditional one-qubit marginals are then mixtures of
one-qubit stabilizer states, yielding $\PStab_n$ as the
benchmark. For Clifford+$T$ inputs, a similar analysis
gives a benchmark between $\PStab_n$ and $\Prod_n$,
depending on the $T$-count.

Whenever the resulting approximation satisfies
$\sigma_\rho\in\operatorname{conv}(\mathcal B_n)$, the same
energy comparison applies with $\mathcal{E}_{n,p}(\mathcal B_n)$
in place of $\mathcal{E}_{n,p}(\Prod_n)$.
The distance-nonincreasing choice of separable approximation can always be made within
$\overline{\operatorname{conv}}(\mathcal B_n)$.
Taking this closure does not change the benchmark energy;
see Lemma~\ref{lem:profile-entropy-comparison}.
Definition~\ref{def:conditioning-benchmark} formalizes
the required approximation property, and
Section~\ref{sec:sharper-conditioning-benchmarks}
develops these examples.

\subsubsection{Bounding the profile entropy}
Once the benchmark is chosen, the only remaining class-specific
input is a bound on the metric entropy of
$\mathcal F_{n,p}(\mathcal C_n)$.
The profile class can admit a much smaller cover than the
underlying state class. We now outline the covering arguments
for a few examples to illustrate how these entropy bounds
are obtained.

Consider first the pure stabilizer states. Write
$m:=|\mathcal P_{n,p}|=3^p\binom np$.
Although a general profile lies in $[-1,1]^m$, a stabilizer
state $\rho$ satisfies $f_\rho(P)\in\{-1,0,1\}$ for any $P\in\mathcal P_{n,p}$. Counting these vectors coordinate by coordinate gives at most
$3^m$ possibilities. The coordinates, however, satisfy strong
algebraic constraints. Specifically, whenever $f_\rho(P)=s\in\{\pm1\}$, the signed Pauli $sP$
stabilizes the state, meaning $(sP)|\psi\rangle=|\psi\rangle$.
Consequently, the subgroup generated by all such signed Paulis
is contained in the stabilizer group of $|\psi\rangle$ and
therefore admits at most $n$ independent generators, chosen
from these signed weight-$p$ Paulis. These generators determine
the entire weight-$p$ profile through membership of $P$ or $-P$
in the subgroup. Each generator has at most $2m$ possible
choices. Padding shorter lists with identities therefore gives at most
$(1+2m)^n$ profiles. Thus, in Theorem~\ref{thm:local-profile-complexity-principle}, we may take $a_{n,p}=n\log(1+2m)=O_p(n\log n)$, $b_{n,p}=0$ and $B_{n,p}=1$. This improves the naive entropy bound $m\log3$ and places stabilizer states in the subquadratic profile-complexity regime.

Even without exploiting additional structure of the profile
class, we can bound its metric entropy using a sufficiently
fine cover of the circuits preparing the underlying states. For example, circuits with $G_n$ arbitrary one- and two-qubit
gates on $M_n$ qubits have at most $M_n^{2G_n}$ possible sequences
of gate locations. Counting these discrete choices and
discretizing the $2\times2$ and $4\times4$ unitary gate matrices
gives an entropy bound of the form in
Theorem~\ref{thm:local-profile-complexity-principle}, with
$a_{n,p}=O(G_n\log M_n)$ and $b_{n,p}=O(G_n)$.
The scale $B_{n,p}$ accounts for the accumulation of gate
approximation errors and their effect on the canonical
profile norm; see Section~\ref{sec:gate-count}. With arbitrarily many discardable ancillas, we first prune
the circuit to the backward lightcones of the retained outputs.
The remaining circuit uses at most $n+G_n$ qubits and produces
the same output state. Applying the circuit-covering argument
to this representation gives a profile-entropy bound independent
of the original number of ancillas.

For reverse magic, a cover obtained from circuit representations
already suffices at controlled width. Although the final Clifford
block can spread local observables across the entire system,
there are only $\exp(O(M_n^2))$ Clifford unitaries on $M_n$
qubits up to global phase. Counting this block contributes
$O(M_n^2)$ to $a_{n,p}$, while the shallow block is covered
by counting gate locations and discretizing its gates as above.
At linear width and fixed shallow depth, this gives quadratic
profile complexity; see Section~\ref{sec:higher-magic-levels}.
With arbitrarily many clean ancillas, we first use stabilizer
constraints to obtain a representation on $O_D(n)$ qubits,
then apply the same counting argument; see
Section~\ref{sec:revmag-clean}.

\subsection{Organization of the paper}
Section~\ref{sec:main-proof} develops the conditioning and
energy-comparison arguments and proves the local-profile
complexity principle.
Section~\ref{sec:structured-classes} applies this framework
to circuits, block-product states, and tensor network states,
while Section~\ref{sec:proofs-magic} treats the magic hierarchy
and its extensions with ancillas.
Section~\ref{sec:sharper-conditioning-benchmarks} develops
the sharper benchmarks for stabilizer and Clifford+$T$ states
and determines the large-$p$ product-stabilizer energy.
Appendices~\ref{app:profile-entropy-growth}
and~\ref{sec:pauli-preliminaries} collect supporting entropy
bounds and Pauli and stabilizer preliminaries.

\section{The local-profile complexity principle}
\label{sec:main-proof}
We now prove Theorem~\ref{thm:local-profile-complexity-principle}.  The
argument is organized around three inputs: a truncated
de Finetti approximation obtained by conditioning, uniform bounds for the benchmark and conditioned-site contributions, and a metric-entropy bound for the residual Gaussian process.  We establish these components in turn and then combine them to prove the theorem.  The ground-state argument uses
the universal product-state benchmark, while the subquadratic argument
allows the more general benchmarks described in the introduction.

For any function $f:\mathcal P_{n,p}\to\mathbb R$, we write $X_f:=\binom np^{-1/2}\sum_{P\in\mathcal P_{n,p}}g_Pf(P)$ for the
associated Gaussian linear functional.  In particular,
$X_{f_\rho}=\Tr(H_{n,p}\rho)$.

\subsection{Conditioning and local approximation}
\label{subsec:revmag-definetti}

Central to the nonsymmetric de Finetti argument of
\cite[Section~2.6]{berta2022semidefinite} is a conditioning procedure,
building on earlier work of Brand{\~a}o and Harrow
\cite{brandao2013quantum,brandao2013product},
in which one measures a controlled set of qubits and, conditional
on the measurement outcomes, approximates the state on the remaining
sites by the product of its conditional one-qubit marginals.
In what follows, we develop this argument in a form that keeps
the conditioning measurement explicit. For any state, we construct a separable approximation by
averaging the conditional product states over the measurement outcomes.
The maximum number \(m_\star\) of measured qubits can be prescribed in
advance, and we bound the squared trace-norm error between the
\(p\)-qubit marginals of the original state and those of the separable
approximation, averaged over all \(p\)-subsets of the remaining unmeasured
qubits.  Keeping the construction explicit allows us to identify more
structured benchmarks in
Section~\ref{sec:sharper-conditioning-benchmarks}.  These measurements
are an analytical device solely used for the proofs here and are not operations assumed to be available to any of the circuit classes preparing states.

\subsubsection{Measurements and conditioning}
\label{subsubsec:conditioning-measurements}

We first provide the necessary background definitions for completeness. For a finite-dimensional Hilbert space \(\mathcal H\), let
\(\mathsf L(\mathcal H)\) denote the space of linear operators on
\(\mathcal H\).  Let \(\mathcal Z\) be a finite set of possible
measurement outcomes. A one-qubit positive operator-valued measure (POVM) with
outcome set $\mathcal Z$ is a collection
$\{M_z\in\mathsf L(\mathbb C^2):z\in\mathcal Z\}$ satisfying
\[
    M_z\succeq0,
    \qquad
    \sum_{z\in\mathcal Z}M_z=I;
\]
see \cite{nielsen2010quantum}. The operators $M_z$ are called the \emph{effects} of the POVM.

For a state \(\omega\) on \(\mathbb C^2\), the numbers
\[
    p_\omega(z):=\Tr(M_z\omega),
    \qquad z\in\mathcal Z,
\]
form the probability distribution of the measurement outcomes.

Let \(\mathcal H_{\mathrm{cl}}\cong\mathbb C^{|\mathcal Z|}\) be an
auxiliary Hilbert space with orthonormal basis
\(\{|z\rangle_{\mathrm{cl}}:z\in\mathcal Z\}\).  The associated
quantum-to-classical measurement channel is the linear map
\[
    \mathcal M:
    \mathsf L(\mathbb C^2)
    \longrightarrow
    \mathsf L(\mathcal H_{\mathrm{cl}}),
    \qquad 
    \mathcal M(X)
    :=
    \sum_{z\in\mathcal Z}\Tr(M_zX)
    |z\rangle\langle z|_{\mathrm{cl}}.
\]
This map is completely positive and trace preserving, and its output
is diagonal in the basis $\{|z\rangle_{\mathrm{cl}}\}$ for every
input operator.  For a state $X=\omega$, the diagonal entries are the
probabilities $p_\omega(z)$, so $\mathcal M(\omega)$ records the
outcome distribution.  For a general Hermitian operator, the diagonal
entries are real but need not form a probability distribution.
In particular, the same formula applies to differences of states.

For a bipartite state $\omega_{RB}$, with $B$ a single qubit,
we apply the one-qubit POVM $\{M_z:z\in\mathcal Z\}$ to $B$.
Each $M_z$ acts on $\mathcal H_B\cong\mathbb C^2$, and the
corresponding POVM element on the joint system $RB$ is
$I_R\otimes M_z$. Here $I_R$ denotes the identity operator on
$R$, while $\operatorname{Id}_R$ denotes the identity map on
its operator space.

The probability of outcome $z$ is
\[
    q(z)
    :=\Tr(M_z\omega_B)
    =\Tr[(I_R\otimes M_z)\omega_{RB}].
\]
Let $Z$ denote the measurement outcome, so that
$\mathbb P(Z=z)=q(z)$. For each $z$ with $q(z)>0$,
define the conditional density matrix
\[
    \omega_{R\mid z}
    :=
    \frac{\Tr_B[(I_R\otimes M_z)\omega_{RB}]}{q(z)}.
\]
We write $\omega_{R\mid Z}
    :=
    \sum_{z:q(z)>0}
    \mathbf 1_{\{Z=z\}}\omega_{R\mid z}$ for the corresponding random conditional density matrix, defined
almost surely. Measuring $B$ and retaining both $R$ and the outcome produces
\[
    (\operatorname{Id}_R\otimes\mathcal M)(\omega_{RB})
    =
    \sum_{z:q(z)>0}q(z)\omega_{R\mid z}
       \otimes|z\rangle\langle z|_{\mathrm{cl}}.
\]
This is a classical--quantum state, that is, it is a single density matrix encoding
the outcome distribution together with the conditional quantum
state on $R$ for each outcome.

We now apply the same measurement to a set of qubits in an $n$-qubit system.
For an $n$-qubit state $\rho$, apply the same measurement separately
to each qubit in $S\subseteq[n]$. For the joint outcome
$z=(z_i)_{i\in S}\in\mathcal Z^S$, define
\[
    M_z^{(S)}
    :=\bigotimes_{i=1}^n B_i(z),
    \qquad
    B_i(z):=
    \begin{cases}
        M_{z_i}, & i\in S,\\
        I, & i\notin S,
    \end{cases}
\]
where $M_{z_i}$ denotes the copy on qubit $i$ of the one-qubit
POVM effect $M_{z_i}$ corresponding to outcome $z_i$. We also set $K_z^{(S)}:=(M_z^{(S)})^{1/2}$, the positive
square root of the product effect $M_z^{(S)}$. The collection $\{M_z^{(S)}:z\in\mathcal Z^S \}$ is a POVM on the full \(n\)-qubit system.

The outcome probability and the
conditional state of the unmeasured qubits are
\begin{equation}
\label{eq:conditioning-state-convention}
    q_S(z):=\Tr(M_z^{(S)}\rho),
    \qquad
    \rho_{S^c\mid z}
    :=\frac{\Tr_S[K_z^{(S)}\rho K_z^{(S)}]}{q_S(z)},
    \qquad q_S(z)>0,
\end{equation}
where $\Tr_S$ denotes the partial trace over the qubits in $S$. For
$R\subseteq S^c$, write $\rho_{R\mid z} := \Tr_{S^c \setminus R}(\rho_{S^c\mid z})$ for the corresponding
marginal.  We omit zero-probability outcomes from all sums involving
conditional states.  When $S=\emptyset$, there is one empty outcome
of probability one and the conditional state is $\rho$.

Cyclicity of the partial trace over the measured qubits gives
$q_S(z)\rho_{S^c\mid z}=\Tr_S[M_z^{(S)}\rho]$.  In particular,
averaging over the outcomes leaves the unmeasured marginal unchanged
\begin{equation}
\label{eq:conditioning-average-marginal}
    \sum_z q_S(z)\rho_{R\mid z}=\rho_R,
    \qquad R\subseteq S^c.
\end{equation} 

A POVM is called \emph{informationally complete} if its outcome
probabilities uniquely determine the input state or equivalently if $\mathcal M$ is
injective on Hermitian operators.  We use the following quantitative
trace-norm comparison, 
\begin{equation}
\label{eq:measurement-trace-comparison}
    \|\xi\|_1
    \le
    \kappa_{\mathcal M}
    \|(\operatorname{Id}_R\otimes\mathcal M)(\xi)\|_1
\end{equation}
for every finite-dimensional system $R$ and every Hermitian operator
$\xi$ on $R\otimes\mathbb C^2$, with a constant
$\kappa_{\mathcal M}<\infty$ independent of $R$; see \cite[Lemma~D.1]{berta2022semidefinite}. Generally, applying a measurement can only decrease trace norm
\[
    \|(\operatorname{Id}_R\otimes\mathcal M)(\xi)\|_1
    \le
    \|\xi\|_1.
\]
The estimate \eqref{eq:measurement-trace-comparison} gives the converse bound, up to the factor
\(\kappa_{\mathcal M}\), for the particular POVM used here.  In
particular, if \(\xi=\omega-\tau\), it says that two joint states cannot
become arbitrarily more difficult to distinguish after measuring the
qubit, provided that one retains both the measurement outcome and the
unmeasured system \(R\).  Equivalently, if the resulting
classical--quantum states are close, then the original joint states were
already close.  The constant \(\kappa_{\mathcal M}\) quantifies how much
distinguishability may be lost in passing from the quantum states to
their measured versions. This is the comparison
needed to transfer the information-theoretic estimate back to the
unmeasured quantum state in the de Finetti argument
\cite[Section~2.4]{berta2022semidefinite}.

\subsubsection{The Pauli POVM}
\label{subsubsec:pauli-povm}

The nonsymmetric de Finetti argument of Berta et al.\ \cite{berta2022semidefinite}
allows any POVM satisfying the distortion estimate
\eqref{eq:measurement-trace-comparison}. We use the Pauli POVM defined
below, which is equivalent to performing classical shadow measurements in a random Pauli basis~\cite{huang2020predicting}. The reason for this choice will become clear in
Section~\ref{sec:sharper-conditioning-benchmarks}.

For each nonidentity Pauli operator \(a\in\{X,Y,Z\}\) and each sign
\(s\in\{+1,-1\}\), let \(\lvert a,s\rangle\) denote a unit vector
satisfying
\[
    a\lvert a,s\rangle=s\lvert a,s\rangle,
    \qquad
    \langle a,s\mid a,s\rangle=1.
\]
Thus \(\lvert a,s\rangle\) is the state associated with eigenvalue
\(s\) when the Pauli observable \(a\) is measured.  The projector onto
this eigenspace is
\[
    \lvert a,s\rangle\langle a,s\rvert
    =
    \frac12(I+sa).
\]
Define
\[
    E_{a,s}
    :=
    \frac13\lvert a,s\rangle\langle a,s\rvert
    =
    \frac16(I+sa).
\]
For each fixed axis \(a\), $E_{a,+1}+E_{a,-1}=(1/3)I$, and summing over the three axes gives $\sum_{a, s}E_{a,s}=I$. The six operators therefore form a POVM whose outcomes are the pairs
\((a,s)\).  Operationally, this POVM first chooses one of the three
Pauli axes uniformly, then measures the corresponding Pauli observable,
and records both the chosen axis \(a\) and the observed eigenvalue \(s\).
For an input state \(\omega\), the probability of outcome
\((a,s)\) is $\Tr(E_{a,s}\omega)
    =
    (1/3)
    \langle a,s\vert\omega\vert a,s\rangle$.
Define the corresponding measurement map
\begin{equation}
\label{eq:pauli-conditioning-measurement}
    \mathcal M_{\mathrm{Pauli}}(X)
    :=\sum_{a,s}\Tr(E_{a,s}X)
       |a,s\rangle\langle a,s|_{\mathrm{cl}},
\end{equation}
where the six classical labels form an orthonormal basis, in contrast
to the six qubit eigenstates $|a,s\rangle$.

We now show that the Pauli POVM satisfies the distortion estimate \eqref{eq:measurement-trace-comparison}.

\begin{lemma}[Trace-norm comparison for the Pauli measurement]
\label{lem:pauli-measurement-comparison}
The measurement $\mathcal M_{\mathrm{Pauli}}$ satisfies
\eqref{eq:measurement-trace-comparison} with
$\kappa_{\mathcal M}=4$.
\end{lemma}

\begin{proof}
Let \(Q\) denote the single-qubit subsystem with Hilbert space
\(\mathcal H_Q=\mathbb C^2\), and set
\(\mathcal P:=\{I,X,Y,Z\}\).  The Pauli matrices form an orthogonal basis
of \(\mathsf L(\mathcal H_Q)\), since
\[
    \Tr(ab)=2\delta_{ab},
    \qquad a,b\in\mathcal P.
\]
Hence every Hermitian operator \(\xi\) on \(R\otimes Q\) has a unique
expansion
\[
    \xi
    =
    A_0\otimes I
    +
    \sum_{a\in\{X,Y,Z\}}A_a\otimes a,
\]
where
\[
    A_0=\frac12\Tr_Q(\xi),
    \qquad
    A_a
    =
    \frac12\Tr_Q\left[(I_R\otimes a)\xi\right].
\]
Each coefficient \(A_0,A_X,A_Y,A_Z\) is Hermitian.

For \(a\in\{X,Y,Z\}\) and \(s\in\{\pm1\}\), define the corresponding
output block on \(R\) by $Y_{a,s}
    :=
    \Tr_Q\left[(I_R\otimes E_{a,s})\xi\right]$.
Using
\[
    \Tr(E_{a,s})=\frac13,
    \qquad
    \Tr(E_{a,s}b)=\frac{s}{3}\delta_{ab},
    \qquad b\in\{X,Y,Z\},
\]
we obtain $Y_{a,s}
    =
    \frac13(A_0+sA_a)$.
It follows that
\[
    A_0=\frac12\sum_{a,s}Y_{a,s},
    \qquad
    A_a=\frac32\bigl(Y_{a,+}-Y_{a,-}\bigr).
\]

The measured operator is
\[
    (\operatorname{Id}_R\otimes\mathcal M_{\mathrm{Pauli}})(\xi)
    =
    \sum_{a,s}
    Y_{a,s}\otimes
    |a,s\rangle\langle a,s|_{\mathrm{cl}}.
\]
Since the classical labels are orthonormal, this is the direct sum of
the six operators \(Y_{a,s}\).  Therefore
\[
    \bigl\|
        (\operatorname{Id}_R\otimes
        \mathcal M_{\mathrm{Pauli}})(\xi)
    \bigr\|_1
    =
    \sum_{a,s}\|Y_{a,s}\|_1.
\]

Finally, the trace norm is multiplicative under tensor products and
\(\|P\|_1=2\) for every \(P\in\mathcal P\).  Thus
\[
\begin{aligned}
    \|\xi\|_1
    &\le
    2\|A_0\|_1
    +
    2\sum_{a\in\{X,Y,Z\}}\|A_a\|_1                                      \\
    &\le
    \sum_{a,s}\|Y_{a,s}\|_1
    +
    3\sum_{a\in\{X,Y,Z\}}
        \bigl(\|Y_{a,+}\|_1+\|Y_{a,-}\|_1\bigr)                            \\
    &=4\sum_{a,s}\|Y_{a,s}\|_1                                           \\
    &=4\bigl\|
        (\operatorname{Id}_R\otimes
        \mathcal M_{\mathrm{Pauli}})(\xi)
      \bigr\|_1.
\end{aligned}
\]
This proves \eqref{eq:measurement-trace-comparison} with
\(\kappa_{\mathcal M}=4\).
\end{proof}

We note once more that this estimate holds for arbitrary input states.  We will use the same
measurement for the universal product-state approximation and later
exploit its additional preservation of stabilizer structure in Section~\ref{sec:sharper-conditioning-benchmarks}.

\subsubsection{Conditioning with a measurement budget}
\label{subsubsec:truncated-conditioning-proof}

The following result applies to any finite-outcome one-qubit POVM
satisfying the distortion bound
\eqref{eq:measurement-trace-comparison}.  In particular, by
Lemma~\ref{lem:pauli-measurement-comparison}, it applies to the Pauli
POVM with \(\kappa_{\mathcal M}=4\).  Without assuming any symmetry of
the input state, it constructs a single separable approximation with small
average squared trace-norm error over the \(p\)-qubit marginals supported
on the unmeasured qubits.

\begin{lemma}[Truncated squared conditioning with a prescribed POVM]
\label{lem:measurement-parametrized-definetti}
There is a universal constant \(C<\infty\) such that the following holds.
Fix a finite-outcome one-qubit POVM
\(\{M_u:u\in\mathcal Z\}\) whose measurement map satisfies
\eqref{eq:measurement-trace-comparison} with constant
\(\kappa_{\mathcal M}\).  Let \(\rho\) be an \(n\)-qubit state, let
\(n\ge p\ge2\), and let the integer \(m_\star\) satisfy
\(0\le m_\star\le n-p\).

Then there exists \(J\subseteq[n]\), with
\(|J^c|\le m_\star\), such that applying the prescribed POVM separately
to each qubit in \(J^c\) has the following property.  Denote the joint
outcome by \(z=(z_i)_{i\in J^c}\), and let
\(q(z):=q_{J^c}(z)\) and \(\rho_{J\mid z}\) be defined by
\eqref{eq:conditioning-state-convention}.  

Then
\[
    \eta_J
    :=
    \sum_{z:q(z)>0}q(z)\bigotimes_{i\in J}\rho_{i\mid z},
    \qquad
    \rho_{i\mid z}
    :=
    \Tr_{J\setminus\{i\}}(\rho_{J\mid z}),
\]
is a fully separable state on \(J\) satisfying
\begin{equation}
\label{eq:measurement-conditioned-trace-error}
    \frac1{\binom{|J|}{p}}
    \sum_{\substack{I\subseteq J\\|I|=p}}
    \|\rho_I-(\eta_J)_I\|_1^2
    \le
    C\kappa_{\mathcal M}^2\frac{p^3}{m_\star+1}.
\end{equation}
\end{lemma}

Before presenting the proof of Lemma~\ref{lem:measurement-parametrized-definetti}, we recall the information-theoretic identities to be used. The von Neumann entropy of a state $\omega$ is $S(\omega):=-\Tr(\omega\log\omega)$,
with natural logarithms and $0\log0:=0$.
For a state on a finite-dimensional Hilbert space $\mathcal H$,
the entropy satisfies $0\le S(\omega)\le\log\dim\mathcal H$.  For a subsystem $A$, write
$S(A)_\omega:=S(\omega_A)$.  Mutual information and conditional
mutual information are defined as
\begin{align*}
    I(A:B)_\omega
    &:=S(A)_\omega+S(B)_\omega-S(AB)_\omega,\\
    I(A:B\mid Z)_\omega
    &:=S(AZ)_\omega+S(BZ)_\omega-S(Z)_\omega-S(ABZ)_\omega.
\end{align*}
When $Z$ is classical, so that
\begin{equation}
\label{eq:cq-state-decomposition}
    \omega_{ABZ}
    =\sum_z q(z)\omega_{AB\mid z}
           \otimes|z\rangle\langle z|_{\mathrm{cl}},
\end{equation}
conditional mutual information is the average
\[
    I(A:B\mid Z)_\omega
    =\sum_zq(z)I(A:B)_{\omega_{AB\mid z}}.
\]
We use the chain rule
\[
    I(A:BC)_\omega=I(A:B)_\omega+I(A:C\mid B)_\omega
\]
and quantum Pinsker's inequality
\[
    \|\omega_{AB}-\omega_A\otimes\omega_B\|_1^2
    \le c I(A:B)_\omega
\]
for a positive universal constant $c$. These identities and this inequality can be found in standard quantum information theory texts, such as \cite[Theorem~5.38 and Eqs.~(5.91), (5.295)]{watrous2018theory}.

\begin{proof}[Proof of Lemma~\ref{lem:measurement-parametrized-definetti}]
We follow the argument of
\cite[Theorem~2.5]{berta2022semidefinite}, keeping the measurement
constant explicit and truncating the average over conditioning depths.

Let $\pi=(i_1,\ldots,i_p,j_1,\ldots,j_{n-p})$ be a uniformly random permutation of $[n]$.  The first $p$ qubits are retained for the comparison, and the $j$-qubits are available for
conditioning.  Write $A_t:=i_1\cdots i_t$ for the subsystem of the
first $t$ retained qubits, for $1\le t\le p$, and write $Z_j$ for
the outcome of the prescribed POVM on qubit $j$.  For \(0\le m\le n-p\), let
\[
    S_m:=\{j_1,\ldots,j_m\},
    \qquad
    Z^{(m)}:=(Z_{j_1},\ldots,Z_{j_m})
\]
denote the measured qubits and their outcomes. We use the convention \(S_0=\emptyset\) and \(Z^{(0)}\) is the unique empty
outcome, which has probability one and satisfies
\(\rho_{A_t\mid Z^{(0)}}=\rho_{A_t}\).

For $2\le t\le p$, define
\begin{equation}
\label{eq:revmag-delta-def}
    \Delta_{t,m}
    :=\|
        \rho_{A_t\mid Z^{(m)}}
        -\rho_{A_{t-1}\mid Z^{(m)}}\otimes
         \rho_{i_t\mid Z^{(m)}}
      \|_1.
\end{equation}
This measures the remaining dependence when adjoining the next
retained qubit.  We use $\E_\pi$ for the average over permutations,
$\E[\,\cdot\mid\pi]$ for the average over measurement outcomes at
a fixed permutation, and $\E$ for the joint average.

We first control the conditional dependence using a bound
on the information shared with the measurement outcomes. Fix \(t\in\{2,\ldots,p\}\).  Imagine also measuring \(i_t\) and
define the classical register $W_t:=(Z_{i_t},Z_{j_1},\ldots,Z_{j_{n-p}})$.

Let $\omega^{\pi,t}_{A_{t-1}W_t}$ be the state obtained by measuring
$i_t,j_1,\ldots,j_{n-p}$ in $\rho$ and keeping only the unmeasured
subsystem $A_{t-1}$ and the classical outcome register $W_t$.

For $0\le m\le n-p$, its marginal on
$A_{t-1}Z_{i_t}Z^{(m)}$ is obtained by discarding the later outcomes
$Z_{j_{m+1}},\ldots,Z_{j_{n-p}}$. Equivalently, we can obtain this
marginal by measuring only $i_t$ and $S_m$, retaining $A_{t-1}$
and those outcomes, and tracing out all other qubits.
These procedures agree because discarding a measurement outcome
means summing over it, which replaces the corresponding POVM
effect by $\sum_z M_z=I$.

All information quantities in this step are computed in
$\omega^{\pi,t}$ or its marginals, with the dependence on $\pi,t$
suppressed from the notation.

For each fixed ordering, the chain rule gives
\begin{align}
\label{eq:revmag-fixed-chain-rule}
    I(A_{t-1}:W_t)
    &=
    I(A_{t-1}:Z_{i_t})
    +
    \sum_{m=1}^{n-p}
    I\left(
        A_{t-1}:Z_{j_m}
        \,\middle|\,
        Z_{i_t},Z_{j_1},\ldots,Z_{j_{m-1}}
    \right).
\end{align}

We now average over the uniformly random ordering.  Fix \(m\ge1\)
and pair each ordering $\pi$ with the ordering obtained by exchanging the entries in positions
\(t\) and \(p+m\).  This pairing is a bijection of the permutation
space and leaves \(i_1,\ldots,i_{t-1}\), and hence \(A_{t-1}\),
unchanged.

For a given ordering \(\pi\), the \(m\)-th chain-rule summand for
the paired ordering, expressed in terms of the physical-site labels
from \(\pi\), is
\[
    I\left(
        A_{t-1}:Z_{i_t}
        \,\middle|\,
        Z_{j_m},Z_{j_1},\ldots,Z_{j_{m-1}}
    \right).
\]
Since the order of the classical conditioning registers is
immaterial, this equals
\[
    I\left(
        A_{t-1}:Z_{i_t}
        \,\middle|\,
        Z_{j_1},\ldots,Z_{j_m}
    \right)
    =
    I(A_{t-1}:Z_{i_t}\mid Z^{(m)}).
\]
Because the uniform average is invariant under this pairing,
\[
    \E_\pi
    I\left(
        A_{t-1}:Z_{j_m}
        \,\middle|\,
        Z_{i_t},Z_{j_1},\ldots,Z_{j_{m-1}}
    \right)
    =
    \E_\pi
    I(A_{t-1}:Z_{i_t}\mid Z^{(m)}).
\]
For \(m=0\), the corresponding term in
\eqref{eq:revmag-fixed-chain-rule} is already
\(I(A_{t-1}:Z_{i_t})\).  Averaging
\eqref{eq:revmag-fixed-chain-rule} therefore yields
\[
    \sum_{m=0}^{n-p}
    \E_\pi I(A_{t-1}:Z_{i_t}\mid Z^{(m)})
    =
    \E_\pi I(A_{t-1}:W_t).
\]
Note that this is an equality of expectations and the corresponding information
quantities need not be equal for an individual fixed ordering.

Since $W_t$ is classical, for every fixed ordering,
\begin{align*}
    I(A_{t-1}:W_t)
    &=
    S(A_{t-1})+S(W_t)-S(A_{t-1}W_t) \\
    &=
    S(A_{t-1})+H(q_\pi)
    -\left[
        H(q_\pi)
        +\sum_w q_\pi(w)S(\omega_{A_{t-1}\mid w})
    \right] \\
    &=
    S(A_{t-1})
    -\sum_w q_\pi(w)S(\omega_{A_{t-1}\mid w}) \\
    &\le
    S(A_{t-1}) 
    \le
    \log\dim\mathcal H_{A_{t-1}}
    =(t-1)\log 2,
\end{align*}
where $q_\pi$ is the outcome distribution of $W_t$ and
$H(q_\pi):=-\sum_w q_\pi(w)\log q_\pi(w)$ is its classical (Shannon) entropy.
The second equality uses the entropy formula for a
classical--quantum state. The first inequality uses the
nonnegativity of von Neumann entropy, and the second uses
$S(\omega_A)\le\log\dim\mathcal H_A$.
Finally, $A_{t-1}$ consists of $t-1$ qubits, so
$\dim\mathcal H_{A_{t-1}}=2^{t-1}$.  It follows
that
\begin{equation}
\label{eq:revmag-mi-budget}
    \sum_{m=0}^{n-p}
    \E_\pi I(A_{t-1}:Z_{i_t}\mid Z^{(m)})
    \le
    (t-1)\log 2.
\end{equation}

We now transfer this information bound to the original conditional
quantum states.  Fix an ordering and an outcome
\(Z^{(m)}=z\) of positive probability.  Applying
\eqref{eq:measurement-trace-comparison}, with \(\mathcal M\) acting
on \(i_t\), to
\[
    \xi := \rho_{A_t\mid z}
    -
    \rho_{A_{t-1}\mid z}\otimes\rho_{i_t\mid z},
\]
and then Pinsker's inequality gives
\[
    \left\|
        \rho_{A_t\mid z}
        -
        \rho_{A_{t-1}\mid z}\otimes\rho_{i_t\mid z}
    \right\|_1^2
    \le
    C\kappa_{\mathcal M}^2
    I(A_{t-1}:Z_{i_t})_{\omega\mid Z^{(m)}=z}.
\]
Averaging over \(Z^{(m)}\), for the fixed ordering, gives
\[
    \E[\Delta_{t,m}^2\mid\pi]
    \le
    C\kappa_{\mathcal M}^2
    I(A_{t-1}:Z_{i_t}\mid Z^{(m)}).
\]
Finally, averaging over \(\pi\), summing over \(m\), and applying
\eqref{eq:revmag-mi-budget} yields
\begin{equation}
\label{eq:revmag-df-chain}
    \sum_{m=0}^{n-p}
    \E\Delta_{t,m}^2
    \le
    C\kappa_{\mathcal M}^2(t-1).
\end{equation}

We next choose a conditioning depth $m\le m_\star$ with
small average error and use it to obtain a product
approximation. The summands in \eqref{eq:revmag-df-chain} are nonnegative, and so restricting
to $m\le m_\star$ and summing over $t=2,\ldots,p$ gives
\[
    \frac1{m_\star+1}\sum_{m=0}^{m_\star}
    \E\sum_{t=2}^p\Delta_{t,m}^2
    \le C\kappa_{\mathcal M}^2\frac{p^2}{m_\star+1}.
\]
Hence, there is a single $m\in\{0,\ldots,m_\star\}$ such that
\begin{equation}
\label{eq:revmag-df-good-m-step}
    \E\sum_{t=2}^p\Delta_{t,m}^2
    \le C\kappa_{\mathcal M}^2\frac{p^2}{m_\star+1}.
\end{equation}

For the selected depth \(m\), fix an outcome \(z\) of \(Z^{(m)}\)
having positive probability.  All states below are conditioned on the
same event \(Z^{(m)}=z\).  Tensor telescoping gives
\begin{align*}
&\rho_{A_p\mid Z^{(m)}=z}
-\bigotimes_{\ell=1}^p
 \rho_{i_\ell\mid Z^{(m)}=z}
\\
&\qquad=
\sum_{t=2}^p
\left(
    \rho_{A_t\mid Z^{(m)}=z}
    -
    \rho_{A_{t-1}\mid Z^{(m)}=z}
    \otimes
    \rho_{i_t\mid Z^{(m)}=z}
\right)
\otimes
\bigotimes_{\ell=t+1}^p
\rho_{i_\ell\mid Z^{(m)}=z}.
\end{align*}

Tensoring with a density matrix preserves the trace norm.  The triangle
inequality and Cauchy--Schwarz therefore imply
\[
    \bigg\|\rho_{A_p\mid Z^{(m)}}
       -\bigotimes_{\ell=1}^p\rho_{i_\ell\mid Z^{(m)}}\bigg\|_1^2
    \le(p-1)\sum_{t=2}^p\Delta_{t,m}^2.
\]
Together with \eqref{eq:revmag-df-good-m-step}, this yields
\begin{equation}
\label{eq:revmag-df-good-m}
    \E\bigg\|\rho_{A_p\mid Z^{(m)}}
       -\bigotimes_{\ell=1}^p\rho_{i_\ell\mid Z^{(m)}}\bigg\|_1^2
    \le C\kappa_{\mathcal M}^2\frac{p^3}{m_\star+1}.
\end{equation}

Finally, we fix the measured set and average the conditional
product states over the measurement outcomes to construct
the separable approximation. The conditioning sites in \eqref{eq:revmag-df-good-m} are still random.
Condition first on their ordered list $(j_1,\ldots,j_m)$.  Since the
total average satisfies the displayed bound, some fixed list satisfies
the same bound after averaging only over the retained tuple and the
measurement outcomes.  Fix this list and set
$J^c:=\{j_1,\ldots,j_m\}$, so $|J^c|=m\le m_\star$.
Conditional on the list, the retained tuple is uniform among the ordered
$p$-tuples of distinct sites in $J$.  Every $p$-subset has exactly
$p!$ orderings, and the trace-norm error is unchanged by reordering
its tensor factors.  Thus, writing $q(z)=q_{J^c}(z)$,
\begin{equation}
\label{eq:conditioning-fixed-set-error}
    \frac1{\binom{|J|}{p}}
    \sum_{\substack{I\subseteq J\\|I|=p}}\sum_zq(z)
    \bigg \|\rho_{I\mid z}
       -\bigotimes_{i\in I}\rho_{i\mid z}\bigg\|_1^2
    \le C\kappa_{\mathcal M}^2\frac{p^3}{m_\star+1}.
\end{equation}
The same measured set is used for every $I$ in this average.

Form $\eta_J$ as in the statement.  Each conditional product is
separable, so $\eta_J$ is fully separable.  By
\eqref{eq:conditioning-average-marginal},
\[
    \rho_I=\sum_zq(z)\rho_{I\mid z},
    \qquad
    (\eta_J)_I=\sum_zq(z)\bigotimes_{i\in I}\rho_{i\mid z}.
\]
The triangle inequality and Jensen's inequality give
\[
    \|\rho_I-(\eta_J)_I\|_1^2
    \le\sum_zq(z)
       \bigg\|\rho_{I\mid z}
          -\bigotimes_{i\in I}\rho_{i\mid z}\bigg \|_1^2.
\]
Average over $I$ and apply \eqref{eq:conditioning-fixed-set-error} to
obtain \eqref{eq:measurement-conditioned-trace-error}.
\end{proof}

The following corollary shows that the conditioning construction of Lemma~\ref{lem:measurement-parametrized-definetti} is compatible with adjoining and subsequently discarding an auxiliary system, a fact we will use later to treat circuits with ancillas.

\begin{corollary}[Conditioning and partial trace]
\label{cor:conditioning-extensions}
Let $\widehat\rho$ be a state on $[n]\cup E$ and set
$\rho:=\Tr_E\widehat\rho$. Fix $S\subseteq[n]$ and put
$J=[n]\setminus S$.

Measure the qubits in $S$ using the prescribed POVM.
Write $\widehat q(z),\widehat\rho_{JE\mid z}$ for the outcome
probabilities and conditional states when $E$ is kept, and
$q(z),\rho_{J\mid z}$ when $E$ is traced out before measurement.
Then $\widehat q(z)=q(z)$ for every $z$, and
\[
    \Tr_E\widehat\rho_{JE\mid z}=\rho_{J\mid z},
    \qquad q(z)>0.
\]
If $\widehat\rho$ is pure and the POVM effects have rank one,
then $\widehat\rho_{JE\mid z}$ is pure for every outcome
of positive probability.
\end{corollary}

\begin{proof}
Write $K_z=(M_z^{(S)})^{1/2}$, using the measurement operators
on the original $n$ qubits. Since the measurement does not act on $E$,
\[
    \widehat q(z)
    =\Tr\left[(M_z^{(S)}\otimes I_E)\widehat\rho\right]
    =\Tr\left[M_z^{(S)}\Tr_E\widehat\rho\right]
    =q(z).
\]
For $q(z)>0$, the conditioning formula gives
\begin{align*}
    \Tr_E\widehat\rho_{JE\mid z}
    &=
    \frac{\Tr_{SE}\left[
        (K_z\otimes I_E)\widehat\rho(K_z\otimes I_E)
    \right]}{q(z)}\\
    &=
    \frac{\Tr_S\left[K_z(\Tr_E\widehat\rho)K_z\right]}{q(z)}
    =
    \frac{\Tr_S[K_z\rho K_z]}{q(z)}
    =\rho_{J\mid z}.
\end{align*}

Now let $\widehat\rho=|\Psi\rangle\langle\Psi|$ and order the
qubits as $S,J,E$.
If the effects have rank one, then
$L_z:=\bigotimes_{i\in S}M_{z_i}^{1/2}$ has rank one.
Hence, for some vectors $|v_z\rangle_S$ and
$|\widetilde\phi_z\rangle_{JE}$,
\[
    (L_z\otimes I_{JE})|\Psi\rangle
    =|v_z\rangle_S\otimes|\widetilde\phi_z\rangle_{JE}.
\]
Tracing out $S$ and normalizing gives $\widehat\rho_{JE\mid z}
    =
    |\widetilde\phi_z\rangle\langle\widetilde\phi_z|/
         \|\widetilde\phi_z\|^2$,
which is pure.
\end{proof}

\subsubsection{Profile approximation and conditioning benchmarks}
\label{subsubsec:canonical-conditioning}

The conditioning results of the previous section control the average squared trace-norm distance between the \(p\)-qubit marginals of the original state and those of a separable approximation. The Hamiltonian comparison, however, is formulated in terms of the canonical norm of their Pauli profiles. We now convert the local trace-norm estimate into the required canonical-norm estimate. Although an analogous result holds for any informationally complete POVM satisfying \eqref{eq:measurement-trace-comparison}, we state only the Pauli-POVM version needed in this work.

For $J\subseteq[n]$, define the coordinate projection
\[
    (\mathsf M_Jf)(P)
    :=f(P)\mathbf1_{\{\operatorname{supp}(P)\subseteq J\}},
    \qquad P\in\mathcal P_{n,p}.
\]
It keeps the interactions supported entirely on $J$ without changing
the original normalization in $\|\cdot\|_{n,p}$.

\begin{corollary}[Canonical form of the conditioning approximation]
\label{cor:conditioning-canonical}
There is a universal constant \(C_0<\infty\) such that the following
holds.  For every \(n\)-qubit state \(\rho\), integers \(n\ge p\ge2\),
and an integer \(0\le m_\star\le n-p\), there exist a set \(J\subseteq[n]\), with
\(|J^c|\le m_\star\), and a fully separable \(n\)-qubit state \(\sigma\)
such that
\begin{equation}
\label{eq:definetti-canonical}
    \|\mathsf M_J(f_\rho-f_\sigma)\|_{n,p}
    \le
    r_{p,m_\star}
    :=
    C_0\frac{2^{p/2}p^{3/2}}{\sqrt{m_\star+1}}.
\end{equation}
\end{corollary}

\begin{proof}
Apply Lemma~\ref{lem:measurement-parametrized-definetti} with the
Pauli POVM. By Lemma~\ref{lem:pauli-measurement-comparison}, we may
take $\kappa_{\mathcal M}=4$. Let $J$ and $\eta_J$ be the resulting
set and separable state, and define
$\sigma:=\eta_J\otimes|0\rangle\langle0|^{\otimes J^c}$.
Then $\sigma$ is fully separable and
$\sigma_{J_0}=(\eta_J)_{J_0}$ for every $J_0\subseteq J$.

For a $p$-element subset $J_0\subseteq J$, set
$A_{J_0}:=\rho_{J_0}-\sigma_{J_0}$. Pauli orthogonality gives
\begin{align*}
    \sum_{a\in\{X,Y,Z\}^p}
    \left|
        \Tr\left[(a_1\otimes\cdots\otimes a_p)A_{J_0}\right]
    \right|^2
    &\le
    \sum_{P\in\{I,X,Y,Z\}^{\otimes p}}
        |\Tr(PA_{J_0})|^2 \\
    &=2^p\Tr(A_{J_0}^2)
    \le 2^p\|A_{J_0}\|_1^2.
\end{align*}
The equality follows from Parseval's identity, since the
normalized Pauli strings $2^{-p/2}P$ form an orthonormal basis
for the Hilbert--Schmidt inner product. This gives a factor $2^p$ instead of the $3^p$ obtained by
bounding each Pauli expectation separately, which is important for the large-$p$ ground-state comparison. Consequently,
\begin{align*}
    \|\mathsf M_J(f_\rho-f_\sigma)\|_{n,p}^2
    &\le
    \frac{2^p}{\binom np}
    \sum_{\substack{J_0\subseteq J\\|J_0|=p}}
        \|\rho_{J_0}-(\eta_J)_{J_0}\|_1^2 \\
    &\le
    C\kappa_{\mathcal M}^2
    \frac{2^pp^3}{m_\star+1}
    \frac{\binom{|J|}{p}}{\binom np} \\
    &\le
    16C\frac{2^pp^3}{m_\star+1},
\end{align*}
where the second inequality uses
\eqref{eq:measurement-conditioned-trace-error}.
Taking square roots and absorbing the universal factor
$4\sqrt C$ into $C_0$ proves \eqref{eq:definetti-canonical}.
\end{proof}

A quantitative sufficient condition for
Definition~\ref{def:conditioning-benchmark} is the following.
For every $n\ge p\ge2$, every integer $0\le m_\star\le n-p$,
and every $\rho\in\mathcal C_n$, one can choose
$J\subseteq[n]$ and $\sigma\in\operatorname{conv}(\mathcal B_n)$ with
\begin{equation}
\label{eq:conditioning-benchmark-sufficient}
    |J^c|\le m_\star,
    \qquad
    \|\mathsf M_J(f_\rho-f_\sigma)\|_{n,p}
    \le
    \frac{A_p}{\sqrt{m_\star+1}},
\end{equation}
where $A_p<\infty$ is independent of $n$, $m_\star$, and $\rho$.
Indeed, for fixed $\alpha\in(0,1)$, take
$m_\star=\lfloor\alpha n\rfloor$, which is at most $n-p$
for all sufficiently large $n$. The squared canonical error is then
at most $A_p^2/(\alpha n)=O_{p,\alpha}(n^{-1})$, uniformly over
$\rho\in\mathcal C_n$, as required by the definition.

Corollary~\ref{cor:conditioning-canonical} gives this condition with
$\mathcal B_n=\Prod_n$ and $A_p=C_0 2^{p/2}p^{3/2}$.
Thus $\Prod_n$ is a conditioning benchmark for \textit{every} state class.

\subsection{Uniform energy comparison}
For a subset $J \subseteq [n]$, we now bound the operator norm of the sum of Hamiltonian terms whose supports intersect $J^c$. As a reminder, we decompose
$H_{n,p} = H_J+H_{\mathrm{hit}(J^c)}$,
where
\[
    H_J
    :=
    \binom np^{-1/2}
    \sum_{\substack{P\in\mathcal P_{n,p}\\\operatorname{supp}(P)\subseteq J}}
    g_P P,
    \qquad
    H_{\mathrm{hit}(J^c)}
    :=
    H_{n,p}-H_J.
\]
For an integer $0\le m_{\star}\le n$, write $\mathcal J_{n,m_{\star}}
    :=
    \{
        J\subseteq[n]:
        |J^c|\le m_{\star}
    \}$.
\label{subsec:conditioning-to-energy}
\begin{lemma}[Uniform cost of a conditioning budget]
\label{lem:revmag-hit-bound}
There is a universal constant $C<\infty$ such that, for every $n,p\ge2$ and
$0\le m_{\star}\le n-p$,
\begin{equation}
\label{eq:variable-budget-hit-bound}
    \E
    \max_{J\in\mathcal J_{n,m_{\star}}}
    \|H_{\mathrm{hit}(J^c)}\|_{\mathrm{op}}
    \le
    C\,3^{p/2}\sqrt{pm_{\star}}.
\end{equation}
\end{lemma}
\begin{proof}
For $S\subseteq[n]$, $|S|\le m_{\star}$, write
\[
    H_{\mathrm{hit}(S)}
    =
    \sum_{P\in\mathcal P_{n,p}}g_PA_{P,S},
    \qquad
    A_{P,S}
    :=
    \binom np^{-1/2}
    P\,\mathbf 1_{\{\operatorname{supp}(P)\cap S\ne\emptyset\}}.
\]
Since $P^2=I$,
\begin{align*}
    \bigg\|\sum_{P\in\mathcal P_{n,p}}A_{P,S}^2
    \bigg\|_{\mathrm{op}}
    &=
    3^p
    \frac{
        \#\{J_0\in\binom{[n]}p: J_0\cap S\ne\emptyset\}
    }{
        \binom np
    }     
    \le
    3^pp\frac{|S|}{n}
    \le
    3^pp\frac{m_{\star}}{n}.
\end{align*}
To control all admissible sets $J$ simultaneously, we use a
block-diagonal dilation.  Namely, define the block-diagonal Gaussian series
\[
    \mathbf H_{m_{\star}}
    :=
    \bigoplus_{J\in\mathcal J_{n,m_{\star}}}
    H_{\mathrm{hit}(J^c)},
\]
which is simply the block-diagonal matrix having one block
$H_{\mathrm{hit}(J^c)}$ for each admissible $J$.  The operator norm
of a block-diagonal matrix is the largest operator norm of its blocks,
so
$\|\mathbf H_{m_{\star}}\|_{\mathrm{op}}
=
\max_{J\in\mathcal J_{n,m_{\star}}}
\|H_{\mathrm{hit}(J^c)}\|_{\mathrm{op}}$. Since  $|\mathcal J_{n,m_{\star}}|\le2^n$, its dimension is
\[
    d_{m_{\star}}
    =
    2^n|\mathcal J_{n,m_{\star}}|
    \le
    2^{2n}.
\]
Set
$\mathbf A_P:=\bigoplus_{J\in\mathcal J_{n,m_\star}}A_{P,J^c}$.
Then $\mathbf H_{m_\star}=\sum_Pg_P\mathbf A_P$ is a Gaussian
matrix series with independent coefficients $g_P$, shared
by all blocks. Its variance parameter satisfies
\[
    \big\|\sum_P\mathbf A_P^2\big\|_{\mathrm{op}}
    \le
    3^pp\frac{m_\star}{n}.
\]
The Gaussian matrix-series estimate
\cite[Theorem~4.1 and Eq.~(4.9)]{tropp2012user} therefore gives
\[
    \E\|\mathbf H_{m_\star}\|_{\mathrm{op}}
    \le
    C\sqrt{\log(2d_{m_\star})}\,
    3^{p/2}\sqrt{\frac{pm_\star}{n}}
    \le
    C\,3^{p/2}\sqrt{pm_\star}.
\]
\end{proof}

The following lemma converts local-profile complexity
into an energy bound relative to a benchmark class.  It shows that the
possible advantage over the benchmark is controlled by the Gaussian width
of the approximation residual together with the contribution of interactions
whose supports intersect the excluded qubits.

\begin{lemma}[Profile entropy and benchmark comparison]
\label{lem:profile-entropy-comparison}
Fix $n\ge p\ge2$.  Let $\mathcal C_n$ and $\mathcal B_n$ be classes
of $n$-qubit states, let $0\le m_{\star}\le n-p$ be an integer, and let $r>0$.
Suppose that, for every $\rho\in\mathcal C_n$, there exist
$J\in\mathcal J_{n,m_{\star}}$ and
$\sigma\in\operatorname{conv}(\mathcal B_n)$ such that $\|\mathsf M_J(f_\rho-f_\sigma)\|_{n,p}\le r$. Then
\begin{align}
\label{eq:profile-entropy-comparison}
    \E\Delta^{\mathrm{ref}}_{n,p}(\mathcal C_n;\mathcal B_n)
    \le{}&
    C
    \int_0^{2r}
    \sqrt{
        \log\left(
            2|\mathcal J_{n,m_{\star}}|
            N\left(
                \varepsilon,
                \mathcal F_{n,p}(\mathcal C_n),
                \|\cdot\|_{n,p}
            \right)
        \right)
    }
    \,d\varepsilon
    +
    C\,3^{p/2}\sqrt{pm_{\star}}.
\end{align}
Moreover, for every $t>0$,
\begin{equation}
\label{eq:profile-gap-concentration}
    \P\left(
        \Delta^{\mathrm{ref}}_{n,p}(\mathcal C_n;\mathcal B_n)
        \ge
        \E\Delta^{\mathrm{ref}}_{n,p}(\mathcal C_n;\mathcal B_n)+t
    \right)
    \le
    \exp\left(-\frac{t^2}{8\cdot3^p}\right).
\end{equation}
\end{lemma}

\begin{proof}
For each $J\in\mathcal J_{n,m_{\star}}$, define $\mathcal K_J
    :=
    \left\{
        \mathsf M_Jf_\sigma:
        \sigma\in\overline{\operatorname{conv}}(\mathcal B_n)
    \right\}$.
This is a nonempty compact convex set.  The closure is taken in the
finite-dimensional space of density matrices and does not change the
supremum of any linear energy functional.  Let $\Pi_J$ denote metric
projection onto $\mathcal K_J$ in the canonical norm, and set $\mathsf{R}_J(f)
    :=
    \mathsf M_Jf-\Pi_J(\mathsf M_Jf)$.

The dependence of these objects on $\mathcal B_n$ is suppressed in the
notation.  By the approximation hypothesis, every $\rho\in\mathcal C_n$
admits some $J\in\mathcal J_{n,m_{\star}}$ for which
\[
    \|\mathsf{R}_J(f_\rho)\|_{n,p}
    =\operatorname{dist}(\mathsf M_Jf_\rho,\mathcal K_J)
    \le r.
\]
Define the class of admissible residuals
\[
    \mathfrak R_{n,p,m_{\star}}(\mathcal C_n)
    :=
    \left\{
        \mathsf{R}_J(f_\rho):
        \rho\in\mathcal C_n,
        \ J\in\mathcal J_{n,m_{\star}},
        \ \|\mathsf{R}_J(f_\rho)\|_{n,p}\le r
    \right\} \cup \{0\},
\]
where the zero profile is included to ensure that suprema of the associated Gaussian process are nonnegative, as required to bound \(\Delta^{\mathrm{ref}}_{n,p}\) from above. The benchmark and $r$ are fixed throughout the proof.

Fix $\rho\in\mathcal C_n$ and choose a corresponding set $J$.  Let
\[
    k_J(\rho)
    :=
    \Pi_J(\mathsf M_Jf_\rho),
    \qquad
    r_J(\rho)
    :=
    \mathsf{R}_J(f_\rho).
\]
Then
\[
    \Tr(H_{n,p}\rho)
    =
    X_{k_J(\rho)}
    +
    X_{r_J(\rho)}
    +
    \Tr(H_{\mathrm{hit}(J^c)}\rho).
\]
Because $k_J(\rho)\in\mathcal K_J$, there exists
$\sigma\in\overline{\operatorname{conv}}(\mathcal B_n)$ such that
$\mathsf M_Jf_\sigma=k_J(\rho)$.  Hence
\begin{align*}
    X_{k_J(\rho)}
    &=
    \Tr(H_J\sigma)                                                   \\
    &=
    \Tr(H_{n,p}\sigma)
    -
    \Tr(H_{\mathrm{hit}(J^c)}\sigma)                                \\
    &\le
    \mathcal{E}_{n,p}(\mathcal B_n)
    +
    \|H_{\mathrm{hit}(J^c)}\|_{\mathrm{op}}.
\end{align*}
Here we used linearity and continuity of the energy on the
closed convex hull of $\mathcal B_n$, together with the bounds
$-\Tr(H_{\mathrm{hit}(J^c)}\sigma)\le
\|H_{\mathrm{hit}(J^c)}\|_{\mathrm{op}}$ and
$\Tr(H_{\mathrm{hit}(J^c)}\rho)\le
\|H_{\mathrm{hit}(J^c)}\|_{\mathrm{op}}$.
Consequently, pointwise in the Gaussian disorder,
\begin{equation}
\label{eq:profile-pointwise-comparison}
    \Delta^{\mathrm{ref}}_{n,p}(\mathcal C_n;\mathcal B_n)
    \le
    \sup_{r'\in\mathfrak R_{n,p,m_{\star}}(\mathcal C_n)}X_{r'}
    +
    2
    \max_{J\in\mathcal J_{n,m_{\star}}}
    \|H_{\mathrm{hit}(J^c)}\|_{\mathrm{op}}.
\end{equation}

We next cover the residual class. Since metric projections
onto closed convex subsets of a Hilbert space are firmly
nonexpansive, $I-\Pi_J$ is nonexpansive
\cite[Fact~1.1]{bauschke2012firmly}. Moreover, $\mathsf M_J$ is an orthogonal coordinate projection,
so
\begin{align}
\label{eq:residual-map-nonexpansive}
    \|\mathsf{R}_J(f)-\mathsf{R}_J(g)\|_{n,p}
    &\le
    \|\mathsf M_Jf-\mathsf M_Jg\|_{n,p}
    \le
    \|f-g\|_{n,p}.
\end{align}
Map an $\varepsilon$-cover of the profile class through each
residual map. These maps are defined on the entire profile
space and are nonexpansive, so the resulting balls cover
all admissible residuals at radius $\varepsilon$.
Taking their union and adding a ball centered at zero gives
\[
    N\left(
        \varepsilon,
        \mathfrak R_{n,p,m_{\star}}(\mathcal C_n),
        \|\cdot\|_{n,p}
    \right)
    \le
    1
    +
    |\mathcal J_{n,m_{\star}}|
    N\left(
        \varepsilon,
        \mathcal F_{n,p}(\mathcal C_n),
        \|\cdot\|_{n,p}
    \right).
\]
The residual class contains zero and lies in the canonical ball of radius
$r$. Dudley's entropy inequality therefore implies
\begin{align*}
    \E
    \sup_{r'\in\mathfrak R_{n,p,m_{\star}}(\mathcal C_n)}X_{r'}
    \le
    C
    \int_0^{2r}
    \sqrt{
        \log\left(
            2|\mathcal J_{n,m_{\star}}|
            N\left(
                \varepsilon,
                \mathcal F_{n,p}(\mathcal C_n),
                \|\cdot\|_{n,p}
            \right)
        \right)
    }
    \,d\varepsilon.
\end{align*}
Taking expectations in \eqref{eq:profile-pointwise-comparison} and applying
Lemma~\ref{lem:revmag-hit-bound} proves
\eqref{eq:profile-entropy-comparison}.

To show the concentration result, see
\cite[Section 3.2]{anschuetz2025bounds} or note that for every
state $\rho$, the coefficient vector of the Gaussian functional
$g\mapsto\Tr(H_{n,p}\rho)$ satisfies
\[
    \frac1{\binom np}
    \sum_{P\in\mathcal P_{n,p}}
    |\Tr(P\rho)|^2
    \le
    3^p.
\]
Thus both $\mathcal{E}_{n,p}(\mathcal C_n)$ and
$\mathcal{E}_{n,p}(\mathcal B_n)$ are
$3^{p/2}$-Lipschitz functions of the standard Gaussian coefficient vector.
Their difference is $2\cdot3^{p/2}$-Lipschitz, and taking the positive part
does not increase the Lipschitz constant.
\eqref{eq:profile-gap-concentration} therefore follows by Gaussian concentration.
\end{proof}

For $\mathcal B_n=\Prod_n$, Corollary~\ref{cor:conditioning-canonical}
verifies the approximation hypothesis with $r=r_{p,m_{\star}}$.  This recovers the
universal product-state comparison used below.  For a general conditioning
benchmark, Definition~\ref{def:conditioning-benchmark} verifies the
hypothesis with $m_{\star}=\lfloor\alpha n\rfloor$ and
$r=A_{p,\alpha}/\sqrt n$ for each fixed $p,\alpha$ and all sufficiently
large $n$.

\subsection{Ground-state separation at quadratic profile complexity}

We now apply Lemma~\ref{lem:profile-entropy-comparison} at the quadratic
profile-complexity scale.  The resulting estimate gives a quantitative
upper bound on the energy attainable by a state class in terms of its
effective profile complexity.  Comparing this bound with the
ground-state energy will yield
Theorem~\ref{thm:local-profile-complexity-principle}\textup{(i)}.

\begin{corollary}[Quadratic profile complexity]
\label{cor:quadratic-profile-complexity}
Fix $p\ge2$, and let $\mathcal C_n$ be a class of $n$-qubit states.
Suppose that there exist sequences $a_n,b_n\ge0$ and $B_n\ge1$ such
that, for every $0<\varepsilon\le1$,
\[
    \log
    N\left(
        \varepsilon,
        \mathcal F_{n,p}(\mathcal C_n),
        \|\cdot\|_{n,p}
    \right)
    \le
    a_n
    +
    b_n\log\left(
        \frac{B_n}{\varepsilon}
    \right),
\]
and that, for some $A>0$,
\[
    \limsup_{n\to\infty}
    \frac{
        a_n+b_n\log(eB_nn)
    }{
        n^2
    }
    \le
    A.
\]
Then, for every fixed $0<\alpha<1/2$,
\[
    \limsup_{n\to\infty}
    \frac{
        \E\Delta^{\mathrm{ref}}_{n,p}(\mathcal C_n;\Prod_n)
    }{
        \sqrt n
    }
    \le
    C\left[
        \frac{
            2^{p/2}p^{3/2}\sqrt A
        }{
            \sqrt\alpha
        }
        +
        3^{p/2}\sqrt{p\alpha}
    \right].
\]
Moreover, for every $\varepsilon>0$, with probability tending to one,
\[
    \mathcal{E}_{n,p}(\mathcal C_n)
    \le
    \mathcal{E}_{n,p}(\Prod_n)
    +
    \left\{
        C\left[
            \frac{
                2^{p/2}p^{3/2}\sqrt A
            }{
                \sqrt\alpha
            }
            +
            3^{p/2}\sqrt{p\alpha}
        \right]
        +
        \varepsilon
    \right\}\sqrt n.
\]
\end{corollary}

\begin{proof}
Fix $0<\alpha<1/2$. For all sufficiently large $n$, apply
Lemma~\ref{lem:profile-entropy-comparison} with benchmark
$\mathcal B_n=\Prod_n$, budget $m_\star:=\lfloor\alpha n\rfloor$,
and radius $r=r_{p,m_\star}$. Then
\[
    r=\frac{C_0 2^{p/2}p^{3/2}}{\sqrt{m_\star+1}},
    \qquad
    \log(2|\mathcal J_{n,m_\star}|)\le(n+1)\log2.
\]
For fixed $p,\alpha$ and all sufficiently large $n$,
$n^{-1}\le r\le1/2$. The assumed entropy bound therefore
bounds the normalized residual contribution by
\begin{align*}
    \frac{C}{\sqrt n}
    \int_0^{2r}
        \sqrt{a_n+b_n\log(B_n/\varepsilon)+Cn}
    \,d\varepsilon
    &\le
    \frac{Cr}{\sqrt n}
        \sqrt{a_n+b_n\log(eB_n/r)+n} \\
    &\le
    C\frac{2^{p/2}p^{3/2}}{\sqrt\alpha}
    \left[
        \frac{\sqrt{a_n+b_n\log(eB_nn)}}{n}
        +\frac1{\sqrt n}
    \right].
\end{align*}
The last inequality uses $m_\star+1\ge\alpha n$ and
$\log(eB_n/r)\le\log(eB_nn)$.
Taking the limsup and using the assumption on $A$ bounds
this contribution by
$C2^{p/2}p^{3/2}\sqrt A/\sqrt\alpha$.

The normalized contribution of interactions intersecting the
conditioned set is at most $C3^{p/2}\sqrt{p\alpha}$.
Lemma~\ref{lem:profile-entropy-comparison} therefore gives the
stated bound on
$\E\Delta^{\mathrm{ref}}_{n,p}(\mathcal C_n;\Prod_n)$.
The high-probability statement follows from
\eqref{eq:profile-gap-concentration}.
\end{proof}

We complete the proof of
Theorem~\ref{thm:local-profile-complexity-principle}\textup{(i)} after
recording the required product-state and unrestricted-energy benchmarks
below.

\subsection{Energy benchmarks}
We next record two energy benchmark estimates for the dense
quantum $p$-spin model that will be used throughout the applications.
The first controls the product-state optimum.
The product-state energy is studied in
\cite[Sections~3.2 and~4.2]{anschuetz2025bounds}, with the normalization
used here.  For completeness, we give a direct proof of the uniform
finite-$n$ bound needed below.

\begin{lemma}[Product-state energy benchmark]
\label{lem:profile-product-benchmark}
There is a universal constant $C<\infty$ such that, for every
$n\ge p\ge2$,
\[
    \E \mathcal{E}_{n,p}(\Prod_n)
    \le
    C\sqrt{n\log(ep)}.
\]
Moreover, for every $t>0$,
\[
    \P\left(
        \mathcal{E}_{n,p}(\Prod_n)
        >
        C\sqrt{n\log(ep)}+t
    \right)
    \le
    \exp\left(-\frac{t^2}{2}\right).
\]
\end{lemma}

\begin{proof}
Let $\mathbb S^2:=\{v\in\mathbb R^3:\|v\|_2=1\}$.
Represent a pure product state by
$u=(u_1,\ldots,u_n)\in(\mathbb S^2)^n$, where
$u_i=(u_{i,X},u_{i,Y},u_{i,Z})$, with its $i$th factor equal to
$\frac12(I+u_{i,X}X+u_{i,Y}Y+u_{i,Z}Z)$.
Write $f_u$ for its profile and $X_u:=X_{f_u}$ for its energy.  On a $p$-element support $J_0$, the vector
of Pauli expectations is $\bigotimes_{i\in J_0}u_i$.  Thus $ \|f_u\|_{n,p}^2=\E X_u^2=1$ and 
\[
    d_{n,p}(f_u,f_v)^2
    =\binom np^{-1}\sum_{|J_0|=p}
      2\left(1-\prod_{i\in J_0}(u_i\cdot v_i)\right).
\]
For $c_1,\ldots,c_p\in[-1,1]$, telescoping gives
\[
    1-\prod_{j=1}^p c_j
    =\sum_{j=1}^p(1-c_j)\prod_{\ell<j}c_\ell
    \le\sum_{j=1}^p(1-c_j).
\]
Consequently,
\[
    d_{n,p}(f_u,f_v)^2
    \le\binom np^{-1}\sum_{|J_0|=p}\sum_{i\in J_0}\|u_i-v_i\|_2^2
    =\frac pn\sum_{i=1}^n\|u_i-v_i\|_2^2.
\]
A Euclidean $\delta$-net of $\mathbb S^2$ has cardinality at most $(C/\delta)^C$.  Taking the product of
$n$ such nets with $\delta=\varepsilon/\sqrt p$ therefore gives
\[
    \log N\left(
        \varepsilon,\mathcal F_{n,p}(\Prod_n),\|\cdot\|_{n,p}
    \right)
    \le Cn\log\left(\frac{C\sqrt p}{\varepsilon}\right),
    \qquad 0<\varepsilon\le2.
\]
The canonical diameter is at most two. Dudley's entropy integral gives
\[
    \E \mathcal{E}_{n,p}(\Prod_n)
    \le C\sqrt n\int_0^2
        \sqrt{\log\left(\frac{C\sqrt p}{\varepsilon}\right)}
        \,d\varepsilon
    \le C\sqrt{n\log(ep)}.
\]
Finally, $\|f_u\|_{n,p}=1$ implies that the energy supremum is
one-Lipschitz in the standard Gaussian coefficient vector.  Gaussian
concentration gives the stated tail bound.
\end{proof}

A sharper asymptotic characterization was given recently by
Bhattacharya~\cite[Theorem~1]{bhattacharya2026parisi}, who proves
that, for every fixed $p\ge2$, the limit
\[
    e_{\mathrm{prod}}(p)
    :=
    \lim_{n\to\infty}
    \frac{\E \mathcal{E}_{n,p}(\Prod_n)}{\sqrt n}
\]
exists and is given by a Parisi variational formula.
Moreover, Gaussian concentration also gives
$\mathcal{E}_{n,p}(\Prod_n)/\sqrt n\to e_{\mathrm{prod}}(p)$
in $L^1$ and in probability.
Consequently,
Theorem~\ref{thm:local-profile-complexity-principle}\textup{(ii)}
gives the same limiting energy for any class
$\mathcal C_n\supseteq\Prod_n$ satisfying its entropy assumptions
with $\Gamma_{n,p}=o(n^2)$.

The second benchmark is a lower bound on the unrestricted optimum.
In Corollary~D.2 of \cite{anschuetz2025strongly}, the
spin-glass Hamiltonian is normalized as
$\widetilde H_{n,p}:=(3^p\binom np)^{-1/2}\sum_{P\in\mathcal
P_{n,p}}g_PP$.  Hence $H_{n,p}=3^{p/2}\widetilde H_{n,p}$, and their
bound $\lambda_{\max}(\widetilde H_{n,p})=\Omega(\sqrt n/p)$
gives the following statement in our normalization.

\begin{lemma}[Ground-state energy benchmark]
\label{lem:revmag-ground-benchmark}
There is a universal constant $c_0>0$ such that, for every fixed
$p\ge3$,
\[
    \P\left(
        \frac{1}{\sqrt n}\lambda_{\max}(H_{n,p})
        \ge
        c_0\frac{3^{p/2}}{p}
    \right)
    \longrightarrow1.
\]
\end{lemma}
\subsection{Completing the quadratic
ground-state argument}
We now combine the quadratic profile estimate with the two
energy benchmarks to prove the ground-state conclusion.  The subquadratic
benchmark comparison is treated separately in the following subsection.

\begin{proof}[Proof of
Theorem~\ref{thm:local-profile-complexity-principle}\textup{(i)}] 
Fix $L\ge0$ and an integer $p \ge 3$ satisfying
$\limsup_{n\to\infty}\Gamma_{n,p}/n^2\le L$.
We require $p$ to be sufficiently large depending only on $L$.
Corollary~\ref{cor:quadratic-profile-complexity}, together with
$\E \mathcal{E}_{n,p}(\mathcal C_n)-\E \mathcal{E}_{n,p}(\Prod_n)
\le\E\Delta^{\mathrm{ref}}_{n,p}(\mathcal C_n;\Prod_n)$, gives, for every
$0<\alpha<1/2$, 
\[
    \limsup_{n\to\infty}
    \frac1{\sqrt n}
    \left[
        \E \mathcal{E}_{n,p}(\mathcal C_n)
        -
        \E \mathcal{E}_{n,p}(\Prod_n)
    \right]
    \le
    C\left[
        \frac{2^{p/2}p^{3/2}\sqrt L}{\sqrt\alpha}
        +
        3^{p/2}\sqrt{p\alpha}
    \right].
\]
If $L>0$, set $\alpha_p := \sqrt L\,p (2/3)^{p/2}$.  Then, for sufficiently large $p$,
$\alpha_p<1/2$, and the preceding bound becomes
\[
    \limsup_{n\to\infty}
    \frac1{\sqrt n}
    \left[
        \E \mathcal{E}_{n,p}(\mathcal C_n)
        -
        \E \mathcal{E}_{n,p}(\Prod_n)
    \right]
    \le
    C L^{1/4}p\,6^{p/4}.
\]
For $L=0$, the same bound follows by letting $\alpha\downarrow0$ for
each fixed $p$ in the unoptimized bound.  Together
with Lemma~\ref{lem:profile-product-benchmark}, this gives
\[
    \limsup_{n\to\infty}
    \frac1{\sqrt n}
    \E \mathcal{E}_{n,p}(\mathcal C_n)
    \le
    C\sqrt{\log(ep)}
    +
    C L^{1/4}p\,6^{p/4}.
\]
On the other hand, Lemma~\ref{lem:revmag-ground-benchmark} gives $\lambda_{\max}(H_{n,p})/\sqrt n
    \ge c_0 3^{p/2}/p$ with probability tending to one.  Since
\[
    \frac{
        L^{1/4}p\,6^{p/4}
    }{
        3^{p/2}/p
    }
    =
    L^{1/4}p^2
    \left(\frac23\right)^{p/4}
    \longrightarrow0,
\]
and likewise
\[
    \frac{\sqrt{\log(ep)}}{3^{p/2}/p}
    \longrightarrow0,
\]
the upper bound for $\mathcal C_n$ is a strict fraction of the
ground-state scale for all sufficiently large fixed $p$.

Finally, $\mathcal{E}_{n,p}(\mathcal C_n)$ is
$3^{p/2}$-Lipschitz in the Gaussian coefficients, so for
every fixed $\delta>0$,
\[
    \P\left(
        \mathcal{E}_{n,p}(\mathcal C_n)
        \ge \E\mathcal{E}_{n,p}(\mathcal C_n)+\delta\sqrt n
    \right)
    \le
    \exp\left(-\frac{\delta^2n}{2\cdot3^p}\right).
\]
Taking $\delta$ sufficiently small and combining this
estimate with Lemma~\ref{lem:revmag-ground-benchmark}
by a union bound gives, for every sufficiently large
fixed $p$ depending only on $L$, a constant $c_{p,L}>0$
such that
$\P(\Delta^{\mathrm{opt}}_{n,p}(\mathcal C_n)
\ge c_{p,L}\sqrt n)\to1$.
\end{proof}

\subsection{Benchmark comparison at subquadratic profile complexity}
\label{subsec:subquadratic-benchmark-comparison}

In the subquadratic regime, the expected supremum of the residual Gaussian process is
$o(\sqrt n)$.  The comparison lemma therefore shows that $\mathcal C_n$
cannot outperform any conditioning benchmark $\mathcal B_n$ at the leading
$\sqrt n$ scale.  Unlike part~\textup{(i)}, this argument neither fixes
$\mathcal B_n=\Prod_n$ nor invokes the unrestricted ground-state benchmark.

\begin{proof}[Proof of
Theorem~\ref{thm:local-profile-complexity-principle}\textup{(ii)}]
Fix $p\ge2$, and let $\mathcal B_n$ be a conditioning
benchmark for $\mathcal C_n$.  Fix $0<\alpha<1/2$.
Definition~\ref{def:conditioning-benchmark} gives a constant
$A_{p,\alpha}$, which we may assume is at least one, and the approximation
hypothesis of Lemma~\ref{lem:profile-entropy-comparison} with $m_{\star}:=\lfloor\alpha n\rfloor$, and $r_n:=A_{p,\alpha}/ \sqrt n$, for all sufficiently large $n$.  In particular, $m_{\star}\le n-p$ and
$2r_n<1$ once $n$ is sufficiently large.

The entropy assumption and $|\mathcal J_{n,m_{\star}}|\le2^n$ imply
\begin{align*}
    &\int_0^{2r_n}
    \sqrt{
        \log\left(
            2|\mathcal J_{n,m_{\star}}|
            {N\left(
                \varepsilon,
                \mathcal F_{n,p}(\mathcal C_n),
                \|\cdot\|_{n,p}
            \right)}
        \right)
    }
    \,d\varepsilon
    \\
    &\qquad\le
    C r_n
    \sqrt{
        n+a_{n,p}
        +b_{n,p}\log\left(\frac{eB_{n,p}}{r_n}\right)
    }
    \le
    C\frac{A_{p,\alpha}}{\sqrt n}
    \sqrt{n+\Gamma_{n,p}}.
\end{align*}
For the first inequality, substitute $\varepsilon=2r_nu$ and use
$\int_0^1\sqrt{\log(1/u)}\,du<\infty$.  For the second, use
$A_{p,\alpha}\ge1$ and
$\log(eB_{n,p}/r_n)\le\log(eB_{n,p}n)$.
Consequently,
\[
    \frac{
        \E\Delta^{\mathrm{ref}}_{n,p}(\mathcal C_n;\mathcal B_n)
    }{\sqrt n}
    \le
    C A_{p,\alpha}\frac{\sqrt{n+\Gamma_{n,p}}}{n}
    +
    C3^{p/2}\sqrt{p\alpha}.
\]
Since $\Gamma_{n,p}=o(n^2)$, the first term tends to zero for each
fixed $p$ and $\alpha$.  It follows that
\[
    \limsup_{n\to\infty}
    \frac{
        \E\Delta^{\mathrm{ref}}_{n,p}(\mathcal C_n;\mathcal B_n)
    }{\sqrt n}
    \le
    C3^{p/2}\sqrt{p\alpha}.
\]
This holds for every $0<\alpha<1/2$.  Letting $\alpha\downarrow0$
proves convergence of the normalized expectation to zero.  Because
$\Delta^{\mathrm{ref}}_{n,p}(\mathcal C_n;\mathcal B_n)\ge0$, this is convergence in
$L^1$, and convergence in probability follows from Markov's inequality.
No uniform bound on $A_{p,\alpha}$ as $\alpha\downarrow0$ is needed,
since $n\to\infty$ is taken first.
\end{proof}

\section{Circuits and structured state classes}
\label{sec:structured-classes}
Throughout this section, parameters without subscript such
as $M,G,D,b,$ and $\chi$ are used in fixed-$n$ definitions, while
$M_n,G_n,D_n,b_n,$ and $\chi_n$ denote sequences in asymptotic statements.
For $d\ge1$, let $I_d$ denote the $d\times d$ identity matrix.  Write
$\mathsf{U}(d):=
\{U\in\mathbb C^{d\times d}:U^\dagger U=UU^\dagger=I_d\}$
for the unitary group.

\subsection{Circuit size}
\label{sec:gate-count}

Our first application of
Theorem~\ref{thm:local-profile-complexity-principle} is a lower bound in
terms of total circuit size.  This is perhaps the most direct illustration
of the local-profile complexity principle as we place no restriction on the
depth or geometry of the circuit, and allow an arbitrary
number of discardable ancillas.  

Throughout, the \emph{gate count} of a quantum circuit is the total number of one- and
two-qubit gates appearing in the circuit.

\begin{definition}
For integers $M\ge n$ and $G\ge0$, let
$\mathsf{Circ}^{(\le G)}_M$ denote the class of $M$-qubit circuits of
gate count at most $G$, with arbitrary one- and two-qubit gates and no
restriction on their arrangement.  We define
\[
    \mathsf{Gate}^{[G],\mathrm{anc}}_n
    :=
    \bigcup_{M\ge n}
    \left\{
        \Tr_{\{n+1,\ldots,M\}}
        \left(
            Q|0^M\rangle\langle0^M|Q^\dagger
        \right)
        :
        Q\in\mathsf{Circ}^{(\le G)}_M
    \right\}
\]
to be the class of all $n$-qubit states obtainable using at most $G$
one- and two-qubit gates, with an arbitrary number of discardable ancillas.
\end{definition}

The following result removes the dependence on the ambient ancilla
workspace.  It is the gate-count analogue of the backward-lightcone
pruning used for shallow circuits in
\cite[proof of Proposition~10]{al2026circuit}.

\begin{lemma}
\label{lem:gate-count-active-width}
Every state
$\rho\in\mathsf{Gate}^{[G],\mathrm{anc}}_n$ admits a representation
using a circuit of gate count at most $G$ on at most $n+G$ qubits.
\end{lemma}

\begin{proof}
By definition, for some $M\ge n$, there is a circuit
$Q=U_L\cdots U_1$ with $L\le G$ such that
\[
    \rho
    =
    \Tr_{\{n+1,\ldots,M\}}
    \left(
        Q|0^M\rangle\langle0^M|Q^\dagger
    \right).
\]
Let $E_j$ be the support of $U_j$, so $|E_j|\le2$.
Starting from $A_L=[n]$, define recursively, for
$j=L,\ldots,1$,
\[
    A_{j-1}
    :=
    \begin{cases}
        A_j\cup E_j, & E_j\cap A_j\ne\emptyset,\\
        A_j, & E_j\cap A_j=\emptyset.
    \end{cases}
\]
Thus $A_j$ is the set of qubits immediately after gate $j$
in the backward lightcone of the retained outputs.
A gate whose support meets $A_j$ can add at most one new
qubit, so $|A_{j-1}|\le|A_j|+1$ and
\[
    |A_0|\le n+L\le n+G.
\]
Let $\barQ$ be the pruned circuit obtained by retaining
precisely those gates $U_j$ for which
$E_j\cap A_j\ne\emptyset$, in their original order.
Every retained gate is supported in
$A_{j-1}\subseteq A_0$, so we regard $\barQ$ as a circuit
on $A_0$. It has at most $L\le G$ gates.

To verify that this pruning preserves the reduced output,
fix an observable $O$ on the retained qubits and define
recursively
\[
    O_L:=O\otimes I_{\{n+1,\ldots,M\}},
    \qquad
    O_{j-1}:=U_j^\dagger O_jU_j,
    \qquad j=L,\ldots,1.
\]
Backward induction gives
$\operatorname{supp}(O_j)\subseteq A_j$.
Indeed, if $E_j\cap A_j=\emptyset$, then $U_j$ commutes
with $O_j$, and hence $O_{j-1}=O_j$.
Otherwise, conjugation by $U_j$ can enlarge the support
only to $A_j\cup E_j=A_{j-1}$.
In particular, each deleted gate leaves the corresponding
observable unchanged. Therefore,
\[
    Q^\dagger
    (O\otimes I_{\{n+1,\ldots,M\}})
    Q
    =
    \left[
        \barQ^\dagger
        (O\otimes I_{A_0\setminus[n]})
        \barQ
    \right]\otimes I_{A_0^c},
\]
where $A_0^c=[M]\setminus A_0$.

Taking expectations in $|0^M\rangle$ shows that the original
and pruned circuits have the same expectation for every
observable on the retained qubits.
The input on $A_0$ is itself a product of $|0\rangle$-states.
Relabeling the qubits in $A_0$ so that the retained outputs
remain the first $n$ qubits therefore gives
\[
    \rho
    =
    \Tr_{\{n+1,\ldots,|A_0|\}}
    \left(
        \barQ|0^{|A_0|}\rangle
        \langle0^{|A_0|}|\barQ^\dagger
    \right).
\]
Since $|A_0|\le n+G$ and $\barQ$ has at most $G$ gates,
this proves the claim.
\end{proof}

We next cover the resulting circuits.  The gate discretization below is
the same operator-norm net and telescoping argument used in
\cite[proof of Proposition~10]{al2026circuit} (see also~\cite{zhao2024learning}). The difference here is that
we count arbitrary gate-support sequences rather than depth-$D$
matching architectures. We provide full details for completeness.

\begin{lemma}
\label{lem:gate-count-profile-entropy}
There is a universal constant $C<\infty$ such that, for every
$n\ge p\ge2$, $G\ge1$, and $0<\varepsilon<1$,
\begin{equation}
\label{eq:gate-count-profile-entropy}
    \log
    N\left(
        \varepsilon,
        \mathcal F_{n,p}
        \left(
            \mathsf{Gate}^{[G],\mathrm{anc}}_n
        \right),
        \|\cdot\|_{n,p}
    \right)
    \le
    CG
    \log\left(
        \frac{
            C\,3^{p/2}G(n+G)
        }{
            \varepsilon
        }
    \right).
\end{equation}
\end{lemma}

\begin{proof}
By Lemma~\ref{lem:gate-count-active-width}, after padding with
unused $|0\rangle$-qubits, every state in the class has a
representation
\[
    \rho
    =
    \Tr_{\{n+1,\ldots,n+G\}}(Q\rho_0Q^\dagger),
    \qquad
    \rho_0:=|0^{n+G}\rangle\langle0^{n+G}|,
\]
where $Q$ has at most $G$ gates.
We construct an operator-norm cover of these circuits and
then bound the corresponding error in their output profiles.

\textit{Covering the circuit class: }
By inserting one-qubit identity gates, we may represent
every circuit using exactly $G$ gate locations.
At each location, the gate acts on either one qubit or
an unordered pair of qubits, giving at most
\[
    (n+G)+\binom{n+G}{2}
    \le(n+G)^2
\]
possible supports.
There are therefore at most $(n+G)^{2G}$ support sequences.

As in \cite[proof of Proposition~10]{al2026circuit},
choose operator-norm $\delta$-nets for $\mathsf{U}(2)$
and $\mathsf{U}(4)$, with centers in the respective unitary
groups and cardinalities at most $(C/\delta)^C$.
For each support sequence, choosing a gate from the
appropriate net at every location gives at most
$(C/\delta)^{CG}$ circuits.
Taking the union over all support sequences produces
a finite set $\mathcal Q$ satisfying
\begin{equation}
\label{eq:gate-count-circuit-cover-size}
    |\mathcal Q|
    \le
    (n+G)^{2G}(C/\delta)^{CG}.
\end{equation}

For every circuit $Q=U_G\cdots U_1$, the set $\mathcal Q$
contains a circuit
$\widetilde Q=\widetilde U_G\cdots\widetilde U_1$
with the same gate supports and
$\|U_j-\widetilde U_j\|_{\mathrm{op}}\le\delta$.
Here each gate acts as the identity outside its support.
The telescoping identity
\[
    Q-\widetilde Q
    =
    \sum_{j=1}^G
        \widetilde U_G\cdots\widetilde U_{j+1}
        (U_j-\widetilde U_j)
        U_{j-1}\cdots U_1
\]
and the unitarity of the surrounding factors give
\[
    \|Q-\widetilde Q\|_{\mathrm{op}}
    \le
    \sum_{j=1}^G
        \|U_j-\widetilde U_j\|_{\mathrm{op}}
    \le G\delta.
\]
Thus $\mathcal Q$ is an operator-norm $G\delta$-cover
of the circuit class.

\textit{Transferring the circuit approximation to profiles: }
Let $\rho$ and $\widetilde\rho$ be the reduced outputs of
$Q$ and $\widetilde Q$ applied to $\rho_0$.
Using contractivity under partial trace and $\|\rho_0\|_1=1$,
we obtain
\begin{align*}
    \|\rho-\widetilde\rho\|_1
    &\le
    \|Q\rho_0Q^\dagger
        -\widetilde Q\rho_0\widetilde Q^\dagger\|_1
    \le
    \|(Q-\widetilde Q)\rho_0Q^\dagger\|_1
    +
    \|\widetilde Q\rho_0(Q^\dagger-\widetilde Q^\dagger)\|_1\\
    &\le
    2\|Q-\widetilde Q\|_{\mathrm{op}}
    \le2G\delta.
\end{align*}
Since every Pauli has operator norm one and
$|\mathcal P_{n,p}|=3^p\binom np$,
\[
    \|f_\rho-f_{\widetilde\rho}\|_{n,p}^2
    =
    \frac1{\binom np}
    \sum_{P\in\mathcal P_{n,p}}
        |\Tr(P(\rho-\widetilde\rho))|^2
    \le
    3^p\|\rho-\widetilde\rho\|_1^2.
\]
Consequently,
\begin{equation}
\label{eq:gate-count-profile-approximation}
    \|f_\rho-f_{\widetilde\rho}\|_{n,p}
    \le
    3^{p/2}\|\rho-\widetilde\rho\|_1
    \le
    2\cdot3^{p/2}G\delta.
\end{equation}

\textit{Combining the estimates: }
Choose $\delta=\varepsilon/(2\cdot3^{p/2}G).$ By \eqref{eq:gate-count-profile-approximation}, the output
profiles of the circuits in $\mathcal Q$ form an
$\varepsilon$-cover of
$\mathcal F_{n,p}(\mathsf{Gate}^{[G],\mathrm{anc}}_n)$.
Its cardinality is at most $|\mathcal Q|$, so
\eqref{eq:gate-count-circuit-cover-size} gives
\begin{align*}
    \log N\left(
        \varepsilon,
        \mathcal F_{n,p}
        \left(\mathsf{Gate}^{[G],\mathrm{anc}}_n\right),
        \|\cdot\|_{n,p}
    \right)
    &\le
    \log|\mathcal Q|\\
    &\le
    2G\log(n+G)
    +
    CG\log\left(
        \frac{C3^{p/2}G}{\varepsilon}
    \right)\\
    &\le
    CG\log\left(
        \frac{C3^{p/2}G(n+G)}{\varepsilon}
    \right),
\end{align*}

after increasing $C$.
This proves \eqref{eq:gate-count-profile-entropy}.
\end{proof}
We may now apply the general local-profile complexity principle.

\begin{proposition}[Circuit-size energy bounds]
\label{prop:gate-count-main}
Let $(G_n)_{n\ge1}$ be a fixed deterministic sequence of nonnegative integers.
\begin{enumerate}[label=\textup{(\roman*)}]
\item Suppose $G_n\le Ln^2/\log n$ for all sufficiently large $n$,
for some fixed $L<\infty$. Then, for every sufficiently large fixed $p$, with the lower threshold depending only on $L$, there is
$c_{p,L}>0$ such that
\[
    \P\left(
        \Delta^{\mathrm{opt}}_{n,p}
        \left(\mathsf{Gate}^{[G_n],\mathrm{anc}}_n\right)
        \ge c_{p,L}\sqrt n
    \right)
    \longrightarrow1.
\]

\item For every fixed $p\ge2$, if $G_n=o(n^2/\log n)$, then
\[
    \frac{
        \Delta^{\mathrm{ref}}_{n,p}
        \left(\mathsf{Gate}^{[G_n],\mathrm{anc}}_n;\Prod_n\right)
    }{\sqrt n}
    \longrightarrow0
\]
in $L^1$ and in probability.
\end{enumerate}
Both conclusions are uniform over the number of discardable ancillas.
\end{proposition}

\begin{proof}
Set $\widehat G_n:=\max\{G_n,1\}$.  Since
$\mathsf{Gate}^{[G_n],\mathrm{anc}}_n\subseteq
\mathsf{Gate}^{[\widehat G_n],\mathrm{anc}}_n$, Lemma~\ref{lem:gate-count-profile-entropy}, after enlarging the universal constant $C$ if necessary, allows us to take
\[
    a_{n,p}=0,
    \qquad
    b_{n,p}=C\widehat G_n,
    \qquad
    B_{n,p}=C3^{p/2}\widehat G_n(n+\widehat G_n)
\]
in Theorem~\ref{thm:local-profile-complexity-principle}.  Hence $\Gamma_{n,p}
    \le
    C\widehat G_n
    \log (
        eC3^{p/2}n\widehat G_n(n+\widehat G_n)
    )$.
Under the hypothesis of part~\textup{(i)}, for every fixed $p$,
$\limsup_{n\to\infty}\Gamma_{n,p}/n^2\le CL$, with a constant on the
right independent of $p$.  Theorem~\ref{thm:local-profile-complexity-principle}\textup{(i)}
therefore gives the first conclusion.  Under the hypothesis of
part~\textup{(ii)}, $\Gamma_{n,p}=o(n^2)$ for every fixed $p$.
Because $\Prod_n$ is a conditioning benchmark for every state class,
Theorem~\ref{thm:local-profile-complexity-principle}\textup{(ii)} gives the
second conclusion.
\end{proof}

\begin{proof}[Proof of Corollary~\ref{cor:intro-circuit-size}]
The two conclusions follow from
Proposition~\ref{prop:gate-count-main}\textup{(i)} and~\textup{(ii)},
respectively, applied to $\mathsf{Gate}^{[G_n],\mathrm{anc}}_n$.
That proposition permits arbitrary circuit depth and geometry and is uniform
over the number of discardable ancillas.
\end{proof}

\subsection{A quadratic-gate construction}
\label{subsec:quadratic-gate-construction}
In this section, we construct a class of states preparable using fewer
than $n^2$ two-qubit gates whose optimal energy exceeds the product-state
optimum by a positive multiple of $\sqrt n$, with high probability as
$n\to\infty$, for sufficiently large fixed $p$. Together with our circuit
lower bound, this shows that the gate-count threshold $n^2/\log n$ is
sharp up to a logarithmic factor. Our construction specializes the
commuting-rotation ansatz considered in \cite{anshu2021improved}.

\begin{definition}\label{def:class}
Let
\begin{equation}\label{eq:class}
 \mathcal C_n^{\operatorname{pair}}:=\left\{
   \rho_\theta
   =U_\theta\ket{0^n}\bra{0^n}U_\theta^\dagger
   :\theta= (\theta_{ij})_{1\le i<j\le n}\in\mathbb{R}^{\binom n2}
 \right\},
\end{equation}
where
\[
 U_\theta:=\prod_{1\le i<j\le n}
             \exp(-i\theta_{ij}Y_iY_j).
\]
\end{definition}

The class consists of states obtained from $\ket{0^n}$ by applying
a two-qubit gate to each pair. Since $(Y_iY_j)^2=I$, we may write
\[
 e^{-i\theta_{ij}Y_iY_j}
 =\cos(\theta_{ij})I-i\sin(\theta_{ij})Y_iY_j.
\]
Each factor therefore performs a joint rotation of qubits $i$ and $j$,
with angle $\theta_{ij}$, and can create entanglement between them. 

The following result shows that a deterministic class preparable with fewer than $n^2$ gates obtains an energy that beats the optimum over all product states by a positive multiple of $\sqrt n$, with high probability, for sufficiently large fixed $p$.

\begin{proposition}
\label{prop:beatingProds}
The class $\mathcal{C}_n^{\operatorname{pair}}$ in Definition~\ref{def:class} has the following properties.
\begin{enumerate}[label=\textup{(\roman*)}]
 \item Every state in $\mathcal{C}_n^{\operatorname{pair}}$ can be prepared exactly from $\ket{0^n}$
 using at most $\binom n2<n^2$ two-qubit gates and no ancillary qubits.

 \item For every fixed integer $p\ge2$, there is a state $\rho^\star_{n,p}\in\mathcal{C}_n^{\operatorname{pair}}$ such that
 \begin{equation}\label{eq:state-limit}
 \frac{\Tr(H_{n,p}\rho^\star_{n,p})}{\sqrt n}
 \longrightarrow\sqrt{\frac{p}{2e}}
 \qquad\text{in $L^2$ and in probability}.
 \end{equation}

 \item There exists an integer $p_0<\infty$ such that, for every fixed
 integer $p\ge p_0$, there is a constant $\delta_p>0$ for which
 \begin{equation}\label{eq:advantage}
 \P\left (\mathcal E_{n,p}(\mathcal{C}_n^{\operatorname{pair}})\ge
              \mathcal E_{n,p}(\Prod_n)+\delta_p\sqrt n\right)
 \longrightarrow1.
 \end{equation}
\end{enumerate}
\end{proposition}

\begin{proof}
 Part~\textup{(i)} follows since each factor $\exp(-i\theta_{ij}Y_iY_j)$ is a single two-qubit gate and there is one factor for each unordered pair. 

To show part~\textup{(ii)}, we construct $\rho^\star_{n,p}$ by choosing the corresponding angles by taking a single gradient-ascent step on the energy $\Tr(H_{n,p}\rho_\theta)$, starting from $\theta=0$, which corresponds to the product state $\ket{0^n}$. The step size is chosen to maximize the expected energy
of the resulting state.

Let $\ket{\psi_\theta} := U_\theta \ket{0^n}$, then $\ket{\psi_0}=\ket{0^n}$.  Since all the generators commute, differentiating gives $\partial_{\theta_{ij}}U_\theta=-iY_iY_jU_\theta$ and $\partial_{\theta_{ij}}U_\theta^\dagger=iU_\theta^\dagger Y_iY_j$. Therefore
\begin{align*}
 \partial_{\theta_{ij}} \Tr(H_{n,p}\rho_\theta)
 &=
 \partial_{\theta_{ij}} \bra{0^n}U_\theta^\dagger H_{n,p}U_\theta\ket{0^n}\\
 &=\bra{0^n}\bigl(iU_\theta^\dagger Y_iY_jH_{n,p}U_\theta
       -iU_\theta^\dagger H_{n,p}Y_iY_jU_\theta\bigr)\ket{0^n}\\
  &=\bra{\psi_\theta}i \bigl( Y_iY_jH_{n,p}
       - H_{n,p}Y_iY_j \bigr)\ket{\psi_\theta}.
\end{align*}
It follows that
\[
 \partial_{\theta_{ij}} \Tr(H_{n,p}\rho_\theta)|_{\theta=0}
 = {\binom{n}{p}}^{-1/2}
   \sum_{P\in\mathcal P_{n,p}} g_P
   \bra{0^n} i (Y_iY_j P - P Y_iY_j)\ket{0^n}.
\]

For $P\in\mathcal P_{n,p}$, let $T(P):= i(\bra{0^n} Y_iY_jP\ket{0^n}-\bra{0^n} PY_iY_j \ket{0^n})$. We claim that
\begin{align}\label{eq:TP}
 T(P)=
 \begin{cases}
  2, & P\in \{ X_iY_jZ_S, Y_iX_jZ_S\},\\
  0, & \text{otherwise},
 \end{cases}
\end{align}
where $S\subseteq[n]\setminus\{i,j\}$, $|S|=p-2$ and $Z_S := \prod_{k \in S} Z_k$.

To see this, note first that for any $ k \notin \{i,j\}$, both expectations in $T(P)$ contain the factor $\bra{0} P_k \ket{0} = \mathbf{1}_{P_k \in \{ I,Z\}}.$ Therefore, any factor outside $\{i,j\}$ must be $I$ or $Z$, and contributes $1$.

For $k = i$, the first expectation contains $\bra{0} Y P_i \ket{0}$ and the second contains $\bra{0} P_i Y \ket{0}$. If $P_i = I$, both are zero. If $P_i =Z$ then using $YZ = iX$ and $ZY = -iX$ gives that both expectations are zero. Applying the same argument to the case $k=j$ shows that $P_i, P_j \in \{X,Y\}.$ It follows that $P_i P_j \in \{ XX, XY, YX, YY\}$. We consider the four cases in turn. Note that the factors outside of $i,j$ must be taken to be $Z$ or $I$ and contribute 1. 
\begin{itemize}
    \item If $P_iP_j=XX$, then, using $YX=-iZ$,
    \[
        \bra{0^n}Y_iY_jP\ket{0^n}
        = \bra{0}YX\ket{0}\bra{0}YX\ket{0}
        = \bigl(\bra{0}(-iZ)\ket{0}\bigr)^2
        = -1.
    \]
    Similarly, $\bra{0^n}PY_iY_j\ket{0^n}=-1$, so $T(P)=0$.

    \item If $P_iP_j=XY$, then
    \[
        \bra{0^n}Y_iY_jP\ket{0^n}
        = \bra{0}YX\ket{0}\bra{0}YY\ket{0}
        = \bra{0}(-iZ)\ket{0}
        = -i.
    \]
    Also, $\bra{0^n}PY_iY_j\ket{0^n}=i$, so
    $T(P)=i(-i-i)=2$.

    \item If $P_iP_j=YX$, a similar argument to the $XY$ case
    gives $T(P)=2$.

    \item If $P_iP_j=YY$, then
    \[
        \bra{0^n}Y_iY_jP\ket{0^n}
        = \bra{0}Y^2\ket{0}\bra{0}Y^2\ket{0}
        = 1.
    \]
    Similarly, $\bra{0^n}PY_iY_j\ket{0^n}=1$, so $T(P)=0$.
\end{itemize}
This proves \eqref{eq:TP}. Consequently,
\[
 \partial_{\theta_{ij}} \operatorname{Tr}(H_{n,p} \rho_\theta)|_{\theta=0}
 =2 \binom{n}{p}^{-1/2}
 \sum_{\substack{S\subseteq[n]\setminus\{i,j\}\\|S|=p-2}}
       \bigl(g_{X_iY_jZ_S}+g_{Y_iX_jZ_S}\bigr).
\]
Let $b_{ij}:=\frac12\partial_{\theta_{ij}}
\operatorname{Tr}(H_{n,p}\rho_\theta)|_{\theta=0}$.
The sums for distinct unordered pairs $\{i,j\}$ involve
disjoint sets of Hamiltonian coefficients, since the locations
of the $X$ and $Y$ factors uniquely identify the pair.
Thus the $b_{ij}$ are independent centered Gaussians.
Each defining sum contains $2\binom{n-2}{p-2}$ independent
standard Gaussians, giving the common variance
\[
    v_{n,p}
    =\frac{2\binom{n-2}{p-2}}{\binom np}
    =\frac{2p(p-1)}{n(n-1)}.
\]
We set $\theta_{ij}(t)=t\sqrt n\,b_{ij}$, which is equivalent
to taking one gradient-ascent step from zero, with step size
$t\sqrt n/2$, where $t$ is a parameter to be chosen later.

We now compute the mean energy of the state produced by the gradient
step. Writing $f_\theta(P) := f_{\rho_\theta}(P) =\Tr(P\rho_\theta)$ and expanding the Hamiltonian, we obtain
\[
 m_{n,p}(t)
 :=\E\Tr(H_{n,p}\rho_{\theta(t)})
 ={\binom{n}{p}}^{-1/2}
   \sum_{P\in\mathcal P_{n,p}}
   \E[g_Pf_{\theta(t)}(P)].
\]
Note that the expectation is over the Hamiltonian coefficients, which also
determine the angles $\theta(t)$. We therefore need to evaluate
the terms $\E[g_Pf_{\theta(t)}(P)]$. By Stein's lemma, this is equivalent to evaluating $\E[\partial_{g_P} f_{\theta(t)}(P)].$

Only strings of the forms $X_iY_jZ_S$ and $Y_iX_jZ_S$,
where $i<j$, $S\subseteq[n]\setminus\{i,j\}$, and $|S|=p-2$,
can contribute to the mean. For any other string $\tilde P$,
the coefficient $g_{\tilde P}$ appears in none of the $b_{ab}$.
It is therefore independent of the entire angle vector
$\theta(t)$ and has mean zero, giving
$\E[g_{\tilde P}f_{\theta(t)}(\tilde P)]=0$.
We first consider $P=X_iY_jZ_S$; the other orientation is
treated below.

The coefficient $g_P$ appears only in $b_{ij}$, and hence only
in $\theta_{ij}(t)$. Thus, with $t$ fixed,
\[
 \frac{\partial\theta_{ij}(t)}{\partial g_P}
 =\frac{t\sqrt n}{\sqrt{\binom{n}{p}}},
 \qquad
 \frac{\partial\theta_{ab}(t)}{\partial g_P}=0
 \quad\text{for }(a,b)\ne(i,j).
\]
Since $f_{\theta(t)}(P)$ depends on the disorder only through
$\theta(t)$, the chain rule gives
\begin{align*}
 \partial_{g_P}\bigl(f_{\theta(t)}(P)\bigr)
 &=\nabla_\theta f_\theta(P)  \cdot |_{\theta = \theta(t)} \frac{\partial \theta(t)}{\partial g_P} \\
 &=\sum_{a<b}
   \left.
   \frac{\partial}{\partial\theta_{ab}}f_\theta(P)
   \right|_{\theta=\theta(t)}
   \frac{\partial\theta_{ab}(t)}{\partial g_P}\\
 &=\frac{t\sqrt n}{\sqrt{\binom{n}{p}}}
   \left.
   \frac{\partial}{\partial\theta_{ij}}f_\theta(P)
   \right|_{\theta=\theta(t)}.
\end{align*}

We next compute this derivative. By commutation of the generators defining the circuit $U_\theta$, we have $\frac{\partial}{\partial\theta_{ij}}\ket{\psi_\theta}
 =-iY_iY_j\ket{\psi_\theta}$, and therefore
\begin{align*}
 \frac{\partial}{\partial\theta_{ij}}f_\theta(P)
 &=
 \frac{\partial}{\partial\theta_{ij}} \bra{\psi_\theta}P\ket{\psi_\theta}
 =i\bra{\psi_\theta}Y_iY_jP\ket{\psi_\theta}
   -i\bra{\psi_\theta}PY_iY_j\ket{\psi_\theta}
 =\Tr\!\left(i[Y_iY_j,P]\rho_\theta\right),
\end{align*}
where $[A,B]=AB-BA.$ For our chosen string $P=X_iY_jZ_S$, the identities
$YX=-iZ$, $XY=iZ$, and $Y^2=I$ give
\[
 \begin{aligned}
 Y_iY_jP
 &=(Y_iX_i)(Y_jY_j)Z_S=-iZ_iZ_S,\\
 PY_iY_j
 &=(X_iY_i)(Y_jY_j)Z_S= iZ_iZ_S.
 \end{aligned}
\]
Consequently,
\[
 i[Y_iY_j,P]
 =i\bigl(-iZ_iZ_S-iZ_iZ_S\bigr)
 =2Z_iZ_S
 =2Z_{S\cup\{i\}},
\]
and hence $ \frac{\partial}{\partial\theta_{ij}}f_\theta(P)
 =2\Tr(Z_{S\cup\{i\}}\rho_\theta)$. The preceding calculations then yield
\begin{align}
 \frac{1}{\sqrt{\binom{n}{p}}}
 \E[g_Pf_{\theta(t)}(P)]
 &=\frac{1}{\sqrt{\binom{n}{p}}}
   \E\!\left[
     \partial_{g_P}\bigl(f_{\theta(t)}(P)\bigr)
   \right]
   \notag\\
 &=\frac{t\sqrt n}{\binom{n}{p}}
   \E\!\left[
     \left.
     \frac{\partial}{\partial\theta_{ij}}f_\theta(P)
     \right|_{\theta=\theta(t)}
   \right]
   \notag\\
 &=\frac{2t\sqrt n}{\binom{n}{p}}
   \E\Tr(Z_{S\cup\{i\}}\rho_{\theta(t)}).
 \label{eq:ibp}
\end{align}
For $P=Y_iX_jZ_S$, interchanging the roles of $i$ and $j$
gives $i[Y_iY_j,P]=2Z_{S\cup\{j\}}$, so the same calculation
holds with $Z_{S\cup\{i\}}$ replaced by $Z_{S\cup\{j\}}$.

Both sets $S\cup\{i\}$ and $S\cup\{j\}$ have cardinality $p-1$.
Moreover, since the angles are independent and identically distributed,
$\E\Tr(Z_W\rho_{\theta(t)})$ depends only on $|W|$. To see this, consider any permutation $\pi$ of the qubit labels.
The initial state $\ket{0^n}$ is unchanged by this permutation,
while the circuit is transformed by relabelling its angles, since the gate order is not a factor due to commutation. Since the angles are i.i.d., this relabelling preserves their
joint distribution. Consequently,
\[
 \E\Tr(Z_{\pi(W)}\rho_{\theta(t)})
 =\E\Tr(Z_W\rho_{\theta(t)}).
\]
Any two subsets of $[n]$ with the same cardinality are related
by such a permutation, so the expectation depends only on $|W|$. Summing over the
$p(p-1)\binom{n}{p}$ contributing strings therefore gives
\[
 m_{n,p}(t)
 =2tp(p-1)\sqrt n\,
   \E\Tr(Z_W\rho_{\theta(t)}),
 \qquad |W|=p-1.
\]
It remains to compute this expectation. For a general set $W\subseteq[n]$ of cardinality $k$, with independent angles $\theta_{ab}\sim N(0,\sigma^2)$,
\[
 \E\Tr(Z_W\rho_\theta)
 =\bra{0^n}\E[U_\theta^\dagger Z_WU_\theta]\ket{0^n}.
\]
Since $YZ = -ZY$, the generator $Y_aY_b$ anticommutes with $Z_W$ exactly when one of
$a,b$ lies in $W$ and the other does not. For such a crossing pair,
\[
 e^{i\theta_{ab}Y_aY_b}Z_We^{-i\theta_{ab}Y_aY_b}
 =\cos(2\theta_{ab})Z_W
  +i\sin(2\theta_{ab})Y_aY_bZ_W,
\]
and using $\E[\cos(2\theta_{ab})]=e^{-2\sigma^2}$ and $\E[\sin(2\theta_{ab})]=0$, gives 
\begin{align*}
    \E[e^{i\theta_{ab}Y_aY_b}Z_We^{-i\theta_{ab}Y_aY_b}] 
     = 
     e^{-2\sigma^2}Z_W. 
\end{align*}
A noncrossing pair commutes with $Z_W$ and leaves it unchanged.
Independence allows us to average over the gates successively.
There are $k(n-k)$ crossing pairs, so
\[
 \E[U_\theta^\dagger Z_WU_\theta]
 =e^{-2\sigma^2k(n-k)}Z_W.
\]
Therefore, 
\begin{align*}
    \E \operatorname{Tr}(Z_W \rho_{\theta})
    &=
    \E \operatorname{Tr}( U_\theta^\dagger Z_W U_\theta \ket{0^n}\bra{0^n} )
    = 
    \operatorname{Tr}( \E [U_\theta^\dagger Z_W U_\theta] \ket{0^n}\bra{0^n} )\\
    &= 
    e^{-2\sigma^2k(n-k)} \operatorname{Tr}( Z_W \ket{0^n}\bra{0^n} )= e^{-2\sigma^2k(n-k)},
\end{align*}
where the last line follows using $\bra{0^n}Z_W\ket{0^n}=1$.  For our choice of angles, we have $k=p-1$ and 
\[
 \sigma^2=\operatorname{var}(\theta_{ij}(t)) = t^2 n v_{n,p}
 =\frac{2t^2p(p-1)}{n-1},
\]
giving $m_{n,p}(t)$
\begin{equation}\label{eq:mean-t}
 m_{n,p}(t)
 =2tp(p-1)\sqrt n\,
  \exp\!\left(
   -\frac{4t^2p(p-1)^2(n-p+1)}{n-1}
  \right).
\end{equation}

Maximizing this expression over $t\ge0$ gives
\[
 t^\star_{n,p}
 =\sqrt{\frac{n-1}{8p(p-1)^2(n-p+1)}}.
\]
Writing $\rho_{n,p}^\star$ for the corresponding state, we obtain
\begin{equation}\label{eq:mean-opt}
 \E\Tr(H_{n,p}\rho^\star_{n,p})
 =\sqrt{\frac{p(n-1)}{2e(n-p+1)}}\,\sqrt n.
\end{equation}
Thus, for fixed $p$, the normalized mean converges to
$\sqrt{p/(2e)}$ as $n\to\infty$.
To establish the $L^2$ convergence in part~\textup{(ii)},
it remains to control the fluctuations around this mean.
We next show that
\[
 \operatorname{var}\!\left(\Tr(H_{n,p}\rho^\star_{n,p})\right)\le C_p
\]
for all sufficiently large $n$, where $C_p$ depends only on $p$.
The normalized variance therefore tends to zero, yielding
convergence in $L^2$ and hence in probability.

Let $g=(g_P)_{P\in\mathcal P_{n,p}}$ denote the Hamiltonian coefficients,
and write
\[
 F(g):=\Tr(H_{n,p}(g)\rho_{\theta^\star(g)})=\frac{1}{\sqrt{\binom{n}{p}}}
       \sum_{P\in\mathcal P_{n,p}}g_Pf_{\theta^\star(g)}(P).
\]
By the Gaussian Poincar\'e inequality, $\operatorname{var}(F)\le\E\|\nabla_gF\|_2^2$, and so it suffices to bound the expected squared-norm of the gradient of $F$ with respect to the disorder.

For each $Q\in\mathcal P_{n,p}$,
\[
 \frac{\partial F}{\partial g_Q}
 =
 \frac{f_{\theta^\star}(Q)}{\sqrt{\binom{n}{p}}}
 +
 \sum_{i<j}
 \left.
 \frac{\partial}{\partial\theta_{ij}} \operatorname{Tr}(H_{n,p} \rho_{\theta})
 \right|_{\theta=\theta^\star}
 \frac{\partial\theta^\star_{ij}}{\partial g_Q},
\]
where $H$ is held fixed when taking the derivative with respect to the angles.

Since the angles are linear functions of $g$, we may write
$\theta^\star= K g$ for a deterministic matrix $K$. Defining
\[
 a(\theta):=
 \left(\frac{f_\theta(P)}{\sqrt{\binom{n}{p}}}\right)
 _{P\in\mathcal P_{n,p}},
\]
we can express the preceding identity as
\begin{equation}\label{eq:chain}
 \nabla_gF
 =a(\theta^\star)
  +K^{\mathsf T}\nabla_\theta \operatorname{Tr}(H_{n,p} \rho_{\theta^\star}).
\end{equation}
We now bound these two terms.

For the first, using $|f_\theta(P)|\le1$ and
$|\mathcal P_{n,p}|=3^p\binom{n}{p}$ gives
\[
 \|a(\theta)\|_2^2
 =\frac{1}{\binom{n}{p}}
   \sum_{P\in\mathcal P_{n,p}}|f_\theta(P)|^2
 \le3^p.
\]

For the second, note that distinct angles depend on
disjoint sets of the Hamiltonian coefficients. Consequently, the rows
of $K$ have disjoint supports. Moreover, since the coefficients
$g_P$ are independent standard Gaussians, the squared norm of each
row is the variance of the corresponding angle:
\[
 \sigma_{n,p}^2
 :=\operatorname{var}(\theta^\star_{ij})
 =\frac{1}{4(p-1)(n-p+1)}.
\]
It follows that $KK^{\mathsf T}=\sigma_{n,p}^2I$, and in particular, for every vector $v$,
\[
 \|K^{\mathsf T}v\|_2^2
 =v^{\mathsf T}KK^{\mathsf T}v
 =\sigma_{n,p}^2\|v\|_2^2.
\]
Using this identity in \eqref{eq:chain}, together with
$\|x+y\|_2^2\le2\|x\|_2^2+2\|y\|_2^2$, gives
\[
 \|\nabla_gF\|_2^2
 \le2\cdot3^p+
    2\sigma_{n,p}^2
    \|\nabla_\theta\operatorname{Tr}(H_{n,p} \rho_{\theta^\star})\|_2^2.
\]

To bound the gradient norm term, note that as in the earlier derivative calculation,
\[
 \frac{\partial}{\partial\theta_{ij}} \operatorname{Tr}(H_{n,p} \rho_{\theta})
 =i\bra{\psi_\theta}[Y_iY_j,H_{n,p}]\ket{\psi_\theta}.
\]
Since both $H_{n,p}$ and $Y_iY_j$ are Hermitian, the two matrix
elements in the commutator satisfy
\[
 \bra{\psi_\theta}H_{n,p}Y_iY_j\ket{\psi_\theta}
 =
 \overline{
 \bra{\psi_\theta}Y_iY_jH_{n,p}\ket{\psi_\theta}
 }.
\]
Thus, the derivative is $i$ times the difference between a complex
number and its conjugate, giving
\[
 \left|
 \frac{\partial}{\partial\theta_{ij}}
 \Tr(H_{n,p}\rho_\theta)
 \right|
 \le2\left|
 \bra{\psi_\theta}Y_iY_jH_{n,p}\ket{\psi_\theta}
 \right|.
\]

To control the sum of these squared bounds, we show that the vectors
$\{Y_iY_j\ket{\psi_\theta}:i<j\}$ are orthonormal.
The inner product between the vectors indexed by $(i,j)$ and
$(k,\ell)$ is
\begin{align*}
 \bigl(Y_iY_j\ket{\psi_\theta}\bigr)^\dagger
 \bigl(Y_kY_\ell\ket{\psi_\theta}\bigr)
 &=
 \bra{0^n}U_\theta^\dagger
 Y_iY_jY_kY_\ell U_\theta\ket{0^n}
 =
 \bra{0^n}Y_iY_jY_kY_\ell\ket{0^n}.
\end{align*}
The last equality follows because each generator commutes with
$U_\theta$, allowing us to cancel $U_\theta^\dagger U_\theta=I$.

If $(i,j)\ne(k,\ell)$, the pairs either share exactly one site
or are disjoint. If they share one site, say $j=k$, then
\[
 Y_iY_jY_kY_\ell=Y_iY_j^2Y_\ell=Y_iY_\ell.
\]
Since $i\ne\ell$ and $\ket{0^n}$ is a product state,
\[
 \bra{0^n}Y_iY_\ell\ket{0^n}
 =\bra0Y\ket0\,\bra0Y\ket0=0.
\]
The same argument applies whichever site is shared.
If the pairs are disjoint, the expectation instead factorizes
into four such factors:
\[
 \bra{0^n}Y_iY_jY_kY_\ell\ket{0^n}
 =\bigl(\bra0Y\ket0\bigr)^4=0.
\]

Finally, if $(i,j)=(k,\ell)$, then
$Y_iY_jY_kY_\ell=Y_i^2Y_j^2=I$, so the expectation is one.
Thus the vectors $\{Y_iY_j\ket{\psi_\theta}:i<j\}$ are orthonormal.

Bessel's inequality, applied to $H_{n,p}\ket{\psi_\theta}$ and
this orthonormal family, therefore gives
\begin{align*}
 \|\nabla_\theta\operatorname{Tr}(H_{n,p} \rho_\theta)\|_2^2
 &\le4\sum_{i<j}
   \left|
   \bra{\psi_\theta}Y_iY_jH_{n,p}\ket{\psi_\theta}
   \right|^2
 \le4\|H_{n,p}\ket{\psi_\theta}\|_2^2
 \le4\|H_{n,p}\|_\op^2.
\end{align*}
It follows that 
\[
 \|\nabla_gF\|_2^2
 \le2\cdot3^p+8\sigma_{n,p}^2\|H_{n,p}\|_\op^2.
\]

We therefore only need an estimate for
$\E\|H_{n,p}\|_\op^2$. The Hamiltonian is a matrix Gaussian
series of dimension $2^n$, with variance parameter $3^p$. The matrix Gaussian tail inequality
\cite[Theorem~4.1]{tropp2012user} gives, for $u\ge0$,
\[
 \P\{\|H_{n,p}\|_\op>u\}
 \le\min\{1,2^{n+1}e^{-u^2/(2\cdot3^p)}\}.
\]
Set $s_0=2\cdot3^p(n+1)\log2$, and using the above bound gives
\begin{align*}
 \E\|H_{n,p}\|_\op^2
 &=\int_0^\infty
   \P\{\|H_{n,p}\|_\op^2>s\}\,ds\notag\\
 &\le s_0+
   \int_{s_0}^\infty2^{n+1}e^{-s/(2\cdot3^p)}\,ds\notag\\
 &=2\cdot3^p\bigl((n+1)\log2+1\bigr).
\end{align*}

Putting it together, we conclude that
\begin{align*}
 \operatorname{var}(F)
 &\le\E\|\nabla_gF\|_2^2\notag\\
 &\le2\cdot3^p+
   8\sigma_{n,p}^2\E\|H_{n,p}\|_\op^2\notag\\
 &\le2\cdot3^p+
   \frac{4\cdot3^p\bigl((n+1)\log2+1\bigr)}
        {(p-1)(n-p+1)}.
\end{align*}
For each fixed $p$, this remains bounded as $n\to\infty$. Combining it with the mean computed in \eqref{eq:mean-opt} gives
\[
 \E\left|
 \frac{F}{\sqrt n}-\sqrt{\frac{p}{2e}}
 \right|^2
 =
 \frac{\operatorname{var}(F)}{n}
 +
 \left(
 \sqrt{\frac{p(n-1)}{2e(n-p+1)}}-\sqrt{\frac{p}{2e}}
 \right)^2
 \longrightarrow0.
\]
This establishes part~\textup{(ii)}.

Finally, to show part~\textup{(iii)}, let $C_0$ denote the
constant in Lemma~\ref{lem:profile-product-benchmark}, and
fix $p$ sufficiently large so that
\[
    \delta_p:=\frac12\left(
        \sqrt{\frac{p}{2e}}-C_0\sqrt{\log(ep)}
    \right)>0.
\]
Let $A_n$ be the event that
\[
    \frac{\mathcal E_{n,p}(\Prod_n)}{\sqrt n}
    \le C_0\sqrt{\log(ep)}+\frac{\delta_p}{2}.
\]
By Lemma~\ref{lem:profile-product-benchmark},
$\P(A_n)\ge1-e^{-\delta_p^2n/8}$.
Let $B_n$ be the event that
$\Tr(H_{n,p}\rho^\star_{n,p})/\sqrt n
\ge\sqrt{p/(2e)}-\delta_p/2$.
Part~\textup{(ii)} gives $\P(B_n)\to1$, so a union bound yields
$\P(A_n\cap B_n)\ge1-\P(A_n^c)-\P(B_n^c)\to1$.

Since $\rho^\star_{n,p}\in\mathcal{C}_n^{\operatorname{pair}}$, on $A_n\cap B_n$ we have
\[
 \mathcal E_{n,p}(\mathcal{C}_n^{\operatorname{pair}})-\mathcal E_{n,p}(\Prod_n)
 \ge\Tr(H_{n,p}\rho^\star_{n,p})-\mathcal E_{n,p}(\Prod_n)
 \ge\delta_p\sqrt n.
\]
This proves \eqref{eq:advantage}.
\end{proof}

\subsection{Circuit depth and width}
\label{subsec:plain-shallow}

Proposition~\ref{prop:gate-count-main} yields substantially stronger depth
bounds when the circuit width is controlled.  With unrestricted ancillas,
one first restricts to the backward lightcone of the retained outputs,
as in \cite[Proposition~10]{al2026circuit}, and then invokes the same
gate-count theorem.  Thus no separate covering argument for shallow
circuits is required.

\begin{definition}[Shallow circuit classes]
For integers $M\ge n$ and $D\ge0$, let
$\mathsf{Circ}^{[D]}_M$ denote the class of $M$-qubit circuits of depth at
most $D$, where every layer consists of one- and two-qubit gates with
pairwise disjoint supports.  Define
\[
    \mathsf{shallow}^{[D]}_{n,M}
    :=
    \left\{
        \Tr_{\{n+1,\ldots,M\}}
        \left(
            Q|0^M\rangle\langle0^M|Q^\dagger
        \right):
        Q\in\mathsf{Circ}^{[D]}_M
    \right\}.
\]
We write
$\mathsf{shallow}^{[D]}_n:=\mathsf{shallow}^{[D]}_{n,n}$ for the
ancilla-free class, and set
\[
    \mathsf{shallow}^{[D],\mathrm{anc}}_n
    :=
    \bigcup_{M\ge n}
    \mathsf{shallow}^{[D]}_{n,M}.
\]
\end{definition}

The gate-count bound gives a continuous tradeoff between circuit depth
and ancillary workspace.  If a depth-$D$ circuit acts on $M$ qubits
but only $n$ outputs are retained, then at most
${\widehat M}:=\min\{M,n2^D\}$ qubits can lie in the backward
lightcone of the retained outputs.

\begin{proposition}[Depth--width energy bounds]
\label{prop:plain-shallow-width-depth-tradeoff}
Let $(M_n)_{n\ge1}$ and $(D_n)_{n\ge1}$ be fixed deterministic integer
sequences with $M_n\ge n$ and $D_n\ge0$, and set
${\widehat M}_n:=\min\{M_n,n2^{D_n}\}$.
\begin{enumerate}[label=\textup{(\roman*)}]
\item Suppose $D_n{\widehat M}_n\le Ln^2/\log n$ for all
sufficiently large $n$, for some fixed $L<\infty$. Then, for every
sufficiently large fixed $p$, with the lower threshold depending
only on $L$, there is $c_{p,L}>0$ such that
\[
    \P\left(
        \Delta^{\mathrm{opt}}_{n,p}
        \left(\mathsf{shallow}^{[D_n]}_{n,M_n}\right)
        \ge c_{p,L}\sqrt n
    \right)
    \longrightarrow1.
\]

\item For every fixed $p\ge2$, if
$D_n{\widehat M}_n=o(n^2/\log n)$, then
\[
    \frac{
        \Delta^{\mathrm{ref}}_{n,p}
        \left(\mathsf{shallow}^{[D_n]}_{n,M_n};\Prod_n\right)
    }{\sqrt n}
    \longrightarrow0
\]
in $L^1$ and in probability.
\end{enumerate}
\end{proposition}

\begin{proof}
Fix $n$ and a state
$\rho\in\mathsf{shallow}^{[D_n]}_{n,M_n}$,
prepared by a circuit $Q=Q_{D_n}\cdots Q_1$.
Let $E_{t,j}$ denote the gate supports in layer $Q_t$.
Starting from the retained qubits $A_{D_n}=[n]$, define recursively
\[
    A_{t-1}
    :=
    A_t\cup
    \bigcup_{j:E_{t,j}\cap A_t\ne\emptyset}E_{t,j},
    \qquad t=D_n,\ldots,1,
\]
so that $A_t$ collects the qubits that can affect the $n$ retained qubits through the remaining layers $Q_{t+1},\ldots, Q_{D_n}.$ The gate supports in each layer are disjoint and have size at most two. Thus at most $|A_t|$ gates meet $A_t$,
and each adds at most one qubit outside $A_t$, giving $|A_{t-1}|\le2|A_t|$. Since $A_0\subseteq[M_n]$, it follows that
\[
    |A_0|
    \le
    \min\{M_n,2^{D_n}|A_{D_n}|\}
    =
    \min\{M_n,n2^{D_n}\}
    =
    \hatM_n.
\]

Retain precisely those gates in layer $t$ whose supports
meet $A_t$.
By the observable argument in the proof of
Lemma~\ref{lem:gate-count-active-width}, deleting the other
gates preserves the reduced state on the retained outputs.
After relabeling the active qubits, the resulting circuit
$\barQ$ acts on $|A_0|$ qubits, starts from
$|0^{|A_0|}\rangle$, and prepares $\rho$ on its first $n$
output qubits.

The retained gates still have disjoint supports within each
layer, so each layer contains at most $|A_0|$ gates.
Consequently, $\barQ$ has at most $D_n|A_0| \le D_n\hatM_n$ gates. Setting $G_n:=D_n\hatM_n$, we obtain
\begin{equation}
\label{eq:shallow-gate-count-inclusion}
    \mathsf{shallow}^{[D_n]}_{n,M_n}
    \subseteq
    \mathsf{Gate}^{[G_n],\mathrm{anc}}_n.
\end{equation}
This also holds when $D_n=0$, in which case $A_0=[n]$
and the pruned circuit has no gates.

By \eqref{eq:shallow-gate-count-inclusion} and monotonicity
with respect to the state class,
\[
    \Delta^{\mathrm{opt}}_{n,p}
    \left(\mathsf{shallow}^{[D_n]}_{n,M_n}\right)
    \ge
    \Delta^{\mathrm{opt}}_{n,p}
    \left(\mathsf{Gate}^{[G_n],\mathrm{anc}}_n\right),
\]
and
\[
    0
    \le
    \Delta^{\mathrm{ref}}_{n,p}
    \left(\mathsf{shallow}^{[D_n]}_{n,M_n};\Prod_n\right)
    \le
    \Delta^{\mathrm{ref}}_{n,p}
    \left(\mathsf{Gate}^{[G_n],\mathrm{anc}}_n;\Prod_n\right).
\]
Under the hypothesis of part~\textup{(i)},
$G_n\le Ln^2/\log n$ eventually.
Proposition~\ref{prop:gate-count-main}\textup{(i)}
and the first inequality therefore give part~\textup{(i)}. Under the hypothesis of part~\textup{(ii)},
$G_n=o(n^2/\log n)$.
Proposition~\ref{prop:gate-count-main}\textup{(ii)}
and the second inequality give the claimed convergence
in $L^1$ and in probability.
\end{proof}

\begin{proof}[Proof of Corollary~\ref{cor:intro-shallow-depth}]
Apply Proposition~\ref{prop:plain-shallow-width-depth-tradeoff} with the
same sequences $D_n,M_n$ and effective width
${\widehat M}_n=\min\{M_n,n2^{D_n}\}$.
Its parts~\textup{(i)} and~\textup{(ii)} give the ground-state
separation and product-state comparison, respectively, at the stated
thresholds.
\end{proof}

\begin{remark}[Comparison with \cite{al2026circuit}]
\cite[Theorem~13]{al2026circuit} proves a comparison with product states,
uniformly over an arbitrary number of ancillas, under the condition
$4^{D_n}D_n(\log n+D_n)=o(n)$.  The present argument gives the same
product-state comparison under the weaker depth restriction $D_n2^{D_n}
    \log(e n(D_n+1)2^{D_n})
    =
    o(n)$, which in particular permits
$D_n\le c\log_2n$ for every fixed $c\in(0,1)$.

The improvement in the accessible depth range does not come from a
stronger backward-lightcone reduction, since we use the same pruning
mechanism developed in \cite{al2026circuit}.  The difference is how the
reduction is used.  The argument in \cite{al2026circuit} controls a
lightcone-localized residual process, whose variance introduces an
additional $2^D$ cost.  Here pruning is used only to bound the
effective gate count, after which the local-profile complexity principle
gives both the product-state comparison and ground-state separation.

For circuits of linear width, pruning is unnecessary altogether, and
the gate-count theorem yields the much larger depth range
$D_n=o(n/\log n)$, with ground-state separation continuing to the
critical scale $D_n=O(n/\log n)$ for sufficiently large $p$.
\end{remark}

\subsection{Block-product states}
\label{subsec:block-product}

The preceding applications constrain states through the resources used to
prepare them.  We now consider a different restriction, based only on the
extent of entanglement.  Within a block of qubits the state may be
completely arbitrary and may require unbounded circuit depth or gate
count.  The only requirement is that the global state factors across
blocks of bounded size.

\begin{definition}[Block-product states]
For integers $1\le b\le n$, let $\mathfrak B_{n,b}$ be the set of partitions
$\pi$ of $[n]$ into nonempty blocks $B$ with $|B|\le b$.  Define
\[
    \BlockProd^{[b]}_n
    :=
    \bigcup_{\pi\in\mathfrak B_{n,b}}
    \left\{
        \bigotimes_{B\in\pi}
        |\psi_B\rangle\langle\psi_B|
        :
        |\psi_B\rangle\in(\mathbb C^2)^{\otimes |B|},
        \ \|\psi_B\|_2=1
        \ \text{for all } B\in\pi
    \right\}.
\]
Thus $\BlockProd^{[b]}_n$ consists of pure states whose entanglement is
confined to blocks of at most $b$ qubits.  The convex hull of
$\BlockProd^{[b]}_n$ is the usual class of mixed $b$-producible
states.
\end{definition}

There is no restriction on the complexity of the state within an
individual block.  Moreover, the block-product restriction is stable
under arbitrary discardable ancillary workspace.

\begin{lemma}[Stability under discardable ancillas]
\label{lem:block-product-ancillas}
Suppose a joint state of the $n$ physical qubits and an arbitrary
number of ancillas factors across blocks, each containing at most $b$
physical qubits.  After discarding the ancillas, the resulting state on
the physical qubits belongs to
$\operatorname{conv}(\BlockProd_n^{[b]})$.
\end{lemma}

\begin{proof}
Tracing out the ancillas preserves the tensor-product decomposition
between blocks.  Each remaining factor is a mixed state on at most $b$
physical qubits.  Decomposing these factors into pure states expresses
the resulting physical state as a convex combination of elements of
$\BlockProd_n^{[b]}$.
\end{proof}

We next bound the local-profile complexity of the class.  There are two
sources of entropy: the choice of the partition and the states assigned
to its blocks.

\begin{lemma}[Profile entropy of block-product states]
\label{lem:block-product-profile-entropy}
There is a universal constant $C<\infty$ such that, for every
$n\ge p\ge2$, integer $1\le b\le n$, and $0<\varepsilon<1$,
\begin{equation}
\label{eq:block-product-profile-entropy}
    \log
    N\left(
        \varepsilon,
        \mathcal F_{n,p}\left(\BlockProd^{[b]}_n\right),
        \|\cdot\|_{n,p}
    \right)
    \le
    n\log n
    +
    \frac{Cn2^b}{b}
    \log\left(
        \frac{C\,3^{p/2}n}{\varepsilon}
    \right).
\end{equation}
\end{lemma}

\begin{proof}
Fix $\pi\in\mathfrak B_{n,b}$.  For each $B\in\pi$, the unit sphere
of $(\mathbb C^2)^{\otimes |B|}$ has a $\delta$-net of cardinality
at most $(C/\delta)^{C2^{|B|}}$.  For unit vectors
$\ket{\psi},\ket{\phi}$, $\|
        |\psi\rangle\langle\psi|
        -
        |\phi\rangle\langle\phi|
    \|_1
    \le
    2\|\psi-\phi\|_2$. Therefore, a Euclidean $\delta/2$-net of the unit sphere induces a
trace-norm $\delta$-net of the pure states, so
\[
    N\left(
        \delta,
        \{\text{pure states on }B\},
        \|\cdot\|_1
    \right)
    \le
    \left(\frac{C}{\delta}\right)^{C2^{|B|}}.
\]

For
$\rho=\bigotimes_{B\in\pi}\rho_B$ and
$\widetilde\rho=\bigotimes_{B\in\pi}\widetilde\rho_B$, tensor
telescoping gives
\[
    \|\rho-\widetilde\rho\|_1
    \le
    \sum_{B\in\pi}
    \|\rho_B-\widetilde\rho_B\|_1.
\]
Moreover, $\|f_\rho-f_{\widetilde\rho}\|_{n,p}
    \le
    3^{p/2}\|\rho-\widetilde\rho\|_1$. Since $|\pi|\le n$, choosing
$\delta=\varepsilon/(3^{p/2}n)$ on each block gives an
$\varepsilon$-net in the canonical metric.  For positive integers $j$,
$2^j/j$ is nondecreasing.  Since $1\le |B|\le b$ and
$\sum_{B\in\pi}|B|=n$, we have
\[
    \sum_{B\in\pi}2^{|B|}
    \le\frac{2^b}{b}\sum_{B\in\pi}|B|
    =\frac{n2^b}{b}.
\]
For the fixed partition, the logarithmic cardinality is therefore at most
\[
    C
    \sum_{B\in\pi}2^{|B|}
    \log\left(
        \frac{C\,3^{p/2}n}{\varepsilon}
    \right)
    \le
    \frac{Cn2^b}{b}
    \log\left(
        \frac{C\,3^{p/2}n}{\varepsilon}
    \right).
\]

Finally,
$|\mathfrak B_{n,b}|\le n^n$, since assigning one of $n$ block
labels to each site only enlarges the number of admissible partitions.
Taking the union over $\pi\in\mathfrak B_{n,b}$ proves
\eqref{eq:block-product-profile-entropy}.
\end{proof}

We may therefore apply the local-profile complexity principle without
placing any restriction on how the states inside the blocks are
prepared.

\begin{proposition}[Block-product energy bounds]
\label{prop:block-product-main}
Let $(b_n)_{n\ge1}$ be a fixed deterministic sequence of integers with
$1\le b_n\le n$.
\begin{enumerate}[label=\textup{(\roman*)}]
\item Suppose $2^{b_n}/b_n\le Ln/\log n$ for all sufficiently
large $n$, for some fixed $L<\infty$. Then, for every sufficiently
large fixed $p$, with the lower threshold depending only on $L$,
there is $c_{p,L}>0$ such that
\[
    \P\left(
        \Delta^{\mathrm{opt}}_{n,p}
        \left(\operatorname{conv}(\BlockProd^{[b_n]}_n)\right)
        \ge c_{p,L}\sqrt n
    \right)
    \longrightarrow1.
\]

\item For every fixed $p\ge2$, if $2^{b_n}\log n/b_n=o(n)$, then
\[
    \frac{
        \Delta^{\mathrm{ref}}_{n,p}
        \left(\operatorname{conv}(\BlockProd^{[b_n]}_n);\Prod_n\right)
    }{\sqrt n}
    \longrightarrow0
\]
in $L^1$ and in probability.
\end{enumerate}
Both conclusions remain valid when each block may use an arbitrary number
of discardable ancillas.
\end{proposition}

\begin{proof}
Lemma~\ref{lem:block-product-profile-entropy} allows us to take
\[
    a_{n,p}=n\log n,
    \qquad
    b_{n,p}=Cn2^{b_n}/b_n,
    \qquad
    B_{n,p}=C3^{p/2}n
\]
for the pure block-product class in
Theorem~\ref{thm:local-profile-complexity-principle}.  Hence
\[
    \Gamma_{n,p}
    \le
    n\log n
    +
    \frac{Cn2^{b_n}}{b_n}
    \log\left(eC3^{p/2}n^2\right).
\]
Under the hypothesis of part~\textup{(i)},
$\limsup_{n\to\infty}\Gamma_{n,p}/n^2\le CL$ for every fixed $p$,
with a constant independent of $p$.  Theorem~\ref{thm:local-profile-complexity-principle}\textup{(i)}
therefore gives ground-state separation.  Under the hypothesis of
part~\textup{(ii)}, $\Gamma_{n,p}=o(n^2)$ for every fixed $p$, so
Theorem~\ref{thm:local-profile-complexity-principle}\textup{(ii)}, with
benchmark $\mathcal B_n=\Prod_n$, gives the product-state comparison.

Passing to the convex hull does not change the maximal energy, since the objective is linear.  Finally, the extension to arbitrary
discardable ancillas follows from
Lemma~\ref{lem:block-product-ancillas}.
\end{proof}

\begin{proof}[Proof of Corollary~\ref{cor:intro-entanglement-depth}]
The class of $b_n$-producible states is
$\operatorname{conv}(\BlockProd^{[b_n]}_n)$.
Both conclusions therefore follow from
Proposition~\ref{prop:block-product-main}.
The extension to arbitrary discardable ancillas within each block is
included there by Lemma~\ref{lem:block-product-ancillas}.
\end{proof}

Proposition~\ref{prop:block-product-main} proves
Corollary~\ref{cor:intro-entanglement-depth} at the exact thresholds stated
there.  In particular, if $b_n\le c\log_2 n$ for some fixed
$c\in(0,1)$, then $2^{b_n}\log n/b_n=o(n)$, so both conclusions apply.

\subsection{Matrix product states}
\label{subsec:tensor-networks}

We now consider a class of tensor network states of a given bond dimension $\chi$.
We focus on matrix product states (MPS) with open boundary conditions
\cite{schon2005sequential,perez2006matrix}.
These states can possess
entanglement which can extend across the full system.  
An open boundary MPS is specified by an ordering of the sites and by bond dimensions $r_0=r_n=1$ and $1\le r_j\le\chi$ for $j \in [n-1]$.
For every site $j \in [n]$ and index $x \in \{0,1\}$, the MPS is parameterized by a collection of matrices $A=(A_x^{[j]})_{j,x}$ where $A_x^{[j]}\in\mathbb C^{r_{j-1}\times r_j}$. The corresponding state is
\begin{equation}
\label{eq:mps-representation}
    |\psi(A)\rangle
    =
    \sum_{x_1,\ldots,x_n\in\{0,1\}}
    A_{x_1}^{[1]}
    A_{x_2}^{[2]}
    \cdots
    A_{x_n}^{[n]}
    |x_1\cdots x_n\rangle .
\end{equation}

For $\chi\ge1$, let $\mathsf{MPS}^{[\chi]}_n$ be the class of pure $n$-qubit states whose Schmidt rank (bond dimension) across every contiguous cut $\{1,\ldots,j\}:\{j+1,\ldots,n\}$ is at most $\chi$.  
Unlike block-product states, this class permits correlations across the entire chain; for example, the GHZ state has bond dimension two.

\begin{lemma}[Profile entropy of matrix product states]
\label{lem:mps-profile-entropy}
There is a universal constant $C<\infty$ such that, for every
$n\ge p\ge2$, $\chi\ge1$, and $0<\varepsilon<1$,
\begin{equation}
\label{eq:mps-profile-entropy}
\log N\left(
    \varepsilon,
    \mathcal F_{n,p}\left(\mathsf{MPS}^{[\chi]}_n\right),
    \|\cdot\|_{n,p}
\right)
\le
C n\chi^2
\log\left(
    \frac{C\,3^{p/2}n\chi}{\varepsilon}
\right).
\end{equation}
The same estimate holds with an added $O(n\log n)$ term if the ordering of
the physical qubits is also allowed to vary.
\end{lemma}

\begin{proof}
We construct a cover by discretizing a canonical MPS
representation, following the approach of
Aharonov, Arad, and Irani \cite{aharonov2010efficient}
and Schuch and Cirac \cite{schuch2010matrix},
and then transfer the approximation to the Pauli-profile norm. Assume first a fixed ordering of the qubits and a sequence of bond dimensions $r_1, \dots, r_n$ where $r_0=r_n=1$ and $r_j\le\chi$.
The Schmidt rank across a cut at site $j$ is at most $r_j$.

Every state in $\mathsf{MPS}^{[\chi]}_n$ can be put into a canonical form where the state is generated by a sequence of isometries (see \cite[Theorem~1 and Section~5.1]{perez2006matrix})
\[
    V_j:\mathbb C^{r_{j-1}}
    \longrightarrow
    \mathbb C^2\otimes\mathbb C^{r_j},
    \qquad
    r_0=r_n=1,
    \quad
    r_j\le\chi.
\]
More specifically, in this canonical form, matrices may be chosen so that $\sum_{x=0}^1 A_x^{[j]}A_x^{[j]\dagger} = I_{r_{j-1}}$.

We now construct a $\delta$-net in operator norm for the possible isometries at each site.  The matrix $V_j$ is a $(2r_j)\times r_{j-1}$ matrix, and therefore belongs to a real vector space of dimension $d_j\le 4\chi^2$. Furthermore, $\|V_j\|_F \le  \sqrt\chi$ since columns of $V_j$ are orthonormal.
Let $\mathcal I_j$ denote the set of admissible isometries at site $j$,
viewed as a subset of this real Euclidean space.  Choose a maximal
$\delta$-separated subset $\mathcal N_j\subseteq\mathcal I_j$ in
Frobenius norm so that $\mathcal N_j$ is a $\delta$-net.
Using the standard packing and covering bounds from~\cite[Lemma~4.2.8 and Proposition~4.2.12]{vershynin2018high}, for $0<\delta<1$,
\[
\begin{aligned}
    |\mathcal N_j| &\le \left(1+\frac{2\sqrt\chi}{\delta}\right)^{d_j} \le \left( \frac{C\sqrt\chi}{\delta}\right)^{4\chi^2}.
\end{aligned}
\]
Since $\|A\|_{\mathrm{op}}\le\|A\|_F$, this is also a
$\delta$-net in operator norm.

If two sequential representations satisfy
$\|V_j-\widetilde V_j\|_{\mathrm{op}}\le\delta$ at every site, telescoping
through the isometries gives
\[
    \bigl\||\psi\rangle-|\widetilde\psi\rangle\bigr\|_2
    \le n\delta,
    \qquad
    \bigl\|
      |\psi\rangle\langle\psi|
      -|\widetilde\psi\rangle\langle\widetilde\psi|
    \bigr\|_1
    \le2n\delta.
\]
Since $\|f_\rho-f_\sigma\|_{n,p}
    \le3^{p/2}\|\rho-\sigma\|_1$, we choose $\delta=\varepsilon/(2\,3^{p/2}n)$.  Multiplying the tensor nets over
all sites gives logarithmic cardinality
$Cn\chi^2\log(C3^{p/2}n\chi/\varepsilon)$.  Finally, there are at most
$\chi^{n-1}$ choices of dimension in the sequences, and at most $n!$ possible site orderings if these are included.  Absorbing these additional terms proves the result.
\end{proof}

\begin{proposition}[Energy bounds for MPS]
\label{prop:mps-energy}
Let $(\chi_n)_{n\ge1}$ be a fixed deterministic sequence of bond dimensions.
\begin{enumerate}[label=\textup{(\roman*)}]
\item Suppose $\chi_n^2\log(n\chi_n)\le Ln$ for all sufficiently
large $n$, for some fixed $L<\infty$. Then there exists $p_0=p_0(L)$ such that, for every fixed $p\ge p_0$, there is $c_{p,L}>0$ for which
\[
    \P\left(
        \Delta^{\mathrm{opt}}_{n,p}
        \left(\mathsf{MPS}^{[\chi_n]}_n\right)
        \ge c_{p,L}\sqrt n
    \right)
    \longrightarrow1.
\]

\item For every fixed $p\ge2$, if $\chi_n^2\log(n\chi_n)=o(n)$,
then
\[
    \frac{
        \Delta^{\mathrm{ref}}_{n,p}
        \left(\mathsf{MPS}^{[\chi_n]}_n;\Prod_n\right)
    }{\sqrt n}
    \longrightarrow0
\]
in $L^1$ and in probability.
\end{enumerate}
The same conclusions hold for the convex hull of the MPS class and for the
union over all one-dimensional orderings.
\end{proposition}

\begin{proof}
Lemma~\ref{lem:mps-profile-entropy} gives $\Gamma_{n,p}
    \le
    Cn\chi_n^2
    \log\left(eC3^{p/2}n^2\chi_n\right)$. If $\chi_n^2\log(n\chi_n)\le Ln$, then, for every fixed $p$,
$\limsup_{n\to\infty}\Gamma_{n,p}/n^2\le CL$, where the constant is
independent of $p$; the fixed factor $3^{p/2}$ contributes only an additive
constant inside the logarithm.  This proves part~\textup{(i)} by
Theorem~\ref{thm:local-profile-complexity-principle}\textup{(i)}.  If
$\chi_n^2\log(n\chi_n)=o(n)$, then $\Gamma_{n,p}=o(n^2)$ for every
fixed $p$, and Theorem~\ref{thm:local-profile-complexity-principle}\textup{(ii)},
with benchmark $\mathcal B_n=\Prod_n$, proves part~\textup{(ii)}.
Optimizing over the convex hull does not change the optimum of the
Hamiltonian objective, and the additional $O(n\log n)$ entropy from varying the ordering of sites is negligible in both regimes.
\end{proof}

\begin{proof}[Proof of Corollary~\ref{cor:intro-mps}]
Both conclusions follow from Proposition~\ref{prop:mps-energy}.
Its final assertion also includes convex mixtures of the MPS class and
the union over all orderings of the physical qubits, so the ordering may
vary with the state.
\end{proof}

\begin{remark}[MERA and other isometric tensor networks]
The same telescoping argument applies to other
tensor networks with fixed architectures.  A binary one-dimensional MERA of bond dimension $\chi$ has
$O(n)$ tensors, including its normalized top tensor, while
each disentangler or isometry lies in an ambient matrix
space of real dimension at most $O(\chi^4)$ \cite{vidal2008class}. Write $\mathsf{MERA}^{[\chi]}_n$ for the class of normalized states
represented by such networks with a fixed architecture and bond
dimension at most $\chi$.  Consequently,
\[
\log N\left(
    \varepsilon,
    \mathcal F_{n,p}\left(\mathsf{MERA}^{[\chi]}_n\right),
    \|\cdot\|_{n,p}
\right)
\le
C n\chi^4
\log\left(
    \frac{C\,3^{p/2}n\chi}{\varepsilon}
\right).
\]
For fixed bond dimension, MERA therefore have no leading-order
energy advantage over product states on the $\sqrt n$ scale.
More generally, the preceding energy bounds extend with
$\chi_n^2$ replaced by $\chi_n^4$.
\end{remark}
\section{Magic hierarchy} 
\label{sec:proofs-magic}

We now apply the local-profile complexity principle to circuits that
alternate unrestricted Clifford operations with shallow quantum circuits.
The shallow blocks supply the non-stabilizer resource, or \emph{magic}
\cite{bravyi2005universal,veitch2014resource,parham2026quantum}.

In Subsection~\ref{sec:higher-magic-levels}, we prove ground-state
separation for both orientations at every fixed level of the magic
hierarchy under a linear-width assumption, allowing linearly many
discardable ancillas.  We then specialize to a single round and
remove the restriction on the number of ancillas.
For forward magic, Subsection~\ref{subsec:forward-profile-complexity}
establishes a sharper profile-entropy bound that allows arbitrarily
many discardable ancillas.
For reverse magic, Subsection~\ref{sec:revmag-clean} establishes
ground-state separation at fixed depth with arbitrarily many clean
ancillas.

\subsection{Stabilizer background and magic classes}
\label{subsec:magic-background}

We first recall the Pauli and stabilizer notation and define the magic
classes used in the subsequent results.  Appendix~\ref{sec:pauli-preliminaries} provides further explanation and algebraic background on the Pauli and stabilizer formalism.

In this section, we need Pauli operators of arbitrary weight, rather than
only the weight-$p$ strings $\mathcal P_{n,p}$ appearing in the Hamiltonian.
We write $\mathcal P_M:=\{I,X,Y,Z\}^{\otimes M}$ for the set of
phase-free Pauli operators on $M$ qubits.  Thus
$\mathcal P_{M,p}\subseteq\mathcal P_M$ consists of those elements that
act nontrivially on exactly $p$ qubits.

Let $\mathbb F_2=\{0,1\}$ be the field with arithmetic modulo two.
For a Pauli operator $P$, write $[P]=(x,z)\in\mathbb F_2^{2M}$
for its binary label, ignoring its phase, with the single-qubit
correspondences $I\leftrightarrow(0,0)$, $X\leftrightarrow(1,0)$,
$Z\leftrightarrow(0,1)$, and $Y\leftrightarrow(1,1)$.
For $u=(x,z)$ and $v=(x',z')$, define the symplectic inner product
\[
    \langle u,v\rangle_{\mathrm{sp}}
    :=x\cdot z'+z\cdot x'\pmod2.
\]
Two Paulis commute exactly when this product is zero.
A binary subspace $W$ is \emph{isotropic} if the product vanishes
between every pair of its elements.  We write
$W^\perp:=\{v:\langle v,w\rangle_{\mathrm{sp}}=0
\text{ for every }w\in W\}$ and
$v^\perp:=(\operatorname{span}\{v\})^\perp$.
All spans and dimensions of binary Pauli spaces are over $\mathbb F_2$.

$\mathsf{U}(1)=\{e^{i\theta}:\theta\in\mathbb R\}$ is the group of
complex phases.  We use the projective $M$-qubit Clifford group
\[
    \Cliff_M
    :=
    \left\{
        U\in\mathsf{U}(2^M):
        UPU^\dagger\in\{\pm R:R\in\mathcal P_M\}
        \text{ for every }P\in\mathcal P_M
    \right\}\big/\mathsf{U}(1).
\]
The quotient consists of equivalence classes
$[U]=\{e^{i\theta}U:\theta\in\mathbb R\}$, so unitaries differing by
a global phase are identified.  Whenever a Clifford acts on a state,
we choose any unitary representative of its class.  This choice does
not affect the output density matrix, since $(e^{i\theta}U)\rho(e^{i\theta}U)^\dagger
    =
    U\rho U^\dagger$.
We write
\[
    \Stab_M
    :=
    \left\{
        C|0^M\rangle\langle0^M|C^\dagger:
        C\in\Cliff_M
    \right\}
\]
for the class of pure $M$-qubit stabilizer states; see, for example,
\cite{gottesman1997stabilizer,gottesman1998heisenberg}.
We also write $|\psi\rangle\in\Stab_M$ when $|\psi\rangle\langle\psi|\in\Stab_M$.

Throughout this section, parameters with subscript $n$ denote fixed deterministic
integer sequences, with $D_n\ge1$ and $M_n\ge n$.

\begin{definition}[Forward and reverse magic]
\label{def:forward-reverse-magic}
For integers $M\ge1$ and $D\ge1$, define
\[
    \mathsf{mag}^{[D]}_M
    :=
    \left\{
        QC|0^M\rangle\langle0^M|C^\dagger Q^\dagger:
        C\in\Cliff_M,\;
        Q\in\mathsf{Circ}^{[D]}_M
    \right\}
\]
and
\[
    \RevMag^{[D]}_M
    :=
    \left\{
        CQ|0^M\rangle\langle0^M|Q^\dagger C^\dagger:
        C\in\Cliff_M,\;
        Q\in\mathsf{Circ}^{[D]}_M
    \right\}.
\]
For $M\ge n$, define their reduced-output classes by
\[
    \begin{aligned}
        \mathsf{mag}^{[D]}_{n,M-n}
        &:=
        \left\{
            \Tr_{\{n+1,\ldots,M\}}(\rho):
            \rho\in\mathsf{mag}^{[D]}_M
        \right\},\\
        \RevMag^{[D]}_{n,M-n}
        &:=
        \left\{
            \Tr_{\{n+1,\ldots,M\}}(\rho):
            \rho\in\RevMag^{[D]}_M
        \right\}.
    \end{aligned}
\]
Here $M$ is the total width, and the $M-n$ ancillary qubits may be
discarded.  Taking $M=n$ recovers the classes without ancillas.
We also set
\[
    \mathsf{mag}^{[D],\mathrm{anc}}_n
    :=
    \bigcup_{M\ge n}\mathsf{mag}^{[D]}_{n,M-n}.
\]
\end{definition}

In the notation of \cite{parham2026quantum},
$\mathsf{mag}^{[D]}_M$ and $\RevMag^{[D]}_M$ are the
state-preparation versions of $\mathcal A_{CQ}^1$ and
$\mathcal A_{QC}^1$, respectively, with shallow blocks of depth
at most $D$.  Taking the shallow block to be the identity includes
all stabilizer states.

We now allow a fixed number of alternating shallow and Clifford blocks.

\begin{definition}[$k$-round magic classes]
\label{def:higher-magic-levels}
\[
    \mathsf{mag}^{[D,k]}_M
    :=
    \left\{
        U\ket{0^M}\bra{0^M}U^\dagger:
            U=Q_kC_{k-1}\cdots Q_1C_0,
            ~
            C_{j-1}\in\Cliff_M, Q_j\in\mathsf{Circ}^{[D]}_M,~ j=1,\ldots,k
    \right\}
\]
and
\[
    \RevMag^{[D,k]}_M
    :=
    \left\{
        U\ket{0^M}\bra{0^M}U^\dagger:
            U=C_kQ_k\cdots C_1Q_1,~
            C_j\in\Cliff_M, Q_j\in\mathsf{Circ}^{[D]}_M,~ j=1,\ldots,k
    \right\}.
\]
For $n\ge1$ and $r\ge0$, define the reduced-output classes by
\[
    \begin{aligned}
        \mathsf{mag}^{[D]}_{n,k,r}
        &:=
        \left\{
            \Tr_{\{n+1,\ldots,n+r\}}(\rho):
            \rho\in\mathsf{mag}^{[D,k]}_{n+r}
        \right\},\\
        \RevMag^{[D]}_{n,k,r}
        &:=
        \left\{
            \Tr_{\{n+1,\ldots,n+r\}}(\rho):
            \rho\in\RevMag^{[D,k]}_{n+r}
        \right\}.
    \end{aligned}
\]
Thus $k$ counts the shallow blocks, the total width is $M=n+r$,
and all $r$ ancillas may be discarded. For $k=1$, these agree
with $\mathsf{mag}^{[D]}_{n,r}$ and $\RevMag^{[D]}_{n,r}$,
respectively.
\end{definition}

\begin{remark}[Relation with Parham's hierarchy]
\label{rem:higher-parham-notation}
With depth at most $D$ imposed on each shallow block,
$\mathsf{mag}^{[D,k]}_M$ and $\RevMag^{[D,k]}_M$
are the output density-matrix classes corresponding to
$\mathcal A_{CQ}^{2k-1}$ and $\mathcal A_{QC}^{2k-1}$
in \cite{parham2026quantum}, respectively.

At even level $2k$, suppose again that every shallow
block has depth at most $D$.
A circuit in the $CQ$ orientation begins and ends with
a Clifford block. Appending an identity shallow block
therefore places its output density matrix in
$\mathsf{mag}^{[D,k+1]}_M$.
Similarly, a circuit in the $QC$ orientation begins and
ends with a shallow block. Appending an identity Clifford
block places its output density matrix in
$\RevMag^{[D,k+1]}_M$.
For $M=n+r$, discarding the last $r$ qubits gives states
in $\mathsf{mag}^{[D]}_{n,k+1,r}$ and
$\RevMag^{[D]}_{n,k+1,r}$, respectively.

Parham's hierarchy is defined for circuit families
admitting some constant shallow-depth bound, independent
of the number of qubits. Each such family therefore
satisfies the depth restriction above for a suitable
fixed $D$. Consequently, our bounds for every fixed
number of rounds apply to both orientations at every
fixed level, under the stated restrictions on circuit
width.
\end{remark}

\subsection{Fixed levels of the magic hierarchy}
\label{sec:higher-magic-levels}

We now bound the profile entropy of the classes in
Definition~\ref{def:higher-magic-levels} and derive ground-state
separation at every fixed level when the total width of the circuit is $M=O(n)$.

In reverse magic, the final Clifford can map a local output Pauli
to a Pauli string of large support, so backward-lightcone
arguments, such as the one used in the proof of Proposition~\ref{prop:plain-shallow-width-depth-tradeoff}, do not give a bound on the relevant workspace in terms of
$n$ and $D$ alone. 

The lightcone approach first reduces the circuit to the qubits
that can affect the retained outputs.  Here we keep the full
$M$-qubit workspace and cover the preparation circuit by counting
the possible Clifford blocks and approximating the local gates
in the shallow blocks.  Each Clifford has at most $\exp(CM^2)$
choices, so $k$ such blocks contribute $O(kM^2)$ to the logarithmic
covering number.  This argument works for both orientations,
but the resulting bound retains its dependence on the total
width $M$.

\begin{lemma}[Profile entropy of $k$ magic rounds]
\label{lem:higher-magic-entropy}
There is a universal constant $C<\infty$ such that, for all integers
$M\ge n\ge p\ge2$, $k,D\ge1$, and all $0<\varepsilon<1$,
\[
    \log N\left(
        \varepsilon,
        \mathcal F_{n,p}
        \left(
            \mathsf{mag}^{[D]}_{n,k,M-n}
            \cup
            \RevMag^{[D]}_{n,k,M-n}
        \right),
        \|\cdot\|_{n,p}
    \right)
    \le
    CkM^2
    +
    CkDM
    \log\left(
        \frac{C3^{p/2}kDM}{\varepsilon}
    \right).
\]
\end{lemma}

\begin{proof}

In either orientation, the preparation circuit contains $k$
Clifford blocks and $k$ shallow blocks. We construct a finite
cover of these circuits and then show that their reduced-output
profiles form the required cover.

We first count the possible choices of Clifford blocks. Identifying unitaries that differ only by a global phase,
the number of $M$-qubit Clifford operations is
\[
    |\Cliff_M|
    =
    2^{M^2+2M}\prod_{j=1}^M(4^j-1)
    \le
    \exp(CM^2);
\]
see \cite[Equation~(1)]{koenig2014efficiently}. Choosing one Clifford operation for each of the $k$ blocks
therefore gives at most
$|\Cliff_M|^k\le\exp(CkM^2)$ possibilities.
For each operation, we fix a representative unitary matrix.
The choice of global phase does not affect the output density
matrix.

An architecture specifies the gate locations in each layer.
For each qubit, there are three possibilities: it participates
in a two-qubit gate, in which case there are $M-1$ choices
for the other qubit; it has a one-qubit gate; or it has no gate.
Thus there are
\[
    (M-1)+1+1=M+1
\]
possible labels per qubit.
Every valid layer determines such a labeling, so there are
at most $(M+1)^M$ layer architectures.
Padding each shallow block with identity layers gives $kD$
layers in total, and hence at most $(M+1)^{kDM}$ choices
of their architectures.
Finally, disjointness of the gate supports implies that
each layer has at most $M$ gates, giving at most $kDM$
local gate locations in total.

We now construct an operator-norm cover of the preparation
circuits. Fix $0<\delta<1$.
As in Lemma~\ref{lem:gate-count-profile-entropy}, choose
operator-norm $\delta$-nets for $\mathsf U(2)$ and
$\mathsf U(4)$, with centers in the respective unitary groups
and cardinalities at most $(C/\delta)^C$.
For each choice of Clifford blocks and matching architectures,
approximate every local gate by a point in the corresponding
net, keeping the Clifford blocks unchanged.
Since there are at most $kDM$ local gates, this gives at most
$(C/\delta)^{CkDM}$ approximating circuits for each such choice.
Taking the union over all Clifford blocks, architectures,
and both orientations gives a finite set $\mathcal U$
of circuits satisfying
\begin{equation}
\label{eq:higher-magic-circuit-cover-size}
    |\mathcal U|
    \le
    2\exp(CkM^2)
    (M+1)^{kDM}
    (C/\delta)^{CkDM}.
\end{equation}
For every preparation circuit $U$, the set $\mathcal U$
contains a circuit $\widetilde U$ with the same Clifford
blocks and matching architectures, and with corresponding
local gates differing by at most $\delta$.

To bound the covering radius, let $Q_i$ and $\widetilde Q_i$
be the corresponding shallow blocks of $U$ and $\widetilde U$.
Each contains at most $DM$ local gates.
The telescoping estimate in the proof of
Lemma~\ref{lem:gate-count-profile-entropy} therefore gives
\[
    \|Q_i-\widetilde Q_i\|_{\mathrm{op}}
    \le DM\delta,
    \qquad i=1,\ldots,k.
\]
Replacing the shallow blocks one at a time in the full
preparation circuit, all surrounding factors are unitary
and the Clifford blocks remain unchanged. Hence
\[
    \|U-\widetilde U\|_{\mathrm{op}}
    \le
    \sum_{i=1}^k
        \|Q_i-\widetilde Q_i\|_{\mathrm{op}}
    \le
    kDM\delta.
\]
Thus $\mathcal U$ is an operator-norm $kDM\delta$-cover
of the preparation circuits with the chosen Clifford
representatives.

We now transfer this circuit approximation to the output
profiles. Let $\rho$ and $\widetilde\rho$ be the reduced
outputs on the first $n$ qubits obtained by applying
$U$ and $\widetilde U$ to the same input $|0^M\rangle$.
The trace-distance and profile estimates in the proof of
Lemma~\ref{lem:gate-count-profile-entropy} give
\[
    \|\rho-\widetilde\rho\|_1
    \le
    2\|U-\widetilde U\|_{\mathrm{op}}
    \le
    2kDM\delta
\]
and consequently
\[
    \|f_\rho-f_{\widetilde\rho}\|_{n,p}
    \le
    3^{p/2}\|\rho-\widetilde\rho\|_1
    \le
    2\cdot3^{p/2}kDM\delta.
\]
Taking $\delta = \varepsilon/(2\cdot3^{p/2}kDM)$ therefore ensures that every target profile is within
$\varepsilon$ of the reduced-output profile of a circuit
in $\mathcal U$.
These profiles form an $\varepsilon$-cover of the target
class with cardinality at most $|\mathcal U|$.

Substituting the chosen value of $\delta$ into
\eqref{eq:higher-magic-circuit-cover-size} gives
\begin{align*}
    &\log N\left(
        \varepsilon,
        \mathcal F_{n,p}\left(
            \mathsf{mag}^{[D]}_{n,k,M-n}
            \cup
            \RevMag^{[D]}_{n,k,M-n}
        \right),
        \|\cdot\|_{n,p}
    \right)
    \le
    \log|\mathcal U|\\
    &\qquad\le
    \log2+CkM^2+kDM\log(M+1)
    +CkDM\log\left(
        \frac{C3^{p/2}kDM}{\varepsilon}
    \right)\\
    &\qquad\le
    CkM^2
    +CkDM\log\left(
        \frac{C3^{p/2}kDM}{\varepsilon}
    \right).
\end{align*}
The last inequality absorbs the architecture term and
$\log2$ by increasing $C$, using $k,D\ge1$ and
$0<\varepsilon<1$.
\end{proof}

\begin{proposition}[Every fixed number of magic rounds is separated from the ground state]
\label{prop:higher-magic-ground-gap}
Fix an integer $k\ge1$ and constants $K\ge1$ and $0<L<\infty$.
Let $M_n\ge n$ and $D_n\ge1$ be deterministic integer sequences
satisfying
\[
    M_n\le Kn,
    \qquad
    D_n\le\frac{Ln}{\log n}
\]
for all sufficiently large $n$.  There is
$p_0=p_0(k,K,L)<\infty$ such that, for every fixed $p\ge p_0$,
there is $c_{p,k,K,L}>0$ for which
\[
    \P\left(
        \Delta^{\mathrm{opt}}_{n,p}
        \left(
            \mathsf{mag}^{[D_n]}_{n,k,M_n-n}
            \cup
            \RevMag^{[D_n]}_{n,k,M_n-n}
        \right)
        \ge
        c_{p,k,K,L}\sqrt n
    \right)
    \longrightarrow1.
\]
\end{proposition}

\begin{proof}
Lemma~\ref{lem:higher-magic-entropy} gives
\[
    \Gamma_{n,p}
    \le
    CkM_n^2+CkD_nM_n
    \log\left(eC3^{p/2}kD_nM_n n\right).
\]
For each fixed $k$ and $p$, the width and depth assumptions bound the
logarithm by $3\log n+O_{p,k,K,L}(1)$.  Hence
\[
    \limsup_{n\to\infty}\frac{\Gamma_{n,p}}{n^2}
    \le Ck(K^2+KL).
\]
The right-hand side is independent of $p$, so
Theorem~\ref{thm:local-profile-complexity-principle}\textup{(i)} gives
the conclusion with a threshold depending only on $k$, $K$, and $L$.
\end{proof}

By Remark~\ref{rem:higher-parham-notation},
Proposition~\ref{prop:higher-magic-ground-gap} gives
ground-state separation for both orientations
at every fixed level of Parham's hierarchy under the
stated assumption of linear total circuit width.
This addresses the reverse and higher-level cases
posed in \cite[Open Problem~4]{parham2026quantum}
subject to that width restriction. The proposition also allows the shallow-block depth to grow
as $D_n=O(n/\log n)$, extending beyond the constant-depth
blocks in the definition of the hierarchy.

\subsection{Forward magic with arbitrary discardable ancillas}
\label{subsec:forward-profile-complexity}

The entropy bound in Subsection~\ref{sec:higher-magic-levels}
depends on the total width $M$ and therefore does not control the
union over all ancillary widths.  In this section,  we obtain a bound for forward
magic that is independent of the original number of discardable
ancillas, together with a sharper bound at fixed total width.

In forward magic, the initial Clifford may entangle the physical qubits
with arbitrarily many ancillas.  After restricting to the backward
lightcone of the retained outputs, the relevant input state is therefore
generally mixed.  It is nevertheless the reduction of a pure stabilizer
state and hence belongs to the mixed-stabilizer class, which we now define.

\begin{definition}[Mixed stabilizer states]
\label{def:mixed-stabilizer}
Let $\mathsf{mixstab}_M$ be the class of $M$-qubit states obtained as
reductions of pure stabilizer states:
\[
    \mathsf{mixstab}_M
    :=
    \left\{
        \Tr_{\{M+1,\ldots,M+s\}}
        |\Psi\rangle\langle\Psi|:
        |\Psi\rangle\in\Stab_{M+s},\;
        s\ge0
    \right\}.
\]
These states are normalized projectors onto stabilizer codes.
Arbitrary convex mixtures of pure stabilizer states need not
belong to $\mathsf{mixstab}_M$.
\end{definition}

A backward-lightcone reduction now removes the dependence on the
original number of ancillas.  The causal part is the same as in
\cite[proof of Proposition~10]{al2026circuit}. The difference here is
that pruning a stabilizer input produces a mixed stabilizer state.

\begin{lemma}[Backward-lightcone pruning for forward magic]
\label{lem:fixed-depth-active-cone-pruning}
Every $\rho\in\mathsf{mag}^{[D],\mathrm{anc}}_n$ has a representation
\[
    \rho
    =
    \Tr_{\{n+1,\ldots,n2^D\}}
    \left(
        \widehat Q\widehat\tau\widehat Q^\dagger
    \right),
    \qquad
    \widehat\tau\in\mathsf{mixstab}_{n2^D},
\]
for some depth-$D$ circuit $\widehat Q$ on $n2^D$ qubits,
independently of the number of ancillas used in the original
preparation. For a fixed total width $M\ge n$, every
$\rho\in\mathsf{mag}^{[D]}_{n,M-n}$ admits the same representation
with $n2^D$ replaced by $\hatM:=\min\{M,n2^D\}$.
\end{lemma}

\begin{proof}
By definition, for some $r\ge0$, there is a Clifford
unitary $C$ and a depth-$D$ circuit $Q=Q_D\cdots Q_1$
on $M:=n+r$ qubits such that
\[
    \rho
    =
    \Tr_{\{n+1,\ldots,M\}}
    (Q\rho_0Q^\dagger),
    \qquad
    \rho_0
    :=
    C|0^M\rangle\langle0^M|C^\dagger.
\]
Write each layer of the circuit as a product of its one- or two-qubit gates,
\[
    Q_t=\prod_{j=1}^{\ell_t}U_{t,j},
    \qquad
    E_{t,j}:=\operatorname{supp}(U_{t,j}).
\]
The sets $E_{t,j}$ are pairwise disjoint and have
size at most two.
Let $A_D=[n]$ be the set of retained output qubits at layer $D$, and for
$t=D,\ldots,1$, define recursively
\[
    J_t
    :=
    \{j\in[\ell_t]:E_{t,j}\cap A_t\ne\emptyset\},
    \qquad
    A_{t-1}
    :=
    A_t\cup\bigcup_{j\in J_t}E_{t,j}.
\]
In words, $A_t$ collects the qubits in layer $t$ that can affect
the retained output qubits through the remaining layers
$Q_{t+1},\dots,Q_D$, and $J_t$ indexes the gates in layer $t$
whose supports intersect the current set $A_t$.
Because the gate supports are disjoint, $|J_t|\le|A_t|$.
Each gate indexed by $J_t$ adds at most one qubit outside
$A_t$, so
\[
    |A_{t-1}|
    =
    |A_t|+\sum_{j\in J_t}|E_{t,j}\setminus A_t|
    \le
    |A_t|+|J_t|
    \le
    2|A_t|.
\]
Consequently, $|A_0|\le2^D|A_D|=n2^D$.
Since $A_0\subseteq[M]$, we also have $|A_0|\le M$,
and hence $|A_0|\le\hatM$.

We now construct a pruned circuit $\barQ$ by retaining only
the gates in the backward lightcone of the retained outputs.
In layer $t$, these are the gates indexed by $J_t$. Define
\[
    V_t:=\prod_{j\in J_t}U_{t,j},
    \qquad
    \barQ:=V_D\cdots V_1.
\]
For every $j\in J_t$, we have
$E_{t,j}\subseteq A_{t-1}\subseteq A_0$, so all retained
gates act on qubits in $A_0$.
We therefore regard $\barQ$ as a circuit on $A_0$.
Since each $V_t$ consists of gates with disjoint supports,
$\barQ$ has depth at most $D$.

The observable identity established in the proof of
Lemma~\ref{lem:gate-count-active-width} is independent
of the input state. Setting
$\tau:=\Tr_{A_0^c}(\rho_0)$, where
$A_0^c=[M]\setminus A_0$, and taking expectations in
$\rho_0$ shows that the original and pruned circuits agree
on every observable supported on the retained outputs.
Consequently,
\[
    \rho
    =
    \Tr_{A_0\setminus[n]}
    \left(\barQ\tau\barQ^\dagger\right).
\]

Since $\rho_0$ is a pure stabilizer state,
$\tau\in\mathsf{mixstab}_{|A_0|}$.
Relabel the qubits in $A_0$ so that the retained outputs
are the first $n$ qubits.
If $|A_0|<\hatM$, append $\hatM-|A_0|$ qubits in the
state $|0\rangle$ and extend $\barQ$ by the identity.
Appending these qubits preserves the mixed-stabilizer class
since tensoring $\rho_0$ with the appended $|0\rangle$-states
produces another pure stabilizer state, whose partial trace
over $A_0^c$ is the padded state.
Thus we obtain $\widehat\tau\in\mathsf{mixstab}_{\hatM}$
and a depth-$D$ circuit $\widehat Q$ on $\hatM$ qubits
such that
\[
    \rho
    =
    \Tr_{\{n+1,\ldots,\hatM\}}
    \left(
        \widehat Q\widehat\tau\widehat Q^\dagger
    \right).
\]
This proves the representation for fixed total width $M$.
Since $\hatM\le n2^D$, further padding with $|0\rangle$-states
and extending the circuit by the identity gives a representation
on exactly $n2^D$ qubits, independent of the original number
$r$ of ancillas.
\end{proof}

Therefore, after pruning, the output state depends on two choices: the mixed
stabilizer input and the shallow circuit. We first fix the circuit
and bound the profile entropy as the stabilizer input varies. Each weight-$p$
output Pauli depends on at most $p2^D$ input qubits, so we only need
to count the possible input Pauli expectations on these qubits.
We then discretize the gates to cover the possible circuits.
Combining these steps gives the following entropy bound.
We state the bound and derive its energy consequences before
giving its proof.

\begin{lemma}[Profile entropy of forward magic]
\label{lem:forward-profile-entropy}
There is a universal constant $C<\infty$ such that, for every
$M\ge n\ge p\ge2$, $D\ge1$, and $0<\varepsilon<1$, setting
${\widehat M}:=\min\{M,n2^D\}$ gives
\[
    \log N\left(
        \varepsilon,
        \mathcal F_{n,p}\left(\mathsf{mag}^{[D]}_{n,M-n}\right),
        \|\cdot\|_{n,p}
    \right)
    \le
    C{\widehat M}\left[
        \min\{p2^D,{\widehat M}\}
        +D\log\left(\frac{C3^{p/2}D{\widehat M}}{\varepsilon}\right)
    \right].
\]
The same bound holds for $\mathsf{mag}^{[D],\mathrm{anc}}_n$ upon setting
${\widehat M}:=n2^D$.
\end{lemma}

\begin{proposition}[Forward-magic energy bounds]
\label{prop:dense-magic}
Let $(M_n)_{n\ge1}$ and $(D_n)_{n\ge1}$ be fixed deterministic
integer sequences with $M_n\ge n$ and $D_n\ge1$.
\begin{enumerate}[label=\textup{(\roman*)}]
    \item Suppose $M_n\le Kn$ and $D_n\le Ln/\log n$ for all
    sufficiently large $n$, for some fixed constants $K\ge1$ and
    $0<L<\infty$. Then, for every fixed $p$ sufficiently large in terms of $K$
and $L$, there is $c_{p,K,L}>0$ for which
\[
    \P\left(
        \Delta^{\mathrm{opt}}_{n,p}
        \left(\mathsf{mag}^{[D_n]}_{n,M_n-n}\right)
        \ge c_{p,K,L}\sqrt n
    \right)
    \longrightarrow1.
\]

    \item For every fixed $p\ge2$, if $M_n=O(n)$ and $2^{D_n}=o(n)$,
    then
    \[
        \frac{
            \Delta^{\mathrm{ref}}_{n,p}
            \left(\mathsf{mag}^{[D_n]}_{n,M_n-n};\Prod_n\right)
        }{\sqrt n}
        \longrightarrow0
    \]
    in $L^1$ and in probability.
\end{enumerate}
With arbitrarily many discardable ancillas, assume that
$4^{D_n}=o(n)$.  Then the convergence in
part~\textup{(ii)} holds for
$\mathsf{mag}^{[D_n],\mathrm{anc}}_n$ for every fixed
$p\ge2$.  Moreover, there is a universal $p_0$ such
that, for every fixed $p\ge p_0$,
\[
    \P\left(
        \Delta^{\mathrm{opt}}_{n,p}
        \left(\mathsf{mag}^{[D_n],\mathrm{anc}}_n\right)
        \ge c_p\sqrt n
    \right)
    \longrightarrow1,
\]
where $c_p>0$ depends only on $p$.
Both conclusions are uniform over the number of
discardable ancillas.
\end{proposition}

\begin{proof}
Set ${\widehat M}_n:=\min\{M_n,n2^{D_n}\}$. Lemma~\ref{lem:forward-profile-entropy} allows us to take
\[
    a_{n,p}=C{\widehat M}_n\min\{p2^{D_n},{\widehat M}_n\},
    \qquad
    b_{n,p}=CD_n{\widehat M}_n,
    \qquad
    B_{n,p}=C3^{p/2}D_n{\widehat M}_n
\]
in Theorem~\ref{thm:local-profile-complexity-principle}.  Hence
\[
    \Gamma_{n,p}
    \le
    C{\widehat M}_n\min\{p2^{D_n},{\widehat M}_n\}
    +CD_n{\widehat M}_n\log\left(eC3^{p/2}D_n{\widehat M}_n n\right).
\]
Under the hypotheses of part~\textup{(i)}, using
${\widehat M}_n\le M_n$ gives
\[
    \Gamma_{n,p}
    \le
    CM_n^2+CD_nM_n\log\left(eC3^{p/2}D_nM_n n\right).
\]
For fixed $p,K,L$, the logarithm is at most
$3\log n+O_{p,K,L}(1)$, so
\[
    \limsup_{n\to\infty}\frac{\Gamma_{n,p}}{n^2}
    \le C(K^2+KL).
\]
The right-hand side is independent of $p$.
Theorem~\ref{thm:local-profile-complexity-principle}\textup{(i)}
therefore gives the result.

For part~\textup{(ii)}, choose a fixed
$K$ such that $M_n\le Kn$ eventually.  Using
$\min\{p2^{D_n},{\widehat M}_n\}\le p2^{D_n}$ instead gives
\[
    \Gamma_{n,p}
    \le
    C_{p,K}n\left[
        2^{D_n}+D_n\log\left(en^2D_n\right)
    \right]
    =o(n^2),
\]
since $2^{D_n}=o(n)$ implies $D_n=O(\log n)$. Theorem~\ref{thm:local-profile-complexity-principle}\textup{(ii)},
with benchmark $\mathcal B_n=\Prod_n$, gives the result.

For the arbitrary-ancilla class, the entropy bound gives
\[
    \frac{\Gamma_{n,p}}{n^2}
    \le
    C_p\left[
        \frac{4^{D_n}}{n}
        +
        \frac{
            D_n2^{D_n}
            \log(en^2D_n2^{D_n})
        }{n}
    \right].
\]
If $4^{D_n}=o(n)$, then eventually
$2^{D_n}\le\sqrt n$ and $D_n=O(\log n)$, so
\[
    \frac{\Gamma_{n,p}}{n^2}
    \le
    C_p\left[
        \frac{4^{D_n}}{n}
        +\frac{(\log n)^2}{\sqrt n}
    \right]
    \longrightarrow0.
\]
Theorem~\ref{thm:local-profile-complexity-principle}
\textup{(ii)} gives the product-state comparison.
Part~\textup{(i)}, applied with $L=0$, gives the
ground-state separation with a universal lower
threshold for $p$ and a gap constant depending only
on $p$.
\end{proof}

We now prove the profile-entropy bound in Lemma~\ref{lem:forward-profile-entropy}.  The first step is to count
how many different expectation patterns mixed stabilizer states can
realize on a prescribed collection of Pauli operators. The description of Pauli expectations through stabilizer-group
membership also underlies Rocchetto's PAC-learning algorithm
\cite[Lemma~1]{rocchetto2017stabiliser}. The following
lemma controls this number without counting all mixed stabilizer
states.

\begin{lemma}[Restriction growth of mixed stabilizer profiles]
\label{lem:fixed-depth-mixed-stabilizer-growth}
For every $\tau\in\mathsf{mixstab}_M$ and every $P\in\mathcal P_M$,
$f_\tau(P) = \Tr(P\tau)\in\{-1,0,1\}$.  Moreover, each of the set classes
\begin{align*}
    \mathcal A_+ &:=
    \bigl\{\{P\in\mathcal P_M:f_\tau(P)=1\}:
    \tau\in\mathsf{mixstab}_M\bigr\},\\
    \mathcal A_- &:=
    \bigl\{\{P\in\mathcal P_M:f_\tau(P)=-1\}:
    \tau\in\mathsf{mixstab}_M\bigr\}
\end{align*}
has VC dimension at most $M$.  Consequently, for arbitrary
$P_1,\ldots,P_q\in\mathcal P_M$,
\[
    \#\left\{
        \bigl(f_\tau(P_1),\ldots,f_\tau(P_q)\bigr):
        \tau\in\mathsf{mixstab}_M
    \right\}
    \le
    \bigg(\sum_{j=0}^M\binom qj\bigg)^2.
\]
\end{lemma}

\begin{proof}
By Definition~\ref{def:mixed-stabilizer}, every
$\tau\in\mathsf{mixstab}_M$ can be written as
$\tau=\Tr_A(|\Psi\rangle\langle\Psi|)$, where
$|\Psi\rangle\in\Stab_{M+s}$ for some $s\ge0$ and $A$ denotes the
$s$ discarded qubits.  Any $(M+s)$-qubit stabilizer state $|\Psi\rangle$ has a corresponding
stabilizer group $\mathcal G_\Psi$, consisting of all signed Pauli
operators $G$ that stabilize it, meaning
$G|\Psi\rangle=|\Psi\rangle$.
This group has $M+s$ independent generators $g_1,\ldots,g_{M+s}$,
which give the representation
\[
    |\Psi\rangle\langle\Psi|
    =
    \prod_{j=1}^{M+s}\frac{I+g_j}{2}
    =
    2^{-(M+s)}\sum_{G\in\mathcal G_\Psi}G;
\]
see Appendix~\ref{app:IsotropicSubspaces}.
The group contains $I$ but not $-I$, since
$(-I)|\Psi\rangle=-|\Psi\rangle\ne|\Psi\rangle$.

Write each $G\in\mathcal G_\Psi$ as $G=S\otimes R$, where $S$ is
a signed Pauli operator on the retained qubits and $R\in\mathcal P_s$
acts on $A$. Then
\[
    \Tr_A(G)
    =
    S\Tr(R)
    =
        2^sS \mathbf 1 \{R=I_A\}.
\]
Define
\[
    \mathcal S_\tau
    :=
    \left\{
        S\in\{\pm P:P\in\mathcal P_M\}:
        S\otimes I_A\in\mathcal G_\Psi
    \right\}.
\]
The stabilizer expansion therefore gives
\[
    \tau
    =
    2^{-(M+s)}
    \sum_{G\in\mathcal G_\Psi}\Tr_A(G)
    =
    2^{-M}\sum_{S\in\mathcal S_\tau}S.
\]
The set $\mathcal S_\tau$ inherits the group structure and
commutation relations of $\mathcal G_\Psi$: for every
$S,S'\in\mathcal S_\tau$,
\[
    SS'=S'S\in\mathcal S_\tau,
    \qquad S^{-1}=S.
\]
Moreover, $I\in\mathcal S_\tau$ and $-I\notin\mathcal S_\tau$.  Pauli orthogonality gives for $P \in \mathcal P_M$
\begin{align} \label{eq:pauli-profile-in-proof}
    f_\tau(P)
    =
    2^{-M}\sum_{S\in\mathcal S_\tau}\Tr(PS)
    =
    \begin{cases}
        +1,  & P\in\mathcal S_\tau,\\
        -1, & -P\in\mathcal S_\tau,\\
        ~0,  & \text{otherwise}.
    \end{cases}
\end{align}
Next, suppose that $P_1,\ldots,P_k$ are \emph{shattered} by $\mathcal A_+$, meaning that for every $B\subseteq[k]$, there exists
$\tau_B\in\mathsf{mixstab}_M$ such that
\[
    f_{\tau_B}(P_i)=1
    \quad\Longleftrightarrow\quad
    i\in B.
\]
First apply shattering with $B=[k]$, and denote the resulting state
by $\tau_{\mathrm{all}}$.
By \eqref{eq:pauli-profile-in-proof}, all the $P_i$ belong to
$\mathcal S_{\tau_{\mathrm{all}}}$, whose elements commute, so the $P_i$ commute.

We next show that $[P_1],\ldots,[P_k]$ are linearly independent
over $\mathbb F_2$. Suppose for the sake of contradiction that there are coefficients
$c_1,\ldots,c_k\in\{0,1\}$, not all zero, such that $\sum_{i=1}^k c_i[P_i]=0$. Setting $T:=\{i:c_i=1\}$, this becomes
\[
    \sum_{i\in T}[P_i]=0,
    \qquad T\ne\emptyset.
\]
Multiplication of Pauli operators corresponds to addition of their
binary labels, with phases ignored. Hence
\[
    \left[\prod_{i\in T}P_i\right]
    =
    \sum_{i\in T}[P_i]
    =
    0.
\]
Since the zero label corresponds to the identity, the product equals
$I$ up to a phase. Moreover, the $P_i$ commute and satisfy $P_i^2=I$,
so
\[
    \bigg(\prod_{i\in T}P_i\bigg)^2
    =
    \prod_{i\in T}P_i^2
    =
    I.
\]
The phase must therefore be $+1$ or $-1$, giving $\prod_{i\in T}P_i=\pm I$. Since the product belongs to $\mathcal S_{\tau_{\mathrm{all}}}$, which does not
contain $-I$, it must equal $I$. Now fix $j\in T$ and apply
shattering again, this time to the subset $[k]\setminus\{j\}$.
Denote the resulting state by $\tau_{-j}$, so that
\[
    P_i\in\mathcal S_{\tau_{-j}}\quad\text{for every }i\ne j,
    \qquad
    P_j\notin\mathcal S_{\tau_{-j}}.
\]
The identity $\prod_{i\in T}P_i=I$ is an identity between operators
and does not depend on the choice of state. By commutation, the identity can be written as 
\begin{align*}
    P_j \prod_{i\in T\setminus \{j\}} P_i=I,
\end{align*}
and multiplying both sides by $P_j$ and using $P_j^2=I$ yields
\begin{align*}
    P_j = \prod_{i\in T\setminus \{j\}} P_i.
\end{align*}
Since every $P_i$ with $i\in T\setminus\{j\}$ belongs to
$\mathcal S_{\tau_{-j}}$, closure under multiplication gives
\[
    \prod_{i\in T\setminus\{j\}}P_i
    \in\mathcal S_{\tau_{-j}}.
\]
This contradicts $P_j\notin\mathcal S_{\tau_{-j}}$, which was
required by the second application of shattering. Thus the labels are linearly independent. 

Hence, $W:=\operatorname{span}\{[P_1],\ldots,[P_k]\}$ has dimension $k$ by linear independence, and is isotropic, because
\[
    \langle[P_i],[P_j]\rangle_{\mathrm{sp}}=0,
    \qquad i,j\in[k].
\]
Every isotropic subspace of $\mathbb F_2^{2M}$ has dimension at
most $M$, so $k\le M$. This proves the VC-dimension bound for
the class $\mathcal A_+$.

For the class $\mathcal A_-$, the same argument applies with
$P_i$ replaced by $-P_i$, using
\[
    f_\tau(P_i)=-1
    \quad\Longleftrightarrow\quad
    -P_i\in\mathcal S_\tau.
\]
Its VC dimension is therefore also at most $M$.

Finally, fix arbitrary $P_1,\ldots,P_q\in\mathcal P_M$ and set
$E:=\{P_1,\ldots,P_q\}$. Since $\mathcal A_+$ and $\mathcal A_-$
each have VC dimension at most $M$, Sauer's lemma
\cite[Lemma~6.10]{shalev2014understanding} gives
\[
    \#\{A\cap E:A\in\mathcal A_+\}
    \le \sum_{j=0}^M\binom qj,
    \qquad
    \#\{A\cap E:A\in\mathcal A_-\}
    \le \sum_{j=0}^M\binom qj,
\]
where we used $|E|\le q$.

For each $\tau\in\mathsf{mixstab}_M$, define
\[
    E_\tau^+:=\{P\in E:f_\tau(P)=1\},
    \qquad
    E_\tau^-:=\{P\in E:f_\tau(P)=-1\}.
\]
Since $f_\tau$ takes values in $\{-1,0,1\}$,
\[
    f_\tau(P_i)
    =
    \mathbf 1_{E_\tau^+}(P_i)
    -
    \mathbf 1_{E_\tau^-}(P_i),
    \qquad i=1,\dots, q.
\]
Thus specifying $E_\tau^+$ and $E_\tau^-$ determines every
coordinate of the restricted profile. Moreover,
\[
    E_\tau^+\in\{A\cap E:A\in\mathcal A_+\},
    \qquad
    E_\tau^-\in\{A\cap E:A\in\mathcal A_-\}.
\]
Consequently,
\begin{align*}
    \#\left\{
        \bigl(f_\tau(P_1),\ldots,f_\tau(P_q)\bigr):
        \tau\in\mathsf{mixstab}_M
    \right\}
    &\le
    \#\{A\cap E:A\in\mathcal A_+\}
    \#\{A\cap E:A\in\mathcal A_-\}\\
    &\le
    \bigg(\sum_{j=0}^M\binom qj\bigg)^2.
\end{align*}
\end{proof}

We now prove Lemma~\ref{lem:forward-profile-entropy}.

\begin{proof}
Recall that ${\widehat M}=\min\{M,n2^D\}$.
For each circuit $Q$ of depth at most $D$ on ${\widehat M}$
qubits, define
\[
    \mathcal F_Q
    :=
    \left\{
        f_\rho:
        \rho=
        \Tr_{\{n+1,\ldots,{\widehat M}\}}(Q\tau Q^\dagger),~~
        \tau\in\mathsf{mixstab}_{{\widehat M}}
    \right\}.
\]
Here $f_\rho(P)=\Tr(P\rho)$ for $P\in\mathcal P_{n,p}$. By Lemma~\ref{lem:fixed-depth-active-cone-pruning}, $\mathcal F_{n,p}(\mathsf{mag}^{[D]}_{n,M-n})
    \subseteq
    \bigcup_Q\mathcal F_Q$. It therefore suffices to bound the metric entropy of
$\bigcup_Q\mathcal F_Q$.
To this end, we construct a finite circuit cover $\mathcal Q$ such that for any $Q$, there exists $\widetilde Q \in \mathcal Q$ such that
\begin{equation}
\label{eq:forward-profile-approximation}
    \sup_{f\in\mathcal F_Q}
    \inf_{g\in\mathcal F_{\widetilde Q}}
    \|f-g\|_{n,p}
    \le\varepsilon/4.
\end{equation}
Combining \eqref{eq:forward-profile-approximation} with an $\varepsilon/4$-cover
of each $\mathcal F_{\widetilde Q}$ then gives
\begin{align}
\label{eq:forward-cover-decomposition}
    \log N\bigg(
        \varepsilon,\bigcup_Q\mathcal F_Q,\|\cdot\|_{n,p}
    \bigg)
    &\le
    \log\left(
        \sum_{\widetilde Q\in\mathcal Q}
        N\left(
            \varepsilon/4,\mathcal F_{\widetilde Q},\|\cdot\|_{n,p}
        \right)
    \right) \nonumber\\
    &\le
    \log\left(
        |\mathcal Q|\,
        \sup_{\widetilde Q\in\mathcal Q}
        N\left(
            \varepsilon/4,\mathcal F_{\widetilde Q},\|\cdot\|_{n,p}
        \right)
    \right)\nonumber\\
    &=
    \log|\mathcal Q|
    +
    \sup_{\widetilde Q\in\mathcal Q}
    \log N\left(
        \varepsilon/4,\mathcal F_{\widetilde Q},\|\cdot\|_{n,p}
    \right).
\end{align}
We first bound the second term in the last line of
\eqref{eq:forward-cover-decomposition} uniformly over $Q$,
then construct $\mathcal Q$, verify
\eqref{eq:forward-profile-approximation},
and bound the first term $\log|\mathcal Q|$.

\textit{1. Covering profiles for fixed $Q$ and varying $\tau$: }
Fix $Q$. For every input $\tau\in\mathsf{mixstab}_{{\widehat M}}$, let
$\rho=\Tr_{\{n+1,\ldots,{\widehat M}\}}(Q\tau Q^\dagger)$
be its reduced output. For every $P\in\mathcal P_{n,p}$,
\[
    \begin{aligned}
    f_\rho(P)
    &=\Tr(P\rho)
    =\Tr\!\left(
        P\,\Tr_{\{n+1,\ldots,{\widehat M}\}}
        (Q\tau Q^\dagger)
    \right)\\
    &=\Tr\bigl((P\otimes I)Q\tau Q^\dagger\bigr)
    =\Tr\bigl(Q^\dagger(P\otimes I)Q\,\tau\bigr).
    \end{aligned}
\]
The third equality is the defining identity for the partial trace,
where $I$ acts on the discarded qubits. We now show that the operator $Q^\dagger(P\otimes I)Q$ is supported on at most $s:=\min\{p2^D,{\widehat M}\}$ qubits, and use its Pauli expansion to express $f_\rho(P)$
in terms of input expectations.

Write $Q=Q_D\cdots Q_1$, where each $Q_t$ is a layer of gates
with disjoint supports of size at most two.
We first consider how conjugation by one layer changes the support
of an operator.

Fix a layer $Q_t=U_1\cdots U_r$, and let $E_j$ be the support
of $U_j$. Suppose $A$ is supported on $S$, meaning
$A=A_S\otimes I_{S^c}$, and define $J:=\{j\in[r]:E_j\cap S\ne\emptyset\}$.

For $j\notin J$, the gate $U_j$ acts outside $S$, so
$U_j^\dagger A U_j=A$.
Moreover, the $U_j$ commute with one another because their
supports are disjoint. We can therefore move every gate
with $j\notin J$ next to $A$ and cancel it with its adjoint,
obtaining
\[
    Q_t^\dagger A Q_t
    =
    \big(\prod_{j\in J}U_j\big)^\dagger
    A
    \big(\prod_{j\in J}U_j\big).
\]
Every factor on the right acts as the identity outside $S':=S\cup\bigcup_{j\in J}E_j$. Hence $\operatorname{supp}(Q_t^\dagger A Q_t)\subseteq S'$.

To bound $|S'|$, note that each $E_j$ with $j\in J$
contains at least one qubit of $S$ and has size at most two.
It therefore contains at most one qubit outside $S$.
Also, the sets $E_j\cap S$, $j\in J$, are nonempty and
pairwise disjoint, so $|J|\le|S|$. Consequently,
\[
    |S'|
    =
    |S|+\sum_{j\in J}|E_j\setminus S|
    \le
    |S|+|J|
    \le
    2|S|.
\]
Thus conjugation by one layer at most doubles the support size.

The operator $P\otimes I$ starts with support of size $p$.
In the expression
\[
    Q^\dagger(P\otimes I)Q
    =
    Q_1^\dagger\cdots Q_D^\dagger
    (P\otimes I)
    Q_D\cdots Q_1,
\]
we first conjugate by $Q_D$, then by $Q_{D-1}$, and continue
through $Q_1$. Applying the preceding bound at each step,
the support has size at most $2p$ after the first step,
$2^2p$ after the second, and $2^Dp$ after all $D$ steps.
Since the support is also a subset of the
${\widehat M}$ qubits, we obtain
\[
    \bigl|\operatorname{supp}(Q^\dagger(P\otimes I)Q)\bigr|
    \le
    \min\{p2^D,{\widehat M}\}
    =s.
\]

Now fix $P_1,\ldots,P_q\in\mathcal P_{n,p}$.
Each operator $Q^\dagger(P_\ell\otimes I)Q$ has a Pauli expansion
with at most $4^s$ terms, since there are four choices
$I,X,Y,Z$ on each qubit of its support and only $I$ elsewhere.
List the distinct Paulis appearing in these $q$ expansions
as $R_1,\ldots,R_N$, where $N\le q4^s$.
Then
\begin{align*}
    Q^\dagger(P_\ell\otimes I)Q
    &=\sum_{j=1}^N a_{\ell j}R_j,\\
    f_\rho(P_\ell)
    &=\sum_{j=1}^N a_{\ell j}f_\tau(R_j),
    \qquad \ell = 1,\dots,q,
\end{align*}
where the coefficients depend on $Q$ and the chosen output
Paulis, but not on $\tau$. If $\tau,\tau'\in\mathsf{mixstab}_{{\widehat M}}$ satisfy
\[
    f_\tau(R_j)=f_{\tau'}(R_j),
    \qquad j = 1,\dots, N,
\]
then their reduced outputs $\rho,\rho'$ under the same circuit $Q$
satisfy
\[
    f_\rho(P_\ell)
    =\sum_{j=1}^N a_{\ell j}f_\tau(R_j) =\sum_{j=1}^N a_{\ell j}f_{\tau'}(R_j)
    =f_{\rho'}(P_\ell),
    \qquad \ell =1,\dots, q.
\]
Thus the input expectations at $R_1,\ldots,R_N$ uniquely determine
the output expectations at $P_1,\ldots,P_q$, so the number of
distinct output vectors cannot exceed the number of distinct
input vectors.
By Lemma~\ref{lem:fixed-depth-mixed-stabilizer-growth},
\begin{align*}
    \#\left\{
        \bigl(f(P_1),\ldots,f(P_q)\bigr):f\in\mathcal F_Q
    \right\}
    &\le
    \#\left\{
        \bigl(f_\tau(R_1),\ldots,f_\tau(R_N)\bigr):
        \tau\in\mathsf{mixstab}_{{\widehat M}}
    \right\}\\
    &\le
    \left(\sum_{j=0}^{{\widehat M}}\binom Nj\right)^2
    \le
    \left(\sum_{j=0}^{{\widehat M}}\binom{q4^s}{j}\right)^2.
\end{align*}
To apply Lemma~\ref{lem:growth-entropy} with
$d=2{\widehat M}$, we need a restriction-growth bound for
every $q\ge2{\widehat M}$.
For such $q$, we have $q4^s\ge {\widehat M}$, so the
standard binomial bound gives
\[
    \left(\sum_{j=0}^{{\widehat M}}\binom{q4^s}{j}\right)^2
    \le
    \left(\frac{eq4^s}{{\widehat M}}\right)^{2{\widehat M}}
    =
    \left(\frac{Aq}{d}\right)^d,
\]
with $d=2{\widehat M}$ and $A=2e4^s$. Thus $\mathcal F_Q$ satisfies the restriction-growth hypothesis
of the lemma with these parameters.

We use the uniform probability measure on $\mathcal P_{n,p}$.
For any function $h$ on this set, the corresponding $L_2$ norm satisfies
\[
    \|h\|_{L_2}^2
    =
    \frac{1}{3^p\binom np}
    \sum_{P\in\mathcal P_{n,p}}|h(P)|^2
    =
    3^{-p}\|h\|_{n,p}^2.
\]
Consequently, an $L_2$-cover of radius $\varepsilon/3^{p/2}$
is a cover of radius $\varepsilon$ in $\|\cdot\|_{n,p}$.
Applying Lemma~\ref{lem:growth-entropy} at this $L_2$ radius gives
\begin{align}
    \log N\left(\varepsilon,\mathcal F_Q,\|\cdot\|_{n,p}\right)
    &=
    \log N\left(
        \frac{\varepsilon}{3^{p/2}},\mathcal F_Q,L_2
    \right)\notag\\
    &\le
    Cd\log\left(\frac{eA3^{p/2}}{\varepsilon}\right)\notag\\
    &\le
    C{\widehat M}\left[
        s+\log\left(\frac{C3^{p/2}}{\varepsilon}\right)
    \right].
    \label{eq:forward-fixed-circuit-entropy}
\end{align}

The bound \eqref{eq:forward-fixed-circuit-entropy} holds uniformly over $Q$.

\textit{2. Covering the circuit class: }
An architecture specifies the gate locations in each layer.
For each qubit, there are three possibilities: it participates
in a two-qubit gate, in which case there are
${\widehat M}-1$ choices for the other qubit; it has a
one-qubit gate; or it has no gate.
Thus there are
\[
    ({\widehat M}-1)+1+1={\widehat M}+1
\]
possible labels per qubit.
Every valid layer determines such a labeling, so there are at most
$({\widehat M}+1)^{{\widehat M}}$ layer architectures
and therefore at most $({\widehat M}+1)^{D{\widehat M}}$ architectures for $D$ layers.
Finally, disjointness of the gate supports implies that each layer
has at most ${\widehat M}$ gates, giving at most
$D{\widehat M}$ gate locations in total.

As in Lemma~\ref{lem:gate-count-profile-entropy}, choose
operator-norm $\delta$-nets for $\mathsf U(2)$ and
$\mathsf U(4)$, with centers in the respective unitary groups
and cardinalities at most $(C/\delta)^C$.
For each architecture, approximate every gate by a point
in the corresponding net. Since there are at most
$D{\widehat M}$ gates, this gives at most
$(C/\delta)^{CD{\widehat M}}$ approximating circuits
per architecture. Taking the union over all architectures
gives a finite set $\mathcal Q$ of circuits satisfying
\begin{equation}
\label{eq:forward-circuit-cover-size}
    |\mathcal Q|\le
    ({\widehat M}+1)^{D{\widehat M}}
    (C/\delta)^{CD{\widehat M}}.
\end{equation}
For every circuit $Q$, the set $\mathcal Q$ contains
a circuit $\widetilde Q$ with the same architecture and
with corresponding gates differing by at most $\delta$.

To bound the covering radius, write $Q$ and $\widetilde Q$
in terms of their $L\le D{\widehat M}$ gates,
\[
    Q=U_L\cdots U_1,
    \qquad
    \widetilde Q=\widetilde U_L\cdots\widetilde U_1,
\]
using the same ordering of gate locations in both circuits.
Here each gate acts as the identity outside its designated
qubits, and
$\|U_j-\widetilde U_j\|_{\mathrm{op}}\le\delta$.
Replacing the gates one at a time gives the telescoping identity
\[
    Q-\widetilde Q
    =
    \sum_{j=1}^L
        \widetilde U_L\cdots\widetilde U_{j+1}
        (U_j-\widetilde U_j)
        U_{j-1}\cdots U_1,
\]
where an empty product is interpreted as the identity.
All factors surrounding $U_j-\widetilde U_j$ are unitary,
so multiplication by these factors preserves its operator norm.
The triangle inequality therefore gives
\[
    \|Q-\widetilde Q\|_{\mathrm{op}}
    \le
    \sum_{j=1}^L\|U_j-\widetilde U_j\|_{\mathrm{op}}
    \le
    L\delta
    \le
    D{\widehat M}\delta.
\]
Thus this construction gives an operator-norm
$D{\widehat M}\delta$-cover of the circuit class.

\textit{3. Transferring the circuit approximation to profiles: }
Let $\rho$ and $\widetilde\rho$ be the reduced outputs of
$Q$ and $\widetilde Q$ applied to the same input
$\tau\in\mathsf{mixstab}_{{\widehat M}}$.
Using contractivity under partial trace and $\|\tau\|_1=1$,
we obtain
\[
    \begin{aligned}
    \|\rho-\widetilde\rho\|_1
    &\le
    \|Q\tau Q^\dagger-\widetilde Q\tau\widetilde Q^\dagger\|_1
    \le
    \|(Q-\widetilde Q)\tau Q^\dagger\|_1
    +
    \|\widetilde Q\tau(Q^\dagger-\widetilde Q^\dagger)\|_1\\
    &\le
    2\|Q-\widetilde Q\|_{\mathrm{op}}
    \le 2D{\widehat M}\delta.
    \end{aligned}
\]
Since every Pauli has operator norm one and
$|\mathcal P_{n,p}|=3^p\binom np$, the definition of the
profile norm gives
\[
    \|f_\rho-f_{\widetilde\rho}\|_{n,p}
    \le
    3^{p/2}\|\rho-\widetilde\rho\|_1
    \le
    2\cdot3^{p/2}D{\widehat M}\delta.
\]
This bound is uniform over $\tau$. Taking $\delta
    =\varepsilon/
         ({8\cdot3^{p/2}D{\widehat M}})$ therefore ensures that every profile in $\mathcal F_Q$
lies within $\varepsilon/4$ of a profile in
$\mathcal F_{\widetilde Q}$, as required by
\eqref{eq:forward-profile-approximation}.

\textit{4. Combining the estimates:}
For each $\widetilde Q\in\mathcal Q$, choose an
$\varepsilon/4$-cover of $\mathcal F_{\widetilde Q}$
using \eqref{eq:forward-fixed-circuit-entropy}.

To verify that their union covers the target class, take any
target profile $f_\rho\in\mathcal F_Q$.
Approximate $Q$ by a circuit $\widetilde Q$ in the circuit cover
and apply it to the same input $\tau$.
For the chosen value of $\delta$, the resulting profile
$f_{\widetilde\rho}\in\mathcal F_{\widetilde Q}$ satisfies $\|f_\rho-f_{\widetilde\rho}\|_{n,p}
    \le\varepsilon/4$. The cover of $\mathcal F_{\widetilde Q}$ then contains a point
$h$ satisfying $\|f_{\widetilde\rho}-h\|_{n,p}
    \le\varepsilon/4$. Consequently,
\[
    \|f_\rho-h\|_{n,p}
    \le
    \|f_\rho-f_{\widetilde\rho}\|_{n,p}
    +
    \|f_{\widetilde\rho}-h\|_{n,p}
    \le\varepsilon/2.
\]
Thus the union of the profile covers is an
$\varepsilon/2$-cover, and hence an $\varepsilon$-cover,
of the target class.

There is one profile cover for each circuit in the circuit
cover. The total cardinality is therefore at most the
circuit-cover cardinality multiplied by the uniform bound
on the profile-cover cardinality for a fixed circuit.
Substituting \eqref{eq:forward-circuit-cover-size} and
\eqref{eq:forward-fixed-circuit-entropy} into the first and second
terms, respectively, in the last line of
\eqref{eq:forward-cover-decomposition} gives
\begin{align*}
    \log N\left(
        \varepsilon,
        \mathcal F_{n,p}(\mathsf{mag}^{[D]}_{n,M-n}),
        \|\cdot\|_{n,p}
    \right)
    &\le
    D{\widehat M}\log({\widehat M}+1)
    +CD{\widehat M}\log(C/\delta)\\
    &+
    C{\widehat M}\left[
        s+\log\left(\frac{C3^{p/2}}{\varepsilon}\right)
    \right].
\end{align*}

Here \eqref{eq:forward-fixed-circuit-entropy} is applied at
radius $\varepsilon/4$, with the factor $4$ absorbed into $C$.
Substituting
$\delta=\varepsilon/(8\cdot3^{p/2}D{\widehat M})$
and using $D\ge1$ yields
\begin{equation}
\label{eq:forward-combined-entropy}
    \log N\left(
        \varepsilon,
        \mathcal F_{n,p}(\mathsf{mag}^{[D]}_{n,M-n}),
        \|\cdot\|_{n,p}
    \right)
    \le
    C{\widehat M}\left[
        s+D\log\left(
            \frac{C3^{p/2}D{\widehat M}}{\varepsilon}
        \right)
    \right].
\end{equation}
Substituting $s=\min\{p2^D,{\widehat M}\}$ into
\eqref{eq:forward-combined-entropy} proves the first estimate.

\textit{5. Arbitrary discardable ancillas: }
For arbitrary discardable ancillas,
Lemma~\ref{lem:fixed-depth-active-cone-pruning} represents
every output using a depth-$D$ circuit on the common width
${\widehat M}=n2^D$, regardless of the original workspace.
The circuit cover counted in \eqref{eq:forward-circuit-cover-size}
and the profile covers obtained from
\eqref{eq:forward-fixed-circuit-entropy}
therefore cover the entire class, uniformly over the number
of ancillas.
Since $p\le n$, we have $s=p2^D$.
Substituting these values of ${\widehat M}$ and $s$
into \eqref{eq:forward-combined-entropy} proves the second estimate.
\end{proof}

\subsection{Reverse magic with arbitrary clean ancillas} \label{sec:revmag-clean}
We now specialize to a single round of reverse magic and extend
the bound from Subsection~\ref{sec:higher-magic-levels} to an
arbitrary number of clean ancillas. Here each ancillary qubit
must be restored to $\ket{0}$ at the output. We will use this
condition to obtain a bound independent of the number of ancillas.

\begin{definition}[Clean-ancilla reverse magic]
\label{def:revmag-clean}
Fix $D\ge1$. Let $\RevMag^{[D],\mathrm{clean}}_n$ be the class
of pure $n$-qubit states $\ket{\psi}$ for which there exist
$r\ge0$, a Clifford $C\in\Cliff_{n+r}$, and a depth-$D$
matching circuit $Q$ on $n+r$ qubits such that $CQ\ket{0^{n+r}}
    =
    \ket{\psi}\ket{0^r}$.
\end{definition}

Our reduction is based on \textit{stabilizer nullity}. To define it,
we first recall the binary representation of Pauli operators.
We use $\mathcal P_M:=\{I,X,Y,Z\}^{\otimes M}$ for the
phase-free Paulis on $M$ qubits and
$\mathcal P_M^\pm:=\{\pm P:P\in\mathcal P_M\}$ for the
signed Hermitian Paulis. Each $P\in\mathcal P_M$
has a binary label $[P]=(x,z)\in\mathbb F_2^{2M}$, obtained
by assigning the pairs $(0,0),(1,0),(1,1),(0,1)$ to the
single-qubit factors $I,X,Y,Z$, respectively. Here
$\mathbb F_2=\{0,1\}$, with arithmetic modulo two.
Labels ignore scalar phases: $[\omega P]=[P]$ for
$\omega\in\{1,-1,i,-i\}$. Pauli multiplication then
corresponds to addition of labels, $[PP']=[P]+[P']$.

The label space carries the symplectic bilinear form
\[
    \langle(x,z),(x',z')\rangle_{\mathrm{sp}}
    :=x\cdot z'+z\cdot x'\pmod2,
\]
which records Pauli commutation through the identity
$PP'=(-1)^{\langle[P],[P']\rangle_{\mathrm{sp}}}P'P$.
Thus the value is zero when the operators commute and one
when they anticommute; see Appendix~\ref{app:PauliOpsandLabels}.
Working with labels lets us express these relations using
linear algebra over $\mathbb F_2$.

A \emph{Pauli stabilizer} of a state $\ket{\xi}$ is an
operator $G\in\mathcal P_M^\pm$ satisfying
$G\ket{\xi}=\ket{\xi}$.

\begin{definition}[Stabilizer nullity~\cite{beverland2020lower}]
\label{def:stabilizer-nullity}
For a pure $M$-qubit state $\ket{\xi}$, its Pauli
stabilizer space is
\[
    \mathsf S(\xi)
    :=
    \left\{
        [P]\in\mathbb F_2^{2M}:
        P\in\mathcal P_M^\pm,\;
        P\ket{\xi}=\ket{\xi}
    \right\},
\]
and its stabilizer nullity is
$\nu_{\mathrm{st}}(\xi):=M-\dim\mathsf S(\xi)$.
\end{definition}

Notably, a pure state has stabilizer nullity zero if and only if it
is a stabilizer state; see Appendix~\ref{app:stabilizer-nullity}. Stabilizer nullity is a monotone measure of magic
\cite{veitch2014resource,gu2025magic}: for pure states,
it is invariant under Clifford operations and cannot
increase upon conditioning on a single-qubit Pauli
measurement outcome of positive probability
\cite{beverland2020lower}. A related pseudo-stabilizer property, formulated
through commuting local restrictions of the global
Pauli stabilizer group, was used in
\cite{coble2023hamiltonians} to prove the No Low-Energy
Almost-Clifford States conjecture.

We next describe the logical qubits associated with the
stabilizer constraints of $\xi$.

Set $S:=\mathsf S(\xi)$ and let
\[
    \mathcal G_\xi
    :=\left\{
        G\in\mathcal P_M^\pm:
        G\ket{\xi}=\ket{\xi}
    \right\}
\]
be its full signed Pauli stabilizer group. Its label space
is $S=\{[G]:G\in\mathcal G_\xi\}$.
The common $+1$ eigenspace of these stabilizers is
\[
    \mathcal K_S
    :=\left\{
        \ket{v}\in(\mathbb C^2)^{\otimes M}:
        G\ket{v}=\ket{v}
        \text{ for every }G\in\mathcal G_\xi
    \right\}.
\]
In particular, $\ket{\xi}\in\mathcal K_S$.
Here $\mathcal K_S$ is defined with the signed group
$\mathcal G_\xi$ fixed; $S$ alone does not specify the
signs of the constraints.

If $\dim S=M-t$, or equivalently $\nu_{\mathrm{st}}(\xi)=t$,
then $\dim\mathcal K_S=2^t$ by
Appendix~\ref{app:IsotropicSubspaces}.
Moreover, Lemma~\ref{lem:app-stabilizer-nullity-decomposition}
gives a Clifford $V_{\mathrm{glob}}\in\Cliff_M$ such that
\[
    V_{\mathrm{glob}}\mathcal K_S
    =\mathbb C^{2^t}
      \otimes\operatorname{span}\{\ket{0^{M-t}}\}.
\]
Thus every $\ket{v}\in\mathcal K_S$ has a unique expression
$V_{\mathrm{glob}}\ket{v}=\ket{u}\otimes\ket{0^{M-t}}$
with $\ket{u}\in\mathbb C^{2^t}$.
The first $t$ unrestricted qubits in this representation
are called the \emph{logical qubits} encoded by
$\mathcal K_S$, and $V_{\mathrm{glob}}$ is a \emph{decoder}.
The corresponding encoding is the isometry
\[
    \ket{u}\longmapsto
    V_{\mathrm{glob}}^\dagger
    \bigl(\ket{u}\otimes\ket{0^{M-t}}\bigr).
\]
The logical qubits need not correspond to a subset of the
original physical qubits, since $V_{\mathrm{glob}}$ may act
on all $M$ qubits. The subscript $\mathrm{glob}$ stands for
``global''. This reduction to a smaller register also underlies the
learning algorithms studied in \cite{grewal2025efficient}.

A \emph{physical Pauli} is an operator $R\in\mathcal P_M$
acting on the original $M$ qubits. It preserves
$\mathcal K_S$ if $R\ket{v}\in\mathcal K_S$ for every
$\ket{v}\in\mathcal K_S$.
This holds precisely when $R$ commutes with every element
of $\mathcal G_\xi$, or equivalently when $[R]\in S^\perp$,
where
\[
    S^\perp
    :=\left\{
        a\in\mathbb F_2^{2M}:
        \langle a,s\rangle_{\mathrm{sp}}=0
        \text{ for every }s\in S
    \right\}
\]
is the symplectic orthogonal complement of $S$.
Lemma~\ref{lem:app-stabilizer-pauli-action} shows that,
for such an $R$, there is $P_{\mathrm L}\in\mathcal P_t^\pm$
such that
\[
    V_{\mathrm{glob}}R\ket{v}
    =(P_{\mathrm L}\ket{u})
      \otimes\ket{0^{M-t}}
\]
whenever $\ket{v}\in\mathcal K_S$ and
$V_{\mathrm{glob}}\ket{v}=\ket{u}\otimes\ket{0^{M-t}}$.
Thus its restriction to $\mathcal K_S$ acts as a signed Pauli
on the $t$ logical qubits; we call
this restriction a \emph{logical Pauli operation}.
Every $t$-qubit Pauli arises in this way up to sign.

For every $G\in\mathcal G_\xi$ and
$\ket{v}\in\mathcal K_S$, we have
$RG\ket{v}=R\ket{v}$. Multiplication by a stabilizer
therefore leaves the action on $\mathcal K_S$ unchanged,
while adding an element of $S$ to the physical Pauli label.
Since the stabilizers commute, $S\subseteq S^\perp$.
Disregarding overall phases, the logical Pauli operations
are indexed by the quotient
\[
    W_S:=S^\perp/S
        =\{a+S:a\in S^\perp\},
    \qquad
    a+S:=\{a+s:s\in S\}.
\]
We call an element $a+S$ of $W_S$ a
\emph{logical Pauli class}. A physical Pauli $R$ is a
\emph{representative} of this class if $[R]\in a+S$,
or equivalently $[R]+S=a+S$; the same definition applies
to signed physical Paulis because labels ignore signs.
Representatives of the same class have the same action
on $\mathcal K_S$ up to sign, and distinct classes give
distinct logical Pauli operations up to sign.

The quotient has dimension $2t$ and carries the
nondegenerate symplectic form
\[
    \langle a+S,b+S\rangle_{\mathrm{sp}}
    :=\langle a,b\rangle_{\mathrm{sp}},
    \qquad a,b\in S^\perp,
\]
which records commutation of its physical representatives.
These statements are proved in
Appendix~\ref{app:stabilizer-nullity}; in particular,
Lemma~\ref{lem:app-stabilizer-pauli-action} identifies $W_S$
with $\mathbb F_2^{2t}$ by a linear bijection preserving
this form.

These code-space and logical-Pauli statements extend to
any group $\mathcal G$ of commuting Hermitian Pauli
operators with $-I\notin\mathcal G$, taking
$S:=\{[G]:G\in\mathcal G\}$ and $\mathcal K_S$ to be its
common $+1$ eigenspace; see
Appendix~\ref{app:stabilizer-nullity}.
We retain the notation $\dim S=M-t$ in this general setting.
The following lemma shows that, if every logical Pauli
class has a representative supported on a set $E$, then
a Clifford acting only on $E$ can recover the logical
state specified by a prescribed decoder.

\begin{lemma}[Local Clifford recovery of a prescribed logical state]
\label{lem:revmag-local-recovery}
Let $\mathcal G\subseteq\mathcal P_M^\pm$ be a group of
commuting Pauli operators with $-I\notin\mathcal G$.
Let $S$ be its label space and $\mathcal K_S$ its common
$+1$ eigenspace, with $\dim S=M-t$ and $1\le t\le M$.
Suppose $V_{\mathrm{glob}}\in\Cliff_M$ satisfies
\[
    V_{\mathrm{glob}}\mathcal K_S
    =\mathbb C^{2^t}
      \otimes\operatorname{span}\{\ket{0^{M-t}}\}.
\]
Let $E\subseteq[M]$ be such that every logical Pauli class
in $W_S=S^\perp/S$ has a representative supported on $E$.
Then there exist a Clifford $V_E$ on $E$ and an ordered set
$T=\{q_1,\ldots,q_t\}\subseteq E$ with the following property:
for every normalized $\ket{v}\in\mathcal K_S$, writing
$V_{\mathrm{glob}}\ket{v}=\ket{u}\otimes\ket{0^{M-t}}$, we have
\[
    \Tr_{E\setminus T}\left[
        V_E\bigl(\Tr_{E^c}\ket{v}\bra{v}\bigr)V_E^\dagger
    \right]
    =\ket{u}\bra{u}.
\]
On the left, the qubits of $T$ are listed in the order
$q_1,\ldots,q_t$.
\end{lemma}

\begin{proof}
For $j\in[t]$, let $X_j,Z_j\in\mathcal P_t$ apply $X,Z$
to the $j$th logical qubit and the identity to the other
logical qubits. Define the physical operators
\[
    R_{X,j}:=V_{\mathrm{glob}}^\dagger
        (X_j\otimes I_{2^{M-t}})V_{\mathrm{glob}},
    \qquad
    R_{Z,j}:=V_{\mathrm{glob}}^\dagger
        (Z_j\otimes I_{2^{M-t}})V_{\mathrm{glob}}.
\]
Clifford conjugation preserves signed Hermitian Paulis,
so $R_{X,j},R_{Z,j}\in\mathcal P_M^\pm$; see
Appendix~\ref{app:PauliOpsandLabels}.
For any $\ket{v}\in\mathcal K_S$, write
$\ket{v}=V_{\mathrm{glob}}^\dagger
(\ket{u}\otimes\ket{0^{M-t}})$. Then
\[
    R_{X,j}\ket{v}
    =V_{\mathrm{glob}}^\dagger
      (X_j\ket{u}\otimes\ket{0^{M-t}})\in\mathcal K_S.
\]
The same identity holds with $X$ replaced by $Z$.
After decoding, these actions become $X_j\ket{u}$ and
$Z_j\ket{u}$, respectively. Thus $R_{X,j}$ and $R_{Z,j}$
implement $X$ and $Z$ on the $j$th logical qubit for
every vector in $\mathcal K_S$. In particular, they
preserve $\mathcal K_S$ and their labels belong to $S^\perp$.

By the hypothesis on $E$, there are stabilizers
$G_{X,j},G_{Z,j}\in\mathcal G$ such that
\[
    \overline X_j:=R_{X,j}G_{X,j},
    \qquad
    \overline Z_j:=R_{Z,j}G_{Z,j}
\]
are supported on $E$. To justify this choice, take a representative
on $E$ of $[R_{X,j}]+S$ and choose $G_{X,j}$ with label equal
to the difference of its label and $[R_{X,j}]$; the product
has the chosen label and hence the same support. The argument
for $Z$ is identical.
Since $R_{X,j}$ commutes with $G_{X,j}$, their product
$\overline X_j$ belongs to $\mathcal P_M^\pm$; the same
argument applies to $\overline Z_j$.
These operators agree exactly with $R_{X,j},R_{Z,j}$ on
$\mathcal K_S$, because every stabilizer acts as the
identity on this space.
Clifford conjugation preserves the commutation relations
of $X_j,Z_j$, and multiplication by stabilizers preserves
them because these stabilizers commute with every
$R_{X,j},R_{Z,j}$ and with one another. Consequently,
\[
    \overline X_j\overline Z_k
    =(-1)^{\delta_{jk}}\overline Z_k\overline X_j,
    \qquad
    [\overline X_j,\overline X_k]
    =[\overline Z_j,\overline Z_k]=0.
\]
Here $[A,B]=AB-BA$ and $\delta_{jk}$ is the Kronecker delta.

Apply Lemma~\ref{lem:app-local-symplectic} to these pairs
on the subsystem $E$. It gives distinct physical qubits
$q_1,\ldots,q_t\in E$ and a Clifford $V_E$ on $E$.
With $\widehat V_E:=V_E\otimes I_{E^c}$, we have
\[
    \widehat V_E\overline X_j\widehat V_E^\dagger=X_{q_j},
    \qquad
    \widehat V_E\overline Z_j\widehat V_E^\dagger=Z_{q_j}.
\]
Here $X_{q_j},Z_{q_j}$ act on physical qubit $q_j$ and
as the identity on the other $M-1$ qubits.
Set $T:=\{q_1,\ldots,q_t\}$, ordered by the pairs.

For any $P\in\mathcal P_t$, write
$P=\gamma\prod_{j=1}^t X_j^{a_j}Z_j^{b_j}$ with
$a_j,b_j\in\{0,1\}$ and $\gamma\in\{+1,-1,+i,-i\}$,
and define
$\overline P:=\gamma\prod_{j=1}^t
\overline X_j^{a_j}\overline Z_j^{b_j}$.
All factors preserve $\mathcal K_S$, so their matching actions
on that space give
\begin{align*}
    \overline P\ket{v}
    &=V_{\mathrm{glob}}^\dagger
      (P\ket{u}\otimes\ket{0^{M-t}}),\\
    \widehat V_E\overline P\widehat V_E^\dagger
    &=P_T\otimes I_{T^c}.
\end{align*}
Here $P_T$ denotes $P$ acting on the qubits of $T$ in
the chosen order. The second identity also shows that
$\overline P\in\mathcal P_M^\pm$.
Let $\rho_T$ denote the reduced state in the conclusion.
Since $V_E$ acts only on $E$, we also have
$\rho_T=\Tr_{T^c}[\widehat V_E\ket{v}\bra{v}\widehat V_E^\dagger]$.
The preceding identities yield
\[
    \Tr(P\rho_T)
    =\langle v|\overline P|v\rangle
    =\langle u|P|u\rangle,
    \qquad P\in\mathcal P_t.
\]
The phase-free Pauli operators form a basis of the space
of matrices on $t$ qubits by
Appendix~\ref{app:PauliOpsandLabels}, so
$\rho_T=\ket{u}\bra{u}$.
\end{proof}

We apply this recovery principle to the intermediate
state $\ket{\eta}:=Q\ket{0^{n+r}}$ in a clean-output
preparation
$C\ket{\eta}=\ket{\psi}\otimes\ket{0^r}$.
This state has a depth-$D$ preparation, while Clifford
invariance of stabilizer nullity and its invariance
under appending zero-state qubits give
\[
    t:=\nu_{\mathrm{st}}(\eta)
    =\nu_{\mathrm{st}}(\psi)\le n;
\]
see Appendix~\ref{app:stabilizer-nullity}.

The following proposition combines these two facts.
The depth bound supplies a small set supporting
all logical Pauli classes, allowing the logical
state selected by decoding $\psi$ to be recovered
from that set. Pruning its preparation then gives
a circuit on at most $2^{2D+1}t$ qubits.
For fixed $D$, this is $O_D(n)$, independently
of the original number of ancillas.

\begin{proposition}[Ancilla-independent logical core reduction]
\label{prop:revmag-clean-decomposition}
Suppose $CQ\ket{0^{n+r}}=\ket{\psi}\ket{0^r}$, where
$Q$ is a matching circuit of depth at most $D$ and
$C\in\Cliff_{n+r}$.
Then there exist an integer $0\le t\le n$, a pure $t$-qubit
state $\ket{\phi}$, and a Clifford $U\in\Cliff_n$ such that
\[
    \ket{\psi}=U\bigl(\ket{\phi}\ket{0^{n-t}}\bigr).
\]
Moreover, for some integer $m$ satisfying $t\le m\le2^{2D+1}t$,
there exist a matching circuit $Q'$ of depth at most $D$
on $m$ qubits and a Clifford $V\in\Cliff_m$ such that
\[
    \ket{\phi}\bra{\phi}
    =\Tr_{\{t+1,\ldots,m\}}\left[
        VQ'\ket{0^m}\bra{0^m}Q'^\dagger V^\dagger
    \right].
\]
\end{proposition}

\begin{proof}
Set $M:=n+r$ and $\ket{\eta}:=Q\ket{0^M}$.
Clifford invariance of stabilizer nullity and its invariance
under appending zero-state qubits give
$t:=\nu_{\mathrm{st}}(\eta)=\nu_{\mathrm{st}}(\psi)\le n$.
If $t=0$, then $\psi$ is a stabilizer state, so
$\ket{\psi}=U\ket{0^n}$ for some $U\in\Cliff_n$.
The second assertion holds with $m=0$, with the zero-qubit
state and identity represented by the scalar $1$.
Assume henceforth that $t\ge1$.

Set $S:=\mathsf S(\eta)$, and let
$\mathcal G_\eta=\{G\in\mathcal P_M^\pm:G\ket{\eta}=\ket{\eta}\}$
be the full signed Pauli stabilizer group of $\eta$.
Use $\mathcal K_S$ and $W_S$ for its code space and logical
Pauli space as defined above. We have $\dim S=M-t$ and
$\ket{\eta}\in\mathcal K_S$.
We first specify the logical state to recover, then find
a small set supporting its logical Pauli operations,
apply the recovery lemma, and prune the preparation circuit.

\medskip
\noindent
\textit{1. Choosing the logical state to recover.}
By Lemma~\ref{lem:app-stabilizer-nullity-decomposition},
there exist $V_{\mathrm{out}}\in\Cliff_n$ and a pure
$t$-qubit state $\ket{\phi}$ with $\mathsf S(\phi)=\{0\}$ such that
\[
    V_{\mathrm{out}}\ket{\psi}
    =\ket{\phi}\otimes\ket{0^{n-t}}.
\]
Thus the first assertion already holds with
$U:=V_{\mathrm{out}}^\dagger$.
Define $V_{\mathrm{glob}}:=(V_{\mathrm{out}}\otimes I_{2^r})C$.
The clean-output identity gives
\[
    V_{\mathrm{glob}}\ket{\eta}
    =(V_{\mathrm{out}}\ket{\psi})\otimes\ket{0^r}
    =\ket{\phi}\otimes\ket{0^{M-t}}.
\]
For $j=t+1,\ldots,M$, the operators
$V_{\mathrm{glob}}^\dagger Z_jV_{\mathrm{glob}}$ stabilize $\eta$.
Their $M-t$ independent labels form a basis of $S$, so
they generate $\mathcal G_\eta$ by
Appendix~\ref{app:IsotropicSubspaces}. It follows that
\[
    V_{\mathrm{glob}}\mathcal K_S
    =\mathbb C^{2^t}
      \otimes\operatorname{span}\{\ket{0^{M-t}}\}.
\]
This specifies the logical state but does not locate it
on a small set of physical qubits, since $V_{\mathrm{glob}}$
may act on all $M$ qubits.

\medskip
\noindent
\textit{2. Finding a small set supporting the logical Pauli operations.}
Set $A_i:=QZ_iQ^\dagger$, $i\in[M]$.
Each $A_i$ is a Hermitian unitary fixing $\eta$.
Moreover, $\eta$ is their unique common $+1$ eigenstate
up to global phase. To see this, let $\ket{v}$ be a normalized
state satisfying $A_i\ket{v}=\ket{v}$ for every $i$.
Multiplication by $Q^\dagger$ gives
\[
    Z_iQ^\dagger\ket{v}=Q^\dagger\ket{v},
    \qquad i\in[M].
\]
These conditions force every qubit of $Q^\dagger\ket{v}$
to be in $\ket{0}$, so
$Q^\dagger\ket{v}=e^{i\theta}\ket{0^M}$ and hence
$\ket{v}=e^{i\theta}\ket{\eta}$.

Since each layer of $Q$ consists of disjoint one- and two-qubit
gates, $A_i$ is supported on at most $2^D$ qubits.
Every nonzero term in its Pauli expansion
$A_i=\sum_{R\in\mathcal P_M}a_i(R)R$ has support contained
in $\supp(A_i)$. Collect the logical classes of the terms
that commute with every stabilizer:
\[
    \mathcal L:=\left\{
        [R]+S:[R]\in S^\perp,\ R\in\mathcal P_M,\
        a_i(R)\ne0\text{ for some }i\in[M]
    \right\}.
\]
We claim that $\operatorname{span}\mathcal L=W_S$.
Otherwise, nondegeneracy of the symplectic form on $W_S$
provides a nonzero class $[P]+S$ orthogonal to this span.
Choose $P\in\mathcal P_M$ representing that class;
in particular, $P^\dagger=P$ and $P^2=I$.

For each term $R$ with $a_i(R)\ne0$, we compare its
expectations in $\eta$ and $P\eta$.
If $[R]\in S^\perp$, then $[R]+S\in\mathcal L$, so $P$
commutes with $R$ and the expectations agree.
If $[R]\notin S^\perp$, some $G\in\mathcal G_\eta$
anticommutes with $R$, giving
\[
    \langle\eta|R|\eta\rangle
    =\langle\eta|GRG|\eta\rangle
    =-\langle\eta|R|\eta\rangle=0.
\]
Since $PRP\in\{R,-R\}$, we also have
$\langle P\eta|R|P\eta\rangle=0$.
Thus every Pauli term has the same expectation
in the two states, and
$\langle P\eta|A_i|P\eta\rangle=1$ for every $i$.
As $A_i$ is a Hermitian unitary,
\[
    \|A_iP\ket{\eta}-P\ket{\eta}\|^2
    =2-2\operatorname{Re}\langle P\eta|A_i|P\eta\rangle=0.
\]
The uniqueness proved above shows that $\ket{\eta}$ is
an eigenvector of $P$. Since $P^2=I$, either $P$ or $-P$
fixes $\ket{\eta}$ and hence belongs to $\mathcal G_\eta$.
Both operators have label $[P]$, so $[P]\in S$,
contradicting the choice of a nonzero class.

Since $\dim W_S=2t$ by Appendix~\ref{app:stabilizer-nullity},
we can select $2t$ classes from $\mathcal L$ forming a basis
of $W_S$. Let $R_1,\ldots,R_{2t}$ be their corresponding
Pauli terms and set
\[
    E:=\bigcup_{j=1}^{2t}\supp(R_j),
    \qquad |E|\le2t\,2^D=2^{D+1}t.
\]
For every $R\in\mathcal P_M^\pm$ preserving $\mathcal K_S$, there are
$c_j\in\{0,1\}$ such that
\[
    [R]+S=\sum_{j=1}^{2t}c_j([R_j]+S).
\]
Thus $[R]=\sum_jc_j[R_j]+[G]$ for some
$G\in\mathcal G_\eta$, and
$R=\gamma R_1^{c_1}\cdots R_{2t}^{c_{2t}}G$
for some $\gamma\in\{+1,-1,+i,-i\}$.
Since $G$ acts as the identity on $\mathcal K_S$,
\begin{equation}
\label{eq:RepresentationofpreservingPauli}
    R\ket{v}=\gamma R_1^{c_1}\cdots R_{2t}^{c_{2t}}\ket{v},
    \qquad \ket{v}\in\mathcal K_S.
\end{equation}
The operator $\gamma R_1^{c_1}\cdots R_{2t}^{c_{2t}}=RG$
belongs to $\mathcal P_M^\pm$, because $R$ and $G$ commute.
It is supported on $E$, so every logical Pauli class has
a signed Hermitian representative supported on $E$.

\medskip
\noindent
\textit{3. Recovering the prescribed logical state from $E$.}
Apply Lemma~\ref{lem:revmag-local-recovery} to this set $E$
and the decoder $V_{\mathrm{glob}}$ chosen in Step~1.
Because $V_{\mathrm{glob}}\ket{\eta}
=\ket{\phi}\otimes\ket{0^{M-t}}$, the lemma gives
a Clifford $V_E$ and an ordered set
$T=\{q_1,\ldots,q_t\}\subseteq E$ such that, with
$\eta_E:=\Tr_{E^c}\ket{\eta}\bra{\eta}$,
\begin{equation}
\label{eq:revmag-clean-local-recovery}
    \Tr_{E\setminus T}[V_E\eta_EV_E^\dagger]
    =\ket{\phi}\bra{\phi}.
\end{equation}
The reduced state $\eta_E$ may be mixed; the identity shows
that applying $V_E$ and retaining $T$ recovers the same
pure state $\phi$ selected in Step~1.

\medskip
\noindent
\textit{4. Pruning the preparation circuit.}
Retain the gates in the backward lightcone of $E$ through $Q$.
They form a circuit $Q'$ of depth at most $D$ on $m$ qubits,
including all qubits in $E$.
Each backward layer increases the number of included qubits
by at most a factor of two. Therefore
\[
    t\le |E|\le m\le2^D|E|\le2^{2D+1}t.
\]
Index the retained qubits with $E$ first and, within $E$,
with $q_1,\ldots,q_t$ first in this order.
The pruning argument from the proof of
Lemma~\ref{lem:fixed-depth-active-cone-pruning} gives
\[
    \eta_E=\Tr_{\{|E|+1,\ldots,m\}}\left[
        Q'\ket{0^m}\bra{0^m}Q'^\dagger
    \right].
\]
The omitted gates do not affect expectations of operators
on $E$, and the retained input qubits are in $\ket{0^m}$.
Extend $V_E$ by setting
$V:=V_E\otimes I_{2^{m-|E|}}\in\Cliff_m$.
Tracing out the remaining qubits and using
\eqref{eq:revmag-clean-local-recovery} yields
\begin{align*}
    \Tr_{\{t+1,\ldots,m\}}\left[
        VQ'\ket{0^m}\bra{0^m}Q'^\dagger V^\dagger
    \right]
    =\Tr_{E\setminus T}[V_E\eta_EV_E^\dagger]
    =\ket{\phi}\bra{\phi}.
\end{align*}
Together with $U=V_{\mathrm{out}}^\dagger$ from Step~1,
this proves both assertions.
\end{proof}

\begin{lemma}[Clean-ancilla reverse-magic entropy]
\label{lem:revmag-clean-profile-entropy}
For every fixed $D\ge1$, there exists $C_D<\infty$ such that
\[
    \log N\left(
        \varepsilon,
        \mathcal F_{n,p}
        \left(
            \RevMag^{[D],\mathrm{clean}}_n
        \right),
        \|\cdot\|_{n,p}
    \right)
    \le
    C_Dn^2
    +
    C_Dn
    \log\left(
        \frac{C_D3^{p/2}n}{\varepsilon}
    \right)
\]
for every $n\ge p\ge2$ and $0<\varepsilon<1$.
\end{lemma}

\begin{proof}
Fix $n\ge p\ge2$ and $0<\varepsilon<1$.
We first use Proposition~\ref{prop:revmag-clean-decomposition}
to represent every state in the class using circuits on a
common number of qubits. We then count the Clifford choices,
discretize the shallow preparation circuit, and bound the
resulting error in the profile norm.

\textit{1. Reducing to a common circuit width: }
Set $m_*:=2^{2D+1}n$.
By Proposition~\ref{prop:revmag-clean-decomposition}, every
$\ket{\psi}\in\RevMag^{[D],\mathrm{clean}}_n$ can be written as
$\ket{\psi}=U(\ket{\phi}\ket{0^{n-t}})$, where
$0\le t\le n$, $U\in\Cliff_n$, and $\ket{\phi}$ is obtained
by retaining the first $t$ qubits of a matching circuit of depth at most $D$ followed by a Clifford on at most
$2^{2D+1}t\le m_*$ qubits.
Append unused qubits in $\ket{0}$ and extend both circuits
by the identity to obtain a realization on exactly $m_*$
qubits. This preserves the depth and the reduced state on
the first $t$ qubits.

We now express the reduced logical state and the final
$n$-qubit state as density matrices. Set $\omega_*:=\ket{0^{m_*}}\bra{0^{m_*}}$ and, for
$U\in\Cliff_n$, $V\in\Cliff_{m_*}$, and a matching circuit
$Q$ of depth at most $D$ on $m_*$ qubits, define
\begin{align*}
    \sigma_t(V,Q)
    &:=\Tr_{\{t+1,\ldots,m_*\}}
        \left[VQ\omega_*Q^\dagger V^\dagger\right],\\
    \rho_t(U,V,Q)
    &:=U\left(
        \sigma_t(V,Q)\otimes
        \ket{0^{n-t}}\bra{0^{n-t}}
    \right)U^\dagger.
\end{align*}
Every target density matrix $\ket{\psi}\bra{\psi}$ has this
form for some $t,U,V,Q$.
For $t=0$, the partial trace defining $\sigma_0(V,Q)$ is
the scalar $1$, so $\rho_0(U,V,Q)=U\ket{0^n}\bra{0^n}U^\dagger$.
Thus the same family also covers the zero-nullity case.
The full family may contain mixed states; we will account
for this when choosing the centers of the final cover.

For $k=n,m_*$, let $\mathcal C_k$ contain one unitary
representative of each $k$-qubit Clifford modulo global phase.
The Clifford count used in the proof of
Lemma~\ref{lem:higher-magic-entropy} gives $|\mathcal C_k|\le\exp(Ck^2)$.

Global phases cancel in the density matrices above, so
we may restrict $U$ and $V$ to $\mathcal C_n$ and
$\mathcal C_{m_*}$, respectively, without changing the family.

\textit{2. Covering the preparation circuits: }
An architecture specifies the gate locations in each layer.
For each qubit in a layer, there are $m_*-1$ possible
partners for a two-qubit gate, one possibility for a
one-qubit gate, and one for no gate.
Every valid layer determines a labeling with $m_*+1$
choices per qubit. There are therefore at most
$(m_*+1)^{Dm_*}$ architectures with $D$ layers, including
circuits of smaller depth by allowing empty layers.
Since gate supports within a layer are disjoint, each
architecture has at most $Dm_*$ gate locations.

Fix $0<\delta<1$, to be chosen below.
As in Lemma~\ref{lem:gate-count-profile-entropy}, choose
operator-norm $\delta$-nets for $\mathsf U(2)$ and
$\mathsf U(4)$, with centers in the respective unitary
groups and cardinalities at most $(C/\delta)^C$.
For each architecture, replace each gate by a point in
the corresponding net. Taking the union over all
architectures gives a finite set $\mathcal Q$ satisfying
\begin{equation}
\label{eq:revmag-clean-circuit-cover-size}
    |\mathcal Q|
    \le
    (m_*+1)^{Dm_*}(C/\delta)^{CDm_*}.
\end{equation}
For every $Q$, there is $\widetilde Q\in\mathcal Q$ with
the same architecture and with corresponding gates
differing by at most $\delta$ in operator norm.

To bound the accumulated error, write
$Q=U_L\cdots U_1$ and
$\widetilde Q=\widetilde U_L\cdots\widetilde U_1$,
where $L\le Dm_*$ and the same ordering of gate locations
is used in both circuits. Each gate is extended by the
identity outside its support. Replacing the gates one
at a time gives
\[
    Q-\widetilde Q
    =\sum_{j=1}^L
        \widetilde U_L\cdots\widetilde U_{j+1}
        (U_j-\widetilde U_j)
        U_{j-1}\cdots U_1,
\]
where empty products are identities.
The factors surrounding $U_j-\widetilde U_j$ are unitary,
so the triangle inequality gives
\[
    \|Q-\widetilde Q\|_{\mathrm{op}}
    \le\sum_{j=1}^L
        \|U_j-\widetilde U_j\|_{\mathrm{op}}
    \le L\delta\le Dm_*\delta.
\]

\textit{3. Transferring the circuit approximation to profiles: }
Fix $t,U,V,Q$ and choose $\widetilde Q$ as above.
Set $\rho:=\rho_t(U,V,Q)$ and
$\widetilde\rho:=\rho_t(U,V,\widetilde Q)$.
Conjugation by a unitary and tensoring with a fixed pure
state preserve the trace norm. Partial trace contracts
the trace norm of the difference of two states.
Keeping $U$ and $V$ unchanged therefore gives
\begin{align*}
    \|\rho-\widetilde\rho\|_1
    &=\|\sigma_t(V,Q)-\sigma_t(V,\widetilde Q)\|_1
    \le\|Q\omega_*Q^\dagger
          -\widetilde Q\omega_*\widetilde Q^\dagger\|_1\\
    &\le\|(Q-\widetilde Q)\omega_*Q^\dagger\|_1
       +\|\widetilde Q\omega_*
           (Q^\dagger-\widetilde Q^\dagger)\|_1
    \le 2\|Q-\widetilde Q\|_{\mathrm{op}}
     \le 2Dm_*\delta,
\end{align*}
where the penultimate inequality uses $\|\omega_*\|_1=1$
and the operator norm one of the unitary factors.

For each $P\in\mathcal P_{n,p}$, we have
$|\Tr(P(\rho-\widetilde\rho))|
\le\|\rho-\widetilde\rho\|_1$, since
$\|P\|_{\mathrm{op}}=1$.
Using $|\mathcal P_{n,p}|=3^p\binom np$ in the definition
of the profile norm gives
\begin{align*}
    \|f_\rho-f_{\widetilde\rho}\|_{n,p}^2
    =\binom np^{-1}
      \sum_{P\in\mathcal P_{n,p}}
        |\Tr(P(\rho-\widetilde\rho))|^2
    \le 3^p\|\rho-\widetilde\rho\|_1^2.
\end{align*}
Consequently, taking
$\delta:=\varepsilon/(8\cdot3^{p/2}Dm_*)$ ensures that
\begin{equation}
\label{eq:revmag-clean-profile-approximation}
    \|f_\rho-f_{\widetilde\rho}\|_{n,p}
    \le 2\cdot3^{p/2}Dm_*\delta
    =\frac{\varepsilon}{4}.
\end{equation}
This estimate holds uniformly over $t,U,V,Q$.

\textit{4. Combining the estimates: }
Let $\mathcal F:=
\mathcal F_{n,p}(\RevMag^{[D],\mathrm{clean}}_n)$ and form
the finite set
\[
    \mathcal H:=\left\{
        f_{\rho_t(U,V,\widetilde Q)}:
        0\le t\le n,\quad
        U\in\mathcal C_n,\quad
        V\in\mathcal C_{m_*},\quad
        \widetilde Q\in\mathcal Q
    \right\}.
\]
By the representation above and
\eqref{eq:revmag-clean-profile-approximation},
$\mathcal H$ is an $\varepsilon/4$-cover, and hence an
$\varepsilon$-cover, of $\mathcal F$.
The centers need not belong to $\mathcal F$ under our
covering-number convention, and their number satisfies
\[
    |\mathcal H|
    \le(n+1)|\mathcal C_n|\,|\mathcal C_{m_*}|\,|\mathcal Q|.
\]
Combining the Clifford counts with
\eqref{eq:revmag-clean-circuit-cover-size} gives
\begin{align*}
    \log N(\varepsilon,\mathcal F,\|\cdot\|_{n,p})
    &\le\log(n+1)+C(n^2+m_*^2)
       +Dm_*\log(m_*+1)
       +CDm_*\log(C/\delta)\\
    &\le\log(n+1)+C(n^2+m_*^2)
       +Dm_*\log(m_*+1)
       +CDm_*\log\left(
           \frac{C3^{p/2}Dm_*}{\varepsilon}
       \right).
\end{align*}
Since $m_*=2^{2D+1}n$, increasing a constant $C_D$
depending only on $D$ yields
\[
    \log N(\varepsilon,\mathcal F,\|\cdot\|_{n,p})
    \le C_Dn^2
       +C_Dn\log\left(
           \frac{C_D3^{p/2}n}{\varepsilon}
       \right).
\]
Here the architecture count and the factor $n+1$ have
been absorbed into the right-hand side.
The representation and the cover apply to every original
ancillary width, so the bound is uniform over the number
of clean ancillas.
\end{proof}

\begin{theorem}[Reverse magic with arbitrary clean ancillas]
\label{thm:revmag-clean-ground-gap}
Fix an integer $D\ge1$.  There exists $p_0=p_0(D)<\infty$ such that,
for every fixed $p\ge p_0$, there is $c_{p,D}>0$ for which
\[
    \P\left(
        \Delta^{\mathrm{opt}}_{n,p}
        \left(
            \RevMag^{[D],\mathrm{clean}}_n
        \right)
        \ge
        c_{p,D}\sqrt n
    \right)
    \longrightarrow1.
\]
The conclusion is uniform over the number of clean ancillas.
\end{theorem}

\begin{proof}
Fix $D\ge1$ and choose $C_D\ge1$ so that the bound in
Lemma~\ref{lem:revmag-clean-profile-entropy} holds for all
$n\ge p\ge2$ and $0<\varepsilon<1$.
To apply Theorem~\ref{thm:local-profile-complexity-principle}
to the class $\RevMag^{[D],\mathrm{clean}}_n$, take
\[
    a_{n,p}:=C_Dn^2,
    \qquad
    b_{n,p}:=C_Dn,
    \qquad
    B_{n,p}:=C_D3^{p/2}n.
\]
These parameters satisfy $a_{n,p},b_{n,p}\ge0$ and
$B_{n,p}\ge1$, and the entropy lemma gives precisely
the required bound
$a_{n,p}+b_{n,p}\log(B_{n,p}/\varepsilon)$.
The corresponding effective profile complexity is
\[
    \Gamma_{n,p}
    :=a_{n,p}+b_{n,p}\log(eB_{n,p}n)
    =C_Dn^2+C_Dn\log\left(eC_D3^{p/2}n^2\right).
\]
For every fixed $p$, we therefore have
\[
    \frac{\Gamma_{n,p}}{n^2}
    =C_D+\frac{C_D}{n}
        \left(2\log n+\log(eC_D)+\frac p2\log3\right)
    \longrightarrow C_D.
\]
In particular, the limiting bound is independent of $p$.
Theorem~\ref{thm:local-profile-complexity-principle}\textup{(i)},
applied with $L=C_D$, gives the claimed ground-state
separation, with a threshold $p_0$ depending only on $D$.
Since the entropy bound holds for the full class with
arbitrarily many clean ancillas, the separation holds
simultaneously for every prescribed ancillary width.
\end{proof}

\begin{proof}[Proof of Corollary~\ref{cor:intro-magic}]
\textup{(i)} follows from
Proposition~\ref{prop:higher-magic-ground-gap}, which covers
both orientations for every fixed number of shallow blocks.

For \textup{(ii)}, the linear-width ground-state
conclusion follows from
Proposition~\ref{prop:dense-magic}\textup{(i)}.
The comparison with $\Prod_n$ in the presence of arbitrary
discardable ancillas follows from
Proposition~\ref{prop:dense-magic}\textup{(ii)}
under the stated logarithmic-depth assumption.

For \textup{(iii)}, the linear-width conclusion follows
from Proposition~\ref{prop:higher-magic-ground-gap} with $k=1$,
and the fixed-depth conclusion with arbitrary clean ancillas
is Theorem~\ref{thm:revmag-clean-ground-gap}.
In each linear-width application, choose a fixed $K$ such
that $M_n\le Kn$ for all sufficiently large $n$.

Depth-zero blocks are included in the corresponding
depth-one classes. Replacing $D_n$ by $\max\{D_n,1\}$
preserves the stated depth bounds for all sufficiently
large $n$ when the prefactors are positive.
When $L=0$, apply the depth-one bound with $L=1$.
The same inclusion covers the clean-ancilla case at depth zero.
The conclusions pass to subclasses because restricting
the state class can only increase $\Delta^{\mathrm{opt}}_{n,p}$
and can only decrease
$\Delta^{\mathrm{ref}}_{n,p}(\,\cdot\,;\Prod_n)$.
\end{proof}

\section{Stabilizer states and Clifford+$T$ circuits}
\label{sec:sharper-conditioning-benchmarks}

In this section, we prove Corollary~\ref{cor:intro-t-count}. For integers $n\ge1$ and $t\ge0$, define the class of Clifford+$T$ states as
\[
    \CliffT^{\mathrm{anc}}_{n,t}
    :=
    \bigcup_{M\ge n}
    \left\{
        \Tr_{\{n+1,\ldots,M\}}
        \bigl(U|0^M\rangle\langle0^M|U^\dagger\bigr)
    \right\},
\]
where, for each $M$, $U$ ranges over circuits of the form
\[
    U=C_sT_{i_s}C_{s-1}\cdots T_{i_1}C_0,
    \qquad s\in\{0,\ldots,t\},
\]
with $C_0,\ldots,C_s\in\Cliff_M$ and
$i_1,\ldots,i_s\in[M]$. For $s=0$, we set $U=C_0$.
The gate $T_i$ acts on qubit $i$ and is defined by
\[
    T_i:=I^{\otimes(i-1)}\otimes T\otimes I^{\otimes(M-i)},
    \qquad
    T:=\begin{pmatrix}
        1&0\\
        0&e^{i\pi/4}
    \end{pmatrix}.
\]
Its inverse is
$T^\dagger=\operatorname{diag}(1,e^{-i\pi/4})$.
Equivalently, we may allow both $T$- and $T^\dagger$-gates,
counting each toward the bound $t$.
This leaves the class unchanged because
$T^\dagger=\operatorname{diag}(1,-i)T$,
and the Clifford factor $\operatorname{diag}(1,-i)$
can be absorbed into an adjacent Clifford block.

The bound $t$ counts both types of gates on retained and
discarded qubits. Discarding ancillas allows mixed output
states. Since Clifford operations and the number of
discardable ancillas are unrestricted, we develop a
representation of the output states whose size depends
only on $n$ and $t$, independently of the original
circuit width.

\subsection{Overview of the argument}
We use the same conditioning and profile-entropy framework as
in the previous sections, but exploit the structure of
Clifford+$T$ outputs to obtain a sharper benchmark than the unrestricted product states. We first explain the argument before giving the formal
statements and proofs.

Recall how the product approximation is constructed. Lemma~\ref{lem:measurement-parametrized-definetti}, applied with
the Pauli POVM from \eqref{eq:pauli-conditioning-measurement},
gives a set $J$ of unmeasured qubits and the state
\[
    \eta_J
    =\sum_{z:q(z)>0}q(z)\bigotimes_{i\in J}\rho_{i\mid z},
    \qquad
    \rho_{i\mid z}
    =\Tr_{J\setminus\{i\}}(\rho_{J\mid z}).
\]
Here $q(z)$ is the distribution of the outcomes obtained by
measuring $J^c$, and $\rho_{J\mid z}$ is the corresponding
conditional state. Lemma~\ref{lem:measurement-parametrized-definetti} controls
the average squared trace-norm error $\|\rho_I-(\eta_J)_I\|_1^2$ over all $p$-element subsets
$I\subseteq J$.

In the previous applications, it was enough that
$\eta_J\in\operatorname{conv}(\Prod_J)$. For Clifford+$T$ outputs, we can say more about the conditional
factors $\rho_{i\mid z}$.

First consider $t=0$, so the circuit contains only Clifford
gates. Its full output, before any qubits are discarded, is
a pure stabilizer state. We measure the qubits in $J^c$ using
the Pauli POVM from \eqref{eq:pauli-conditioning-measurement}
and condition on their joint outcome $z$. Conditioning a pure stabilizer state on single-qubit Pauli measurement outcomes of positive probability leaves a pure
stabilizer state on the unmeasured qubits
\cite[Section~10.5.3]{nielsen2010quantum}; see also Lemma~\ref{lem:app-nullity-conditioning}\textup{(ii)}. 

Therefore, letting $E$ denote the ancillas, we have that the joint conditional state
$\rho_{JE\mid z}$ is a pure stabilizer state. Its reduced
state on qubit $i\in J$ is $\rho_{i\mid z}
    =\Tr_{(J\setminus\{i\})\cup E}(\rho_{JE\mid z})$. To identify this marginal, write its Pauli expansion as
\[
    \rho_{i\mid z}
    =\frac12(I_i+r_X X_i+r_Y Y_i+r_Z Z_i),
    \qquad
    r_a:=\Tr(a_i\rho_{i\mid z}).
\]
Since the joint state is a stabilizer state, each Pauli
expectation $r_a$ belongs to $\{0,\pm1\}$.
Moreover, positivity of $\rho_{i\mid z}$ requires
$r_X^2+r_Y^2+r_Z^2\le1$. Thus either all three coefficients
vanish, or exactly one equals $\pm1$, giving $\rho_{i\mid z}=I_i/ 2$, or $\rho_{i\mid z}=(I_i+s a_i)/2$ for $a\in\{X,Y,Z\}$, $s\in\{\pm1\}$. The latter is the rank-one projector onto the eigenstate
of $a$ with eigenvalue $s$. These six eigenstates are
precisely the pure one-qubit stabilizer states
\cite[Section~10.5.1]{nielsen2010quantum}.
Since $I_i/2=(|0\rangle\langle0|_i+
|1\rangle\langle1|_i)/2$, every $\rho_{i\mid z}$ is therefore
a mixture of one-qubit stabilizer states.

Products of these pure states belong to
$\PStab_A:=\Stab_A\cap\Prod_A$, where $A$ indexes the qubits
in the product. Consequently,
\[
    \bigotimes_{i\in J}\rho_{i\mid z}
    \in\operatorname{conv}(\PStab_J),
\]
and averaging over the outcomes gives
\[
    \eta_J
    =\sum_{z:q(z)>0}q(z)\bigotimes_{i\in J}\rho_{i\mid z}
    \in\operatorname{conv}(\PStab_J).
\]
Thus $\eta_J$ belongs to the smaller class of mixtures of
product stabilizer states, sharpening the general inclusion
$\eta_J\in\operatorname{conv}(\Prod_J)$ used previously.

When $t>0$, the full output need not be a stabilizer state,
and the conditional marginals $\rho_{i\mid z}$, $i\in J$,
need not be mixtures of pure one-qubit stabilizer states.
Nevertheless, Lemma~\ref{lem:nullity-preserving-conditioning}
shows that, for each outcome $z$ with $q(z)>0$, all but
at most $t$ of these marginals still have this property. We refer to the remaining at most $t$ sites as being \textit{exceptional}.

Fix such an outcome $z$. For each $\rho_{i\mid z}$ that is
a mixture of pure one-qubit stabilizer states, choose such
a decomposition. For the remaining at most $t$ exceptional marginals,
use their spectral decompositions to express them as
mixtures of arbitrary pure one-qubit states. Taking their tensor product therefore gives
a mixture of pure product states, each with at most $t$
factors that are not Pauli eigenstates.
Writing $\mathcal B_{J,t}$ for this class of pure product
states, we obtain
\[
    \bigotimes_{i\in J}\rho_{i\mid z}
    \in\operatorname{conv}(\mathcal B_{J,t}).
\]
Averaging over $z$ then gives
\[
    \eta_J
    =\sum_{z:q(z)>0}q(z)\bigotimes_{i\in J}\rho_{i\mid z}
    \in\operatorname{conv}(\mathcal B_{J,t}).
\]
The exceptional sites may vary between terms in the mixture.
Thus $\mathcal B_{J,t}$ enlarges $\PStab_J$ by allowing
at most $t$ arbitrary factors, with
\[
    \mathcal B_{J,0}=\PStab_J,
    \qquad
    \PStab_J\subseteq\mathcal B_{J,t}\subseteq\Prod_J.
\]

The profile-entropy estimate in
Lemma~\ref{lem:t-count-profile-entropy} allows us to apply
Theorem~\ref{thm:local-profile-complexity-principle}\textup{(ii)}
with this benchmark when $t=o(n)$.
Lemma~\ref{lem:pstab-exceptional-site-cost} then shows that,
for fixed $p$, allowing the $t$ exceptional factors increases
the expected optimal energy by only $o(\sqrt n)$.
Together, these results give the comparison with $\PStab_n$.

Finally, Proposition~\ref{prop:pstab-value} determines the
sharp large-$p$ limit of the normalized product-stabilizer
energy as $\sqrt{2\log6}$. This result is of independent
interest and, together with the comparison above, yields
the large-$p$ conclusion of
Corollary~\ref{cor:intro-t-count}\textup{(ii)}. 
The constant $\sqrt{2\log6}$ is also the limiting expected
maximum of $6^n$ independent standard Gaussian energies,
divided by $\sqrt n$, in the random energy model.
Our comparison concerns the sequential limit
$n\to\infty$ followed by $p\to\infty$, although the
product-stabilizer energies are correlated at fixed $p$.
Related large-$p$ connections between spin glasses and
the random energy model are discussed in
\cite{gamarnik2025shattering,anschuetz2025bounds}.

\subsection{Pauli conditioning and sharper benchmarks}
\label{subsec:stabilizer-conditioning}

For $a\in\{X,Y,Z\}$ and $s\in\{\pm1\}$, let $|a,s\rangle$ be a normalized eigenvector of $a$ with
eigenvalue $s$, so that
\[
    a|a,s\rangle=s|a,s\rangle,
    \qquad
    \langle a,s|a,s\rangle=1.
\]
For a finite set $A$ of qubits, write
\[
    \PStab_A
    :=
    \left\{
        \bigotimes_{i\in A}|a_i,s_i\rangle\langle a_i,s_i|:
        a_i\in\{X,Y,Z\},\ s_i\in\{\pm1\}
    \right\}.
\]

For a pure product state $\sigma\in\Prod_A$, define its set of exceptional sites by
\[
    E_A(\sigma)
    :=\left\{
        i\in A:\sigma_i\notin\PStab_{\{i\}}
    \right\},
\]
where $\sigma_i:=\Tr_{A\setminus\{i\}}(\sigma)$.
For an integer $t\ge0$, set
\[
    \mathcal B_{A,t}
    :=\left\{
        \sigma\in\Prod_A:|E_A(\sigma)|\le t
    \right\},
    \qquad
    \mathcal B_{n,t}:=\mathcal B_{[n],t}.
\]
Thus $\mathcal B_{A,t}$ contains all pure product states on
$A$ with at most $t$ factors that are not one-qubit
stabilizer states.

Recall from Appendix~\ref{app:stabilizer-nullity} that the
stabilizer nullity of a pure $M$-qubit state $|\Psi\rangle$ is
$\nu_{\mathrm{st}}(\Psi)=M-\dim\mathsf S(\Psi)$.
For the full circuit output, this is at most the $T$-count
by Lemma~\ref{lem:app-nullity-conditioning}\textup{(i)}.
We use this bound to control the number of exceptional factors. By Corollary~\ref{cor:conditioning-extensions}, applying the
Pauli POVM while retaining the ancillas gives a pure joint
conditional state, whose output marginals agree with those
obtained by conditioning $\rho$ directly.
The following lemma turns this nullity bound into an
approximation by $\operatorname{conv}(\mathcal B_{n,t})$.

\begin{lemma}[Product approximation after conditioning]
\label{lem:nullity-preserving-conditioning}
Let $n\ge p\ge2$ and $t\ge0$ be integers, and suppose $\rho=\Tr_E|\Psi\rangle\langle\Psi|$ with $\nu_{\mathrm{st}}(\Psi)\le t$, where $|\Psi\rangle$ is a pure state on the $n$ qubits and finitely many additional qubits $E$.
For every integer $0\le m_\star\le n-p$, there are
$J\subseteq[n]$ and $\sigma\in\operatorname{conv}(\mathcal B_{n,t})$
such that $|J|\ge n-m_\star$ and 
\[
    \|\mathsf M_J(f_\rho-f_\sigma)\|_{n,p}
    \le c\frac{2^{p/2}p^{3/2}}{\sqrt{m_\star+1}},
\]
where $c>0$ is a universal constant.
\end{lemma}

\begin{proof}
We first construct a mixture of product states approximating
the local marginals of $\rho$. We then show that each product
state in the mixture can be chosen with at most $t$ factors
that are not one-qubit stabilizer states.

\textit{1. Construct the conditioning approximation.}
Apply Lemma~\ref{lem:measurement-parametrized-definetti}
to $\rho$ with measurement budget $m_\star$ and the Pauli POVM in \eqref{eq:pauli-conditioning-measurement}.
This gives a set $S$ of measured qubits with $|S|\le m_\star$.
Let $J=[n]\setminus S$ denote the unmeasured qubits.

Let $q(z)$ and $\rho_{J\mid z}$ be the outcome probabilities
and conditional states defined in
\eqref{eq:conditioning-state-convention}.
For each $z$ with $q(z)>0$, define
\begin{equation}
\label{eq:t-count-conditional-products}
    \rho_{i\mid z}
    :=\Tr_{J\setminus\{i\}}(\rho_{J\mid z}),
    \qquad
    \theta_{J,z}
    :=\bigotimes_{i\in J}\rho_{i\mid z}.
\end{equation}
Thus $\theta_{J,z}$ is the product of the one-qubit marginals
of $\rho_{J\mid z}$. Averaging these products gives
\[
    \eta_J:=\sum_{z:q(z)>0}q(z)\theta_{J,z}.
\]
Lemma~\ref{lem:measurement-parametrized-definetti} then guarantees
\begin{equation}
\label{eq:t-count-conditioned-marginal-error}
    \frac1{\binom{|J|}{p}}
    \sum_{\substack{I\subseteq J\\|I|=p}}
        \|\rho_I-(\eta_J)_I\|_1^2
    \le C\frac{p^3}{m_\star+1}.
\end{equation}

\textit{2. Identify the exceptional factors.}
We now perform the same measurements on $S$ in the full pure state
$|\Psi\rangle\langle\Psi|$, retaining the extra qubits $E$.
By Corollary~\ref{cor:conditioning-extensions}, the outcome
probabilities are again $q(z)$. Since the Pauli POVM has
rank-one effects, the joint conditional state is pure.
Writing it as
$\widehat\rho_{JE\mid z}=|\Psi_z\rangle\langle\Psi_z|$, we have
\begin{equation}
\label{eq:t-count-conditional-marginals}
\begin{aligned}
    \rho_{J\mid z}
    &=\Tr_E\widehat\rho_{JE\mid z},\\
    \rho_{i\mid z}
    &=\Tr_{J\setminus\{i\}}(\rho_{J\mid z})
      =\Tr_{(J\setminus\{i\})\cup E}
          \widehat\rho_{JE\mid z},
    \qquad i\in J.
\end{aligned}
\end{equation}
By Lemma~\ref{lem:app-nullity-conditioning}\textup{(ii)}, conditioning on Pauli measurement outcomes does not increase
nullity, so
\[
    \nu_{\mathrm{st}}(\Psi_z)
    \le\nu_{\mathrm{st}}(\Psi)\le t.
\]
Applying part~\textup{(iii)} of that lemma to $|\Psi_z\rangle$
and using \eqref{eq:t-count-conditional-marginals} gives
\[
    F_z:=\left\{
        i\in J:
        \rho_{i\mid z}\notin\operatorname{conv}(\PStab_{\{i\}})
    \right\},
    \qquad
    |F_z|\le\nu_{\mathrm{st}}(\Psi_z)\le t.
\]
Here $F_z$ is the set of sites $i\in J$ whose conditional
marginals $\rho_{i\mid z}$ cannot be written as mixtures of
pure one-qubit stabilizer states. These are the sites at
which the decomposition below requires non-stabilizer factors.

Fix an outcome $z$ with $q(z)>0$. For $i\in J\setminus F_z$,
decompose $\rho_{i\mid z}$ into pure one-qubit stabilizer states.
For $i\in F_z$, use its spectral decomposition into arbitrary
pure states. In both cases, write
\[
    \rho_{i\mid z}
    =\sum_\ell\lambda_{i,z,\ell}
        |\phi_{i,z,\ell}\rangle\langle\phi_{i,z,\ell}|,
    \qquad
    \lambda_{i,z,\ell}\ge0,\quad
    \sum_\ell\lambda_{i,z,\ell}=1.
\]
Expanding the tensor product in
\eqref{eq:t-count-conditional-products} gives
\[
    \theta_{J,z}
    =
    \sum_{(\ell_i)_{i\in J}}
        \left(\prod_{i\in J}\lambda_{i,z,\ell_i}\right)
        \bigotimes_{i\in J}
        |\phi_{i,z,\ell_i}\rangle\langle\phi_{i,z,\ell_i}|.
\]
For every pure product state $\tau$ appearing in this sum,
\[
    E_J(\tau)\subseteq F_z,
    \qquad
    |E_J(\tau)|\le t,
    \qquad
    \tau\in\mathcal B_{J,t}.
\]
The coefficients are nonnegative and sum to one. Consequently,
\[
    \theta_{J,z}\in\operatorname{conv}(\mathcal B_{J,t}),
    \qquad
    \eta_J=\sum_{z:q(z)>0}q(z)\theta_{J,z}
    \in\operatorname{conv}(\mathcal B_{J,t}).
\]

\textit{3. Extend the approximation and bound the profile error.}
Append $|0\rangle\langle0|$ at every site in $S$, defining
\[
    \sigma
    :=\eta_J\otimes\bigotimes_{i\in S}|0\rangle\langle0|_i.
\]
These additional factors are one-qubit stabilizer states, so
\[
    \sigma\in\operatorname{conv}(\mathcal B_{n,t}),
    \qquad
    \sigma_I=(\eta_J)_I
    \quad\text{for }I\subseteq J.
\]
Pauli orthogonality and
\eqref{eq:t-count-conditioned-marginal-error} therefore give
\begin{align*}
    \|\mathsf M_J(f_\rho-f_\sigma)\|_{n,p}^2
    &\le\frac{2^p}{\binom np}
        \sum_{\substack{I\subseteq J\\|I|=p}}
        \|\rho_I-(\eta_J)_I\|_1^2
    \le c\frac{2^pp^3}{m_\star+1}
        \frac{\binom{|J|}{p}}{\binom np}
    \le c\frac{2^pp^3}{m_\star+1}.
\end{align*}
Taking square roots proves the claim.
\end{proof}

The following shows that Clifford+$T$ outputs with at most $t_n$ $T$-gates admit
conditioning approximations by mixtures of product states,
each with at most $t_n$ non-stabilizer factors.

\begin{corollary}[Conditioning benchmark for Clifford+$T$ outputs]
\label{cor:t-count-conditioning-benchmark}
For every sequence $(t_n)_{n\ge1}$ of nonnegative integers,
$(\mathcal B_{n,t_n})_{n\ge1}$ is a conditioning benchmark for
$(\CliffT^{\mathrm{anc}}_{n,t_n})_{n\ge1}$ in the sense of
Definition~\ref{def:conditioning-benchmark}.
\end{corollary}

\begin{proof}
Fix $p\ge2$ and $\alpha\in(0,1)$.
Let $\rho\in\CliffT^{\mathrm{anc}}_{n,t_n}$, and let
$|\Psi\rangle$ be its full circuit output before the
additional qubits $E$ are discarded.
The initial state has nullity zero, Clifford gates preserve
nullity, and each $T$-gate increases it by at most one,
by Lemma~\ref{lem:app-nullity-conditioning}\textup{(i)}.
Consequently,
\[
    \rho=\Tr_E|\Psi\rangle\langle\Psi|,
    \qquad
    \nu_{\mathrm{st}}(\Psi)\le t_n.
\]

We may therefore apply
Lemma~\ref{lem:nullity-preserving-conditioning}
with $t=t_n$ and $m_\star=\lfloor\alpha n\rfloor$.
For all sufficiently large $n$, this budget satisfies
$m_\star\le n-p$. The lemma gives $J\subseteq[n]$ and
$\sigma\in\operatorname{conv}(\mathcal B_{n,t_n})$ such that
\[
    |J|\ge n-\lfloor\alpha n\rfloor\ge(1-\alpha)n
\]
and
\[
    \|\mathsf M_J(f_\rho-f_\sigma)\|_{n,p}^2
    \le
    C\frac{2^pp^3}{\lfloor\alpha n\rfloor+1}
    \le
    C\frac{2^pp^3}{\alpha n}
    =O_{p,\alpha}(n^{-1}).
\]
The bound is uniform over $\rho$ and independent of $t_n$,
which verifies the conditioning-benchmark property.
\end{proof}

\subsection{Profile entropy of Clifford+$T$ circuits} 
\label{subsec:t-count-profile-entropy}

Corollary~\ref{cor:t-count-conditioning-benchmark} shows that
$\mathcal B_{n,t_n}$ is a conditioning benchmark for
$\CliffT^{\mathrm{anc}}_{n,t_n}$.
The remaining input to
Theorem~\ref{thm:local-profile-complexity-principle}\textup{(ii)}
is a bound on the profile entropy of these circuit outputs.

To obtain this bound, we reformulate the postselected $T$-gate
construction used in \cite[Section~IV]{bravyi2016improved}
to express every $\rho\in\CliffT^{\mathrm{anc}}_{n,t}$ as a
fixed projection of a mixed stabilizer state on $n+t$ qubits. A related representation of quantum-verification acceptance
operators as contractions of stabilizer projectors with magic
states appears in \cite[Section~2]{arunachalam2022parameterized}. Our representation allows us to apply the mixed-stabilizer counting bound
of Lemma~\ref{lem:fixed-depth-mixed-stabilizer-growth}
to the output profiles, independently of the number of
original ancillas. We include a proof to make the normalization
and the treatment of discarded ancillary qubits explicit.

\begin{lemma}[Fixed-projection representation]
\label{lem:t-count-fixed-projection}
Set
\[
    |a\rangle:=\frac{|0\rangle+e^{-i\pi/4}|1\rangle}{\sqrt2},
    \qquad \Pi_a:=|a\rangle\langle a|.
\]
For every $\rho\in\CliffT^{\mathrm{anc}}_{n,t}$, there exists
$\tau\in\mathsf{mixstab}_{n+t}$ such that
\begin{equation}
\label{eq:t-count-fixed-projection}
    \rho
    =2^t\bigl(I_{2^n}\otimes\langle a|^{\otimes t}\bigr)
      \tau\bigl(I_{2^n}\otimes|a\rangle^{\otimes t}\bigr).
\end{equation}
In particular,
$\Tr[(I_{2^n}\otimes\Pi_a^{\otimes t})\tau]=2^{-t}$.
\end{lemma}

\begin{proof}
We first express a single $T$-gate using an extra qubit, a CNOT,
and a fixed contraction. We then apply this construction throughout
the circuit and trace out the original ancillary qubits.

\textit{1. Replacing a single $T$-gate.}
Let $|\psi\rangle_i=\alpha|0\rangle_i+\beta|1\rangle_i$ be an arbitrary
one-qubit state. By definition,
\[
    T_i|\psi\rangle_i
    =
    \alpha|0\rangle_i+e^{i\pi/4}\beta|1\rangle_i.
\]

Appending an extra qubit $j$ in the state $|0\rangle_j$ gives
\[
    (I\otimes|0\rangle_j)|\psi\rangle_i
    =
    |\psi\rangle_i\otimes|0\rangle_j
    =
    \alpha|0\rangle_i|0\rangle_j
    +\beta|1\rangle_i|0\rangle_j.
\]
The controlled-NOT gate with control qubit $i$ and target qubit $j$
is defined by
\[
    \operatorname{CNOT}_{i\to j}
    (|x\rangle_i\otimes|y\rangle_j)
    =
    |x\rangle_i\otimes|x\oplus y\rangle_j,
    \qquad x,y\in\{0,1\},
\]
where $\oplus$ denotes addition modulo $2$. Applying it gives
\[
    \operatorname{CNOT}_{i\to j}
    (|\psi\rangle_i\otimes|0\rangle_j)
    =
    \alpha|0\rangle_i|0\rangle_j
    +\beta|1\rangle_i|1\rangle_j.
\]

We next apply $I\otimes\langle a|_j$, which takes the inner product
on qubit $j$ and leaves a vector on qubit $i$. We obtain
\begin{align*}
    (I\otimes\langle a|_j)
    \bigl(
        \alpha|0\rangle_i|0\rangle_j
        +\beta|1\rangle_i|1\rangle_j
    \bigr) 
    &=
    \alpha\langle a|0\rangle\,|0\rangle_i
    +\beta\langle a|1\rangle\,|1\rangle_i \\
    &=
    \frac1{\sqrt2}
    \bigl(
        \alpha|0\rangle_i
        +e^{i\pi/4}\beta|1\rangle_i
    \bigr)
    =
    \frac1{\sqrt2}\,T_i|\psi\rangle_i.
\end{align*}
We have therefore shown that
\[
    \sqrt2\,
    (I\otimes\langle a|_j)
    \operatorname{CNOT}_{i\to j}
    (I\otimes|0\rangle_j)|\psi\rangle_i
    =
    T_i|\psi\rangle_i.
\]
Since $|\psi\rangle_i$ was arbitrary, this proves the operator identity
\begin{equation}
\label{eq:t-gate-contraction}
    T_i
    =
    \sqrt2\,
    (I\otimes\langle a|_j)
    \operatorname{CNOT}_{i\to j}
    (I\otimes|0\rangle_j).
\end{equation}
If qubit $i$ belongs to an $M$-qubit system, tensoring
\eqref{eq:t-gate-contraction} with the identity on the other
$M-1$ qubits gives
\begin{equation}
\label{eq:t-gate-contraction-extended}
    T_i
    =
    \sqrt2\,
    (I^{\otimes M}\otimes\langle a|_j)
    \operatorname{CNOT}_{i\to j}
    (I^{\otimes M}\otimes|0\rangle_j),
\end{equation}
where $T_i=I^{\otimes(i-1)}\otimes T\otimes I^{\otimes(M-i)}$.

\textit{2. Replacing all $T$-gates.}
Choose $M\ge n$ and a Clifford+$T$ circuit $U$ on $M$ qubits,
containing $s\le t$ $T$-gates, such that
\[
    |\Psi\rangle:=U|0^M\rangle,
    \qquad
    \rho=\Tr_{n+1,\ldots,M}(|\Psi\rangle\langle\Psi|).
\]

First consider a circuit with one $T$-gate, $U=C_1T_iC_0$.
Since $C_0$ and $C_1$ act only on the original $M$ qubits,
the extra qubit can be appended before $C_0$ is applied and contracted
after $C_1$ is applied. Specifically, by \eqref{eq:t-gate-contraction-extended} 
\begin{align*}
    C_1T_iC_0
    &=
    \sqrt2\,C_1
    (I^{\otimes M}\otimes\langle a|_j)
    \operatorname{CNOT}_{i\to j}
    (I^{\otimes M}\otimes|0\rangle_j)C_0 \\
    &=
    \sqrt2\,
    (I^{\otimes M}\otimes\langle a|_j)
    (C_1\otimes I)
    \operatorname{CNOT}_{i\to j}
    (I^{\otimes M}\otimes|0\rangle_j)C_0 \\
    &=
    \sqrt2\,
    (I^{\otimes M}\otimes\langle a|_j)
    (C_1\otimes I)
    \operatorname{CNOT}_{i\to j}
    (C_0\otimes I)
    (I^{\otimes M}\otimes|0\rangle_j).
\end{align*}

For the general circuit, write
\[
    U=C_sT_{i_s}C_{s-1}\cdots T_{i_1}C_0,
    \qquad C_0,\ldots,C_s\in\Cliff_M.
\]
We introduce $t$ extra qubits, labelled $M+1,\ldots,M+t$,
all initialized in $|0\rangle$ and replace the $k$-th $T$-gate
by $\operatorname{CNOT}_{i_k\to M+k}$. Each occurrence uses
a different extra qubit, even if some of the indices $i_k$
coincide. This gives the Clifford circuit
\[
    \widetilde U
    :=
    (C_s\otimes I^{\otimes t})
    \operatorname{CNOT}_{i_s\to M+s}
    (C_{s-1}\otimes I^{\otimes t})
    \cdots
    \operatorname{CNOT}_{i_1\to M+1}
    (C_0\otimes I^{\otimes t}).
\]
The last $t-s$ extra qubits are untouched. For $s=0$, take
$\widetilde U=C_0\otimes I^{\otimes t}$.

No later gate acts on an extra qubit after its CNOT.
Consequently, as in the one-gate example, all extra qubits
can be appended at the beginning and all contractions
postponed until the end. Each used extra qubit implements
its $T$-gate with a factor $2^{-1/2}$, while each unused
extra qubit contributes $\langle a|0\rangle=2^{-1/2}$.
Therefore,
\[
    U
    =
    2^{t/2}
    (I^{\otimes M}\otimes\langle a|^{\otimes t})
    \widetilde U
    (I^{\otimes M}\otimes|0^t\rangle).
\]

\textit{3. Tracing out the original ancillary qubits.}
Using the representation of $U$ from Step~2, we obtain
\begin{align*}
    |\Psi\rangle
    &=U|0^M\rangle \\
    &=
    2^{t/2}
    (I^{\otimes M}\otimes\langle a|^{\otimes t})
    \widetilde U
    (I^{\otimes M}\otimes|0^t\rangle)|0^M\rangle \\
    &=
    2^{t/2}
    (I^{\otimes M}\otimes\langle a|^{\otimes t})
    \widetilde U
    (|0^M\rangle\otimes|0^t\rangle) \\
    &=
    2^{t/2}
    (I^{\otimes M}\otimes\langle a|^{\otimes t})
    \widetilde U|0^{M+t}\rangle.
\end{align*}
Set $|\Phi\rangle:=\widetilde U|0^{M+t}\rangle$.
Since $\widetilde U$ is Clifford, $|\Phi\rangle$ is a pure
stabilizer state. Taking outer products gives
\[
    |\Psi\rangle\langle\Psi|
    =
    2^t
    (I^{\otimes M}\otimes\langle a|^{\otimes t})
    |\Phi\rangle\langle\Phi|
    (I^{\otimes M}\otimes|a\rangle^{\otimes t}).
\]
Substituting into the definition of $\rho$, we have
\begin{align*}
    \rho
    &=
    \Tr_{n+1,\ldots,M}(|\Psi\rangle\langle\Psi|) \\
    &=
    2^t\Tr_{n+1,\ldots,M}\!\left[
        (I^{\otimes M}\otimes\langle a|^{\otimes t})
        |\Phi\rangle\langle\Phi|
        (I^{\otimes M}\otimes|a\rangle^{\otimes t})
    \right] \\
    &=
    2^t
    (I^{\otimes n}\otimes\langle a|^{\otimes t})
    \Tr_{n+1,\ldots,M}(|\Phi\rangle\langle\Phi|)
    (I^{\otimes n}\otimes|a\rangle^{\otimes t}).
\end{align*}
The last equality holds because the partial trace acts on
the original ancillary qubits $n+1,\ldots,M$, whereas the
contractions act on the extra qubits $M+1,\ldots,M+t$.

Now define $\tau
    :=
    \Tr_{n+1,\ldots,M}(|\Phi\rangle\langle\Phi|)$. This state retains the first $n$ original qubits and all $t$
extra qubits. Since it is a reduction of the pure stabilizer
state $|\Phi\rangle$, we have
$\tau\in\mathsf{mixstab}_{n+t}$. The preceding calculation gives
\[
    \rho
    =
    2^t
    (I^{\otimes n}\otimes\langle a|^{\otimes t})
    \tau
    (I^{\otimes n}\otimes|a\rangle^{\otimes t}),
\]
which proves \eqref{eq:t-count-fixed-projection}.
Taking traces yields
\[
    1=\Tr\rho
    =
    2^t\Tr\bigl[
        (I^{\otimes n}\otimes\Pi_a^{\otimes t})\tau
    \bigr],
\]
and hence the stated probability $2^{-t}$.
\end{proof}

We now use this representation to control the profile entropy of
$\CliffT^{\mathrm{anc}}_{n,t}$.

Related counting bounds independent of ancillary workspace
were obtained by Gosset, Kothari, and Wu
\cite[Section~4.1]{gosset2026quantum}, using the canonical
form for postselected stabilizer computations of
Beverland et al.~\cite[Appendix~A.7]{beverland2020lower}.
Here we use the fixed-projection representation above to
control the metric entropy of the local Pauli profiles of
$\CliffT^{\mathrm{anc}}_{n,t}$.

\begin{lemma}[Profile entropy of Clifford+$T$ circuits]
\label{lem:t-count-profile-entropy}
There is a universal constant $C<\infty$ such that, for every
$n\ge p\ge2$, integer $t\ge0$, and $0<\varepsilon<1$,
\[
    \log N\left(
        \varepsilon,
        \mathcal F_{n,p}(\CliffT^{\mathrm{anc}}_{n,t}),
        \|\cdot\|_{n,p}
    \right)
    \le
    C(n+t)\left[
        t+\log\left(\frac{C3^{p/2}}{\varepsilon}\right)
    \right].
\]
\end{lemma}

\begin{proof}
Fix $n\ge p\ge2$, $t\ge0$, and $0<\varepsilon<1$, and set
$m:=n+t$ and
$\mathcal F:=\mathcal F_{n,p}(\CliffT^{\mathrm{anc}}_{n,t})$.

\textit{1. Expressing output expectations through stabilizer profiles.}
By Lemma~\ref{lem:t-count-fixed-projection}, every
$\rho\in\CliffT^{\mathrm{anc}}_{n,t}$ has a representing
state $\tau\in\mathsf{mixstab}_m$ satisfying
\eqref{eq:t-count-fixed-projection}.
For any $Q\in\mathcal P_{n,p}$, 
\begin{align*}
    f_\rho(Q)
    &=
    2^t\Tr\left[
        Q(I^{\otimes n}\otimes\langle a|^{\otimes t})
        \tau(I^{\otimes n}\otimes|a\rangle^{\otimes t})
    \right] \\
    &=
    2^t\Tr\left[
        (I^{\otimes n}\otimes|a\rangle^{\otimes t})
        Q(I^{\otimes n}\otimes\langle a|^{\otimes t})\tau
    \right] \\
    &=
    2^t\Tr\left[
        (Q\otimes\Pi_a^{\otimes t})\tau
    \right] \\
    &=
    \Tr\left[
        \bigl(Q\otimes(2\Pi_a)^{\otimes t}\bigr)\tau
    \right],
\end{align*}
where $\Pi_a = \ket{a}\bra{a}$.

By definition of $|a\rangle$, we have
\[
    2\Pi_a
    =I+\frac{X-Y}{\sqrt2}
    =\sum_{P\in\mcP_1}b_PP,
\]
where $b_I:=1$, $b_X:=1/\sqrt2$, $b_Y:=-1/\sqrt2$,
and $b_Z:=0$.
Expanding the tensor product amounts to choosing a Pauli operator
at each position and multiplying its coefficients. Thus
\[
    (2\Pi_a)^{\otimes t}
    =
    \sum_{R\in\mcP_t}
        \left(\prod_{k=1}^t b_{R_k}\right)R.
\]
Since $b_Z=0$, only the $3^t$ strings with factors in
$\{I,X,Y\}$ have nonzero coefficients.
For $t=0$, the expansion consists of the scalar $1$,
with the empty product interpreted as $1$.

Now fix $Q_1,\ldots,Q_q\in\mathcal P_{n,p}$.
Each operator $Q_i\otimes(2\Pi_a)^{\otimes t}$ is a linear
combination of $3^t$ Pauli strings of the form $Q_i\otimes R$,
where $R\in\mcP_t$ contains no $Z$ factors.
Collect the distinct Pauli strings appearing across these
$q$ expansions into a single list $P_1,\ldots,P_N$.
Each acts on $m=n+t$ qubits, and $N\le q3^t$. Write
\[
    Q_i\otimes(2\Pi_a)^{\otimes t}
    =
    \sum_{j=1}^N a_{ij}P_j,
    \qquad i=1,\ldots,q,
\]
with $a_{ij}=0$ whenever $P_j$ does not appear in the $i$-th
expansion. Taking expectations against $\tau$ gives
\begin{equation}
\label{eq:t-count-profile-linear-map}
    f_\rho(Q_i)
    =
    \Tr\left[
        \bigl(Q_i\otimes(2\Pi_a)^{\otimes t}\bigr)\tau
    \right] 
    =
    \sum_{j=1}^N a_{ij}\Tr(P_j\tau)
    =
    \sum_{j=1}^N a_{ij}f_\tau(P_j),
    \qquad i=1,\ldots,q.
\end{equation}
The collection $P_1,\ldots,P_N$ and the coefficients $a_{ij}$ are
fixed once $t$ and $Q_1,\ldots,Q_q$ are chosen and they do not
depend on $\rho$ or $\tau$.

\textit{2. Counting the possible response vectors.}
Equation~\eqref{eq:t-count-profile-linear-map} shows that
the values $f_\tau(P_j)$, $1\le j\le N$, determine all
the queried output expectations $f_\rho(Q_i)$, $1\le i\le q$. Counting every $\tau\in\mathsf{mixstab}_m$ can only enlarge
the set of possible input response vectors. Hence
Lemma~\ref{lem:fixed-depth-mixed-stabilizer-growth} gives
\begin{align*}
    \#\left\{
        (f(Q_1),\ldots,f(Q_q)):f\in\mathcal F
    \right\}
    &\le
    \#\left\{
        (f_\tau(P_1),\ldots,f_\tau(P_N)):
        \tau\in\mathsf{mixstab}_m
    \right\}\\
    &\le\left(\sum_{j=0}^m\binom Nj\right)^2
     \le\left(\sum_{j=0}^m\binom{q3^t}{j}\right)^2.
\end{align*}
For every $q\ge2m$, we have $q3^t\ge m$, and so a binomial
bound yields
\[
    \left(\sum_{j=0}^m\binom{q3^t}{j}\right)^2
    \le\left(\frac{eq3^t}{m}\right)^{2m}
    =\left(\frac{Aq}{d}\right)^d,
\]
where $d:=2m$, $A:=2e3^t$. This verifies the restriction-growth hypothesis of
Lemma~\ref{lem:growth-entropy} for every required $q$.

\textit{3. Converting the growth bound to profile entropy.}
Every $f\in\mathcal F$ takes values in $[-1,1]$, since it
is the profile of a normalized output state and every
Pauli has operator norm one. 

Let $\mu_{n,p}$ be the uniform probability measure on
$\mathcal P_{n,p}$. Lemma~\ref{lem:growth-entropy} then gives, for
$0<\delta<1$,
\begin{align*}
    \log N(\delta,\mathcal F,L_2(\mu_{n,p}))
    &\le Cd\log\left(\frac{eA}{\delta}\right)
    \le Cm\left[t+\log\left(\frac C\delta\right)\right].
\end{align*}
For any function $h$ on the retained output Paulis,
\[
    \|h\|_{L_2(\mu_{n,p})}^2
    =\frac1{3^p\binom np}
       \sum_{Q\in\mathcal P_{n,p}}|h(Q)|^2
    =3^{-p}\|h\|_{n,p}^2.
\]
An $L_2(\mu_{n,p})$-cover of radius
$\delta:=\varepsilon/3^{p/2}$ is therefore an
$\varepsilon$-cover in $\|\cdot\|_{n,p}$.
Substituting this value and $m=n+t$ gives the stated bound.
\end{proof}

\subsection{Energy bounds for Clifford+\(T\) states}
\label{subsec:sublinear-t-count}

We now prove Corollary~\ref{cor:intro-t-count}\textup{(i)}
and establish the comparison with $\PStab_n$ needed for
part~\textup{(ii)}. The large-$p$ conclusion in part~\textup{(ii)}
will follow from the product-stabilizer bounds in
Subsection~\ref{subsec:product-stabilizer-value}.

For part~\textup{(i)}, we combine
Theorem~\ref{thm:local-profile-complexity-principle}\textup{(i)}
with the profile-entropy bound in
Lemma~\ref{lem:t-count-profile-entropy}.

\begin{proof}[Proof of Corollary~\ref{cor:intro-t-count}\textup{(i)}]
Suppose $t_n\le\kappa n$ for all sufficiently large $n$.
By Lemma~\ref{lem:t-count-profile-entropy}, we may take
\[
    a_{n,p}:=C(n+t_n)t_n,
    \qquad b_{n,p}:=C(n+t_n),
    \qquad B_{n,p}:=C3^{p/2}
\]
in Theorem~\ref{thm:local-profile-complexity-principle},
where $C\ge1$ is universal.
The effective complexity
$\Gamma_{n,p}:=a_{n,p}+b_{n,p}\log(eB_{n,p}n)$ therefore satisfies
\begin{equation}
\label{eq:t-count-effective-complexity}
    \frac{\Gamma_{n,p}}{n^2}
    =C\left(1+\frac{t_n}{n}\right)
       \left[
           \frac{t_n}{n}
           +\frac{\log(eC3^{p/2}n)}{n}
       \right].
\end{equation}
For every fixed $p$, it follows that
\[
    \limsup_{n\to\infty}\frac{\Gamma_{n,p}}{n^2}
    \le C\kappa(1+\kappa).
\]
Since this bound is independent of $p$,
Theorem~\ref{thm:local-profile-complexity-principle}\textup{(i)}
implies that, for every sufficiently large fixed $p$
depending on $\kappa$, there is $c_{p,\kappa}>0$ such that
\[
    \P\left(
        \Delta^{\mathrm{opt}}_{n,p}
        (\CliffT^{\mathrm{anc}}_{n,t_n})
        \ge c_{p,\kappa}\sqrt n
    \right)\longrightarrow1.
\]
\end{proof}

For the sublinear comparison, we use the same entropy bound
and the conditioning benchmark $\mathcal B_{n,t_n}$ from
Corollary~\ref{cor:t-count-conditioning-benchmark} to apply
Theorem~\ref{thm:local-profile-complexity-principle}\textup{(ii)}.
We then use Lemma~\ref{lem:pstab-exceptional-site-cost}
to show that the exceptional factors increase the expected
optimal energy over $\PStab_n$ by only $o(\sqrt n)$
when $t_n=o(n)$.

\begin{lemma}[Energy of the exceptional factors]
\label{lem:pstab-exceptional-site-cost}
For integers $n\ge p\ge2$ and $0\le t\le n-p$,
\[
    0\le
    \E\left[\mathcal{E}_{n,p}(\mathcal B_{n,t})-\mathcal{E}_{n,p}(\PStab_n)\right]
    \le C3^{p/2}\sqrt{pt},
\]
where $C>0$ is a universal constant.
\end{lemma}

\begin{proof}
Let $\omega\in\mathcal B_{n,t}$.
Let $F:=E_{[n]}(\omega)$, so $|F|\le t$, and define
\[
    \omega'
    :=\bigotimes_{i=1}^n\omega_i',
    \qquad
    \omega_i':=
    \begin{cases}
        |0\rangle\langle0|_i, & i\in F,\\
        \omega_i, & i\notin F,
    \end{cases}
\]
where $\omega_i$ is the one-qubit marginal of $\omega$.
Every factor of $\omega'$ is a pure one-qubit stabilizer
state, so $\omega'\in\PStab_n$.

The states $\omega$ and $\omega'$ have identical marginals
on $F^c$. Write
\[
    H_{\mathrm{hit}(F)}
    :=
    \frac{1}{\sqrt{\binom np}}
    \sum_{\substack{P\in\mathcal P_{n,p}\\
                    \operatorname{supp}(P)\cap F\ne\emptyset}}
        g_P P,
\]
where $g_P$ are the same Gaussian coefficients as in $H_{n,p}$.
Thus $H_{\mathrm{hit}(F)}$ contains the interactions
involving at least one qubit in $F$.

We then have
\begin{align*}
    \left|\Tr[H_{n,p}(\omega-\omega')]\right|
    &=\left|\Tr[H_{\mathrm{hit}(F)}(\omega-\omega')]\right|
    \le \|H_{\mathrm{hit}(F)}\|_{\mathrm{op}}
        \|\omega-\omega'\|_1
    \le 2\|H_{\mathrm{hit}(F)}\|_{\mathrm{op}}.
\end{align*}
The first inequality uses
$|\Tr(AB)|\le\|A\|_{\mathrm{op}}\|B\|_1$. The second follows from
$\|\omega-\omega'\|_1\le\|\omega\|_1+\|\omega'\|_1=2$.
Since $\PStab_n\subseteq\mathcal B_{n,t}$, taking the supremum
over $\omega$ yields
\[
    0\le \mathcal{E}_{n,p}(\mathcal B_{n,t})-\mathcal{E}_{n,p}(\PStab_n)
    \le2\max_{\substack{F\subseteq[n]\\|F|\le t}}
        \|H_{\mathrm{hit}(F)}\|_{\mathrm{op}}.
\]
Taking expectations and applying
Lemma~\ref{lem:revmag-hit-bound} with budget $t$
proves the claim.
\end{proof}

\begin{proof}[Proof of the energy comparison in Corollary~\ref{cor:intro-t-count}\textup{(ii)}]
Suppose $t_n=o(n)$ and fix $p\ge2$.
Set $\mathcal C_n=\CliffT^{\mathrm{anc}}_{n,t_n}$ and
$\mathcal B_n=\mathcal B_{n,t_n}$.
We have shown that $\mathcal B_n$ is a conditioning benchmark
for $\mathcal C_n$, and
\eqref{eq:t-count-effective-complexity} gives
$\Gamma_{n,p}=o(n^2)$.
Theorem~\ref{thm:local-profile-complexity-principle}\textup{(ii)}
therefore implies
\[
    \frac{[\mathcal{E}_{n,p}(\mathcal C_n)-\mathcal{E}_{n,p}(\mathcal B_n)]_+}
         {\sqrt n}
    \longrightarrow0
\]
in $L^1$.
For sufficiently large $n$, we also have $t_n\le n-p$, so
Lemma~\ref{lem:pstab-exceptional-site-cost} gives
\[
    0\le
    \frac{\E[\mathcal{E}_{n,p}(\mathcal B_n)-\mathcal{E}_{n,p}(\PStab_n)]}{\sqrt n}
    \le C3^{p/2}\sqrt{\frac{pt_n}{n}}
    \longrightarrow0.
\]
Product Pauli eigenstates belong to both $\mathcal B_n$ and
$\mathcal C_n$. Hence
\begin{align*}
    0
    &\le \mathcal{E}_{n,p}(\mathcal C_n)-\mathcal{E}_{n,p}(\PStab_n)\\
    &\le[\mathcal{E}_{n,p}(\mathcal C_n)-\mathcal{E}_{n,p}(\mathcal B_n)]_+
       +\mathcal{E}_{n,p}(\mathcal B_n)-\mathcal{E}_{n,p}(\PStab_n).
\end{align*}
After division by $\sqrt n$, both terms on the right are
nonnegative and have expectations tending to zero. Thus
\begin{equation}
\label{eq:t-count-product-stabilizer-comparison}
    \frac{\mathcal{E}_{n,p}(\mathcal C_n)-\mathcal{E}_{n,p}(\PStab_n)}{\sqrt n}
    \longrightarrow0
\end{equation}
in $L^1$ and in probability.
The large-$p$ assertion follows from the benchmark calculation
in Subsection~\ref{subsec:product-stabilizer-value}.
\end{proof}

Taking $t_n=0$ gives the stabilizer comparison as a special case.

\begin{corollary}[Stabilizer states reduce to product stabilizers]
\label{thm:clifford-product-stabilizer}
For every fixed $p\ge2$,
\[
    \frac{
        \mathcal{E}_{n,p}(\Stab_n)-\mathcal{E}_{n,p}(\PStab_n)
    }{\sqrt n}
    \longrightarrow0
\]
in $L^1$ and in probability.
\end{corollary}

\begin{proof}
With no $T$ gates, the output is a reduced state of a pure stabilizer
state, so $\CliffT^{\mathrm{anc}}_{n,0}=\mathsf{mixstab}_n$.
Such reduced states are mixtures of pure stabilizer states; see
Appendix~\ref{app:IsotropicSubspaces}. Consequently,
\[
    \Stab_n
    \subseteq\CliffT^{\mathrm{anc}}_{n,0}
    \subseteq\operatorname{conv}(\Stab_n).
\]
Linearity of the energy gives
\begin{equation}
\label{eq:t-count-zero-energy}
    \mathcal{E}_{n,p}(\CliffT^{\mathrm{anc}}_{n,0})
    =\mathcal{E}_{n,p}(\Stab_n).
\end{equation}
The claim follows by setting $t_n=0$ in
\eqref{eq:t-count-product-stabilizer-comparison}.
\end{proof}

\subsection{Product-stabilizer energy}
\label{subsec:product-stabilizer-value}

To obtain the large-$p$ energy in
Corollary~\ref{cor:intro-t-count}\textup{(ii)}, it remains to
evaluate the product-stabilizer optimum. A product stabilizer
state is specified by choosing one of six Pauli eigenstates at
each qubit. We use this finite description to express its energy
as a Gaussian process and estimate the maximum.

Let
\(\sigma_1=X\), \(\sigma_2=Y\), and \(\sigma_3=Z\), and set $\mathcal A:=\{\pm e_1,\pm e_2,\pm e_3\}\subset\mathbb R^3$. Then every element of \(\PStab_n\) can be written as
\[
    \rho_u
    :=
    \bigotimes_{i=1}^n
    \frac{I+s_i\sigma_{a_i}}{2},
    \qquad
    u_i:=s_i e_{a_i}\in\mathcal A,
\]
so \(u=(u_1,\ldots,u_n)\in\mathcal A^n\) records the Pauli eigenstate
chosen at each qubit.  Define its energy by $X_u:=\Tr(H_{n,p}\rho_u)$. Then \(\E X_u=0\), and independence gives
\begin{align}
\label{eq:pstab-covariance}
    \E X_uX_v
    =
    \binom np^{-1}
    \sum_{\substack{J_0\subseteq[n]\\|J_0|=p}}
    \prod_{i\in J_0}(u_i\cdot v_i).
\end{align}
In particular, since \(\|u_i\|_2=1\), we have \(\E X_u^2=1\).

\begin{lemma}[Product-stabilizer thermodynamic limit]
\label{lem:pstab-thermodynamic-limit}
For every fixed $p\ge2$, the limit
\[
    \lim_{n\to\infty}
    \frac{\E \mathcal{E}_{n,p}(\PStab_n)}{\sqrt n}
\]
exists.
\end{lemma}

\begin{proof}
Set $G_n(u):=\sqrt n\,X_u$. For $\beta>0$, define the
normalized expected free energy
\[
    F_{n,\beta}
    :=
    \frac1{\beta n}\E\log
    \sum_{u\in\mathcal A^n}\exp\bigl(\beta G_n(u)\bigr).
\]
The elementary bounds
\[
    \max_u G_n(u)
    \le
    \frac1\beta\log\sum_u e^{\beta G_n(u)}
    \le
    \max_u G_n(u)+\frac{n\log6}{\beta}
\]
give
\[
    0\le
    F_{n,\beta}
    -\frac{\E \mathcal{E}_{n,p}(\PStab_n)}{\sqrt n}
    \le\frac{\log6}{\beta}.
\]
Thus, it suffices to prove that $F_{n,\beta}$ converges
as $n\to\infty$ for every fixed $\beta>0$.

We first compare the original process with an auxiliary
process that allows repeated site indices.
For $u,v\in\mathcal A^n$, define $q_n(u,v):=\frac1n\sum_{i=1}^n u_i\cdot v_i$. Expanding its $p$th power gives
\[
    q_n(u,v)^p
    =
    \frac1{n^p}
    \sum_{i_1,\ldots,i_p\in[n]}
        \prod_{k=1}^p(u_{i_k}\cdot v_{i_k}).
\]
On the other hand, writing
$(n)_p:=n(n-1)\cdots(n-p+1)$,
\eqref{eq:pstab-covariance} gives
\[
    \frac1n\E G_n(u)G_n(v)
    =
    \frac1{(n)_p}
    \sum_{\substack{i_1,\ldots,i_p\in[n]\\
                    \text{all distinct}}}
        \prod_{k=1}^p(u_{i_k}\cdot v_{i_k}).
\]
The first expression averages over all ordered tuples,
whereas the second averages only over tuples with
distinct indices. These correspond to sampling with
and without replacement, respectively.

Extend $q_n$ to $(\mathbb R^3)^n$ by the same formula, and let $\widetilde G_n$ be a centered Gaussian process on this space with
covariance $\E\widetilde G_n(u)\widetilde G_n(v)
    =
    nq_n(u,v)^p$. Note that for any finite collection of configurations, the overlap matrix $(q_n(u,v))_{u,v}$
is positive semidefinite, and its entrywise $p$th power
is positive semidefinite by the Schur product theorem \cite[Section~7.5]{horn2012matrix}.

To compare the two processes, fix $u,v\in\mathcal A^n$
and let $J_1,\ldots,J_p$ be independent uniform indices
in $[n]$, with $\mathcal D$ denoting the event that they
are all distinct. Setting
$W:=\prod_{k=1}^p(u_{J_k}\cdot v_{J_k})$,
the two covariance formulas above give
\[
    \frac1n\E\widetilde G_n(u)\widetilde G_n(v)
    =q_n(u,v)^p=\E W,
    \qquad
    \frac1n\E G_n(u)G_n(v)
    =\E[W\mid\mathcal D].
\]
By conditioning on $\mathcal D$,
\[
\begin{aligned}
    \frac1n\E G_n(u)G_n(v)-q_n(u,v)^p
    &=
    \P(\mathcal D^c)
    \left(
        \E[W\mid\mathcal D]
        -
        \E[W\mid\mathcal D^c]
    \right).
\end{aligned}
\]
Since $|W|\le1$, and since a union bound gives
\[
    \P(\mathcal D^c)
    \le
    \sum_{1\le k<\ell\le p}\P(J_k=J_\ell)
    =
    \binom p2\frac1n,
\]
we obtain
\begin{equation}
\label{eq:pstab-with-replacement}
    \left|
        \frac1n\E G_n(u)G_n(v)-q_n(u,v)^p
    \right|
    \le
    \frac{p(p-1)}n.
\end{equation}
Consequently,
\[
    \sup_{u,v\in\mathcal A^n}
    \left|
        \E G_n(u)G_n(v)
        -
        \E\widetilde G_n(u)\widetilde G_n(v)
    \right|
    \le p(p-1).
\]

Define the corresponding auxiliary free energy by
\[
    \widetilde F_{n,\beta}
    :=
    \frac1{\beta n}\E\log
    \sum_{u\in\mathcal A^n}
        \exp\bigl(\beta\widetilde G_n(u)\bigr).
\]
Both Gaussian processes have variance $n$ at every
configuration, so the diagonal covariance terms cancel
in Gaussian interpolation. As in the proof of
\cite[Lemma~3.2]{bhattacharya2026parisi}, this gives
\begin{align}
\label{eq:interpolation-bd}
    |F_{n,\beta}-\widetilde F_{n,\beta}|
    \le \frac{\beta}{2n}
    \sup_{u,v\in\mathcal A^n}
    \left|
        \E G_n(u)G_n(v)
        -\E\widetilde G_n(u)\widetilde G_n(v)
    \right| 
    \le \frac{\beta p(p-1)}{2n}.
\end{align}
For fixed $p$ and $\beta$, this difference tends to zero.
It remains to prove convergence of
$\widetilde F_{n,\beta}$.

We now use \cite[Proposition~8.4]{chen2025free},
following the approach of
\cite[Lemma~3.4]{bhattacharya2026parisi}.
We state explicitly the consequence we need.
The result concerns centered Gaussian processes $Y_n$
with covariance
\[
    \E Y_n(u)Y_n(v)
    =
    n\xi\bigl(R_n(u,v)\bigr),
    \qquad
    R_n(u,v):=\frac1n\sum_{i=1}^n u_iv_i^\top.
\]
It averages $e^{Y_n(u)}$ over configurations
$u=(u_1,\ldots,u_n)$ by integrating each $u_i$
independently against the same fixed compactly supported probability measure
on $\mathbb R^3$. We use the uniform measure on the finite set $\mathcal A$,
so this condition holds and the configuration average is $6^{-n}\sum_{u\in\mathcal A^n}e^{Y_n(u)}$.

In the polynomial case used here, the remaining
conditions are that $\xi(0)=0$ and that $\xi$ is convex
on the cone of positive semidefinite $3\times3$ matrices.
Under these conditions, the cited result guarantees
that
\[
    \frac1n\E\log\left(
        6^{-n}\sum_{u\in\mathcal A^n}e^{Y_n(u)}
    \right)
\]
converges as $n\to\infty$.

We apply this result with $Y_n=\beta\widetilde G_n$.
Since $q_n(u,v)=\Tr R_n(u,v)$, its covariance has
the required form with $\xi_\beta(A):=\beta^2(\Tr A)^p$.

This is a polynomial that vanishes at zero.
On the positive semidefinite cone, the trace is
nonnegative, so convexity of $x\mapsto x^p$ on
$[0,\infty)$ implies convexity of $\xi_\beta$.
Thus all the conditions hold for every integer $p\ge2$.

Proposition~8.4 of \cite{chen2025free} therefore implies
that
\[
    \frac1n\E\log\left(
        6^{-n}\sum_{u\in\mathcal A^n}
            e^{\beta\widetilde G_n(u)}
    \right)
    =
    \beta\widetilde F_{n,\beta}-\log6
\]
converges. Hence $\widetilde F_{n,\beta}$ converges
for every fixed $\beta>0$, and
\eqref{eq:interpolation-bd} gives the same conclusion
for $F_{n,\beta}$.

Since $F_{n,\beta}$ converges for every fixed $\beta>0$,
the preceding inequality implies
\[
    0\le
    \limsup_{n\to\infty}
        \frac{\E \mathcal{E}_{n,p}(\PStab_n)}{\sqrt n}
    -
    \liminf_{n\to\infty}
        \frac{\E \mathcal{E}_{n,p}(\PStab_n)}{\sqrt n}
    \le \frac{\log6}{\beta}.
\]
Letting $\beta\to\infty$ shows that the limit superior
and limit inferior coincide, proving the claim.
\end{proof}

We next show that the fixed-$p$ limit converges to
$\sqrt{2\log6}$ as $p\to\infty$, matching the limiting
normalized extremal energy of the Random Energy Model
with $6^n$ configurations~\cite{derrida1981random}. The related proof of the large-$p$ lower bound for unrestricted product states in~\cite{anschuetz2025bounds} used a second moment method argument. Our proof instead uses the Sudakov--Fernique Gaussian
comparison inequality to reduce the lower bound to
the expected maximum of independent Gaussian variables,
as in the Random Energy Model. \cite{gamarnik2023shattering} also used comparison inequalities to establish the overlap gap property for classical spin glasses.

\begin{proposition}[Large-$p$ product-stabilizer energy]
\label{prop:pstab-value}
There exist universal constants $C>0$ and $p_0\ge2$ such that, for
every integer $p\ge p_0$,
\[
    \sqrt{2\log6}
    -
    C\frac{(\log p)^2}{p}
    \le
    \lim_{n\to\infty}
    \frac{\E \mathcal{E}_{n,p}(\PStab_n)}{\sqrt n}
    \le
    \sqrt{2\log6}.
\]
Consequently,
\[
   \lim_{p\to\infty}
   \lim_{n\to\infty}
    \frac{\E \mathcal{E}_{n,p}(\PStab_n)}{\sqrt n}
    =
    \sqrt{2\log6}.
\]
\end{proposition}

\begin{proof}
The existence of the limit follows from Lemma~\ref{lem:pstab-thermodynamic-limit}.

The upper bound follows immediately from the Gaussian maximal
inequality. Since $|\mathcal A^n|=6^n$ and every $X_u$ has variance
one, we have
\[
    \E \mathcal{E}_{n,p}(\PStab_n)
    \le
    \sqrt{2\log(6^n)}
    =
    \sqrt{2n\log6}.
\]
Therefore
\begin{equation}
\label{eq:pstab-upper}
    \frac{\E \mathcal{E}_{n,p}(\PStab_n)}{\sqrt n}
    \le
    \sqrt{2\log6}.
\end{equation}

For the lower bound, we construct a packing in the canonical
metric of the Gaussian process, $d_X(u,v)
    :=
    (\E(X_u-X_v)^2)^{1/2}$. Since $\E X_u^2=1$ for every $u$,
\[
    d_X(u,v)^2
    =
    2\bigl(1-\E X_uX_v\bigr).
\]
Thus, constructing an $\varepsilon$-packing amounts to finding a set
whose distinct elements satisfy
\[
    \E X_uX_v
    \le
    1-\frac{\varepsilon^2}{2}.
\]

Recall the overlaps $q_n(u,v)$ from the proof of
Lemma~\ref{lem:pstab-thermodynamic-limit}.
Fix $\delta\in(0,1/2)$. We will construct a set
$\mathcal C_{n,\delta}\subseteq\mathcal A^n$ such that for all distinct
$u,v\in\mathcal C_{n,\delta}$, $|q_n(u,v)|\le1-\delta$.
For every integer $p\ge2$,
\[
    q_n(u,v)^p
    \le |q_n(u,v)|^p
    \le (1-\delta)^p.
\]
Together with \eqref{eq:pstab-with-replacement} and
$G_n(u)=\sqrt n\,X_u$, this gives
\[
    \E X_uX_v
    \le
    (1-\delta)^p+\frac{p(p-1)}n
    =:
    \beta_n.
\]
Consequently, for all sufficiently large $n$,
$\mathcal C_{n,\delta}$ is a
$\sqrt{2(1-\beta_n)}$-packing in the canonical metric.

For even $p$, the configurations $u$ and $-u$ have identical
energies, despite differing at every site. Thus Hamming separation
alone does not guarantee separation in the canonical metric.
To control $|q_n(u,v)|$, we require separation from both
$v$ and $-v$. For $u,v\in\mathcal A^n$, define
\[
    d_{\pm}(u,v)
    :=
    \min\{d_H(u,v),d_H(u,-v)\},
\]
where $d_H$ denotes Hamming distance.
Let $r=\lfloor\delta n\rfloor$.

For $w\in\mathcal A^n$, write the radius-$r$ Hamming ball
\[
    B_H(w,r)
    :=
    \{v\in\mathcal A^n:d_H(w,v)\le r\}.
\]
Since $d_H(u,-v)=d_H(-u,v)$,
\[
    \{v\in\mathcal A^n:d_{\pm}(u,v)\le r\}
    =
    B_H(u,r)\cup B_H(-u,r).
\]
To construct a vector at Hamming distance exactly $k$ from $w$, one
chooses the $k$ coordinates to be changed and, at each such coordinate,
one of the five elements of $\mathcal A$ different from the original
entry. Hence
\[
    |B_H(w,r)|
    =
    \sum_{k=0}^r\binom nk5^k
    :=V_6(n,r).
\]
The standard $6$-ary Hamming-volume estimate gives
\[
    V_6(n,r)
    \le
    \exp\left\{
        n\bigl(h(\delta)+\delta\log5\bigr)
    \right\},
\]
where $h(\delta) := -\delta\log\delta-(1-\delta)\log(1-\delta)$; see \cite[Proposition~3.3.3]{guruswami2025essential}.

We now construct $\mathcal C_{n,\delta}$ greedily by adapting the argument in \cite[Sec.~4.2.1]{guruswami2025essential}. Initialize
$\mathcal C_{n,\delta}=\emptyset$ and let
$\mathcal R=\mathcal A^n$ be the set of available points. While
$\mathcal R$ is nonempty, choose any $u\in\mathcal R$, add $u$ to
$\mathcal C_{n,\delta}$, and remove from $\mathcal R$ every $v$
satisfying $d_{\pm}(u,v)\le r$. Continue until $\mathcal R$ is empty. Each selected
point removes at most $2V_6(n,r)$ points. Since all $6^n$ points of
$\mathcal A^n$ have been removed when the procedure terminates,
\[
    |\mathcal C_{n,\delta}|
    \ge
    \frac{6^n}{2V_6(n,r)}.
\]
Moreover, any two distinct selected points satisfy
$d_{\pm}(u,v)>r$. Since $d_{\pm}(u,v)$ is integer-valued and
$r=\lfloor\delta n\rfloor$, this implies $d_{\pm}(u,v)>\delta n$.

The Hamming-volume bound therefore gives
\begin{equation}
\label{eq:pstab-code-rate}
    \frac1n\log|\mathcal C_{n,\delta}|
    \ge
    \log6-h(\delta)-\delta\log5-\frac{\log2}{n}.
\end{equation}
It remains to verify that the Hamming separation gives the desired separation in the canonical metric. Fix distinct $u,v\in\mathcal C_{n,\delta}$ and define
\[
N_+(u,v) := \bigl|\{i:u_i=v_i\}\bigr|,
\qquad N_-(u,v) := \bigl|\{i:u_i=-v_i\}\bigr|.
\]
Since
\[
d_H(u,v)=n-N_+(u,v), \qquad d_H(u,-v)=n-N_-(u,v),
\]
the inequality $d_{\pm}(u,v)>\delta n$ implies
\[
N_+(u,v)<(1-\delta)n, \qquad N_-(u,v)<(1-\delta)n.
\]
Moreover, $u_i\cdot v_i$ equals $1$, $-1$, or $0$ as $u_i=v_i$, $u_i=-v_i$, or neither. Hence
\[
|q_n(u,v)|
= \frac{|N_+(u,v)-N_-(u,v)|}{n}
\le
\frac{\max\{N_+(u,v),N_-(u,v)\}}{n} < 1-\delta.
\]
Thus $\mathcal C_{n,\delta}$ has the covariance separation claimed above, i.e. $\E X_uX_v\le\beta_n$ for all distinct $u,v\in\mathcal C_{n,\delta}$.

Let $K_n:=|\mathcal C_{n,\delta}|$, and let $Z_0$ and $(Z_u)_{u\in\mathcal C_{n,\delta}}$ be independent standard Gaussian variables. For all sufficiently large $n$, we have $0\le\beta_n<1$, so define the comparison process
\[
Y_u := \sqrt{\beta_n}\,Z_0 + \sqrt{1-\beta_n}\,Z_u, \qquad u\in\mathcal C_{n,\delta}.
\]
The process $(Y_u)$ has variance one and pairwise covariance $\beta_n$. Consequently, for distinct $u,v\in\mathcal C_{n,\delta}$,
\[
\E(X_u-X_v)^2 = 2\bigl(1-\E X_uX_v\bigr) \ge 2(1-\beta_n) = \E(Y_u-Y_v)^2.
\]
Sudakov--Fernique \cite[Theorem~7.2.11]{vershynin2018high} therefore gives
\begin{align*}
\E \mathcal{E}_{n,p}(\PStab_n) &\ge \E\max_{u\in\mathcal C_{n,\delta}}X_u \\
&\ge \E\max_{u\in\mathcal C_{n,\delta}}Y_u \\
&= \sqrt{1-\beta_n}\, \E\max_{1\le j\le K_n}Z_j.
\end{align*}
The final equality holds because the term $\sqrt{\beta_n}Z_0$ is common to every $Y_u$ and has mean zero. The standard asymptotic for the maximum of independent standard Gaussian variables gives
\[
\E\max_{1\le j\le K_n}Z_j = (1+o_n(1))\sqrt{2\log K_n}.
\]
Since $\beta_n\to(1-\delta)^p$ as $n\to\infty$, it follows from \eqref{eq:pstab-code-rate} that
\begin{equation}
\label{eq:pstab-lower-delta}
\lim_{n\to\infty}
\frac{\E \mathcal{E}_{n,p}(\PStab_n)}{\sqrt n}
\ge \sqrt{ 2\bigl(1-(1-\delta)^p\bigr) \bigl( \log6-h(\delta)-\delta\log5 \bigr) }.
\end{equation}
For sufficiently large $p$, take $\delta=\frac{2\log p}{p}$. Then $\delta\in(0,1/2)$ and
\[
(1-\delta)^p \le e^{-p\delta} = p^{-2}.
\]
Moreover, using $h(\delta)\le\delta\log(e/\delta)$,
\[
h(\delta)+\delta\log5 \le \frac{2\log p}{p} \log\left(\frac{5ep}{2\log p}\right) \le C\frac{(\log p)^2}{p}.
\]
Substituting these estimates into \eqref{eq:pstab-lower-delta} gives
\[
\lim_{n\to\infty}
\frac{\E \mathcal{E}_{n,p}(\PStab_n)}{\sqrt n}
\ge \sqrt{2\log6} - C\frac{(\log p)^2}{p}.
\]
Taking $n\to\infty$ in \eqref{eq:pstab-upper} gives
\[
\lim_{n\to\infty}
\frac{\E \mathcal{E}_{n,p}(\PStab_n)}{\sqrt n}\le\sqrt{2\log6}.
\]
Combining the two estimates proves the claimed bounds. Letting $p\to\infty$ gives the result.
\end{proof}

\begin{proof}[Completion of the proof of
Corollary~\ref{cor:intro-t-count}\textup{(ii)}]
Let $t_n=o(n)$ and set
$\mathcal C_n:=\CliffT^{\mathrm{anc}}_{n,t_n}$.
By \eqref{eq:t-count-product-stabilizer-comparison},
it remains to prove the large-$p$ assertion.

Fix $\varepsilon>0$.
By Proposition~\ref{prop:pstab-value}, for every sufficiently
large fixed integer $p$ and all sufficiently large $n$,
\[
    \left|
        \frac{\E\mathcal{E}_{n,p}(\PStab_n)}{\sqrt n}
        -\sqrt{2\log6}
    \right|
    \le\frac{\varepsilon}{3}.
\]
Each product-stabilizer energy is a linear functional of the
Gaussian disorder with coefficient vector of norm one.
Their supremum is therefore $1$-Lipschitz, so Gaussian
concentration gives
\[
    \P\left(
        |\mathcal{E}_{n,p}(\PStab_n)
        -\E\mathcal{E}_{n,p}(\PStab_n)|>u
    \right)
    \le2e^{-u^2/2},
    \qquad u>0.
\]
Taking $u=(\varepsilon/3)\sqrt n$ and using
\eqref{eq:t-count-product-stabilizer-comparison}, we obtain,
with probability tending to one,
\begin{align*}
    \left|
        \frac{\mathcal{E}_{n,p}(\mathcal C_n)}{\sqrt n}
        -\sqrt{2\log6}
    \right|
    &\le
    \frac{\mathcal{E}_{n,p}(\mathcal C_n)
        -\mathcal{E}_{n,p}(\PStab_n)}{\sqrt n}\\
    &\quad+
    \frac{|\mathcal{E}_{n,p}(\PStab_n)
        -\E\mathcal{E}_{n,p}(\PStab_n)|}{\sqrt n}\\
    &\quad+
    \left|
        \frac{\E\mathcal{E}_{n,p}(\PStab_n)}{\sqrt n}
        -\sqrt{2\log6}
    \right|
    \le\varepsilon.
\end{align*}
\end{proof}
The case $t_n=0$ gives the following stabilizer energy value.

\begin{corollary}[Sharp normalized stabilizer energy]
\label{cor:sharp-clifford-energy}
For every $\varepsilon>0$, there exists $p_\varepsilon\ge2$ such that,
for every fixed integer $p\ge p_\varepsilon$,
\[
    \P\left(
        \left|
            \frac{\mathcal{E}_{n,p}(\Stab_n)}{\sqrt n}
            -
            \sqrt{2\log6}
        \right|
        \le\varepsilon
    \right)
    \longrightarrow1
\]
as $n\to\infty$.
\end{corollary}

\begin{proof}
Apply the large-$p$ assertion of
Corollary~\ref{cor:intro-t-count}\textup{(ii)} with $t_n=0$
and use \eqref{eq:t-count-zero-energy}.
\end{proof}

\begin{remark}
The constant $\sqrt{2\log6}$ reflects the six possible one-qubit
stabilizer states.  It does not conflict with the unrestricted
product-state value \cite{anschuetz2025bounds}, which grows with $p$, because arbitrary product
states may choose continuously many Bloch-sphere directions whereas
product stabilizer states have only six local choices.
\end{remark}

\section*{Acknowledgements}
OAG gratefully acknowledges support from the U.S. National Science Foundation under Award No. DMS-2602099 and from the Eric and Wendy Schmidt Center at the Broad Institute of MIT and Harvard. DG acknowledges support from ONR grant N000142512545. 

\section*{Use of AI}
For the reverse-magic result with clean ancillas in
Theorem~\ref{thm:revmag-clean-ground-gap}, GPT 5.6 Sol suggested using
Clifford invariance and the clean-output condition to bound stabilizer
nullity independently of the number of ancillas, after being prompted
with an incorrect proof in the discardable-ancilla setting. The authors
developed this observation into the full proof.

After we provided our proof for stabilizer states and the sharper
benchmark comparison, we asked GPT 5.6 Sol to suggest other classes to
investigate. It suggested extending the analysis to Clifford+$T$
circuits, studied in Section~\ref{sec:sharper-conditioning-benchmarks}.

We additionally used GPT 5.6 Sol to assist with proof checking.
OpenAI Codex was used for a pre-submission review of proofs, references,
and presentation. The authors take full responsibility for the paper's
content.

\begin{appendix}

\section{Profile entropy from restriction growth}
\label{app:profile-entropy-growth}

We record the standard entropy consequence of polynomial
restriction growth used in Sections~\ref{subsec:forward-profile-complexity}
and~\ref{sec:sharper-conditioning-benchmarks}.

\begin{lemma}\label{lem:growth-entropy}
Let $\Omega$ be a finite set, let $\mu$ be a probability measure on $\Omega$,
and let $\cF\subseteq[-1,1]^\Omega$. Suppose that there exist a positive integer $d$ and
$A\ge1$ such that, for every $m\ge d$ and every
$\omega_1,\ldots,\omega_m\in\Omega$,
\[
  \#\{
    (f(\omega_1),\ldots,f(\omega_m)):
    f\in\cF
  \}
  \le
  \left(\frac{Am}{d}\right)^d.
\]
Then, for every $0<\eps<1$,
\[
  \log N(\eps,\cF,L_2(\mu))
  \le
  C d\log\left(\frac{eA}{\eps}\right),
\]
where $C>0$ is a universal constant.
\end{lemma}

\begin{proof}
This is the standard approximation argument for classes with polynomial
restriction growth; see, for example,
\cite[Chapter~II, Lemma~25]{poll}.

If $\mathcal F=\emptyset$, there is nothing to prove, so
assume $\mathcal F\ne\emptyset$.  For completeness, apply the restriction-growth hypothesis to
a finite list containing every point of $\Omega$, with repetitions if
needed to make its length at least $d$.  It follows that $\mathcal F$
contains only finitely many distinct profiles, so we may choose a maximal
$\eps$-separated family $\mathcal G\subseteq\mathcal F$ in $L_2(\mu)$,
and write $M=|\mathcal G|$. Since
$\mathcal F\subseteq[-1,1]^\Omega$, any distinct $f,g\in\mathcal G$ satisfy
\[
    \mu\{f\ne g\}
    \ge \frac{1}{4}\|f-g\|_{L_2(\mu)}^2
    > \frac{\eps^2}{4}.
\]
Take
$m:=d+\lceil8\eps^{-2}\log(\max\{M,2\})\rceil$.  A union bound shows
that $m$ independent points sampled from $\mu$ distinguish every pair
in $\mathcal G$ with positive probability. For such a sample, the restriction-growth assumption
gives $M\le (Am/d )^d$.

Writing $x=(\log M)/d$, the last inequality implies
$x\le\log(CA(1+x/\eps^2))$.  If $x\le2$ there is nothing to prove;
if $x>2$, using $\log x\le x/2$ and rearranging gives
$x\le C\log(eA/\eps)$.  Thus $\log M
    \le
    C d\log (eA/\eps)$. Since a maximal $\eps$-separated family is an $\eps$-net, the result follows.
\end{proof}

\section{Pauli and stabilizer
preliminaries}
\label{sec:pauli-preliminaries}

This appendix reviews the standard Pauli and stabilizer facts
used in Sections~\ref{sec:proofs-magic} and
\ref{sec:sharper-conditioning-benchmarks}.  General references for the
stabilizer formalism include \cite{gottesman1997stabilizer}; the binary
symplectic description used here is developed in
\cite{dehaene2003clifford}.  For background on characters of finite
Abelian groups and Pauli measurements, see
\cite[Secs.~4.2 and~7]{bermejo2012classical}.

\subsection{Pauli operators and binary labels}
\label{app:PauliOpsandLabels}

Let $\mathcal P_M
    :=
    \{I,X,Y,Z\}^{\otimes M}$ denote the set of phase-free Pauli operators on $M$ qubits.
Every $R\in\mathcal P_M$ is Hermitian, satisfies $R^2=I$,
and has eigenvalues in $\{\pm1\}$. We also write
$\mathcal P_M^\pm:=\{\pm R:R\in\mathcal P_M\}$ for the
signed Hermitian Paulis, as in Section~\ref{sec:revmag-clean}.

The set $\mathcal P_M$ is not closed under matrix multiplication:
for example, $XZ=-iY$.  To obtain a group, we include the possible
phase factors and define the full Pauli group
\[
    \widetilde{\mathcal P}_M
    :=
    \left\{
        \omega R:
        R\in\mathcal P_M,\;
        \omega\in\{1,-1,i,-i\}
    \right\},
\]
whose group operation is ordinary matrix multiplication.

A phase-free Pauli operator
$R=R_1\otimes\cdots\otimes R_M$ is represented by a binary label
\[
    [R]=(x,z)\in\mathbb F_2^M\times\mathbb F_2^M
    =
    \mathbb F_2^{2M}
\]
according to the single-qubit correspondence
\[
    I\longleftrightarrow(0,0),
    \qquad
    X\longleftrightarrow(1,0),
    \qquad
    Z\longleftrightarrow(0,1),
    \qquad
    Y\longleftrightarrow(1,1).
\]
Thus $x_j=1$ precisely when $R_j\in\{X,Y\}$, while $z_j=1$
precisely when $R_j\in\{Z,Y\}$.  These bits record the presence
of $X$- and $Z$-components, respectively, and $Y$ has both components
because $Y=iXZ$.

We extend the label notation to the full Pauli group by ignoring
the global phase:
\[
    [\omega R]:=[R],
    \qquad
    R\in\mathcal P_M,\quad
    \omega\in\{1,-1,i,-i\}.
\]
With this convention, Pauli multiplication corresponds to addition
of binary labels:
\[
    [QQ']=[Q]+[Q'],
    \qquad
    Q,Q'\in\widetilde{\mathcal P}_M,
\]
where addition is taken modulo two.  Identifying Pauli operators
that differ only by a global phase gives the \emph{projective Pauli
group}, which is therefore isomorphic to the additive group
$(\mathbb F_2^{2M},+)$.

For $u=(x,z),v=(x',z')\in\mathbb F_2^{2M}$, define the
symplectic inner product
\[
    \langle u,v\rangle_{\mathrm{sp}}
    :=
    x\cdot z'+z\cdot x'
    \pmod 2.
\]
The commutation relation between Pauli operators is
\[
    QQ'
    =
    (-1)^{
        \langle[Q],[Q']\rangle_{\mathrm{sp}}
    }
    Q'Q,
    \qquad
    Q,Q'\in\widetilde{\mathcal P}_M.
\]
Consequently,
\[
    \langle[Q],[Q']\rangle_{\mathrm{sp}}=0
    \quad\Longleftrightarrow\quad
    Q\text{ and }Q'\text{ commute},
\]
whereas symplectic inner product one corresponds to anticommutation.
These properties depend only on the labels, not on the global phases.

The phase-free Pauli operators form an orthogonal basis for the
space of operators on $M$ qubits:
\[
    2^{-M}\Tr(RR')
    =
    \mathbf 1_{\{R=R'\}},
    \qquad
    R,R'\in\mathcal P_M.
\]

For a Pauli operator $R=R_1\otimes\cdots\otimes R_M$, its
\emph{support} is the set of qubits $j$ for which $R_j\ne I$.
Multiplying by a global phase does not change this support.

A unitary $C$ on $M$ qubits is called a \emph{Clifford unitary}
if conjugation by $C$ maps the full Pauli group to itself:
\[
    C\widetilde{\mathcal P}_M C^\dagger
    =\widetilde{\mathcal P}_M.
\]
As in Section~\ref{subsec:magic-background}, $\Cliff_M$ denotes
the group of these unitaries modulo global phase. Whenever a
Clifford acts on a vector or operator, we choose a unitary
representative of its class.
In particular, if $R\in\mathcal P_M$, then
$CRC^\dagger=\pm R'$ for some $R'\in\mathcal P_M$.
The induced map on binary labels is linear and invertible,
and it preserves the symplectic inner product because
unitary conjugation preserves commutation and anticommutation.

A linear map $F$ between spaces equipped with symplectic
inner products is called \emph{symplectic} if
\[
    \langle F(u),F(v)\rangle_{\mathrm{sp}}
    =
    \langle u,v\rangle_{\mathrm{sp}}
\]
for every pair of vectors $u,v$ in its domain.
For Pauli labels, this means that the map preserves
commutation and anticommutation relations.

We also use the converse. If an invertible linear map
$F:\mathbb F_2^{2M}\to\mathbb F_2^{2M}$ preserves the
symplectic inner product, then there exists
$C\in\Cliff_M$ such that
\begin{equation}
\label{eq:app-clifford-symplectic-action}
    [CQC^\dagger]=F([Q]),
    \qquad Q\in\widetilde{\mathcal P}_M;
\end{equation}
see \cite[Sections~II--IV]{dehaene2003clifford}.
Thus a linear change of Pauli labels preserving their
commutation relations can be implemented by a Clifford
unitary. The labels determine the resulting Hermitian
Pauli operators only up to sign, which we handle
explicitly when imposing stabilizer constraints.

\subsection{Isotropic subspaces and commuting Pauli groups}
\label{app:IsotropicSubspaces}

A linear subspace $W\subseteq\mathbb F_2^{2M}$ is called
\emph{isotropic} if
\[
    \langle w,w'\rangle_{\mathrm{sp}}=0
    \qquad
    \text{for every }w,w'\in W.
\]
By the commutation relation, this means that Pauli operators
whose labels lie in $W$ commute pairwise, regardless of their
signs or global phases.

The symplectic orthogonal complement of $W$ is
\[
    W^\perp
    :=
    \left\{
        v\in\mathbb F_2^{2M}:
        \langle v,w\rangle_{\mathrm{sp}}=0
        \text{ for every }w\in W
    \right\}.
\]
For a label $v$, we abbreviate
$v^\perp:=(\operatorname{span}\{v\})^\perp$.
If $W$ is isotropic, then $W\subseteq W^\perp$.  Since the
symplectic form is nondegenerate,
\[
    \dim W+\dim W^\perp=2M,
\]
and therefore $\dim W\leq M$.

We next relate these subspaces to simultaneous eigenspaces.
Let $\mathcal G$ be a subgroup of $\widetilde{\mathcal P}_M$
such that its elements commute pairwise and $-I\notin\mathcal G$.
Such a subgroup is called a \emph{stabilizer group}.
Every element of $\mathcal G$ has the form $\pm R$ with
$R\in\mathcal P_M$, because an element of the form $\pm iR$
would square to $-I$.  Moreover, $\mathcal G$ cannot contain
both $R$ and $-R$, since their product is $-I$.
Thus each binary label corresponding to an element of $\mathcal G$
has exactly one signed representative in $\mathcal G$.
Equivalently, if $G,H\in\mathcal G$ satisfy $[G]=[H]$,
then $G=H$.

Define the \emph{label space} of $\mathcal G$ by
\[
    L_{\mathcal G}
    :=
    \{[G]:G\in\mathcal G\}
    \subseteq\mathbb F_2^{2M}.
\]
This is a vector subspace because $[I]=0$ and
$[GG']=[G]+[G']$.  It is isotropic because the elements of
$\mathcal G$ commute.

Write $k:=\dim L_{\mathcal G}$.
This means that there are $k$ linearly independent binary labels
such that every label in $L_{\mathcal G}$ is a unique sum,
modulo two, of these labels.
Choose $G_1,\ldots,G_k\in\mathcal G$ whose labels form such a basis.
Since each label has a unique representative in $\mathcal G$,
every element of $\mathcal G$ can be written uniquely as
\[
    G_1^{a_1}\cdots G_k^{a_k},
    \qquad
    (a_1,\ldots,a_k)\in\mathbb F_2^k.
\]
In particular, $|\mathcal G|=2^k$, with $k\le M$.
We say that $G_1,\ldots,G_k$ generate $\mathcal G$ and
write $\mathcal G=\langle G_1,\ldots,G_k\rangle$.
An empty list of generators gives the trivial group $\{I\}$.

Define the common $+1$ eigenspace of $\mathcal G$ by
\[
    \mathcal K_{\mathcal G}
    :=
    \left\{
        \ket{\psi}\in(\mathbb C^2)^{\otimes M}:
        G\ket{\psi}=\ket{\psi}
        \text{ for every }G\in\mathcal G
    \right\}.
\]
Since $G_1,\ldots,G_k$ generate $\mathcal G$, a vector belongs
to $\mathcal K_{\mathcal G}$ exactly when it lies in the $+1$
eigenspace of every generator.

Each $G_j$ has eigenvalues in $\{ \pm 1\}$. Consequently, $(I+G_j)/2$ acts as the identity on its $+1$
eigenspace and as zero on its $-1$ eigenspace, so it is the
orthogonal projector onto the former. These projectors commute, and applying one therefore preserves
the constraints imposed by the others. Their product
\[
    \Pi_{\mathcal G}
    :=
    \prod_{j=1}^k\frac{I+G_j}{2}
\]
is thus the orthogonal projector onto $\mathcal K_{\mathcal G}$.

Expanding the product amounts to choosing either $I$ or $G_j$
from each factor. Writing these choices as $a_j=0$ or $a_j=1$,
respectively, gives
\[
    \Pi_{\mathcal G}
    =
    2^{-k}
    \sum_{(a_1,\ldots,a_k)\in\{0,1\}^k}
    G_1^{a_1}\cdots G_k^{a_k}
    =
    2^{-k}\sum_{G\in\mathcal G}G.
\]
The last equality holds because every element of $\mathcal G$
has exactly one such representation, so each group element
appears once in the expansion.

Only the identity has nonzero trace among the elements of
$\mathcal G$, and $\Tr(I)=2^M$. Since the trace of an orthogonal
projector equals the dimension of its range, we conclude that
\[
    \dim\mathcal K_{\mathcal G}
    =
    \Tr(\Pi_{\mathcal G})
    =
    2^{-k}\Tr(I)
    =
    2^{M-k}.
\]

The normalized projector
\[
    \tau_{\mathcal G}:=\frac{\Pi_{\mathcal G}}{2^{M-k}}
    =2^{-M}\sum_{G\in\mathcal G}G
\]
is a mixture of pure stabilizer states. Tracing out qubits in the Pauli expansion of a pure stabilizer state retains exactly the stabilizers acting trivially on the discarded qubits, giving the same form on the remaining qubits.
Thus the reduced states in Definition~\ref{def:mixed-stabilizer}
are mixtures of pure stabilizer states.

\subsection{Stabilizer nullity and logical Pauli spaces}
\label{app:stabilizer-nullity}

In Subsection~\ref{app:IsotropicSubspaces}, we started with a
stabilizer group and considered the vectors satisfying all of
its constraints. Here we reverse the starting point: given a
pure state, we collect its Pauli constraints and ask how much
freedom those constraints leave. We then describe the Pauli
operations acting on the resulting space and explain how this
description behaves under gates and conditioning.

Let $|\eta\rangle$ be a normalized pure $M$-qubit state.
Its \emph{Pauli stabilizer group} is
\[
    \mathcal G_\eta
    :=
    \left\{
        G\in\mathcal P_M^\pm:
        G|\eta\rangle=|\eta\rangle
    \right\}.
\]
This is a stabilizer group in the sense of the preceding
subsection. It contains $I$, and if
$G,H\in\mathcal G_\eta$, then
\[
    GH|\eta\rangle=|\eta\rangle,
    \qquad
    G^{-1}|\eta\rangle=|\eta\rangle.
\]
It cannot contain $-I$ because $|\eta\rangle\ne0$.
Moreover, any two Pauli operators either commute or
anticommute, and anticommutation would give
\[
    |\eta\rangle
    =GH|\eta\rangle
    =-HG|\eta\rangle
    =-|\eta\rangle,
\]
which is impossible.

The \emph{Pauli stabilizer space} of $\eta$ is the binary
label space of this group:
\[
    \mathsf S(\eta)
    :=L_{\mathcal G_\eta}
    =\left\{[G]:G\in\mathcal G_\eta\right\}.
\]
Equivalently,
\[
    \mathsf S(\eta)
    =
    \left\{
        [R]:R\in\mathcal P_M,\quad
        R|\eta\rangle=\pm|\eta\rangle
    \right\}.
\]
To see the equivalence, if
$R|\eta\rangle=\varepsilon|\eta\rangle$ with
$\varepsilon\in\{\pm1\}$, then
$\varepsilon R\in\mathcal G_\eta$ and
$[\varepsilon R]=[R]$.
Thus $\mathcal G_\eta$ retains the signs of the constraints,
while $\mathsf S(\eta)$ records which phase-free Paulis have
$|\eta\rangle$ as an eigenvector.

Since $\mathsf S(\eta)$ is isotropic, its dimension is at most
$M$. The \emph{stabilizer nullity}, introduced in
\cite{beverland2020lower}, is
\[
    \nu_{\mathrm{st}}(\eta)
    :=M-\dim\mathsf S(\eta).
\]
Write
\[
    \mathcal K_\eta:=\mathcal K_{\mathcal G_\eta},
    \qquad
    \Pi_\eta:=\Pi_{\mathcal G_\eta}.
\]
The space $\mathcal K_\eta$ consists of all vectors satisfying
every Pauli constraint of $|\eta\rangle$; some of these vectors
may also satisfy additional constraints. In particular,
$|\eta\rangle\in\mathcal K_\eta$, but the constraints need not
determine $|\eta\rangle$ uniquely.

If $\nu_{\mathrm{st}}(\eta)=t$, then there are $M-t$
independent Pauli constraints. The eigenspace formula from
Subsection~\ref{app:IsotropicSubspaces} gives
\[
    \dim\mathcal K_\eta=2^t,
    \qquad
    \Pi_\eta
    =2^{-(M-t)}\sum_{G\in\mathcal G_\eta}G.
\]
Thus the constraints leave $t$ qubits' worth of freedom.
A pure state with nullity zero is called a
\emph{pure stabilizer state}: its Pauli constraints determine
a one-dimensional space, and hence determine the normalized
state up to a global phase.

The following decomposition makes the remaining freedom
explicit. A Clifford transformation can turn the $M-t$
independent constraints into the requirement that the last
$M-t$ qubits are in $|0\rangle$. The first $t$ qubits are then
unrestricted. This identifies the entire space
$\mathcal K_\eta$ with a $t$-qubit state space.

\begin{lemma}[Stabilizer-nullity decomposition
{\cite{beverland2020lower}}]
\label{lem:app-stabilizer-nullity-decomposition}
Let $|\eta\rangle$ be a pure $M$-qubit state with
$\nu_{\mathrm{st}}(\eta)=t$.
There exists $V\in\Cliff_M$ such that
\[
    V\mathcal G_\eta V^\dagger
    =\langle Z_{t+1},\ldots,Z_M\rangle.
\]
Consequently,
\[
\begin{aligned}
    V\mathcal K_\eta
    &=(\mathbb C^2)^{\otimes t}
      \otimes\operatorname{span}\{|0^{M-t}\rangle\},\\
    V\Pi_\eta V^\dagger
    &=I_{2^t}\otimes|0^{M-t}\rangle\langle0^{M-t}|.
\end{aligned}
\]
In particular, there is a pure $t$-qubit state $|\chi\rangle$
such that
\[
    V|\eta\rangle=|\chi\rangle\otimes|0^{M-t}\rangle,
    \qquad \mathsf S(\chi)=\{0\}.
\]
\end{lemma}

\begin{proof}
Choose generators $G_1,\ldots,G_{M-t}$ whose labels form a
basis of $\mathsf S(\eta)$. Since this subspace is
isotropic, we can complete these labels to a symplectic
basis. Mapping this basis to the standard symplectic
basis, with a suitable ordering, gives an invertible
linear map $F$ preserving the symplectic inner product
and satisfying
\[
    F([G_j])=[Z_{t+j}],
    \qquad j=1,\ldots,M-t.
\]
By \eqref{eq:app-clifford-symplectic-action}, there is
$V_0\in\Cliff_M$ implementing this map on Pauli labels.
Consequently,
\[
    V_0G_jV_0^\dagger=\varepsilon_jZ_{t+j},
    \qquad \varepsilon_j\in\{\pm1\}.
\]
Since $XZX=-Z$, the Clifford
\[
    V:=\left(\prod_{j:\,\varepsilon_j=-1}X_{t+j}\right)V_0
\]
corrects all signs, so that
$VG_jV^\dagger=Z_{t+j}$ for every $j$.
Because the $G_j$ generate $\mathcal G_\eta$, this proves
the claimed group identity.

The common $+1$ eigenspace of $Z_{t+1},\ldots,Z_M$
consists of arbitrary states on the first $t$ qubits
tensored with $|0^{M-t}\rangle$.
This gives the stated eigenspace and projector formulas.
Since $|\eta\rangle\in\mathcal K_\eta$, it also gives
the factorization of $V|\eta\rangle$.

Finally, suppose that a nonidentity Pauli $P$ satisfies
$P|\chi\rangle=\varepsilon|\chi\rangle$ for some
$\varepsilon\in\{\pm1\}$. Then
\[
    V^\dagger(\varepsilon P\otimes I)V
    \in\mathcal G_\eta.
\]
This contradicts the group identity, since
$\varepsilon P\otimes I$ acts nontrivially on the first
$t$ qubits, whereas every element of
$\langle Z_{t+1},\ldots,Z_M\rangle$ acts trivially there.
Thus $\mathsf S(\chi)=\{0\}$. If $t=M$, then $\mathsf S(\eta)=\{0\}$, and we may
simply take $V=I$ and $|\chi\rangle=|\eta\rangle$.
\end{proof}

The same proof applies to any stabilizer group
$\mathcal G\subseteq\mathcal P_M^\pm$. Writing
$S:=L_{\mathcal G}$ and $t:=M-\dim S$, it gives a Clifford
$V$ with
\[
    V\mathcal G V^\dagger
    =\langle Z_{t+1},\ldots,Z_M\rangle,
    \qquad
    V\mathcal K_{\mathcal G}
    =(\mathbb C^2)^{\otimes t}
      \otimes\operatorname{span}\{|0^{M-t}\rangle\}.
\]
With the signed group fixed, we also write
$\mathcal K_S:=\mathcal K_{\mathcal G}$, as in
Section~\ref{sec:revmag-clean}; $S$ alone does not specify
the signs of the constraints. The first $t$ qubits in these
coordinates are the \emph{logical qubits}, and $V$ is a
\emph{decoder}. The corresponding encoding is the isometry
\[
    |u\rangle\longmapsto
    V^\dagger\bigl(|u\rangle\otimes|0^{M-t}\rangle\bigr)
\]
onto $\mathcal K_S$. Thus every vector in $\mathcal K_S$
has a unique logical vector, although the original physical
qubits may be entangled. The condition
$\mathsf S(u)=\{0\}$ holds when $\mathcal G$ is the full
Pauli stabilizer group of the encoded state.
For $t=0$, the decomposition identifies pure stabilizer
states with Clifford images of $|0^M\rangle$.

A physical Pauli preserves $\mathcal K_S$ exactly when it
commutes with $\mathcal G$, or equivalently when its label
lies in $S^\perp$. If it anticommutes with a stabilizer,
it sends $\mathcal K_S$ into that stabilizer's $-1$
eigenspace. Moreover,
$(RG)|v\rangle=R|v\rangle$ for $G\in\mathcal G$ and
$|v\rangle\in\mathcal K_S$. We therefore define the
\emph{logical Pauli label space}
\[
    W_S:=S^\perp/S,
    \qquad
    \dim W_S=(M+t)-(M-t)=2t.
\]
Its elements are cosets $a+S=\{a+s:s\in S\}$,
$a\in S^\perp$. The symplectic form descends to $W_S$ by
\[
    \langle a+S,b+S\rangle_{\mathrm{sp}}
    :=\langle a,b\rangle_{\mathrm{sp}}.
\]
This is well-defined because $a,b\in S^\perp$, and
nondegenerate because
$S^\perp\cap(S^\perp)^\perp=S$.
For $\mathcal G=\mathcal G_\eta$, this is
$W_\eta:=\mathsf S(\eta)^\perp/\mathsf S(\eta)$.

\begin{lemma}[Pauli action on the stabilizer eigenspace]
\label{lem:app-stabilizer-pauli-action}
Let $\mathcal G$, $S$, $t$, and $V$ be as above, and let
$R\in\mathcal P_M^\pm$.
\begin{enumerate}[label=\textup{(\roman*)}]
    \item If $[R]\notin S^\perp$, then
    $\Pi_{\mathcal G}R\Pi_{\mathcal G}=0$.

    \item If $[R]\in S^\perp$, then $R$ preserves
    $\mathcal K_S$ and
    \[
        VRV^\dagger=\varepsilon_RP_R\otimes B_R,
        \qquad
        P_R\in\mathcal P_t,\quad
        \varepsilon_R\in\{\pm1\},\quad
        B_R\in\{I,Z\}^{\otimes(M-t)}.
    \]
    Its action after decoding is therefore
    \[
        VRV^\dagger
        \bigl(|u\rangle\otimes|0^{M-t}\rangle\bigr)
        =(\varepsilon_RP_R|u\rangle)
          \otimes|0^{M-t}\rangle.
    \]
    The map $[R]+S\mapsto[P_R]$ is a linear symplectic
    bijection from $W_S$ to $\mathbb F_2^{2t}$.
    Thus two preserving Paulis have the same action on
    $\mathcal K_S$ up to sign if and only if their labels
    belong to the same coset of $S$.
\end{enumerate}
\end{lemma}

\begin{proof}
For \textup{(i)}, choose $G\in\mathcal G$ anticommuting
with $R$. Since $G\Pi_{\mathcal G}=\Pi_{\mathcal G}$,
\[
    \Pi_{\mathcal G}R\Pi_{\mathcal G}
    =\Pi_{\mathcal G}GR\Pi_{\mathcal G}
    =-\Pi_{\mathcal G}RG\Pi_{\mathcal G}
    =-\Pi_{\mathcal G}R\Pi_{\mathcal G}.
\]
For \textup{(ii)}, the signed Pauli $VRV^\dagger$ commutes
with $Z_{t+1},\ldots,Z_M$, so its last $M-t$ factors are
$I$ or $Z$. These factors fix $|0^{M-t}\rangle$, giving
the stated logical action. Conjugation by $V$, followed by
restriction of binary labels to the first $t$ qubits, defines
a linear map $S^\perp\to\mathbb F_2^{2t}$ with kernel $S$.
It is surjective, since
$V^\dagger(P\otimes I)V$ realizes every $P\in\mathcal P_t$.
It preserves the symplectic form because the discarded
$I,Z$ factors commute pairwise. This proves the quotient
identification and the final assertion.
\end{proof}

In particular, with $\eta$ and $\chi$ as in
Lemma~\ref{lem:app-stabilizer-nullity-decomposition},
\[
    \langle\eta|R|\eta\rangle
    =\begin{cases}
        0,&[R]\notin\mathsf S(\eta)^\perp,\\[1mm]
        \varepsilon_R\langle\chi|P_R|\chi\rangle,
          &[R]\in\mathsf S(\eta)^\perp.
      \end{cases}
\]
The $4^t$ logical Pauli labels are generated by $2t$
symplectic basis vectors and their phase-free representatives
form a basis for operators on the $2^t$-dimensional logical
state space.

When representatives of the logical generators are supported
on a specified set, their coordinates can be realized by a
Clifford acting only on that set.

\begin{lemma}[Local realization of logical Pauli coordinates]
\label{lem:app-local-symplectic}
Let $t\ge0$ be an integer, and let
$\overline X_1,\overline Z_1,\ldots,
\overline X_t,\overline Z_t\in\mathcal P_M^\pm$
be $t$ pairs of signed Pauli operators supported on
$\Lambda\subseteq[M]$, satisfying
\[
\begin{aligned}
    \langle[\overline X_i],[\overline X_j]\rangle_{\mathrm{sp}}
    &=\langle[\overline Z_i],[\overline Z_j]\rangle_{\mathrm{sp}}=0,\\
    \langle[\overline X_i],[\overline Z_j]\rangle_{\mathrm{sp}}
    &=\delta_{ij},\qquad i,j\in[t].
\end{aligned}
\]
Then $t\le|\Lambda|$. For any distinct designated qubits
$\ell_1,\ldots,\ell_t\in\Lambda$, there is a Clifford $C$
supported on $\Lambda$ such that
\[
    C\overline X_jC^\dagger=X_{\ell_j},
    \qquad C\overline Z_jC^\dagger=Z_{\ell_j},
    \qquad j\in[t].
\]
\end{lemma}

\begin{proof}
Restrict the labels to $\mathbb F_2^{2|\Lambda|}$.
The displayed pairings show that these $2t$ vectors are
independent and span a nondegenerate symplectic subspace
$U$. Hence $t\le|\Lambda|$, and a symplectic basis of
$U^\perp$ completes them to one of the full label space.
The Clifford--symplectic correspondence gives a Clifford
$C_0$ supported on $\Lambda$ with
\[
    C_0\overline X_jC_0^\dagger=(-1)^{a_j}X_{\ell_j},
    \qquad
    C_0\overline Z_jC_0^\dagger=(-1)^{b_j}Z_{\ell_j},
    \qquad a_j,b_j\in\{0,1\}.
\]
Set $Q:=\prod_{j=1}^tX_{\ell_j}^{b_j}Z_{\ell_j}^{a_j}$
and $C:=QC_0$. Conjugation by $Q$ corrects both signs
independently for each $j$, giving the required identities.
\end{proof}

We next record the nullity properties used in the paper.
Clifford invariance and tensor additivity give
\cite[Section~2.1]{beverland2020lower}
\[
\begin{aligned}
    \nu_{\mathrm{st}}(C\eta)&=\nu_{\mathrm{st}}(\eta),
        \qquad C\in\Cliff_M,\\
    \nu_{\mathrm{st}}(\eta\otimes\xi)
        &=\nu_{\mathrm{st}}(\eta)+\nu_{\mathrm{st}}(\xi).
\end{aligned}
\]

Let $\operatorname{conv}(\PStab_1)$ denote the convex hull of the pure one-qubit stabilizer states.

\begin{lemma}[Nullity under Pauli rotations and conditioning]
\label{lem:app-nullity-conditioning}
Let $|\psi\rangle$ be a pure $M$-qubit state.
\begin{enumerate}[label=\textup{(\roman*)}]
    \item For every $P\in\mathcal P_M^\pm$ and
    $\theta\in\mathbb R$,
    \[
        \nu_{\mathrm{st}}(e^{-i\theta P}\psi)
        \le\nu_{\mathrm{st}}(\psi)+1.
    \]
    In particular, a $T$-gate increases nullity by at most one.

    \item Measure one qubit in a Pauli basis, as in the
    Pauli POVM defined in
    Section~\ref{subsubsec:pauli-povm}.
    Condition on any outcome of positive probability
    and discard the measured qubit. The remaining qubits
    are then in a pure state $|\psi'\rangle$ satisfying
    \[
        \nu_{\mathrm{st}}(\psi')
        \le\nu_{\mathrm{st}}(\psi).
    \]

    \item For the one-qubit reduced density matrices
    $\psi_i:=\Tr_{[M]\setminus\{i\}}|\psi\rangle\langle\psi|$,
    \[
        \#\{i\in[M]:\psi_i\notin\operatorname{conv}(\PStab_1)\}
        \le\nu_{\mathrm{st}}(\psi).
    \]
\end{enumerate}
\end{lemma}

\begin{proof}
Let $S:=\mathsf S(\psi)$ and $k:=\dim S$.
For \textup{(i)}, every $G\in\mathcal G_\psi$ commuting
with $P$ also stabilizes $e^{-i\theta P}|\psi\rangle$,
since
\[
    Ge^{-i\theta P}|\psi\rangle
    =e^{-i\theta P}G|\psi\rangle
    =e^{-i\theta P}|\psi\rangle.
\]
The labels of these operators form $S\cap[P]^\perp$,
whose dimension is at least $k-1$, because commutation
with $P$ imposes at most one linear constraint.
Thus the output nullity is at most
$M-(k-1)=\nu_{\mathrm{st}}(\psi)+1$.
The $T$-gate claim follows from
$T_i=e^{i\pi/8}e^{-i\pi Z_i/8}$.

For \textup{(ii)}, the normalized state conditional on
the measurement outcome $s$, before discarding the
measured qubit, is
\[
    |\phi\rangle
    :=\frac{(I+sP)|\psi\rangle}
            {\|(I+sP)|\psi\rangle\|},
    \qquad P=a_i,\quad a\in\{X,Y,Z\},\quad s\in\{\pm1\}.
\]
If $[P]\in S$, then $|\psi\rangle$ is an eigenvector
of $P$, and the only positive-probability outcome
leaves the state unchanged.

If $[P]\notin S$, every $G\in\mathcal G_\psi$ commuting
with $P$ also stabilizes $|\phi\rangle$, since $G$
commutes with $I+sP$. The labels of these operators
form $S\cap[P]^\perp$, whose dimension is at least $k-1$.
Moreover, $P|\phi\rangle=s|\phi\rangle$, so $sP$ is
a stabilizer of $|\phi\rangle$. Its label $[P]$ lies
outside $S$ and therefore gives an additional
independent stabilizer.

In either case, $\dim\mathsf S(\phi)\ge k$, so
$\nu_{\mathrm{st}}(\phi)\le\nu_{\mathrm{st}}(\psi)$; see
~\cite[Proposition~2.3]{beverland2020lower}.
The conditional state is the tensor product of the
measured qubit's pure Pauli eigenstate and the pure
state $|\psi'\rangle$ on the remaining qubits.
The Pauli eigenstate has nullity zero, so tensor
additivity gives
$\nu_{\mathrm{st}}(\psi')=\nu_{\mathrm{st}}(\phi)$.

For \textup{(iii)}, let $F$ be the sites on which every
element of $\mathcal G_\psi$ acts trivially. If $i\notin F$,
choose $G\in\mathcal G_\psi$ with a nonidentity factor
$a\in\{X,Y,Z\}$ at $i$. For each $b \in \{X,Y,Z\} \setminus \{a\}$, the operator $b_i$ anticommutes with $G$, so
\[
    \langle\psi|b_i|\psi\rangle
    =\langle\psi|Gb_iG|\psi\rangle
    =-\langle\psi|b_i|\psi\rangle=0.
\]
Hence, for some $\alpha\in[-1,1]$,
\[
    \psi_i
    =\frac{I+\alpha a}{2}
    =\frac{1+\alpha}{2}\frac{I+a}{2}
     +\frac{1-\alpha}{2}\frac{I-a}{2}.
\]
Since $(I+a)/2$ and $(I-a)/2$ are pure one-qubit
stabilizer states, $\psi_i\in\operatorname{conv}(\PStab_1)$.
All exceptional sites therefore lie in $F$. Restricting
$S$ to $F^c$ identifies it with an isotropic subspace on
$M-|F|$ qubits, so $k\le M-|F|$. Consequently,
\[
    \#\{i:\psi_i\notin\operatorname{conv}(\PStab_1)\}
    \le|F|\le M-k=\nu_{\mathrm{st}}(\psi).
\]
This is the support-counting argument of
\cite[Lemma~9]{coble2023hamiltonians}, expressed in terms
of nullity.
\end{proof}

Part~\textup{(iii)} applies to the one-qubit marginals of
$|\psi\rangle$ itself, without first applying a Clifford
transformation. Even if every $\psi_i$ belongs to
$\operatorname{conv}(\PStab_1)$, the full state $|\psi\rangle$
need not be a stabilizer state. Together with Clifford invariance, these properties show
that a Clifford+$T$ circuit starting from a pure stabilizer
state and using at most $s$ $T$-gates has nullity at most $s$,
even after any sequence of positive-probability
single-qubit Pauli postselections. For every sequence of measurement outcomes with positive probability, at most $s$ of the remaining one-qubit marginals of the conditional state lie outside
$\operatorname{conv}(\PStab_1)$, independently of the number
of Clifford gates or pure stabilizer ancillas.

\subsection{\texorpdfstring{Exact workspace compression for Clifford+$T$
circuits}{Exact workspace compression for Clifford+T circuits}}
\label{app:t-count-workspace}

The fixed-projection representation in
Lemma~\ref{lem:t-count-fixed-projection} suffices for the
profile-entropy bound. The preparation itself also admits
a bounded-workspace representation with no increase in
$T$-count.

Let $\CliffT^{\mathrm{pure}}_{M,t}$ be the rank-one projectors
onto states prepared from $|0^M\rangle$ using at most $t$
$T$-gates and unrestricted Clifford operations. Then
\[
    \CliffT^{\mathrm{anc}}_{n,t}
    =\bigcup_{M\ge n}\left\{
        \Tr_{\{n+1,\ldots,M\}}\omega:
        \omega\in\CliffT^{\mathrm{pure}}_{M,t}
      \right\}.
\]
For $P\in\mathcal P_M^\pm$, write
$R(P):=e^{-i\pi P/8}$, so that
$T_i=e^{i\pi/8}R(Z_i)$.

\begin{lemma}[Pauli-rotation normal form and nullity cost]
\label{lem:t-count-nullity-cost}
A projector $|\Psi\rangle\langle\Psi|$ belongs to
$\CliffT^{\mathrm{pure}}_{M,t}$ if and only if, up to
global phase,
\[
    |\Psi\rangle=R(P_s)\cdots R(P_1)|\phi\rangle,
    \qquad 0\le s\le t,
\]
where $|\phi\rangle$ is a pure stabilizer state and
$P_j\in\mathcal P_M^\pm\setminus\{\pm I\}$.
The empty product is the identity. Every such state has
$\nu_{\mathrm{st}}(\Psi)\le t$.
\end{lemma}

\begin{proof}
Replace each $T_i$ by $R(Z_i)$ up to phase and move
Clifford factors to the right using
$CR(P)=R(CPC^\dagger)C$; see
\cite[Proposition~1]{gosset2013algorithm}.
Conversely, every signed nonidentity Hermitian Pauli is
$P=VZ_1V^\dagger$ for a Clifford $V$, so
$R(P)=e^{-i\pi/8}VT_1V^\dagger$ uses one $T$-gate.
The nullity bound follows by iterating
Lemma~\ref{lem:app-nullity-conditioning}\textup{(i)}
from $\nu_{\mathrm{st}}(\phi)=0$.
\end{proof}

For $t=0$, the following compression bound follows from
the bipartite stabilizer normal form of Fattal et al.\
\cite[Theorem~1]{fattal2004entanglement}.
We extend this bound to Clifford+$T$ circuits by retaining
the stabilizers supported on the discarded ancillas
that commute with every Pauli rotation.

\begin{lemma}[Exact compression of discardable workspace]
\label{lem:t-count-workspace-compression}
For integers $n\ge1$ and $t\ge0$, every state in
$\CliffT^{\mathrm{anc}}_{n,t}$ can be prepared using at most
$t$ $T$-gates on at most $2n+t$ total qubits. Equivalently,
\[
    \CliffT^{\mathrm{anc}}_{n,t}
    =\left\{
        \Tr_{\{n+1,\ldots,2n+t\}}\omega:
        \omega\in\CliffT^{\mathrm{pure}}_{2n+t,t}
      \right\}.
\]
\end{lemma}

\begin{proof}
Let $A=[n]$ index the retained qubits, and let $E$ index
the $r$ ancillary qubits discarded at the output.
By Lemma~\ref{lem:t-count-nullity-cost}, a preparation on
$M=n+r$ qubits has the form
\[
    \rho_A=\Tr_E|\Psi\rangle\langle\Psi|,
    \qquad
    |\Psi\rangle=R(P_s)\cdots R(P_1)|\phi\rangle,
    \qquad s\le t,
\]
with $\phi$ a pure stabilizer state. Consider the stabilizers of $\phi$ supported on $E$
that commute with each $P_j$. Their labels form the
subspace
\[
    L:=\{(x,z)\in\mathsf S(\phi):x_A=z_A=0\}
       \cap\bigcap_{j=1}^s[P_j]^\perp.
\]
Since $\dim\mathsf S(\phi)=M$, the $2n$ support constraints
and $s$ commutation constraints give
\begin{equation}
\label{eq:t-count-environment-stabilizers}
    d:=\dim L\ge M-2n-s=r-n-s.
\end{equation}
Choose signed stabilizers $G_1,\ldots,G_d$ with labels
forming a basis of $L$. Apply the stabilizer-group
version of Lemma~\ref{lem:app-stabilizer-nullity-decomposition}
within $E$ to obtain a Clifford $V=I_A\otimes V_E$
acting only on the qubits in $E$ and sending these
generators to $Z_{e_1},\ldots,Z_{e_d}$.
Writing $E_0:=\{e_1,\ldots,e_d\}$ and $E':=E\setminus E_0$,
we have
\[
    V|\phi\rangle=|\phi'\rangle_{AE'}\otimes|0^d\rangle_{E_0},
    \qquad
    VP_jV^\dagger=P'_j\otimes Z^{z_j},
\]
where $\phi'$ is a pure stabilizer state,
$P'_j$ is a signed Hermitian Pauli on $AE'$, and
$Z^{z_j}\in\{I,Z\}^{\otimes d}$.
The first assertion follows from the $+1$ constraints
and nullity additivity. The second follows because each
$P_j$ commutes with all $G_\ell$.

Since $Z^{z_j}|0^d\rangle=|0^d\rangle$, for every vector
$|\xi\rangle$ on $AE'$,
\[
    R(VP_jV^\dagger)(|\xi\rangle\otimes|0^d\rangle)
    =R(P'_j)|\xi\rangle\otimes|0^d\rangle.
\]
Applying this successively gives
\[
    V|\Psi\rangle=|\Psi'\rangle_{AE'}\otimes|0^d\rangle_{E_0},
    \qquad
    |\Psi'\rangle:=R(P'_s)\cdots R(P'_1)|\phi'\rangle.
\]
As $V$ acts only on the discarded qubits,
$\rho_A=\Tr_{E'}|\Psi'\rangle\langle\Psi'|$.
The remaining width satisfies
\[
    M-d\le2n+s\le2n+t.
\]
Since $\phi'$ is a pure stabilizer state, it can be
prepared using only Clifford gates. Rotations with
$P'_j=\pm I$ change only the global phase and can be
omitted. By Lemma~\ref{lem:t-count-nullity-cost}, each
remaining rotation can be implemented using one $T$-gate
and Clifford gates. This prepares $\rho_A$ on at most
$2n+t$ qubits using at most $s\le t$ $T$-gates.

Adding unused ancillary qubits initialized in $|0\rangle$
gives a preparation on exactly $2n+t$ qubits without
changing $\rho_A$ or the $T$-count. Conversely, every
$n$-qubit state obtained from such a preparation by
discarding the ancillary qubits belongs to
$\CliffT^{\mathrm{anc}}_{n,t}$ by definition.
\end{proof}
\end{appendix}

\bibliographystyle{alpha}
\bibliography{bibliography-05.2026}

@article{baiguera2026quantum,
  title={Quantum complexity in gravity, quantum field theory, and quantum information science},
  author={Baiguera, Stefano and Balasubramanian, Vijay and Caputa, Pawel and Chapman, Shira and Haferkamp, Jonas and Heller, Michal P and Halpern, Nicole Yunger},
  journal={Physics Reports},
  volume={1159},
  pages={1--77},
  year={2026},
  publisher={Elsevier}
}

@article{huang2020predicting,
  title={Predicting many properties of a quantum system from very few measurements},
  author={Huang, Hsin-Yuan and Kueng, Richard and Preskill, John},
  journal={Nature Physics},
  volume={16},
  number={10},
  pages={1050--1057},
  year={2020},
  publisher={Nature Publishing Group UK London}
}

@inproceedings{parham2026quantum,
  title={Quantum circuit lower bounds in the magic hierarchy},
  author={Parham, Natalie},
  booktitle={Proceedings of the 58th Annual ACM Symposium on Theory of Computing},
  pages={734--743},
  year={2026}
}

@article{al2026circuit,
  title={Circuit complexity lower bounds for quantum spin glasses},
  author={Al-Ghattas, Omar and Gamarnik, David},
  journal={arXiv preprint arXiv:2607.14384},
  year={2026}
}

@article{berta2022semidefinite,
  title={Semidefinite programming hierarchies for constrained bilinear optimization},
  author={Berta, Mario and Borderi, Francesco and Fawzi, Omar and Scholz, Volkher B},
  journal={Mathematical Programming},
  volume={194},
  number={1},
  pages={781--829},
  year={2022},
  publisher={Springer}
}

@article{bravyi2005universal,
  title={Universal quantum computation with ideal Clifford gates and noisy ancillas},
  author={Bravyi, Sergey and Kitaev, Alexei},
  journal={Physical Review A—Atomic, Molecular, and Optical Physics},
  volume={71},
  number={2},
  pages={022316},
  year={2005},
  publisher={APS}
}

@article{veitch2014resource,
  title={The resource theory of stabilizer quantum computation},
  author={Veitch, Victor and Hamed Mousavian, SA and Gottesman, Daniel and Emerson, Joseph},
  journal={New Journal of Physics},
  volume={16},
  number={1},
  pages={013009},
  year={2014},
  publisher={IOP Publishing}
}

@article{beverland2020lower,
  title={Lower bounds on the non-Clifford resources for quantum computations},
  author={Beverland, Michael and Campbell, Earl and Howard, Mark and Kliuchnikov, Vadym},
  journal={Quantum Science \& Technology},
  volume={5},
  number={3},
  pages={035009},
  year={2020},
  publisher={IOP Publishing}
}

@book{gottesman1997stabilizer,
  title={Stabilizer codes and quantum error correction},
  author={Gottesman, Daniel},
  year={1997},
  publisher={California Institute of Technology}
}

@inproceedings{brandao2013quantum,
  title={Quantum de Finetti theorems under local measurements with applications},
  author={Brandao, Fernando GSL and Harrow, Aram W},
  booktitle={Proceedings of the forty-fifth annual ACM symposium on Theory of computing},
  pages={861--870},
  year={2013}
}

@book{watrous2018theory,
  title={The theory of quantum information},
  author={Watrous, John},
  volume={1},
  year={2018},
  publisher={Cambridge university press Cambridge}
}

@article{dehaene2003clifford,
  title={Clifford group, stabilizer states, and linear and quadratic operations over GF (2)},
  author={Dehaene, Jeroen and De Moor, Bart},
  journal={Physical Review A},
  volume={68},
  number={4},
  pages={042318},
  year={2003},
  publisher={APS}
}

@article{chapman2022quantum,
  title={Quantum computational complexity from quantum information to black holes and back},
  author={Chapman, Shira and Policastro, Giuseppe},
  journal={The European Physical Journal C},
  volume={82},
  number={2},
  pages={128},
  year={2022},
  publisher={Springer}
}

@article{bermejo2012classical,
  title={Classical simulations of Abelian-group normalizer circuits with intermediate measurements},
  author={Bermejo-Vega, Juan and Nest, Maarten Van den},
  journal={arXiv preprint arXiv:1210.3637},
  year={2012}
}

@article{anschuetz2025strongly,
  title={Strongly interacting fermions are nontrivial yet nonglassy},
  author={Anschuetz, Eric R and Chen, Chi-Fang and Kiani, Bobak T and King, Robbie},
  journal={Physical Review Letters},
  volume={135},
  number={3},
  pages={030602},
  year={2025},
  publisher={APS}
}

@article{tropp2012user,
  title={User-friendly tail bounds for sums of random matrices},
  author={Tropp, Joel A},
  journal={Foundations of computational mathematics},
  volume={12},
  number={4},
  pages={389--434},
  year={2012},
  publisher={Springer}
}

@article{anschuetz2025bounds,
  title={Bounds on the Ground State Energy of Quantum p-Spin Hamiltonians},
  author={Anschuetz, Eric R and Gamarnik, David and Kiani, Bobak T},
  journal={Communications in Mathematical Physics},
  volume={406},
  number={10},
  pages={232},
  year={2025},
  publisher={Springer}
}

@article{anshu2022nlts,
  title={{NLTS Hamiltonians from good quantum codes}},
  author={Anshu, Anurag and Breuckmann, Nikolas and Nirkhe, Chinmay},
  journal={arXiv preprint arXiv:2206.13228},
  year={2022}
}

@article{anshu2020circuit,
  title={Circuit lower bounds for low-energy states of quantum code Hamiltonians},
  author={Anshu, Anurag and Nirkhe, Chinmay},
  journal={arXiv preprint arXiv:2011.02044},
  year={2020}
}

@book{nielsen2010quantum,
  title={Quantum computation and quantum information},
  author={Nielsen, Michael A and Chuang, Isaac L},
  year={2010},
  publisher={Cambridge university press}
}

@article{derrida1981random,
  title={Random-energy model: An exactly solvable model of disordered systems},
  author={Derrida, Bernard},
  journal={Physical Review B},
  volume={24},
  number={5},
  pages={2613},
  year={1981},
  publisher={APS}
}

@article{zhao2024learning,
  title={Learning quantum states and unitaries of bounded gate complexity},
  author={Zhao, Haimeng and Lewis, Laura and Kannan, Ishaan and Quek, Yihui and Huang, Hsin-Yuan and Caro, Matthias C},
  journal={PRX Quantum},
  volume={5},
  number={4},
  pages={040306},
  year={2024},
  publisher={APS}
}

@article{gu2025magic,
  title={Magic-induced computational separation in entanglement theory},
  author={Gu, Andi and Oliviero, Salvatore FE and Leone, Lorenzo},
  journal={PRX Quantum},
  volume={6},
  number={2},
  pages={020324},
  year={2025},
  publisher={APS}
}

@article{navascues2018bond,
  title={Bond dimension witnesses and the structure of homogeneous matrix product states},
  author={Navascues, Miguel and Vertesi, Tamas},
  journal={Quantum},
  volume={2},
  pages={50},
  year={2018},
  publisher={Verein zur F{\"o}rderung des Open Access Publizierens in den Quantenwissenschaften}
}

@article{cirac2021matrix,
  title={Matrix product states and projected entangled pair states: Concepts, symmetries, theorems},
  author={Cirac, J Ignacio and Perez-Garcia, David and Schuch, Norbert and Verstraete, Frank},
  journal={Reviews of Modern Physics},
  volume={93},
  number={4},
  pages={045003},
  year={2021},
  publisher={APS}
}

@article{fannes1992finitely,
  title={Finitely correlated states on quantum spin chains},
  author={Fannes, Mark and Nachtergaele, Bruno and Werner, Reinhard F},
  journal={Communications in mathematical physics},
  volume={144},
  number={3},
  pages={443--490},
  year={1992},
  publisher={Springer}
}

@article{anshu2023concentration,
  title={Concentration bounds for quantum states and limitations on the QAOA from polynomial approximations},
  author={Anshu, Anurag and Metger, Tony},
  journal={Quantum},
  volume={7},
  pages={999},
  year={2023},
  publisher={Verein zur F{\"o}rderung des Open Access Publizierens in den Quantenwissenschaften}
}

@article{erdHos2014phase,
  title={Phase transition in the density of states of quantum spin glasses},
  author={Erd{\H{o}}s, L{\'a}szl{\'o} and Schr{\"o}der, Dominik},
  journal={Mathematical Physics, Analysis and Geometry},
  volume={17},
  number={3},
  pages={441--464},
  year={2014},
  publisher={Springer}
}

@article{freedman2014quantum,
  title={Quantum systems on non-k-hyperfinite complexes: a generalization of classical statistical mechanics on expander graphs},
  author={Freedman, Michael H and Hastings, Matthew B},
  journal={Quantum Information \& Computation},
  volume={14},
  number={1-2},
  pages={144--180},
  year={2014},
  publisher={Rinton Press, Incorporated Paramus, NJ}
}

@article{gamarnik2023shattering,
  title={Shattering in the Ising Pure $ p $-Spin Model},
  author={Gamarnik, David and Jagannath, Aukosh and K{\i}z{\i}lda{\u{g}}, Eren C},
  journal={arXiv preprint arXiv:2307.07461},
  year={2023}
}

@Book{MezardMontanariBook,
   author =   "M. Mezard and A. Montanari",
   title =   "Information, Physics and Computation",
   publisher =   "Oxford graduate texts",
 year =   "2009"
}

@book{horn2012matrix,
  title={Matrix analysis},
  author={Horn, Roger A and Johnson, Charles R},
  year={2012},
  publisher={Cambridge university press}
}

@article{chen2025free,
  title={On the free energy of vector spin glasses with nonconvex interactions},
  author={Chen, Hong-Bin and Mourrat, Jean-Christophe},
  journal={Probability and Mathematical Physics},
  volume={6},
  number={1},
  pages={1--80},
  year={2025},
  publisher={Mathematical Sciences Publishers}
}

@article{gosset2013algorithm,
  title={An algorithm for the T-count},
  author={Gosset, David and Kliuchnikov, Vadym and Mosca, Michele and Russo, Vincent},
  journal={arXiv preprint arXiv:1308.4134},
  year={2013}
}

@article{bhattacharya2026parisi,
  title={Parisi Formula for the ground state energy of quantum p-Spin Hamiltonians},
  author={Bhattacharya, Sohom},
  journal={arXiv preprint arXiv:2609.20431},
  year={2026}
}

@book{shalev2014understanding,
  title={Understanding machine learning: From theory to algorithms},
  author={Shalev-Shwartz, Shai and Ben-David, Shai},
  year={2014},
  publisher={Cambridge university press}
}

@article{koenig2014efficiently,
  title={How to efficiently select an arbitrary Clifford group element},
  author={Koenig, Robert and Smolin, John A},
  journal={Journal of Mathematical Physics},
  volume={55},
  number={12},
  year={2014},
  publisher={AIP Publishing}
}

@Book{poll,
 author =   "D. Pollard",
   title =   "Convergence of Stochastic Processes",
publisher =   "Springer-Verlag",
 year =   "1984",
}

@Book{TalagrandBook,
   author =   "M. Talagrand",
   title =   "Mean Field Models for Spin Glasses: Volume {I}: Basic Examples ",
   publisher =   "Springer",
   edition = "",
 year =   "2010"
}

@book{vershynin2018high,
  title={High-dimensional probability: An introduction with applications in data science},
  author={Vershynin, Roman},
  volume={47},
  year={2018},
  publisher={Cambridge University Press}
}

@article{hastings2022perturbation,
  title={Perturbation theory and the sum of squares},
  author={Hastings, Matthew B},
  journal={arXiv preprint arXiv:2205.12325},
  year={2022}
}

@article{zlokapa2025average,
  title={Average-case quantum complexity from glassiness},
  author={Zlokapa, Alexander and Kiani, Bobak T and Anschuetz, Eric R},
  journal={arXiv preprint arXiv:2510.08497},
  year={2025}
}

@article{bravyi2006lieb,
  title={Lieb-Robinson bounds and the generation of correlations and topological quantum order},
  author={Bravyi, Sergey and Hastings, Matthew B and Verstraete, Frank},
  journal={Physical review letters},
  volume={97},
  number={5},
  pages={050401},
  year={2006},
  publisher={APS}
}

@inproceedings{eldar2017local,
  title={Local Hamiltonians whose ground states are hard to approximate},
  author={Eldar, Lior and Harrow, Aram W},
  booktitle={2017 IEEE 58th annual symposium on foundations of computer science (FOCS)},
  pages={427--438},
  year={2017},
  organization={IEEE}
}

@article{wei2026longrangenonstabilizernessquantumcodes,
  title={Long-range nonstabilizerness and quantum codes, phases, and complexity},
  author={Wei, Fuchuan and Liu, Zi-Wen},
  journal={arXiv preprint arXiv:2503.04566},
  year={2025}
}

@article{chen2024sparse,
  title={Sparse random Hamiltonians are quantumly easy},
  author={Chen, Chi-Fang and Dalzell, Alexander M and Berta, Mario and Brand{\~a}o, Fernando GSL and Tropp, Joel A},
  journal={Physical Review X},
  volume={14},
  number={1},
  pages={011014},
  year={2024},
  publisher={APS}
}

@article{schon2005sequential,
  title={Sequential generation of entangled multiqubit states},
  author={Sch{\"o}n, Christian and Solano, Enrique and Verstraete, Frank and Cirac, J Ignacio and Wolf, Michael M},
  journal={Physical review letters},
  volume={95},
  number={11},
  pages={110503},
  year={2005},
  publisher={APS}
}

@article{perez2006matrix,
  title={Matrix product state representations},
  author={Perez-Garcia, David and Verstraete, Frank and Wolf, Michael M and Cirac, J Ignacio},
  journal={arXiv preprint quant-ph/0608197},
  year={2006}
}

@article{vidal2008class,
  title={Class of quantum many-body states that can be efficiently simulated},
  author={Vidal, Guifr{\'e}},
  journal={Physical review letters},
  volume={101},
  number={11},
  pages={110501},
  year={2008},
  publisher={APS}
}

@article{baldwin2020quenched,
  title={Quenched vs Annealed: Glassiness from SK to SYK},
  author={Baldwin, CL and Swingle, B},
  journal={Physical Review X},
  volume={10},
  number={3},
  pages={031026},
  year={2020},
  publisher={APS}
}

@article{gottesman1998heisenberg,
  title={The Heisenberg representation of quantum computers},
  author={Gottesman, Daniel},
  journal={arXiv preprint quant-ph/9807006},
  year={1998}
}

@inproceedings{chen2025incompressibility,
  title={Incompressibility and spectral gaps of random circuits},
  author={Chen, Chi-Fang and Haah, Jeongwan and Haferkamp, Jonas and Liu, Yunchao and Metger, Tony and Tan, Xinyu},
  booktitle={2025 IEEE 66th Annual Symposium on Foundations of Computer Science (FOCS)},
  pages={1304--1312},
  year={2025},
  organization={IEEE}
}

@article{Brown2018second,
   title={Second law of quantum complexity},
   volume={97},
   ISSN={2470-0029},
   url={http://dx.doi.org/10.1103/PhysRevD.97.086015},
   DOI={10.1103/physrevd.97.086015},
   number={8},
   journal={Physical Review D},
   publisher={American Physical Society (APS)},
   author={Brown, Adam R. and Susskind, Leonard},
   year={2018},
   month=Apr }

@article{gamarnik2025shattering,
  author  = {Gamarnik, David and Jagannath, Aukosh and K{\i}z{\i}lda{\u g}, Eren C.},
  title   = {Shattering in the {Ising} $p$-spin glass model},
  journal = {Probability Theory and Related Fields},
  volume  = {193},
  pages   = {89--141},
  year    = {2025},
  doi     = {10.1007/s00440-025-01414-4},
  eprint  = {2307.07461},
  archivePrefix = {arXiv}
}

@article{bravyi2016improved,
  title={Improved classical simulation of quantum circuits dominated by Clifford gates},
  author={Bravyi, Sergey and Gosset, David},
  journal={Physical review letters},
  volume={116},
  number={25},
  pages={250501},
  year={2016},
  publisher={APS}
}

@article{leone2024learning,
  title={Learning t-doped stabilizer states},
  author={Leone, Lorenzo and Oliviero, Salvatore FE and Hamma, Alioscia},
  journal={Quantum},
  volume={8},
  pages={1361},
  year={2024},
  publisher={Verein zur F{\"o}rderung des Open Access Publizierens in den Quantenwissenschaften}
}

@article{coble2023local,
  title={Local hamiltonians with no low-energy stabilizer states},
  author={Coble, Nolan J and Coudron, Matthew and Nelson, Jon and Nezhadi, Seyed Sajjad},
  journal={arXiv preprint arXiv:2302.14755},
  year={2023}
}

@article{coble2023hamiltonians,
  title   = {Hamiltonians whose low-energy states require {$\Omega(n)$} {T} gates},
  author  = {Coble, Nolan J. and Coudron, Matthew and Nelson, Jon and Nezhadi, Seyed Sajjad},
  journal = {arXiv preprint arXiv:2310.01347},
  year    = {2023}
}

@article{bray1980replica,
  title={Replica theory of quantum spin glasses},
  author={Bray, AJ and Moore, MA},
  journal={Journal of Physics C: Solid State Physics},
  volume={13},
  number={24},
  pages={L655--L660},
  year={1980}
}

@article{anshu2021improved,
  title={Improved approximation algorithms for bounded-degree local hamiltonians},
  author={Anshu, Anurag and Gosset, David and Morenz Korol, Karen J and Soleimanifar, Mehdi},
  journal={Physical Review Letters},
  volume={127},
  number={25},
  pages={250502},
  year={2021},
  publisher={APS}
}

@article{fattal2004entanglement,
  title={Entanglement in the stabilizer formalism},
  author={Fattal, David and Cubitt, Toby S and Yamamoto, Yoshihisa and Bravyi, Sergey and Chuang, Isaac L},
  journal={arXiv preprint quant-ph/0406168},
  year={2004}
}

@article{rocchetto2017stabiliser,
  title={Stabiliser states are efficiently PAC-learnable},
  author={Rocchetto, Andrea},
  journal={arXiv preprint arXiv:1705.00345},
  year={2017}
}

@article{grewal2025efficient,
  title={Efficient learning of quantum states prepared with few non-Clifford gates},
  author={Grewal, Sabee and Iyer, Vishnu and Kretschmer, William and Liang, Daniel},
  journal={Quantum},
  volume={9},
  pages={1907},
  year={2025},
  publisher={Verein zur F{\"o}rderung des Open Access Publizierens in den Quantenwissenschaften}
}

@article{bauschke2012firmly,
  title={Firmly nonexpansive mappings and maximally monotone operators: correspondence and duality},
  author={Bauschke, Heinz H and Moffat, Sarah M and Wang, Xianfu},
  journal={Set-Valued and Variational Analysis},
  volume={20},
  number={1},
  pages={131--153},
  year={2012},
  publisher={Springer}
}

@article{arunachalam2022parameterized,
  title={The parameterized complexity of quantum verification},
  author={Arunachalam, Srinivasan and Bravyi, Sergey and Nirkhe, Chinmay and O'Gorman, Bryan},
  journal={arXiv preprint arXiv:2202.08119},
  year={2022}
}

@article{gosset2026quantum,
  title={Quantum state preparation with optimal T-count},
  author={Gosset, David and Kothari, Robin and Wu, Kewen},
  journal={Quantum},
  volume={10},
  pages={2168},
  year={2026},
  publisher={Verein zur F{\"o}rderung des Open Access Publizierens in den Quantenwissenschaften}
}

@article{schuch2010matrix,
  title={Matrix product state and mean-field solutions for one-dimensional systems can be found efficiently},
  author={Schuch, Norbert and Cirac, J Ignacio},
  journal={Physical Review A—Atomic, Molecular, and Optical Physics},
  volume={82},
  number={1},
  pages={012314},
  year={2010},
  publisher={APS}
}

@article{aharonov2010efficient,
  title={Efficient algorithm for approximating one-dimensional ground states},
  author={Aharonov, Dorit and Arad, Itai and Irani, Sandy},
  journal={Physical Review A—Atomic, Molecular, and Optical Physics},
  volume={82},
  number={1},
  pages={012315},
  year={2010},
  publisher={APS}
}

@article{knill1995approximation,
  title={Approximation by quantum circuits},
  author={Knill, Emanuel},
  journal={arXiv preprint quant-ph/9508006},
  year={1995}
}

@article{fomichev2024initial,
  title={Initial state preparation for quantum chemistry on quantum computers},
  author={Fomichev, Stepan and Hejazi, Kasra and Zini, Modjtaba Shokrian and Kiser, Matthew and Fraxanet, Joana and Casares, Pablo Antonio Moreno and Delgado, Alain and Huh, Joonsuk and Voigt, Arne-Christian and Mueller, Jonathan E and others},
  journal={PRX Quantum},
  volume={5},
  number={4},
  pages={040339},
  year={2024},
  publisher={APS}
}

@unpublished{guruswami2025essential,
  author = {Guruswami, Venkatesan and Rudra, Atri and Sudan, Madhu},
  title  = {Essential Coding Theory},
  note   = {Book draft, August 26, 2025},
  year   = {2025},
  url    = {https://cse.buffalo.edu/faculty/atri/courses/coding-theory/book/web-coding-book.pdf}
}

@inproceedings{brandao2013product,
  title={Product-state approximations to quantum ground states},
  author={Brandao, Fernando GSL and Harrow, Aram W},
  booktitle={Proceedings of the forty-fifth annual ACM symposium on Theory of computing},
  pages={871--880},
  year={2013}
}

\end{document}